\documentclass{elsarticle}
\usepackage{latexsym,xspace,calc,amsthm}
\usepackage{amsmath,amssymb} %comment out for llncs ,amsthm}

\makeatletter
\renewcommand\section{\@startsection {section}{1}{\z@}%
           {18\p@ \@plus 6\p@ \@minus 3\p@}%
           {9\p@ \@plus 6\p@ \@minus 3\p@}%
           {\normalsize\bfseries}}
\makeatother

\usepackage{enumerate,mdwlist} % for itemize* environment
\makecompactlist{enumeratei*}{enumeratei}
\usepackage{makeidx}

\usepackage{microtype}
\usepackage{tgtermes}
\makeatletter
\@ifpackageloaded{fdsymbol}\@tempswafalse\@tempswatrue
\if@tempswa
  \newcommand{\fdsy@scale}{1.0}
  \newcommand\fdsy@mweight@normal{Book}
  \newcommand\fdsy@mweight@small{Book}
  \newcommand\fdsy@bweight@normal{Medium}
  \newcommand\fdsy@bweight@small{Medium}
  \DeclareFontFamily{U}{FdSymbolA}{}
  \DeclareSymbolFont{fdsymbols}{U}{FdSymbolA}{m}{n}%
  \SetSymbolFont{symbols}{bold}{U}{FdSymbolA}{b}{n}%
  \DeclareFontShape{U}{FdSymbolA}{m}{n}{
      <-7.1> s * [\fdsy@scale] FdSymbolA-\fdsy@mweight@small
      <7.1-> s * [\fdsy@scale] FdSymbolA-\fdsy@mweight@normal
  }{}
  \DeclareFontShape{U}{FdSymbolA}{b}{n}{
      <-7.1> s * [\fdsy@scale] FdSymbolA-\fdsy@bweight@small
      <7.1-> s * [\fdsy@scale] FdSymbolA-\fdsy@bweight@normal
  }{}
  \DeclareMathSymbol{\aleph}{\mathord}{fdsymbols}{"C7}
  \DeclareMathSymbol{\beth}{\mathord}{fdsymbols}{"C8}
  \DeclareMathSymbol{\gimel}{\mathord}{fdsymbols}{"C9}
  \DeclareMathSymbol{\daleth}{\mathord}{fdsymbols}{"CA}
\fi
\makeatother

\usepackage{xparse}
\usepackage{tikz}
\usetikzlibrary{calc}
\DeclareDocumentCommand{\hcancel}{mO{0pt}O{1pt}O{0pt}O{-1pt}}{%
    \tikz[baseline=(tocancel.base)]{
        \node[inner sep=0pt,outer sep=0pt] (tocancel) {#1};
        \draw[gray] ($(tocancel.south west)+(#2,#3)$) -- ($(tocancel.north east)+(#4,#5)$);
    }%
}%

\usepackage{float}
\floatstyle{boxed} \restylefloat{figure}

\usepackage[all]{xy}

\let\myfresh\#
\def\#{\ensuremath{\text{\tt\myfresh}}}

\newcommand\nomathcal{\mathcal}
\newcommand\finite{small\xspace}

\newcommand\ttop{{\pmb\top}}
\newcommand\tbot{{\pmb\bot}}

\newcommand\tand{{\pmb\wedge}}
\newcommand\tor{{\pmb\vee}}

\newcommand\tlam{{\pmb\lambda}}
\newcommand\tall{{\pmb\forall}}
\newcommand\texi{{\pmb\exists}}

\renewcommand\land{\wedge}
\renewcommand\lor{\vee}
\newcommand\limp{\Rightarrow}
\newcommand\nslprg{\ns L^{\hspace{-1pt}\prg}}
\newcommand\nompow{{\f{Nom}}}
\newcommand\powerset{{\f{Pow}}}
\newcommand\strict{{\f{Strct}}}

\usepackage{yfonts}
\newcommand\idiom{\mho} %\textrm{\textfrak{I}}}
\newcommand\idiomprg{\idiom}

\newcommand\uparrowp[1]{{\uparrow}_{\hspace{-.15em}\scalebox{.6}{$#1$}}}

\newcommand\arrowp[1]{\mathop{{\rightarrow}_{\hspace{-.05em}\scalebox{.6}{$#1$}}}}
\newcommand\eqarrowp[1]{\mathop{{=}_{\hspace{-.05em}\scalebox{.6}{$#1$}}}}

\newcommand\Func{{\Rightarrow}}
\newcommand\equivarto{\to}
\newcommand{\fa}{\f{fa}}
\newcommand\den[1]{{\hspace{.00ex}\scalebox{.55}{$#1$}}}
\newcommand{\ddenot}[1]{\denot{\ns D}{}{#1}}
\newcommand{\pdenot}[1]{\denot{G(\points_\Pi)}{}{#1}}

\usepackage[mathscr]{eucal}

\newcommand{\denot}[3]{\llbracket #3 \rrbracket_{\scalebox{.6}{$#2$}}^\den{#1}} %^{\hspace{-.1ex}\scalebox{.4}{$#1$}}}

 \renewenvironment{thebibliography}[1]{%
   \begin{odlthebibliography}{#1}%
     \setlength{\parskip}{0ex}%
     \setlength{\itemsep}{3pt}%
     \fontsize{9.5}{9.5} %Change these numbers to
     \selectfont
}%
 {%
   \end{odlthebibliography}%
 }

\usepackage{hyperref}
\hypersetup{urlcolor=black,colorlinks}

\usepackage{fancybox}
\newlength{\mylength}
\newenvironment{frameqn}%
{\setlength{\fboxsep}{5pt}
\setlength{\mylength}{\linewidth}%
\addtolength{\mylength}{-2\fboxsep}%
\addtolength{\mylength}{-2\fboxrule}%
\Sbox
\minipage{\mylength}%
\setlength{\abovedisplayskip}{0pt}%
\setlength{\belowdisplayskip}{0pt}%
$$}%
{$$\endminipage\endSbox
{\setlength{\abovedisplayskip}{1pt}%
\setlength{\belowdisplayskip}{0pt}%
\[\fbox{\TheSbox}\]}}

\newenvironment{frametxt}%
{\setlength{\fboxsep}{5pt}
\setlength{\mylength}{\linewidth}%
\addtolength{\mylength}{-2\fboxsep}%
\addtolength{\mylength}{-2\fboxrule}%
\Sbox
\minipage{\mylength}%
\setlength{\abovedisplayskip}{5pt}%
\setlength{\belowdisplayskip}{5pt}%
}%
{\endminipage\endSbox
{\setlength{\abovedisplayskip}{1pt}%
\setlength{\belowdisplayskip}{0pt}%
\[\fbox{\TheSbox}\]}}

\usepackage{graphics}% needed for `new' quantifier
\usepackage{ifthen}
\usepackage{verbatim}
\newdimen\proofrulebreadth \proofrulebreadth=.05em
\newdimen\proofdotseparation \proofdotseparation=1.25ex
\newdimen\proofrulebaseline \proofrulebaseline=2ex
\newcount\proofdotnumber \proofdotnumber=3
\let\then\relax
\def\hfi{\hskip0pt plus.0001fil}
\mathchardef\squigto="3A3B
\newif\ifinsideprooftree\insideprooftreefalse
\newif\ifonleftofproofrule\onleftofproofrulefalse
\newif\ifproofdots\proofdotsfalse
\newif\ifdoubleproof\doubleprooffalse
\let\wereinproofbit\relax
\newdimen\shortenproofleft
\newdimen\shortenproofright
\newdimen\proofbelowshift
\newbox\proofabove
\newbox\proofbelow
\newbox\proofrulename
\def\shiftproofbelow{\let\next\relax\afterassignment\setshiftproofbelow\dimen0 }
\def\shiftproofbelowneg{\def\next{\multiply\dimen0 by-1 }%
\afterassignment\setshiftproofbelow\dimen0 }
\def\setshiftproofbelow{\next\proofbelowshift=\dimen0 }
\def\setproofrulebreadth{\proofrulebreadth}

\def\prooftree{% NESTED ZERO (\ifonleftofproofrule)
\ifnum  \lastpenalty=1
\then   \unpenalty
\else   \onleftofproofrulefalse
\fi
\ifonleftofproofrule
\else   \ifinsideprooftree
        \then   \hskip.5em plus1fil
        \fi
\fi
\bgroup% NESTED ONE (\proofbelow, \proofrulename, \proofabove,
\setbox\proofbelow=\hbox{}\setbox\proofrulename=\hbox{}%
\let\justifies\proofover\let\leadsto\proofoverdots\let\Justifies\proofoverdbl
\let\using\proofusing\let\[\prooftree
\ifinsideprooftree\let\]\endprooftree\fi
\proofdotsfalse\doubleprooffalse
\let\thickness\setproofrulebreadth
\let\shiftright\shiftproofbelow \let\shift\shiftproofbelow
\let\shiftleft\shiftproofbelowneg
\let\ifwasinsideprooftree\ifinsideprooftree
\insideprooftreetrue
\setbox\proofabove=\hbox\bgroup$\displaystyle % NESTED TWO
\let\wereinproofbit\prooftree
\shortenproofleft=0pt \shortenproofright=0pt \proofbelowshift=0pt
\onleftofproofruletrue\penalty1
}

\def\eproofbit{% NESTED TWO
\ifx    \wereinproofbit\prooftree
\then   \ifcase \lastpenalty
        \then   \shortenproofright=0pt  % 0: some other object, no indentation
        \or     \unpenalty\hfil         % 1: empty hypotheses, just glue
        \or     \unpenalty\unskip       % 2: just had a tree, remove glue
        \else   \shortenproofright=0pt  % eh?
        \fi
\fi
\global\dimen0=\shortenproofleft
\global\dimen1=\shortenproofright
\global\dimen2=\proofrulebreadth
\global\dimen3=\proofbelowshift
\global\dimen4=\proofdotseparation
\global\count255=\proofdotnumber
$\egroup  % NESTED ONE
\shortenproofleft=\dimen0
\shortenproofright=\dimen1
\proofrulebreadth=\dimen2
\proofbelowshift=\dimen3
\proofdotseparation=\dimen4
\proofdotnumber=\count255
}

\def\proofover{% NESTED TWO
\eproofbit % NESTED ONE
\setbox\proofbelow=\hbox\bgroup % NESTED TWO
\let\wereinproofbit\proofover
$\displaystyle
}%
\def\proofoverdbl{% NESTED TWO
\eproofbit % NESTED ONE
\doubleprooftrue
\setbox\proofbelow=\hbox\bgroup % NESTED TWO
\let\wereinproofbit\proofoverdbl
$\displaystyle
}%
\def\proofoverdots{% NESTED TWO
\eproofbit % NESTED ONE
\proofdotstrue
\setbox\proofbelow=\hbox\bgroup % NESTED TWO
\let\wereinproofbit\proofoverdots
$\displaystyle
}%
\def\proofusing{% NESTED TWO
\eproofbit % NESTED ONE
\setbox\proofrulename=\hbox\bgroup % NESTED TWO
\let\wereinproofbit\proofusing
\kern0.3em$
}

\def\endprooftree{% NESTED TWO
\eproofbit % NESTED ONE
  \dimen5 =0pt% spread of hypotheses
\dimen0=\wd\proofabove \advance\dimen0-\shortenproofleft
\advance\dimen0-\shortenproofright
\dimen1=.5\dimen0 \advance\dimen1-.5\wd\proofbelow
\dimen4=\dimen1
\advance\dimen1\proofbelowshift \advance\dimen4-\proofbelowshift
\ifdim  \dimen1<0pt
\then   \advance\shortenproofleft\dimen1
        \advance\dimen0-\dimen1
        \dimen1=0pt
        \ifdim  \shortenproofleft<0pt
        \then   \setbox\proofabove=\hbox{%
                        \kern-\shortenproofleft\unhbox\proofabove}%
                \shortenproofleft=0pt
        \fi
\fi
\ifdim  \dimen4<0pt
\then   \advance\shortenproofright\dimen4
        \advance\dimen0-\dimen4
        \dimen4=0pt
\fi
\ifdim  \shortenproofright<\wd\proofrulename
\then   \shortenproofright=\wd\proofrulename
\fi
\dimen2=\shortenproofleft \advance\dimen2 by\dimen1
\dimen3=\shortenproofright\advance\dimen3 by\dimen4
\ifproofdots
\then
        \dimen6=\shortenproofleft \advance\dimen6 .5\dimen0
        \setbox1=\vbox to\proofdotseparation{\vss\hbox{$\cdot$}\vss}%
        \setbox0=\hbox{%
                \advance\dimen6-.5\wd1
                \kern\dimen6
                $\vcenter to\proofdotnumber\proofdotseparation
                        {\leaders\box1\vfill}$%
                \unhbox\proofrulename}%
\else   \dimen6=\fontdimen22\the\textfont2 % height of maths axis
        \dimen7=\dimen6
        \advance\dimen6by.5\proofrulebreadth
        \advance\dimen7by-.5\proofrulebreadth
        \setbox0=\hbox{%
                \kern\shortenproofleft
                \ifdoubleproof
                \then   \hbox to\dimen0{%
                        $\mathsurround0pt\mathord=\mkern-6mu%
                        \cleaders\hbox{$\mkern-2mu=\mkern-2mu$}\hfill
                        \mkern-6mu\mathord=$}%
                \else   \vrule height\dimen6 depth-\dimen7 width\dimen0
                \fi
                \unhbox\proofrulename}%
        \ht0=\dimen6 \dp0=-\dimen7
\fi
\let\doll\relax
\ifwasinsideprooftree
\then   \let\VBOX\vbox
\else   \ifmmode\else$\let\doll=$\fi
        \let\VBOX\vcenter
\fi
\VBOX   {\baselineskip\proofrulebaseline \lineskip.2ex
        \expandafter\lineskiplimit\ifproofdots0ex\else-0.6ex\fi
        \hbox   spread\dimen5   {\hfi\unhbox\proofabove\hfi}%
        \hbox{\box0}%
        \hbox   {\kern\dimen2 \box\proofbelow}}\doll%
\global\dimen2=\dimen2
\global\dimen3=\dimen3
\egroup % NESTED ZERO
\ifonleftofproofrule
\then   \shortenproofleft=\dimen2
\fi
\shortenproofright=\dimen3
\onleftofproofrulefalse
\ifinsideprooftree
\then   \hskip.5em plus 1fil \penalty2
\fi
}

\usepackage{array}
\newcolumntype{L}[1]{>{$}p{#1}<{$}}
\newcolumntype{C}[1]{>{\centering$}p{#1}<{$}}
\newcolumntype{R}[1]{>{\raggedleft$}p{#1}<{$}}
\newcommand\maketab[2]
   {\newenvironment{#1}{\begin{quotation}\noindent\begin{tabular}{#2}}{\end{tabular}\end{quotation} }
    \newenvironment{#1noquote}{\noindent\begin{tabular}{#2}}{\end{tabular}}
    }

\usepackage[mathscr]{eucal}
\newcommand\ns{\mathscr}

\newcommand{\freshcup}[1]{\mbox{$\bigcup^{\hspace{-.1ex}\raisebox{-.2ex}{\scalebox{.6}{$\# #1$}}}$}}

\newcommand{\freshcap}[1]{\mbox{$\bigcap^{\hspace{-.1ex}\raisebox{-.2ex}{\scalebox{.6}{$\# #1$}}}$}}
\newcommand{\freshwedge}[1]{\tall #1.}
\newcommand{\freshwedges}[1]{\mbox{$\bigwedge^{\hspace{-.5ex}\raisebox{-.2ex}{\scalebox{.6}{${\subseteq}#1$}}}$}}

\newcommand{\otop}[1]{{\f{Opn}(#1)}}
\newcommand{\ctop}[1]{{\f{Cpct}(#1)}}

\newcommand{\powsigma}{\f{Pow}_{\hspace{-1pt}\sigma}}
\newcommand{\powamgis}{\f{Pow}_{\hspace{-1pt}\scalebox{.74}{$\amgis$}}}

\newcommand{\india}{\theory{inDi\tall}}
\newcommand{\ndia}{\theory{nDi\tall}}
\newcommand{\indiapp}{\theory{inDi\hspace{.1ex}\tall}_{\hspace{-.3ex}\bullet}}

\newcommand{\inspecta}{\theory{inSpect\tall}}
\newcommand{\inspectapp}{\theory{inSpect\tall}_{\hspace{-.3ex}\bullet}}
\newcommand{\lamtrm}{\theory{Lm\hspace{-1.5pt}T\hspace{-1pt}m}}

\newcommand\topel{t}

\def\:{{\hspace{-1pt}{:}\hspace{-1.25pt}{:}\hspace{-.5pt}}}
\newcommand\theory[1]{\ensuremath{\mathsf{#1}}}
\newcommand\points{{\f{P}\hspace{-1.75pt}\f{nt}}}
\usepackage{relsize}
\newcommand\prg{{\smaller\partial}}

\newcommand\pp[1]{#1^{\scalebox{.5}{$\bullet$}}}

\newcommand\ppmone[1]{#1^{\scalebox{.6}{$\bullet\mone$}}}
\newcommand\qq[1]{#1^{\scalebox{.7}{$\ast$}}}

\def\id{\mathtxt{id}}

\newcommand\minus{{\text{-}}}

\newcommand\ssm{{{:}\text{=}}}

\newcommand{\jamiepart}[1]{\newpage\part{#1}}
\newcommand\deffont[1]{{\bf #1}}

\newcommand\mone{{{\text{-}1}}}
\newcommand\liff{\Leftrightarrow}

\newcommand\size{\mathit{size}}
\newcommand\supp{\f{supp}}
\newcommand\nontriv{\f{nontriv}}

\newcommand\f[1]{\mathit{#1}}
\newcommand\tf[1]{\mathsf{#1}}
\newcommand\cent{\vdash}

\newcommand{\curvedarrowtop}{\hspace{-.17ex}\raisebox{1.85ex}{\scalebox{1.15}{\rotatebox{270}{$\curvearrowleft$}}}}

\newcommand\ii[1]{{\hspace{.5pt}\raisebox{.4pt}{$\curvedarrowtop$}^{\hspace{-1pt}#1}}} %_{{\scalebox{.6}{$#1$}}}

\def\atoms{\ensuremath{\mathbb{A}}\xspace}

\newcommand\somerel{\mathrel{\mathcal R}}

\newcommand\dact[1]{}

\newcommand\clostordown[1]{#1_{{\tor}\hspace{-1.5pt}{\downarrow}}}

\newcommand\clostandup[1]{#1_{{\tand}\hspace{-1.5pt}{\uparrow}}}

\newcommand\vect[1]{\overline{#1}}

\newcommand\act[0]{{\cdot}}

\newcommand{\Defiff}% definitional bi-implication
 {\mathrel{\ \ \stackrel{\scriptstyle \mathrm{def}}{\Leftrightarrow}\ \ }}

\newcommand{\defeq}% definitional equality
  {\stackrel{\mathrm{def}}{\,=\,}}

\newcommand\fix{\f{fix}}

\newcommand\fresh{\,\sharp_{\hspace{-.5pt}\sigma}\,}

\newcommand\ment[0]{\ \vDash\ }

\newcounter{jamieitemcounter}

\newtheoremstyle{jamiestyle}% name of the style to be used
  {4pt}% measure of space to leave above the theorem. E.g.: 3pt
  {0pt}% measure of space to leave below the theorem. E.g.: 3pt
  {\it}% name of font to use in the body of the theorem
  {0pt}% measure of space to indent
  {\sc}% name of head font
  {.}% punctuation between head and body
  { }% space after theorem head; " " = normal interword space
  {}% Manually specify head
\theoremstyle{jamiestyle}
\newtheorem{thrm}{Theorem}[subsection]
\newtheorem{prop}[thrm]{Proposition}
\newtheorem{lemm}[thrm]{Lemma}
\newtheorem{corr}[thrm]{Corollary}
\newtheoremstyle{jamienfstyle}% name of the style to be used
  {4pt}% measure of space to leave above the theorem. E.g.: 3pt
  {0pt}% measure of space to leave below the theorem. E.g.: 3pt
  {\normalfont}% name of font to use in the body of the theorem
  {0pt}% measure of space to indent
  {\sc}% name of head font
  {.}% punctuation between head and body
  { }% space after theorem head; " " = normal interword space
  {}% Manually specify head
\theoremstyle{jamienfstyle}
\newtheorem{nttn}[thrm]{Notation}
\newtheorem{defn}[thrm]{Definition}
\newtheorem{xmpl}[thrm]{Example}
\newtheorem{rmrk}[thrm]{Remark}

\newcommand\sm{{\mapsto}}
\newcommand\lsm{{\sm}}
\usepackage{stmaryrd}\newcommand\ms{{\mapsfrom}}

\newcommand\mathtxt[1]{ \ensuremath{\mathrm{#1}} }

\newcommand\rulefont[1]{\ensuremath{\smaller{\mathrm{\bf (#1)}}}}

\newcommand\Forall[1]{\forall #1.}
\newcommand\Exists[1]{\exists #1.}

\newcommand\amgis{\reflectbox{$\sigma$}}

\newcommand\new{\reflectbox{$\mathsf{N}$}}

\newcommand\nw{\scalebox{.7}{$\new$}}
\newcommand\New[1]{\new #1.}
\newcommand\lam[1]{\lambda #1.}
\newcommand\bpp{{\circ}}
\newcommand\ppa{{{\multimap}\hspace{-5pt}{\bullet}}}
\newcommand\app{{\bullet}}

\newcommand\blfootnote[1]{%
  \begingroup
  \renewcommand\thefootnote{}\footnote{#1}%
  \addtocounter{footnote}{-1}%
  \endgroup
}

\begin{document}

\begin{frontmatter}

\title{Representation and duality of the untyped $\lambda$-calculus in nominal lattice and topological semantics, with a proof of topological completeness}

\author[murdochaddress]{Murdoch J. Gabbay}
\address[murdochaddress]{Heriot-Watt University, Scotland, UK}
\ead[url]{www.gabbay.org.uk}

\author[michaeladdress]{Michael J. Gabbay}
\address[michaeladdress]{University of Cambridge, UK}

\begin{abstract}
We give a semantics for the $\lambda$-calculus based on a topological duality theorem in nominal sets.
A novel interpretation of $\lambda$ is given in terms of adjoints, and $\lambda$-terms are interpreted absolutely as sets (no valuation is necessary).
\end{abstract}

\begin{keyword}
Nominal algebras, fresh-finite limits, lambda-calculus, spectral spaces, lattices and order, variables, nominal techniques, mathematical foundations, Fraenkel-Mostowski set theory
%\texttt{elsarticle.cls}\sep \LaTeX\sep Elsevier \sep template
%\MSC[2010] 00-01\sep  99-00
\end{keyword}

\end{frontmatter}

%\makeindex

\tableofcontents
%%%%%%%%%%%%%%%%%%%%%%%%%%%%%%%%%%%
\section{Introduction}

In this paper we build a topological duality result for the untyped $\lambda$-calculus in nominal sets and prove soundness, completeness, and topological completeness.\blfootnote{We are extremely grateful to the editor Phil Scott and to the anonymous referees for their time and help in improving the mathematics and readability of what follows.}

This means the following:
\begin{itemize*}
\item
We give a lattice-style axiomatisation of the untyped $\lambda$-calculus and prove it sound and complete.
\item
We define a notion of topological space whose compact open sets have notions of application and $\lambda$-abstraction.
\item
We prove that the categories of lattices-with-$\lambda$ and topological-spaces-with-$\lambda$ are dual.
\item
We give a complete topological semantics for the $\lambda$-calculus.
\end{itemize*}
So this paper does what universal algebra and Stone duality do for Boolean algebras \cite{johnstone:stos}, but for the $\lambda$-calculus.
We do this by combining three threads of previous work:
\begin{itemize*}
\item
the \emph{nominal algebra} of \cite{gabbay:noma-jv};
\item
the study of the logical structure of \emph{nominal powersets} from \cite{gabbay:semooc} (applied in \cite{gabbay:semooc} to first-order logic with equality, and in this paper to the $\lambda$-calculus); and
\item
the \emph{Kripke-style models of the untyped $\lambda$-calculus} from \cite{gabbay:simcmt,mgabbay:simmtu}.
\end{itemize*}

%%%%%%%%%%%%%%%%%%%%%%%%%%%%%%%%
\subsection{A very brief summary of the contributions of this paper}

We summarise some contributions of this paper; this list will be fleshed out in the rest of the Introduction:
\begin{enumerate*}
\item
No previous duality result for $\lambda$-calculus theories exists.

Duality results are interesting in themselves (see next subsection), and it is interesting to see how nominal techniques help to manage the technical demands of such a result.
\item
The lattice semantics of this paper is abstract, yet complete.
The well-known domains-based denotational semantics are \emph{incomplete} (whereas semantics such as term models quotiented by equivalence are complete, but concrete).
\item
We prove a topological completeness result---but this should be impossible: topological \emph{in}completeness results exist in the literature.

This depends on the topology, so the fact that a notion of topology that `works' for the $\lambda$-calculus exists, is surprising given the current state of the art.
One would not expect this to work.
\item
The topological representations are concrete, being based on nominal sets.
\item
The representation of open terms does not use valuations; possibly open $\lambda$-terms are interpreted as open sets in a nominal topological space (this is sometimes called an \emph{absolute} semantics).

Function application and also $\lambda$-abstraction get interpreted as concrete sets operations on nominal-style atoms.
\item
$\beta$-reduction and $\eta$-expansion are exhibited as adjoint properties.
\item
This paper is a nontrivial application of nominal ideas, and the techniques on nominal sets which we use are original and have independent technical interest.
\item
We make nominal-style atoms---urelemente in Fraenkel-Mostowski set theory---behave like variables of the $\lambda$-calculus.
Urelemente in set theory come equipped with very few properties;
indeed, by design urelemente have virtually no properties at all.
It is remarkable that they can nevertheless acquire such rich structure.\footnote{We carry out a similar programme for variables of first-order logic in \cite{gabbay:stodfo,gabbay:semooc}.}
\item
The fine structure of the canonical models (those of the form $\points_\Pi$) is very rich, as we shall see.
The use of canonical models in this paper probably does not exhaust their interest.
\end{enumerate*}

%%%%%%%%%%%%%%%%%%%%%%%%%%%%%%%%
\subsection{The point of duality results}

What is the point of a duality result, and why bother doing it for the $\lambda$-calculus?

\begin{enumerate}
\item
Duality results are a strong form of completeness: for a given class of abstract models, every model has a concrete topological representation (i.e. in terms of sets with a few consistency conditions) and every map between models has a concrete representation as a continuous map (i.e. a map that has to respect those conditions).

Intuitively, Boolean algebras, Heyting algebras, and distributive lattices look like they all have to do with powersets (negation is some kind of sets complement, conjunction is sets intersection, disjunction is sets union, and so on).
But is this true?
Is it possible to construct some sufficiently bizarre model such that for instance conjunction \emph{must} mean something other than sets intersection?

The answer is no: duality theorems tell us that no matter how bizarre the model, it \emph{can} be represented topologically.
In topologies conjunction is sets intersection.

The same analysis is applicable to our duality result for the $\lambda$-calculus.
\item
Of course, it is not obvious how conjunction enters into the untyped $\lambda$-calculus.
Indeed, no duality result has been achieved for the $\lambda$-calculus before and it is not obvious even how to begin to go about this.

One contribution of this paper is that we embed the $\lambda$-calculus in an impredicative logic which we characterise in two ways: in nominal algebra, and using a nominal generalisation of finite limits which we call \emph{fresh-finite limits} (Definition~\ref{defn.fresh.finite.limit}).

An example observation that comes out of this is that we exhibit $\beta$-reduction and $\eta$-expansion as adjoint maps (counit and unit respectively; see Proposition~\ref{prop.lambda.beta.eta}).

Another example is that $\lambda$ is exhibited as a compound object made out of the fresh-finite limit $\forall$, and a right adjoint to application $\ppa$.
See Notation~\ref{nttn.lambda} and the subsequent discussion.

The quantifier $\forall$ is itself an interesting entity, a kind of arbitrary conjunction, which relies heavily on nominal techniques.
More on this later.

Thus, in the process of defining and proving our results---representation, duality, and completeness---we uncover a wealth of structure in the untyped $\lambda$-calculus (to add to the wealth of structure already known).
The technical definitions and lemmas which our `main results' depend on, are as interesting as the results themselves.
\item
Finally, we note that our topological semantics for the $\lambda$-calculus is \emph{complete} (Theorem~\ref{thrm.Pi.completeness}).

This is remarkable because Salibra has shown that all known semantics for the $\lambda$-calculus based on partial orders, are incomplete \cite{salibra:topioi}.
The fact that our semantics is topological (thus ordered) and complete, is unexpected.\footnote{Our paper handles only the case where we have $\eta$-expansion.
We believe this could be generalised to the fully non-extensional case, at some cost in complexity.  See Subsection~\ref{subsect.extensional} or \cite{gabbay:simcmt}.}
We discuss this apparent paradox in Subsection~\ref{subsect.no.conflict}.

In any case, new semantics for the untyped $\lambda$-calculus do not come along very often, and as mentioned above, no duality result for the $\lambda$-calculus has been proved before.
\end{enumerate}

More interesting structure will be uncovered by this way of approaching the $\lambda$-calculus;
the list above justifies why it is \emph{a priori} interesting to try.

\subsection{Map of the paper}

Section~\ref{sect.basic.defs} sets up some basic nominal theory.
The reader might like to skim this at first, since the definitions might only make sense in terms of their application later on in the paper.
Highlights are the notions of nominal set (Definition~\ref{defn.nominal.set}), small-supported and strictly-small-supported powersets (Subsections~\ref{subsect.finsupp.pow} and~\ref{subsect.strict.pow}), equivariance properties of atoms and the $\new$-quantifier (Theorem~\ref{thrm.equivar} and Definition~\ref{defn.New}), and the $\new$-quantifier for sets (Subsection~\ref{subsect.nua}).

Section~\ref{sect.fol-algebra} introduces $\sigma$- and $\amgis$-algebras.
These are the basic building blocks from which our models will be constructed.
The definitions are already non-trivial; highlights are the axioms of Figure~\ref{fig.nom.sigma} (which go back to \cite{gabbay:capasn}, where nominal algebra was introduced to axiomatise substitution) and \cite{gabbay:stodfo} (which introduced $\amgis$-algebras), and the precise definition of the $\sigma$-action on nominal powersets in Definition~\ref{defn.sub.sets}, which uses the $\new$-quantifier.

Section~\ref{sect.nom.pow} considers lattices over nominal sets.
\emph{The} technical highlight here is the characterisation of universal quantification in terms of \emph{fresh-finite limits} (Definition~\ref{defn.nom.poset} and subsequent results).
Combined with impredicativity and the $\sigma$-action we arrive at Definitions~\ref{defn.D.impredicative} and~\ref{defn.india}, which are the lattice-theoretic structure within which we eventually build models of the $\lambda$-calculus which we write $\india$ (pronounced `India').

Section~\ref{sect.sigma.foleq} shows that (simplifying) every nominal powerset is a model of Definition~\ref{defn.D.impredicative}.
Technical highlights here are Definition~\ref{defn.nu.U}, Lemma~\ref{lemm.technical}, and Proposition~\ref{prop.all.sub.commute} which show how to interpret universal quantification and check that it commutes with the $\sigma$-action.

Section~\ref{sect.completeness} uses filters and prime filters to give a nominal sets representation of any $\india$, and Section~\ref{sect.stone} extends this to a full duality.

This duality is for a propositional logic with quantifiers.
To handle the $\lambda$-calculus we need more: this happens in Section~\ref{sect.operators}.
The most important points are probably the introduction of $\app$ and its right adjoint $\ppa$ in Definition~\ref{defn.FOLeq.pp}, and the observation that their topological dual is a \emph{combination operator} $\bpp$ in Definition~\ref{defn.bpp}.

It now becomes fairly easy to show that every $\indiapp$ is also a semantics for the untyped $\lambda$-calculus.
This is Section~\ref{sect.lambda.calculus}, culminating in Definition~\ref{defn.ddenot} and Theorem~\ref{thrm.lambda.soundness}; we include an interlude in Subsection~\ref{subsect.interlude.lambda} where we pause to take stock of what we have been doing so far.

Slightly harder is proving completeness, for which we must construct an object in $\indiapp/\inspectapp$ for any given $\lambda$-theory.
This is Section~\ref{sect.lambda.representation}: we have the tools (nominal and otherwise) required in principle to carry out the constructions (just build filters, etcetera)---but in practice the amount of detailed structure required to make this work is quite striking, involving amongst many other things the construction of $\app$ and $\ppa$ on points (Definition~\ref{defn.some.ops}) and a left adjoint to the $\amgis$-algebra on filters (Definition~\ref{defn.qasmu}).
If two results should illustrate how tightly knit this part of the mathematics can be, then the technical results of Proposition~\ref{prop.fresh.point} and Lemma~\ref{lemm.lsm.id} are good examples.
The final Completeness result is Theorem~\ref{thrm.Pi.completeness}.

We conclude in an Appendix with a nominal axiomatisation of $\forall$, to go with the lattice-theoretic one of Definition~\ref{defn.fresh.finite.limit}, and some nice additional observations on the structure of points.

\subsection{A list of interesting technical features}
\label{subsect.list.of.technical.features}

This list is not of the major results, nor is it an exhaustive list of technical definitions.
But one technical definition or proof looks very much like another, so here are suggestions of which technical highlights might be worth looking at first:
\begin{itemize*}
\item
We characterise universal quantification in various ways, most notably in Proposition~\ref{prop.these.are.equivalent}, and see the connection made between universal quantification and the $\new$-quantifier in condition~\ref{filter.new} of Definition~\ref{defn.filter}.
See also the discussion opening Appendix~\ref{subsect.ffl}.
\item
As mentioned above, $\beta$-reduction and $\eta$-expansion are derived from a counit and unit respectively in Proposition~\ref{prop.lambda.beta.eta}.
\item
We decompose of $\lambda$ into $\forall$ and $\ppa$ in Notation~\ref{nttn.lambda}, and we decompose $\forall$ further using the $\new$-quantifier and freshness in Proposition~\ref{prop.these.are.equivalent}.
\item
Lemmas~\ref{lemm.p.sigma.pi} and~\ref{lemm.tall.unions} are inherently surprising results.
\item
Three little proof gems are in Propositions~\ref{prop.points.bigcap} and~\ref{prop.pumsa.point}, and Lemma~\ref{lemm.lam.point.fresh}.
\item
We need two notions of filter: one is Definition~\ref{defn.filter} (particularly condition~\ref{filter.new}, mentioned above), the other is Definition~\ref{defn.pi.point} (we call it a \emph{point}).
Both notions use the $\new$-quantifier in interesting ways.
\item
Because of condition~\ref{filter.new} of Definition~\ref{defn.filter}, the proof of existence of maximal filters requires a delicate argument; see Remark~\ref{rmrk.hand.rolled.zorn} and Theorem~\ref{thrm.maxfilt.zorn}.
%This is key to proving completeness.
\item
Propositions~\ref{prop.extra.filter} and~\ref{prop.even.stronger} are key to proving compactness in the duality proof.
\item
The pointwise definition of substitution is in Definition~\ref{defn.sub.sets};
a different approach to substitution than the reader has likely seen, which makes use of the $\new$-quantifier and $\amgis$-algebra (Definition~\ref{defn.bus.algebra}).
The two concrete characterisations of it in Subsection~\ref{subsect.two.chars.lsm} are also interesting.
\item
The reader will see much use made of the nominal $\new$-quantifier, %meaning `for all but finitely many atoms',
and of nominal equivariance properties and notions of small support and \emph{strict} small support.
We mention two (connected) examples: the use of $\new$ in defining the $\sigma$-action in Definition~\ref{defn.sub.sets}, and the treatment of the universal quantifier in condition~\ref{filter.new} of Definition~\ref{defn.filter} and the dual condition~\ref{sober.all} of Definition~\ref{defn.sober}.

We use $\new$ deep in the proofs of many technical results: the interested reader could peruse for instance Lemma~\ref{lemm.technical}, Proposition~\ref{prop.technical.contradiction}, Lemma~\ref{lemm.bus.filter}, Proposition~\ref{prop.even.stronger}, Lemma~\ref{lemm.tall.unions}, and Proposition~\ref{prop.points.sober}.

A search of the paper for uses of Corollary~\ref{corr.stuff}, Theorems~\ref{thrm.equivar} and~\ref{thrm.new.equiv}, and Lemma~\ref{lemm.strict.support} will find many uses made of equivariance and (strict) small support.
\item
A nominal algebra axiomatisation in Appendix~\ref{subsect.alg} underlies the design of the maths.
\end{itemize*}

\subsection{Why is this paper so long?}

\begin{enumerate*}
\item
This paper contains a duality result \emph{and} gives a semantics and completeness proof for the untyped $\lambda$-calculus.

It is a fact that duality results are hard, completeness proofs are also hard, and semantics for the untyped $\lambda$-calculus are not trivial to construct.\footnote{There are actually two duality results: one for impredicative distributive lattices (corresponding to a propositional logic with $\land$, $\lor$, and propositional quantification) from Part~\ref{part.indias}, and the other when we add the combination structure in Part~\ref{part.application}.
The second duality piggy-backs on the first, and is shorter.}
\item
We obtain our duality result by a far from obvious combination of nominal algebra \cite{gabbay:noma-jv}, logic in nominal powersets \cite{gabbay:semooc}, and the modal model theory of the lambda calculus \cite{gabbay:simcmt}.
Even individually these techniques are not well known, so to be self-contained and rigorous in combining them, this paper must be long and detailed.
Where we can rely on the material being familiar, we will be more brief.
\end{enumerate*}
So our starting point is urelemente (atoms) of Fraenkel-Mostowski set theory; our target is the untyped $\lambda$-calculus, and between them there is a lot of ground to cover.
Everything in this paper is there because its has to be, it is worth the effort, and
the material has natural momentum which propels us from the first definitions to the final results.

%%%%%%%%%%%%%%%%%%%%%%%%%%%%%%%%%%%
\section{Background on nominal techniques}
\label{sect.basic.defs}

A nominal set is a `set with names'.
The notion of a name being `in' an element is given by support $\supp(x)$ (Definition~\ref{defn.supp}).
For more details of nominal sets, see \cite{gabbay:newaas-jv,gabbay:fountl,pitts:nomsns}.

Here we just give necessary background information.
The reader not interested in nominal techniques \emph{per se} might like to read this section only briefly in the first instance, and use it as a reference for the later sections, where the ideas get applied.

For the reader's convenience we take a moment to note the overall message of this section:
\begin{itemize}
\item
\emph{To the category-theorist} we say that we work mostly in the category of nominal sets, or equivalently in the Schanuel topos (more on this in \cite[Section III.9]{MLM:sgl},\ \cite[A.21, page 79]{johnstone:skeett},\ or \cite[Theorem~9.14]{gabbay:fountl}), and occasionally also in the category of sets with a permutation action.
\item
\emph{To the set-theorist} we say that our constructions can be carried out in Fraenkel-Mostowski set theory (\deffont{FM sets}) and Zermelo-Fraenkel set theory with atoms (\deffont{ZFA}).
A discussion of such sets foundations, tailored to nominal techniques, can be found in \cite[Section~10]{gabbay:fountl}).
\item
\emph{To the reader not interested in foundations} we say that
the apparently inconsequential step of assuming names as primitive entities in Definition~\ref{defn.atoms}
yields a remarkable clutch of definitions and results, notably Theorem~\ref{thrm.supp} and Corollary~\ref{corr.stuff}, and Theorems~\ref{thrm.equivar} and~\ref{thrm.new.equiv}.
These properties are phrased abstractly but will quickly make themselves very useful in the body of this paper.
See previous work for more background \cite{gabbay:newaas-jv,gabbay:fountl,gabbay:nomtnl,pitts:nomsns}.
\end{itemize}

\subsection{Basic definitions}

\begin{defn}
\label{defn.atoms}
For this paper we fix the following:
\begin{itemize*}
\item
Fix an infinite set of \deffont{atoms} $\mathbb A$ and write $\f{size}(\mathbb A)$ for the cardinality of $\mathbb A$.
\item
Call any set whose cardinality is strictly lesser than $\f{size}(\mathbb A)$ \deffont{\finite}.
\item
Call $A\subseteq\mathbb A$ \deffont{cosmall} when $\mathbb A{\setminus} A$ is \finite.
\item
We use a \deffont{permutative convention} that $a,b,c,\ldots$ range over \emph{distinct} atoms.
\end{itemize*}
\end{defn}

\begin{rmrk}
Traditionally $\mathbb A$ is taken to be countable so that `small' corresponds to being finite.
It will be useful (and no harder) to work with the generalisation of the theory for potentially larger sets of atoms and larger notions of `small'.
\end{rmrk}

\begin{defn}
\label{defn.permutation}
A \deffont{permutation} $\pi$ is a bijection on atoms such that $\nontriv(\pi)=\{a\mid \pi(a)\neq a\}$ is finite.

Write $\id$ for the \deffont{identity} permutation such that $\id(a)=a$ for all $a$.
Write $\pi'\circ\pi$ for composition, so that $(\pi'\circ\pi)(a)=\pi'(\pi(a))$.
Write $\pi^\mone$ for inverse, so that $\pi^\mone\circ\pi=\id=\pi\circ\pi^\mone$.
Write $(a\;b)$ for the \deffont{swapping} (terminology from \cite{gabbay:newaas-jv}) mapping $a$ to $b$,\ $b$ to $a$,\ and all other $c$ to themselves, and take $(a\;a)=\id$.
\end{defn}

\begin{nttn}
\label{nttn.fix}
If $A\subseteq\mathbb A$ write
$$
\fix(A)=\{\pi\mid \Forall{a{\in} A}\pi(a)=a\}.
$$
\end{nttn}

\begin{defn}
\label{defn.fin.supp}
\begin{enumerate}
\item
A \deffont{set with a permutation action} $\ns X$ is a pair $(|\ns X|,\act)$ of an \deffont{underlying set} $|\ns X|$ and a \deffont{permutation action} written $\pi\act_{\ns X} x$ or just $\pi\act x$
 which is a group action on $|\ns X|$, so that $\id\act x=x$ and $\pi\act(\pi'\act x)=(\pi\circ\pi')\act x$ for all $x\in|\ns X|$ and permutations $\pi$ and $\pi'$.
\item\label{fin.supp.supports}
Say that $A\subseteq\mathbb A$ \deffont{supports} $x\in|\ns X|$ when $\Forall{\pi}\pi\in\fix(A)\limp \pi\act x=x$.
\item
If a \finite $A$ supporting $x$ exists, call $x$ \deffont{small-supported}.
\end{enumerate}
\end{defn}

\begin{frametxt}
\begin{defn}
\label{defn.nominal.set}
Call a set with a permutation action $\ns X$ a \deffont{nominal set} when every $x\in|\ns X|$ has small support.
$\ns X$, $\ns Y$, $\ns Z$ will range over nominal sets.
\end{defn}
\end{frametxt}

\begin{rmrk}
\label{rmrk.finite.perm}
Permutations are \emph{finite} in Definition~\ref{defn.permutation}, yet support is \emph{small} in Definition~\ref{defn.nominal.set}.
Why are permutations not taken to be small (or even arbitrary) instead?

In this paper we are interested in modelling binders that abstract a single atom at a time.
$\lambda$-syntax is finite so only finitely many binders are ever applied to an element, thus, we never need to $\alpha$-rename more than finitely many atoms in any given element and we only need finite permutations.\footnote{This paper has many binders.  For instance: $\lambda a$, $\nw a$, $\new a$, $\tall a$, $\lambda_\idiom$, and $\sigma$-actions (`substitutions') $[a\sm\text{-}]$.  Nevertheless, each of them only abstracts one atom at a time.}

Generalisations to infinite permutations are possible, but they introduce complexity for no technical benefit to the specific concerns of this paper.
More on this in Subsection~\ref{subsect.perm.size}.

For discussion of why we need small support rather than finite support, see Remark~\ref{rmrk.size.issues} and Subsection~\ref{subsect.cardinality}.
\end{rmrk}

\begin{defn}
\label{defn.equivariant}
Call a function $f$ from $|\ns X|$ to $|\ns Y|$ \deffont{equivariant} when $\pi\act(f(x))=f(\pi\act x)$ for all permutations $\pi$ and $x\in|\ns X|$.
In this case write $f:\ns X\equivarto\ns Y$.

The category of nominal sets and equivariant functions between them is usually called the category of \emph{nominal sets}.
\end{defn}

\begin{defn}
\label{defn.supp}
Suppose $\ns X$ is a nominal set and $x\in|\ns X|$.
Define the \deffont{support} of $x$ by
$$
\supp(x)=\bigcap\{A\mid A\text{ \finite and supports }x\} .
$$
If $\supp(x)=\varnothing$ we call $x$ \deffont{equivariant}.
\end{defn}

\begin{nttn}
\label{nttn.fresh}
Write $a\#x$ as shorthand for $a\not\in\supp(x)$ and read this as $a$ is \deffont{fresh for} $x$.

Given atoms $a_1,\dots,a_n$ and elements $x_1,\dots,x_m$ write $a_1,\dots,a_n\#x_1,\dots,x_m$ as shorthand for $\{a_1,\dots,a_n\}\cap\bigcup_{1{\leq}j{\leq}m}\supp(x_j)=\varnothing$, or to put it more plainly: $a_i\#x_j$ for every $i$ and $j$.
\end{nttn}

\begin{prop}
\label{prop.intersect.AB}
If $A\subseteq\atoms$ is small and supports $x$, and $a\in\atoms$ and $a\#x$, then $A{\setminus}\{a\}$ supports $x$.
\end{prop}
\begin{proof}
Suppose $\pi\in\fix(A{\setminus}\{a\})$.
We assumed $a\#x$ so choose an $A'\subseteq\atoms$ such that $A'$ is small, $A'$ supports $x$, and $a\not\in A'$.
Also, choose some fresh $a'$ (so $a'\not\in A{\cup}\{a\}{\cup}A'$).

Write $\tau=(a'\ a)$.
Note that $\tau\act x=x$ by Definition~\ref{defn.fin.supp}(\ref{fin.supp.supports}), because $A'$ supports $x$ and $\tau\in\fix(A')$.

It is a fact that $(\tau\circ\pi\circ\tau)(a)=a$ for every $a{\in}A$, so $\tau\circ\pi\circ\tau\in\fix(A)$.
Also by the group action $(\tau\circ\pi\circ\tau)\act x = \tau\act(\pi\act(\tau\act x))$.
Since $A$ supports $x$ we have $\tau\act(\pi\act(\tau\act x))=x$ by Definition~\ref{defn.fin.supp}(\ref{fin.supp.supports}).

We apply $\tau$ to both sides, recall that $\tau\act x=x$, and conclude that $\pi\act x=x$ as required.
\end{proof}

\begin{thrm}
\label{thrm.supp}
Suppose $\ns X$ is a nominal set and $x\in|\ns X|$.
Then $\supp(x)$ is the unique least \finite set of atoms that supports $x$.
\end{thrm}
\begin{proof}
Consider a permutation $\pi\in\fix(\supp(x))$ and write $\{a_1,\dots,a_n\}=\f{nontriv}(\pi)$.
Choose any small $A{\subseteq}\atoms$ that supports $x$, so by construction $\supp(x){\subseteq}A$.

By Proposition~\ref{prop.intersect.AB} $A{\setminus}\nontriv(\pi)$ supports $x$.
By construction $\pi\in\fix(A{\setminus}\nontriv(\pi))$, so by Definition~\ref{defn.fin.supp}(\ref{fin.supp.supports}) $\pi\act x=x$ as required.
\end{proof}

\subsection{Examples}
\label{subsect.pow}

Suppose $\ns X$ and $\ns Y$ are nominal sets, and suppose $\ns Z$ is a set with a permutation action.
We consider some examples of sets with a permutation action and of nominal sets.
These will be useful later on in the paper.

\subsubsection{Atoms and booleans}
\label{subsect.xmpl.atoms}

$\mathbb A$ is a nominal set with the \emph{natural permutation action} $\pi\act a=\pi(a)$.

For the case of $\mathbb A$ only we will be lax about the difference between $\mathbb A$ (the set of atoms) and $(|\mathbb A|,\act)$ (the nominal set of atoms with its natural permutation action).
In practice this means we may write $a\in\mathbb A$ for $a\in|\mathbb A|$.

The only equivariant function from $\mathbb A$ to itself (Definition~\ref{defn.equivariant})
is the identity map $a\mapsto a$.
There are more small-supported maps from $\mathbb A$ to itself; see the \emph{small-supported function space} below.

Write $\mathbb B$ for the nominal set of \deffont{Booleans}, which has elements $\{\bot,\top\}$ and the \deffont{trivial} permutation action that $\pi\act x=x$ for all $\pi$ and $x\in|\mathbb B|$.

\subsubsection{Cartesian product}
\label{subsect.cartesian.product}

$\ns X\times\ns Y$ is a nominal set with underlying set $\{(x,y)\mid x\in|\ns X|, y\in|\ns Y|\}$ and the \emph{pointwise} action $\pi\act(x,y)=(\pi\act x,\pi\act y)$.

An equivariant $f:(\ns X\times\ns Y)\equivarto\mathbb B$ corresponds to a relation $\somerel$ such that $x\somerel y$ if and only if $\pi\act x\somerel\pi\act y$.

\subsubsection{Tensor product}
\label{subsect.otimes}

$\ns X\otimes\ns Y$ is a nominal set with underlying set $\{(x,y)\mid x\in|\ns X|, y\in|\ns Y|,\ \supp(x){\cap}\supp(y)=\varnothing\}$ and the pointwise action.
For the pointwise action here to be well-defined depends on $\pi$ being a permutation and the fact (Proposition~\ref{prop.pi.supp} below) that $\supp(\pi\act x)=\pi\act \supp(x)$.

\subsubsection{Full function space}
\label{subsect.full.function.space}

Functions from $|\ns X|$ to $|\ns Y|$ form a set with a permutation action with underlying set all functions from $|\ns X|$ to $|\ns Y|$, and the \deffont{conjugation} permutation action
$$
(\pi\act f)(x)=\pi\act(f(\pi^\mone\act x)) .
$$

The conjugation action can be rephrased as `permutations distribute' (cf. Theorem~\ref{thrm.equivar} below):
\begin{lemm}
\label{lemm.useful.distrib}
If $f{\in}|\ns X{\to}\ns Y|$ then $\pi\act f(x)=(\pi\act f)(\pi\act x)$.
\end{lemm}
\begin{proof}
By easy calculations.
\end{proof}

\subsubsection{Small-supported function space}

$\ns X\Func\ns Y$ is a nominal set with underlying set the functions from $|\ns X|$ to $|\ns Y|$ with small support under the conjugation action, and the conjugation permutation action.

A complete description of the small-supported functions from $\mathbb A$ to $\mathbb A$ is as follows:
\begin{enumerate*}
\item
Any function $f$ such that there exists some \finite $U\subseteq\mathbb A$ such that if $b\not\in U$ then $f(b)=b$ (so $f$ is `eventually the identity').
\item
Any function $f$ such that there exists some $a\in\mathbb A$ and some \finite $U\subseteq\mathbb A$ such that if $b\not\in U$ then $f(b)=a$ (so $f$ is `eventually constant').
\end{enumerate*}
Looking ahead to Definition~\ref{defn.New}, we can write the possibilities as $\New{b}f(b){=}b$ and $\Exists{a}\New{b}f(b){=}a$ respectively.

\begin{lemm}
\label{lemm.equivar.equivar}
$f\in |\ns X\Func\ns Y|$ is equivariant in the sense of Definition~\ref{defn.equivariant} if and only if it is equivariant in the sense of Definition~\ref{defn.supp}.
\end{lemm}
\begin{proof}
We sketch the proof: If $\pi\act f=f$ then for any $x{\in}\ns X$ we have by Lemma~\ref{lemm.useful.distrib} that $\pi\act (f(x))=(\pi\act f)(\pi\act x)=f(\pi\act x)$.
Conversely if for any $x{\in}\ns X$ we have $\pi\act(f(x))=f(\pi\act x)$ then by the conjugation action
$
(\pi\act f)(x)=
\pi\act (f(\pi^\mone\act x))=f(\pi\act(\pi^\mone\act x))=f(x).
$
\end{proof}

\subsubsection{Full powerset}

\begin{defn}
\label{defn.pointwise.action}
Suppose $\ns X$ is a set with a permutation action.
Give subsets $X\subseteq|\ns X|$ the \deffont{pointwise} permutation action
$$
\pi\act X=\{\pi\act x\mid x\in X\} .
$$
\end{defn}

Then $\powerset(\ns X)$ (the full powerset of $\ns X$) is a set with a permutation action with
\begin{itemize*}
\item
underlying set $\{X\mid X\subseteq|\ns X|\}$ (the set of all subsets of $|\ns X|$), and
\item
the pointwise action $\pi\act X=\{\pi\act x\mid x\in X\}$.
\end{itemize*}

A useful instance of the pointwise action is for sets of atoms.
As discussed in Subsection~\ref{subsect.xmpl.atoms} above, if $a\in\mathbb A$ then $\pi\act a=\pi(a)$.
Thus if $A\subseteq\mathbb A$ then
$$
\pi\act A\quad\text{means}\quad \{\pi(a)\mid a\in A\} .
$$

In Definition~\ref{defn.pointwise.action} we gave $X\subseteq|\ns X|$ a permutation action.
Thus by Definition~\ref{defn.supp} we also gave them a notion of equivariance (which $X$ may or may not satisfy, of course).
It is useful to unpack what this means:
\begin{lemm}
\label{lemm.when.set.equivar}
Suppose $\ns X$ is a set with a permutation action and $X\subseteq|\ns X|$.
Then the following are equivalent:
\begin{enumerate*}
\item
$X$ is equivariant in the sense of Definition~\ref{defn.supp}.
\item
$\pi\act X=X$ for any permutation $\pi$.
\item
$x\in X\liff \pi\act x\in X$ for any permutation $\pi$ and $x\in X$.
\end{enumerate*}
\end{lemm}
\begin{proof}
If $X$ is equivariant then $X$ is supported by $\varnothing$ and by Theorem~\ref{thrm.supp}
and Definition~\ref{defn.fin.supp}(\ref{fin.supp.supports})
$\Forall{\pi}\pi\act X=X$.
It follows by Definition~\ref{defn.pointwise.action} that $x\in X\liff x\in\pi^\mone\act X \liff \pi\act x\in X$.

If $\Forall{\pi}\Forall{x{\in}|\ns X|}(x\in X\liff \pi\act x\in X)$ then also $x\in X\liff \pi^\mone\act x\in X \liff x\in\pi\act X$, so that $\pi\act X=X$.
This is for any $\pi$, so by Definition~\ref{defn.fin.supp}(\ref{fin.supp.supports}) $\varnothing$ supports $X$ and by
Definition~\ref{defn.supp} we have $\supp(X)=\varnothing$.
\end{proof}

\begin{lemm}
\label{lemm.the.comb}
It is not the case that if $\ns X$ is a nominal set then $\powerset(\ns X)$ is a nominal set.
%That is, $\powerset(\ns X)$ is not necessarily a nominal set.
\end{lemm}
\begin{proof}
It suffices to provide a counterexample.
Take $\ns X=\atoms$ and enumerate atoms as $(a_\omega)_\omega$, and consider the set
$$
\f{comb}=\{a_{2*\omega} \mid \omega \}
$$
of `every other atom'.
This does not have small support, though permutations still act on it pointwise.
For more discussion of this point, see \cite[Remark~2.18]{gabbay:fountl}.
\end{proof}

We consider further examples in Subsection~\ref{subsect.more.examples}, including the small-supported and strictly small-supported powersets.

\subsection{The principle of equivariance and the NEW quantifier}
\label{subsect.pre-equivar}

We come to Theorem~\ref{thrm.equivar}, a result which is central to the `look and feel' of nominal techniques.
It enables a particularly efficient management of renaming and $\alpha$-conversion in syntax and semantics and captures why it is so useful to use \emph{names} in the foundations of our semantics and not some other infinite set, such as numbers.

Names are by definition symmetric (i.e. can be permuted).
Taking names and permutations as \emph{primitive} implies that permutations propagate to the things we build using them.
This is the \emph{principle of equivariance} (Theorem~\ref{thrm.equivar} below; see also \cite[Subsection~4.2]{gabbay:fountl} and \cite[Lemma~4.7]{gabbay:newaas-jv}).

The principle of equivariance implies that, provided we permute names uniformly in all the parameters of our definitions and theorems, we then get another valid set of definitions and theorems.
This is not true of e.g. numbers because our mathematical foundation equips numbers by construction with numerical properties such as \emph{less than or equal to $\leq$}, which can be defined from first principles with no parameters.

So if we use (for instance) numbers for names then we do not care about $\leq$ because we just needed an infinite set of elements, but then we repeatedly have to \emph{prove} that we did not use an asymmetric property like $\leq$.
In contrast, with nominal foundations and atoms, we do not have to explicitly prove symmetry because we can just look at our mathematical foundation and note that it is naturally symmetric under permuting names; we reserve numbers for naturally \emph{a}symmetric activities, such as counting.

This style of name management is characteristic of nominal techniques.
The reader can find it applied often, e.g. in Lemmas~\ref{lemm.fresh.sub} and~\ref{lemm.pow.closed}, Propositions~\ref{prop.amgis.2} and~\ref{prop.char.freshwedge}, and Lemma~\ref{lemm.pow.nu.closed}.

\begin{rmrk}
The languages of ZFA set theory and FM set theory are identical: first-order logic with equality $=$ and sets membership $\in$.
\end{rmrk}

\begin{thrm}
\label{thrm.equivar}
\label{thrm.no.increase.of.supp}
If $\vect x$ is a list $x_1,\ldots,x_n$, write $\pi\act \vect x$ for $\pi\act x_1,\ldots,\pi\act x_n$.
Suppose $\Phi(\vect x)$ is a predicate in the language of ZFA/FM set theory, with free variables $\vect x$.
Suppose $\Upsilon(\vect x)$ is a function specified in the language of ZFA/FM set theory, with free variables $\vect x$.
Then we have the following principles:
\begin{enumerate*}
\item\label{equivar.pred}
\deffont{Equivariance of predicates.} \ $\Phi(\vect x) \liff \Phi(\pi\act \vect x)$.\footnote{Here $\vect x$ is understood to contain \emph{all} the variables mentioned in the predicate.
It is not the case that $a=a$ if and only if $a=b$---but it is the case that $a=b$ if and only if $b=a$.}
\item\label{equivar.func}
\deffont{Equivariance of functions.}\quad\,$\pi\act \Upsilon(\vect x) = \Upsilon(\pi\act \vect x)$ (cf. Definition~\ref{defn.equivariant}).\footnote{Parts~1 and~2 of Theorem~\ref{thrm.equivar} are morally the same result: by considering $\Phi$ to be a function from its arguments to the nominal set of Booleans $\mathbb B$ from Subsection~\ref{subsect.pow}; and by treating a function as a functional relation, i.e. as a binary predicate.}
\item\label{item.conservation.of.support}
\deffont{Conservation of support.}\quad\ \
If $\vect x$ denotes elements with small support \\ then
$\supp(\Upsilon(\vect x)) \subseteq \supp(x_1){\cup}\cdots{\cup}\supp(x_n)$.

If in addition $\Upsilon$ is injective, then
$\supp(\Upsilon(\vect x)) = \supp(x_1){\cup}\cdots{\cup}\supp(x_n)$.
\end{enumerate*}
\end{thrm}
\begin{proof}
See Theorem~4.4, Corollary~4.6, and Theorem~4.7 from \cite{gabbay:fountl}.
\end{proof}

\begin{rmrk}
Theorem~\ref{thrm.equivar} is three fancy ways of observing that if a specification is symmetric in atoms,
the the result must be at least as symmetric as the inputs.
Using atoms (instead of e.g. numbers) to model names makes this a one-line argument.\footnote{%
The reasoning in this paper could in principle be fully formalised in a sets foundation with atoms, such as Zermelo-Fraenkel set theory with atoms \deffont{ZFA}.
Nominal sets can be implemented in ZFA sets such that nominal sets map to equivariant elements (elements with empty support) and the permutation action maps to `real' permutation of atoms in the model.
See \cite[Subsection~9.3]{gabbay:fountl} and \cite[Section~4]{gabbay:fountl}.
}
\end{rmrk}

\begin{prop}
\label{prop.pi.supp}
$\supp(\pi\act x)=\pi\act\supp(x)$ (which means $\{\pi(a)\mid a\in\supp(x)\}$).

Using Notation~\ref{nttn.fresh}, $a\#\pi\act x$ if and only if $\pi^\mone(a)\#x$, and $a\#x$ if and only if $\pi(a)\#\pi\act x$.
\end{prop}
\begin{proof}
Immediate consequence of part~2 of Theorem~\ref{thrm.equivar}.\footnote{There is also a nice proof of this fact by direct calculations; see \cite[Theorem~2.19]{gabbay:fountl}.
However, it just instantiates Theorem~\ref{thrm.equivar} to the particular $\Upsilon$ specifying support.}
\end{proof}

\begin{corr}
\label{corr.stuff}
Suppose $\ns X$ is a set with a permutation action and suppose $A\subseteq\mathbb A$ supports $x$.
\begin{enumerate*}
\item\label{stuff.fixsupp.fixelt}
If $\pi(a)=a$ for all $a\in A$ then $\pi\act x=x$.

In particular, if $x$ has small support and $\pi(a)=a$ for all $a\in\supp(x)$, then $\pi\act x=x$.
\item\label{stuff.two.pi}
If $\pi(a)=\pi'(a)$ for every $a{\in}A$ then $\pi\act x=\pi'\act x$.
\item\label{stuff.freshness.criterion}
If $x$ has small support then $a\#x$ if and only if $\Exists{b}(b\#x\land (b\;a)\act x=x)$.
\end{enumerate*}
\end{corr}
\begin{proof}
\begin{enumerate}
\item
Direct from Definition~\ref{defn.fin.supp}(\ref{fin.supp.supports}) and Theorem~\ref{thrm.supp}.
\item
By properties of the group action, $\pi\act x=\pi'\act x$ precisely when $x=(\pi^\mone\circ\pi')\act x$.
So it suffices by part~1 of this result that $a=\pi^\mone\circ\pi'(a)$ for every $a\in A$, in other words, that $\pi(a)=\pi'(a)$.
\item
Suppose $x$ has small support.
We prove two implications:
\begin{itemize}
\item
Suppose $a\#x$, meaning by Notation~\ref{nttn.fresh} that $a\not\in\supp(x)$.
Choose any other $b\not\in\supp(x)$; then $(a\ b)\in\fix(\supp(x))$ and by part~1 of this result, $(b\ a)\act x=x$.
\item
Suppose $b\#x$, meaning by Notation~\ref{nttn.fresh} that $b\not\in\supp(x)$, and suppose $(b\;a)\act x=x$.
It follows by Proposition~\ref{prop.pi.supp} that $(b\ a)\act\supp(x)=\supp(x)$ and by facts of sets, $a\not\in\supp(x)$.
\qedhere\end{itemize}
\end{enumerate}
\end{proof}

\begin{frametxt}
\begin{defn}
\label{defn.New}
Write
$\New{a}\Phi(a)$ for `$\{a\mid \neg\Phi(a)\}$ is \finite'.
We call this the \deffont{$\new$ quantifier}.
\end{defn}
\end{frametxt}

\begin{rmrk}
We can read $\new$ as `for all but a small number of $a$', `for fresh $a$', or `for new $a$'.
It captures a \emph{generative} aspect of names, that for any $x$ we can find plenty of atoms $a$ such that $a\not\in\supp(x)$.
$\new$ was designed in \cite{gabbay:newaas-jv} to model the quantifier being used when we informally write
``rename $x$ in $\lam{x}t$ to be fresh'', or ``emit a fresh channel name'' or ``generate a fresh memory cell''.
\end{rmrk}

\begin{rmrk}
$\new$ belongs to a family of `for most' quantifiers \cite{westerstahl:quafnl}, and is a \emph{generalised quantifier} \cite[Section 1.2.1]{keenan:genqll}.

Specifically over nominal sets, however, $\new$ displays special properties.
In particular, it satisfies the \emph{some/any property} that to prove a $\new$-quantified property we test it for \emph{one} fresh atom;
we may then use it for \emph{any} fresh atom.
This is Theorem~\ref{thrm.new.equiv}, which appears in the literature for instance as \cite[Theorem~6.5]{gabbay:fountl} and \cite[Proposition~4.10]{gabbay:newaas-jv}:
\end{rmrk}

\begin{thrm}
\label{thrm.new.equiv}
Suppose $\Phi(\vect z,a)$ is a predicate in the language of ZFA/FM set theory, with free variables $\vect z,a$. 
Suppose $\vect z$ denotes elements with small support.
Then the following are equivalent:
$$
\Forall{a{\in}\atoms}\bigl(a\# \vect z \limp \Phi(\vect z,a)\bigr)
\qquad
\New{a}\Phi(\vect z,a)
\qquad
\Exists{a{\in}\atoms}\bigl(a\#\vect z  \land \Phi(\vect z,a)\bigr)
$$
\end{thrm}
\begin{proof}
Where convenient we may write $\vect z$ as $z_1,\dots,z_n$.
\begin{itemize}
\item
Suppose $\Phi(\vect z,a)$ holds for every atom $a\in\atoms{\setminus}\bigcup_{1{\leq}i{\leq}n}\supp(z_i)$.

By assumption $\vect z$ denotes elements with small support, and it is a fact that a finite union of small sets is small, so $\atoms\setminus\bigcup_{1{\leq}i{\leq}n}\supp(z_i)$ is cosmall.

It follows that $\New{a}\Phi(\vect z,a)$ holds.
\item
Suppose $A\subseteq\atoms$ is cosmall and $\Phi(\vect z,a)$ for every $a{\in}A$.
As in the previous point, there exists some $a{\in}A$ such that $a\#z_i$ for every $1{\leq}i{\leq}n$.

It follows that
$\Exists{a{\in}\atoms}\bigl(a\#\vect z \land \Phi(\vect z,a)\bigr)$.
\item
Now suppose $\Phi(\vect z,a)$ holds for some $a\in\atoms{\setminus}\bigcup_{1{\leq}i{\leq}n}\supp(z_i)$.

By Theorem~\ref{thrm.equivar}(\ref{equivar.pred}) $\Phi((a'\ a)\act\vect z,a')$ holds for any $a'{\in}\atoms$.
Choosing $a'\#\vect z$ we have by Corollary~\ref{corr.stuff}(\ref{stuff.fixsupp.fixelt}) that $(a'\ a)\act z_i=z_i$ for every $1{\leq}i{\leq}n$.

Thus $\Forall{a{\in}\atoms}\bigl(a\# \vect z \limp \Phi(\vect z,a)\bigr)$ holds.
\qedhere\end{itemize}
\end{proof}

We mention a characterisation of small support using $\new$:
\begin{lemm}
\label{lemm.old.school}
Suppose $\ns X$ is a set with a permutation action.
Then
$$
x{\in}\ns X\text{ is small-supported}
\quad\liff\quad
\New{a}\New{b}(b\ a)\act x=x.
$$
\end{lemm}
\begin{proof}
If $x$ is small-supported then by Corollary~\ref{corr.stuff}(\ref{stuff.fixsupp.fixelt}) it suffices to choose $a,b\in\atoms{\setminus}\supp(x)$.

Conversely suppose $A\subseteq\atoms$ is small and $(b\ a)\act x=x$ for all distinct $a,b\in\atoms{\setminus}A$.
Consider any permutation $\pi\in\fix(A)$ and note from Definition~\ref{defn.permutation} that $\nontriv(\pi)$ is \emph{finite}, so we can write $\pi$ as a finite composition of swappings of atoms not in $A$.
It follows that $\pi\act x=x$ and so by Definition~\ref{defn.fin.supp}(\ref{fin.supp.supports}) that $A$ supports $x$.
\end{proof}

\subsection{Two lemmas}

We conclude with two technical but general lemmas (they are also beautiful; especially the first one) which will help us later in Lemma~\ref{lemm.filter.adjoint.compat} and Proposition~\ref{prop.qappx.filter}.

Lemma~\ref{lemm.some.fresh.s} goes back to \cite[Lemma~5.2]{gabbay:forcie} and \cite[Corollary~4.30]{gabbay:nomuae}; see \cite[Lemma~7.6.2]{gabbay:nomtnl} for a recent presentation:
\begin{lemm}
\label{lemm.some.fresh.s}
If $\Upsilon$ is an equivariant function (Definition~\ref{defn.equivariant}) from $\ns X$ to $\ns Y$ and $a\#\Upsilon(x)$ then there exists some $x'\in|\ns X|$ such that $a\#x'$ and $\Upsilon(x)=\Upsilon(x')$.
\end{lemm}
\begin{proof}
Choose fresh $b$ (so $b\#x$).
By Corollary~\ref{corr.stuff}(\ref{stuff.fixsupp.fixelt}) $(b\ a)\act\Upsilon(x)=\Upsilon(x)$ and by Definition~\ref{defn.equivariant} $(b\ a)\act\Upsilon(x)=\Upsilon{((b\ a)\act x)}$.
We take $x'=(b\ a)\act x$.
\end{proof}

\begin{lemm}
\label{lemm.F.magic}
Suppose $F:\ns X\times\ns Y\to \ns Z$ is an equivariant function and $x\in|\ns X|$ and $y\in|\ns Y|$.
Suppose further that $a,b\#y$.
Then $F(x,(b\ a)\act y)=(b\ a)\act F(x,y)$.
\end{lemm}
\begin{proof}
By equivariance $(b\ a)\act F(x,y)=F((b\ a)\act x,(b\ a)\act y)$.
By Corollary~\ref{corr.stuff}(\ref{stuff.fixsupp.fixelt}) (since $a,b\#y$) also $(b\ a)\act y=y$.
\end{proof}

\subsection{Further examples}
\label{subsect.more.examples}

We now consider the small-supported powerset and the strictly small-supported powerset.
These examples are more technically challenging and will be key to the later constructions.

\subsubsection{Small-supported powerset}
\label{subsect.finsupp.pow}

Suppose $\ns X$ is a set with a permutation action (it does not have to be a nominal set).

Then $\nompow(\ns X)$, the \deffont{nominal powerset}, is a nominal set, with
\begin{itemize*}
\item
underlying set those $X\in|\powerset(\ns X)|$ that are small-supported, and
\item
with the \deffont{pointwise} action $\pi\act X=\{\pi\act x\mid x\in X\}$ inherited from Definition~\ref{defn.pointwise.action}.
\end{itemize*}

Unpacking the definitions and using Corollary~\ref{corr.stuff} and Lemma~\ref{lemm.old.school}, $X\subseteq|\ns X|$ is small-supported when, equivalently:
\begin{itemize*}
\item
There exists \finite $A\subseteq\mathbb A$ such that if $\pi\in\fix(A)$ then $\pi\act X=X$.
\item
There exists \finite $A\subseteq\mathbb A$ such that if $\pi\in\fix(A)$ and $x\in X$ then $\pi\act x\in X$.
\item
$\New{a}\New{b}(b\ a)\act X=X$.
\item
$\New{a}\New{b}\Forall{x}(x{\in} X\limp (b\ a)\act x{\in} X)$.
\end{itemize*}

For instance:
\begin{itemize*}
\item
$\nompow(\mathbb A)$ is the set of \finite and cosmall (Definition~\ref{defn.atoms}) sets of atoms.
\item
$X\in\nompow(\powerset(\mathbb A))$ is a set of sets of atoms with small support, though the elements $x\in X$ need not have small support.

For instance, if we set $x=\f{comb}$ from Lemma~\ref{lemm.the.comb} then we can take $X=\{\pi\act x\mid \text{all permutations }\pi\}$.
Here $X$ has \finite (indeed, empty) support, even though none of its elements $\pi\act x$ have small support.
\end{itemize*}

It is useful to formalise these observations as a lemma.
A common source of confusion is to suppose that if $A$ supports $X\in|\nompow(\ns X)|$ then $A$ must support every $x\in X$.
This is incorrect---and compare Lemma~\ref{lemm.pow.not.true} with Lemma~\ref{lemm.strict.support}:
\begin{lemm}
\label{lemm.pow.not.true}
It is not true in general that if $X\in|\nompow(\ns X)|$ and $x\in X$ then $\supp(x)\subseteq\supp(X)$.

In other words, $a\#X$ and $x\in X$ does not imply $a\#x$.
\end{lemm}
\begin{proof}
It suffices to provide a counterexample.
Take $\ns X=\mathbb A$ (the nominal set of atoms with the natural permutation action, from Subsection~\ref{subsect.xmpl.atoms}) and $X=\mathbb A\subseteq|\mathbb A|$ (the underlying set of the nominal set of all atoms, i.e. the set of all atoms!).

It is easy to check that $\supp(X)=\varnothing$ and $a\in X$ and $\supp(a)=\{a\}\not\subseteq\varnothing$.
\end{proof}

%%%%%%%%%%%%%%%%%%%%%%%%%%%%%%%%%%%%%%%%%%%%%%%%%%%%%%
\subsubsection{Strictly small-supported powerset}
\label{subsect.strict.pow}

Suppose $\ns X$ is a nominal set.

\begin{defn}
\label{defn.strictpow}
Call $X\subseteq|\ns X|$ \deffont{strictly supported} by $A\subseteq\mathbb A$ when
$$
\Forall{x{\in} X} \supp(x)\subseteq A .
$$
If there exists some \finite $A$ that strictly supports $X$, then call $X$ \deffont{strictly small-supported} (see \cite[Theorem~2.29]{gabbay:fountl}).

Write $\strict(\ns X)$ for the set of strictly small-supported $X\subseteq|\ns X|$.
That is:
$$
\strict(\ns X)=\{X\subseteq|\ns X|\mid \Exists{A{\subseteq}\mathbb A}A\text{ \finite}\land X\text{ strictly supported by }A\}
$$
\end{defn}

\begin{lemm}
\label{lemm.strict.support}
If $X\in\strict(\ns X)$ then:
\begin{enumerate*}
\item
$\bigcup\{\supp(x)\mid x{\in}X\}$ is \finite.
\item\label{strict.union}
$\bigcup\{\supp(x)\mid x{\in}X\}=\supp(X)$.
\item\label{strict.implies.finite}
If $X\subseteq|\ns X|$ is strictly small-supported then it is small-supported.
\item\label{strict.to.supp.elements}
$x\in X$ implies $\supp(x)\subseteq\supp(X)$ (contrast this with Lemma~\ref{lemm.pow.not.true}).

Equivalently using Notation~\ref{nttn.fresh}, if $a\in\mathbb A$ then $a\#X$ implies $a\#x$ for every $x\in X$.
\item
$\strict(\ns X)$ with the pointwise permutation action is a nominal set.
\end{enumerate*}
\end{lemm}
\begin{proof}
The first part is immediate since by assumption there is some \finite $A{\subseteq}\mathbb A$ that bounds $\supp(x)$ for all $x\in X$.
The second part follows by an easy calculation using Corollary~\ref{corr.stuff}(\ref{stuff.freshness.criterion}); full details are in \cite[Theorem~2.29]{gabbay:fountl}, of which Lemma~\ref{lemm.strict.support} is a special case.
The other parts follow by definitions from the first and second parts.
\end{proof}

\begin{xmpl}
\begin{enumerate*}
\item
$\varnothing\subseteq|\mathbb A|$ is small-supported and strictly small-supported by $\varnothing$.
\item
$\{a\}$ is small-supported by $\{a\}$ and also strictly small-supported by $\{a\}$.
\item
$\mathbb A\subseteq|\mathbb A|$ is small-supported by $\varnothing$ but not strictly small-supported.
\item
$\mathbb A{\setminus}\{a\}$ is small-supported by $\{a\}$ but not strictly small-supported.
\end{enumerate*}
\end{xmpl}

\begin{lemm}
\label{lemm.finite.strict}
If $X{\subseteq}|\ns X|$ is finite then $X\in\strict(\ns X)$.
\end{lemm}
\begin{proof}
By Theorem~\ref{thrm.no.increase.of.supp}(3) $\supp(X)\subseteq\bigcup\{\supp(x)\mid x{\in}X\}$.
\end{proof}

%%%%%%%%%%%%%%%%%%%%%%%%%%%%%%%
\subsection{The NEW-quantifier for nominal sets}
\label{subsect.nua}

Suppose $\ns X$ is a set with a permutation action.

\begin{defn}
\label{defn.nua}
Given a small-supported $X\subseteq|\ns X|$ define the \deffont{new-quantifier for (nominal) sets} by
$$
\nw a.X=\{x\mid \New{b}(b\ a)\act x\in X\} .
$$
\end{defn}

\begin{lemm}
\label{lemm.nua.iff}
Suppose $x\in|\ns X|$.
Suppose $X\subseteq|\ns X|$ is small-supported.
Then
$$
x\in \nw a.X\liff \New{b}(b\ a)\act x\in X.
$$
\end{lemm}
\begin{proof}
Immediate from Definition~\ref{defn.nua}.
\end{proof}

$\nw a.X$ was written $\tf n a.X$ in \cite[Definition~5.2]{gabbay:stodnb}, and goes back to \cite{gabbay:stusun} where it was written $X{-}a$.
We will use $\nw a$ in
Lemma~\ref{lemm.nw.on.points} to prove things about a $\sigma$-action.

\begin{lemm}
\label{lemm.supp.nua}
$\supp(\nw a.X)\subseteq\supp(X){\setminus}\{a\}$.
\end{lemm}
\begin{proof}
By a routine calculation using Corollary~\ref{corr.stuff}.
\end{proof}

Recall from Subsection~\ref{subsect.finsupp.pow} the notion of nominal powerset $\nompow(\ns X)$.
\begin{lemm}
\label{lemm.nua.finsupp}
If $X\in|\nompow(\ns X)|$ then $\nw a.X\in|\nompow(\ns X)|$.
\end{lemm}
\begin{proof}
This amounts to showing that if $X\subseteq|\ns X|$ has small support then so does $\nw a.X\subseteq|\ns X|$.
This follows by Lemma~\ref{lemm.supp.nua} or direct from Theorem~\ref{thrm.no.increase.of.supp}.
\end{proof}

%%%%%%%%%%%%%%%%%%%%%%%%%%%%%%%%%%%%%%%%%%%%
\jamiepart{Nominal distributive lattices with quantification}
\label{part.indias}

%%%%%%%%%%%%%%%%%%%%%%%%%%%%%%%%%%%%%%%%%%%%%%%%%%%%%%%%%%%%%%%%%%
\section{Nominal algebras over nominal sets}
\label{sect.fol-algebra}

\subsection{Definition of a sigma-algebra ($\sigma$-algebra)}
\label{subsect.sigma.amgis}

\maketab{tab0}{@{\hspace{2em}}L{6em}@{\ }L{3em}@{\ }R{7em}@{\ }L{10em}@{\ }L{10em}}

\subsubsection{A termlike $\sigma$-algebra}

\begin{figure}[b]
\begin{minipage}{\textwidth}
\begin{tab0}
\rulefont{\sigma a} && a[a\sm u]=&u
\\[2ex]
\rulefont{\sigma id} && x[a \sm a]=&x
\\
\rulefont{\sigma\#} &a\#x\limp& x[a \sm u]=&x
\\
\rulefont{\sigma\alpha}&b\#x\limp&x[a \sm u]=&((b\;a)\act x)[b \sm u]
\\
\rulefont{\sigma\sigma} &a\#v\limp & x[a \sm u][b \sm v]=&x[b \sm v][a \sm u[b \sm v]]
\\[2ex]
\rulefont{\amgis\sigma}&a\#v\limp& p[v \ms b][u \ms a]=&p[u[b \sm v] \ms a][v \ms b]
\end{tab0}
\end{minipage}
\caption{Nominal algebra axioms for $\sigma$ and $\protect\amgis$}
\label{fig.nom.sigma}
\label{fig.amgis}
\end{figure}

Definitions~\ref{defn.term.sub.alg},\ \ref{defn.sub.algebra},\ and~\ref{defn.bus.algebra} assemble three key technical structures (see also Definitions~\ref{defn.powamgis} and~\ref{defn.powsigma}).

\begin{defn}
\label{defn.term.sub.alg}
A \deffont{termlike $\sigma$-algebra} is a tuple $\ns X=(|\ns X|,\act,\tf{sub}_{\ns X},\tf{atm}_{\ns X})$ of:
\begin{itemize*}
\item
a nominal set $(|\ns X|,\act)$ which we may write just as $\ns X$; and
\item
an equivariant \deffont{$\sigma$-action} $\tf{sub}_{\ns X}:(\ns X\times\mathbb A\times\ns X)\equivarto\ns X$, written $x[a\sm u]_{\ns X}$ or just $x[a\sm u]$; and
\item
an equivariant injection $\tf{atm}_{\ns X}:\mathbb A\equivarto\ns X$ written $a_{\ns X}$ or just $a$,
\end{itemize*}
such that the equalities \rulefont{\sigma a}, \rulefont{\sigma id}, \rulefont{\sigma\#}, \rulefont{\sigma\alpha}, and \rulefont{\sigma\sigma} of Figure~\ref{fig.nom.sigma} hold, where $x$, $u$, and $v$ range over elements of $|\ns X|$.\footnote{Axiom \rulefont{\sigma id} might be more pedantically written as $x[a\sm a_{\ns X}]=x$.}
\end{defn}

We may omit subscripts where $\ns X$ is understood.

\begin{rmrk}
\label{rmrk.what.is.equivar.for.term}
We unpack what equivariance from Definition~\ref{defn.equivariant} means for the $\sigma$-action from Definition~\ref{defn.sub.algebra}: for every $x\in|\ns X|$,\ atom $a$, and $u\in|\ns X|$, and for every permutation $\pi$, we have that
$$
\pi\act(x[a\sm u])=(\pi\act x)[\pi(a)\sm\pi\act u] .
$$
Similarly for the equivariant $\amgis$-action in Definition~\ref{defn.bus.algebra} below.
\end{rmrk}

\begin{rmrk}
Definition~\ref{defn.term.sub.alg} is abstract.
It is an axiom system.
We use \emph{nominal} algebra, because the axioms require freshness side-conditions.

Examples of termlike $\sigma$-algebras include plenty of syntax: for instance the set of terms of first-order logic with substitution;
or the syntax of the untyped $\lambda$-calculus quotiented by $\alpha$-equivalence with capture-avoiding substitution;
or the syntax of propositional logic with quantifiers (syntax generated by $\phi::=a\mid \bot\mid\phi\limp\phi\mid\Forall{a}\phi$, with capture-avoiding substitution $[a\ssm \phi]$).

However, not all termlike $\sigma$-algebras are syntax.
For a huge class of extremely non-syntactic termlike $\sigma$-algebras, consider models of FM sets \cite{gabbay:stusun}.
\end{rmrk}

%%%%%%%%%%%%%%%%%%%%%%%%%%%%%%%%%%%%
\subsubsection{A $\sigma$-algebra}

\begin{defn}
\label{defn.sub.algebra}
A \deffont{$\sigma$-algebra} is a tuple $\ns X=(|\ns X|,\act,\ns X^\prg,\tf{sub})$ of:
\begin{itemize*}
\item
A nominal set $(|\ns X|,\act)$ which we may write just as $\ns X$.
\item
A termlike $\sigma$-algebra $\ns X^\prg$.
\item
An equivariant \deffont{$\sigma$-action} $\tf{sub}_{\ns X}:(\ns X\times\mathbb A\times\ns X^\prg)\equivarto\ns X$, written infix $x[a\sm u]_{\ns X}$ or $x[a\sm u]$.
\end{itemize*}
such that the equalities \rulefont{\sigma id}, \rulefont{\sigma\#}, \rulefont{\sigma\alpha}, and \rulefont{\sigma\sigma} of Figure~\ref{fig.nom.sigma} hold,%
\footnote{That is, the $\sigma$ axioms except \rulefont{\sigma a}, since we do not assume a function $\tf{atm}_{\ns X}$.
Axiom \rulefont{\sigma id}
can be more pedantically written as $x[a\sm a_{\ns X^\prg}]=x$.
}
where $x$ ranges over elements of $|\ns X|$ and $u$ and $v$ range over elements of $|\ns X^\prg|$.
\end{defn}

As for termlike $\sigma$-algebras, we may omit the subscript $\ns X$.
We may slightly informally say that \emph{$\ns X$ has a $\sigma$-algebra structure over $\ns X^\prg$}.

\begin{rmrk}
Every termlike $\sigma$-algebra is a $\sigma$-algebra over itself.
The canonical `interesting' example of a $\sigma$-algebra is the syntax of predicates of first-order logic, whose substitution action is not over predicates but over the termlike $\sigma$-algebra of terms.

Not all $\sigma$-algebras are syntactic.
In this paper we will see many examples of non-syntactic $\sigma$-algebras, based on the $\sigma$-powersets of Definition~\ref{defn.powsigma}.
\end{rmrk}

%%%%%%%%%%%%%%%%%%%%%%%%%%%%%%%%%%%%
\subsection{Definition of an amgis-algebra}

\begin{defn}
\label{defn.bus.algebra}
An \deffont{$\amgis$-algebra} (spoken: \deffont{amgis}-algebra) is a tuple $\ns P=(|\ns P|,\act,\ns P^\prg,\tf{amgis}_{\ns P})$ of:
\begin{itemize*}
\item
A set with a permutation action $(|\ns P|,\act)$ which we may write just as $\ns P$.
\item
A termlike $\sigma$-algebra $\ns P^\prg$.
\item
An equivariant \deffont{amgis}-action $\tf{amgis}_{\ns P}:(\ns P\times\ns P^\prg\times\mathbb A)\equivarto\ns P$, written infix $p[u\ms a]_{\ns P}$ or $p[u\ms a]$.
\end{itemize*}
such that the equality
\rulefont{\amgis \sigma} of Figure~\ref{fig.nom.sigma} holds, where $p$ ranges over elements of $|\ns P|$ and $u$ and $v$ range over elements of $|\ns P^\prg|$.
We may omit the subscript $\ns P$.
\end{defn}

\begin{rmrk}
$[u\ms a]$ looks like $[a\sm u]$ written backwards, and a casual glance at \rulefont{\amgis\sigma} suggests that it is just \rulefont{\sigma\sigma} written backwards.
This is not quite true: we have $u[b\sm v]$ on the right in \rulefont{\amgis\sigma} and not `$u[v\ms b]$' (and this would make no sense, since $\ns P^\prg$ in Definition~\ref{defn.bus.algebra} is a $\sigma$-algebra and is not equipped with an amgis-action).

Discussion of the origin of the axioms of $\amgis$-algebras is in Subsections~\ref{subsect.sigma.to.amgis} and~\ref{subsect.amgis.to.sigma}; see also Proposition~\ref{prop.amgis.2} and Subsection~\ref{subsect.further.remarks}.
\end{rmrk}

\begin{rmrk}
\label{rmrk.classes.of.amgis}
In this paper we will see three main classes of $\amgis$-algebras in the sense of Definition~\ref{defn.bus.algebra}:
\begin{itemize}
\item
$\amgis$-algebras constructed directly from a $\sigma$-powerset.
See Subsection~\ref{subsect.sigma.to.amgis}.
We do not require $p$ to have small support (and this does not prevent us from getting all the results that we need).
\item
$\amgis$-algebras underlying spaces of the form $F(\ns D)$ in Definition~\ref{defn.Ff} (which is really a specialisation of the previous case).
In Theorem~\ref{thrm.maxfilt.zorn} we may make many choices when we construct the points of these spaces, so we cannot require $p$ to have small support, and in general this will not hold.
\item
$\amgis$-algebras underlying the concrete model $\points_\Pi$ in Definition~\ref{defn.pi.point}.
These do have small support. See the discussion opening Subsection~\ref{subsect.Pi.points} and Definition~\ref{defn.qasmu} and the surrounding results.
\end{itemize}
\end{rmrk}

We conclude with three technical lemmas which will be useful later:

\begin{lemm}
\label{lemm.sub.alpha}
If $\ns X$ is a $\sigma$-algebra and $b\#x$ then $x[a\sm b]=(b\ a)\act x$.
\end{lemm}
\begin{proof}
By \rulefont{\sigma\alpha} $x[a\sm b]=((b\ a)\act x)[b\sm b]$.
We use \rulefont{\sigma id}.
\end{proof}

\begin{rmrk}
In the two papers that introduced nominal algebra \cite{gabbay:capasn,gabbay:oneaah}, Lemma~\ref{lemm.sub.alpha} was taken as an axiom (it was called \rulefont{ren{\mapsto}}) and \rulefont{\sigma id} was the lemma.
In the presence of the other axioms of substitution, the two are equivalent.
\end{rmrk}

\begin{lemm}
\label{lemm.fresh.sub}
If $a\#u$ then $a\#x[a\sm u]$.
\end{lemm}
\begin{proof}
Choose fresh $b$ (so $b\#x,u$).
By \rulefont{\sigma\alpha} $x[a\sm u]=((b\ a)\act x)[b\sm u]$.
Also by Corollary~\ref{corr.stuff}(\ref{stuff.fixsupp.fixelt}) $(b\ a)\act u=u$ and by Theorem~\ref{thrm.equivar} $(b\ a)\act(x[a\sm u])=((b\ a)\act x)[b\sm (b\ a)\act u]$.
We put this all together and we deduce that $(b\ a)\act (x[a\sm u])=x[a\sm u]$.
It follows by Corollary~\ref{corr.stuff}(\ref{stuff.freshness.criterion}) that $a\not\in\supp(x[a\sm u])$.
\end{proof}

\begin{lemm}
\label{lemm.sm.to.pi}
If $c\#x$ then $x[a\sm c][b\sm a][c\sm b]=(b\ a)\act x$.
\end{lemm}
\begin{proof}
We use Lemma~\ref{lemm.sub.alpha}, Proposition~\ref{prop.pi.supp}, and Corollary~\ref{corr.stuff}(\ref{stuff.two.pi}) to reason as follows:
$$
x[a\sm c][b\sm a][c\sm b]
=
((c\ a)\act x)[b\sm a][c\sm b]
=
(c\ b)\act (b\ a)\act ((c\ a)\act x)
=
(b\ a)\act x
\qedhere$$
\end{proof}

\maketab{tab1}{@{\hspace{-2em}}R{10em}@{\ }L{12em}L{12em}}
\maketab{tab2}{@{\hspace{-2.3em}}R{10em}@{\ }L{14em}L{12em}}
\maketab{tab2a}{@{\hspace{-4.3em}}R{10em}@{\ }L{14em}L{8em}}
\maketab{tab2b}{@{\hspace{-4.3em}}R{12em}@{\ }L{10em}L{12em}}
\maketab{tab2r}{@{\hspace{-2.3em}}R{12em}@{\ }L{12em}L{12em}}
\maketab{tab2r2}{@{\hspace{-2.3em}}R{12em}@{\ }L{11em}L{9em}}
\maketab{tab2rr}{@{\hspace{-3.3em}}R{12em}@{\ }L{8em}L{10.9em}}
\maketab{tab2rrr}{@{\hspace{-4em}}R{12em}@{\ }L{10em}L{9.5em}}
\maketab{tab2d}{@{\hspace{-2em}}R{10em}@{\ }L{15em}L{12em}}
\maketab{tab2c}{@{\hspace{-2em}}R{11em}@{\ }L{9em}L{12em}}

In Subsection~\ref{subsect.sigma.to.amgis} we move from a $\sigma$-algebra to an $\amgis$-algebra using nominal powersets.
In Subsection~\ref{subsect.amgis.to.sigma} we move from an $\amgis$-algebra to a $\sigma$-algebra, again using nominal powersets.

%%%%%%%%%%%%%%%%%%%%%%%%%%%%%%%%%%%%%%%%%%%%%%%%%%%%%%%%%%%%%
\subsection{Duality I: sigma to amgis}
\label{subsect.sigma.to.amgis}

Given a $\sigma$-algebra we generate an $\amgis$-algebra out of its subsets.
This is Proposition~\ref{prop.amgis.2}.

\begin{defn}
\label{defn.p.action}
Suppose $\ns X=(|\ns X|,\act,\ns X^\prg,\tf{sub}_{\ns X})$ is a $\sigma$-algebra.

Give subsets $p\subseteq|\ns X|$ \deffont{pointwise} actions as follows:
\begin{frameqn}
\begin{array}{r@{\ }l@{\qquad}l}
\pi\act p=&\{\pi\act x\mid x\in p\}
\\
p[u\ms a]=&\{x\mid x[a\sm u]\in p\}
& u\in|\ns X^\prg|
\end{array}
\end{frameqn}
\end{defn}

\begin{prop}
\label{prop.sigma.iff}
Suppose $\ns X$ is a $\sigma$-algebra and $p\subseteq|\ns X|$ and $u\in|\ns X^\prg|$.
Then:
\begin{itemize*}
\item
$x\in \pi\act p$ if and only if $\pi^\mone\act x\in p$.
\item
$x\in p[u\ms a]$ if and only if $x[a\sm u]\in p$.
\end{itemize*}
\end{prop}
\begin{proof}
By easy calculations on the pointwise actions in Definition~\ref{defn.p.action}.
\end{proof}

\begin{defn}
\label{defn.powamgis}
Suppose $\ns X$ is a $\sigma$-algebra.
Define the \deffont{$\amgis$-powerset} algebra $\powamgis(\ns X)$ by setting:
\begin{itemize*}
\item
$|\powamgis(\ns X)|$ is the set of subsets $p\subseteq|\ns X|$ with permutation action $\pi\act p$ following Definition~\ref{defn.p.action}.
\item
$(\powamgis(\ns X))^\prg = \ns X^\prg$.
\item
The amgis-action $p[u\ms a]$ follows Definition~\ref{defn.p.action}.
\end{itemize*}
\end{defn}

\begin{rmrk}
In Definition~\ref{defn.powamgis} $p$ is \emph{not} necessarily small-supported.
We will need this: the construction of Theorem~\ref{thrm.maxfilt.zorn} will generate $p$ with potentially large support.
\end{rmrk}

\begin{prop}
\label{prop.amgis.2}
If $\ns X$ is a $\sigma$-algebra then $\powamgis(\ns X)$ from Definition~\ref{defn.powamgis} is an $\amgis$-algebra.
\end{prop}
\begin{proof}
By Theorem~\ref{thrm.equivar} the operations are equivariant.
We verify rule
\rulefont{\amgis\sigma} from Figure~\ref{fig.amgis}:
\begin{itemize*}
\item
\emph{Property \rulefont{\amgis\sigma}.}\quad
Suppose $a\#v$.
Then:
$$
\begin{array}[b]{r@{\ }l@{\quad}l}
x\in p[v\ms b][u\ms a]\liff &x[a{\sm}u][b{\sm}v]\in p
&\text{Proposition~\ref{prop.sigma.iff}}
\\
\liff& x[b{\sm}v][a{\sm}u[b{\sm}v]]\in p
&\rulefont{\sigma\sigma},\ a\#v
\\
\liff& x\in p[u[b{\sm}v]\ms a][v\ms b]
&\text{Proposition~\ref{prop.sigma.iff}}
\end{array}
\qedhere$$
\end{itemize*}
\end{proof}

\maketab{tab6}{@{\hspace{1em}}R{10em}@{\ }L{12em}L{12em}}
\maketab{tab7}{@{\hspace{1em}}R{12em}@{\ }L{12em}}

%%%%%%%%%%%%%%%%%%%%%%%%%%%%%%%%%%%%%%%%%%
\subsection{Duality II: amgis to sigma}
\label{subsect.amgis.to.sigma}

\subsubsection{The pointwise sigma-action on subsets of an amgis-algebra}

\begin{defn}
\label{defn.sub.sets}
Suppose $\ns P=(|\ns P|,\act,\ns P^\prg,\tf{amgis}_{\ns P})$ is an $\amgis$-algebra.
Give subsets $X\subseteq|\ns P|$ \deffont{pointwise} actions as follows:
\begin{frameqn}
\begin{array}{r@{\ }l@{\qquad}l}
\pi\act X=&\{\pi\act p \mid p\in X\}
\\
X[a{\sm}u]=&\{p \mid \New{c} p[u\ms c]\in (c\ a)\act X\}
&u\in|\ns P^\prg|
\end{array}
\end{frameqn}
\end{defn}

\begin{rmrk}
\label{rmrk.not.just.flipping.over}
We call the action in Definition~\ref{defn.sub.sets} \emph{pointwise}.
We should note:
\begin{itemize*}
\item
In the case of $\pi\act X$ this is true in the conventional sense; see Proposition~\ref{prop.amgis.iff}(3).
\item
However, in the case of $X[a{\sm}u]$ this is true only subject to a freshness side-condition that is encoded in the use of $\new$ above; see Proposition~\ref{prop.amgis.iff}(1\&2).

The extra technical machinery of the $\new c$ and the freshening permutation $(c\ a)$ is needed to give us Lemma~\ref{lemm.sigma.alpha}.
\end{itemize*}
\end{rmrk}

\begin{prop}
\label{prop.amgis.iff}
Suppose $\ns P$ is an $\amgis$-algebra and $X\subseteq|\ns P|$.
Suppose $p\in|\ns P|$ and $u\in|\ns P^\prg|$. % and $a\#u$.
Then:
\begin{enumerate*}
\item
$p\in X[a\sm u]$ if and only if $\New{c}p[u\ms c]\in (c\ a)\act X$.
\item
If furthermore $p\in|\ns P|$ has small support
and $a\#u,p$, then we can simplify part~1 of this result to $p\in X[a\sm u]$ if and only if $p[u\ms a]\in X$.
\item
$p\in \pi\act X$ if and only if $\pi^\mone\act p\in X$.
\end{enumerate*}
\end{prop}
\begin{proof}
\begin{enumerate*}
\item
Direct from Definition~\ref{defn.sub.sets}.
\item
Suppose $a\#u,p$.
From part~1 of this result $p\in X[a\sm u]$ if and only if $\New{c}p[u\ms c]\in (c\ a)\act X$.
By Corollary~\ref{corr.stuff}(\ref{stuff.fixsupp.fixelt}) $(c\ a)\act u=u$ and $(c\ a)\act p=p$, so (applying $(c\ a)$ to both sides of the equality) this is if and only if $\New{c}p[u\ms a]\in X$, which means that $p[u\ms a]\in X$.
\item
Direct from Theorem~\ref{thrm.equivar}.
\qedhere\end{enumerate*}
\end{proof}

\begin{lemm}[$\alpha$-equivalence]
\label{lemm.sigma.alpha}
Suppose $\ns P$ is an $\amgis$-algebra and $X\subseteq|\ns P|$. %has small support.
Then
$$
b\#X\quad\text{implies}\quad X[a\sm u]=((b\ a)\act X)[b\sm u].
$$
\end{lemm}
\begin{proof}
By part~1 of Proposition~\ref{prop.amgis.iff} $p\in X[a\sm u]$ if and only $\New{c}p[u\ms c]\in (c\ a)\act X$, and $p\in ((b\ a)\act X)[b\sm u]$ if and only if $\New{c}p[u\ms c]\in (c\ b)\act((b\ a)\act X)$.
By Corollary~\ref{corr.stuff}(\ref{stuff.fixsupp.fixelt}) $(c\ a)\act X=(c\ b)\act ((b\ a)\act X)$ since $b\#X$.
The first part follows.
\end{proof}

Proposition~\ref{prop.sub.sub} is useful, amongst other things, in Lemma~\ref{lemm.pow.closed}.
On syntax it is known as the \emph{substitution lemma}, but here it is about an action on sets $X$, and the proof is different:
\begin{prop}
\label{prop.sub.sub}
Suppose $\ns P$ is an $\amgis$-algebra and $X\subseteq|\ns P|$. % has small support.
Suppose $u,v\in|\ns P^\prg|$.
Then
$$
a\#v\quad\text{implies}\quad
X[a{\sm}u][b{\sm}v]=X[b{\sm}v][a{\sm}u[b{\sm}v]].
$$
\end{prop}
\begin{proof}
We reason as follows, where we write $\pi=(a'\ a)\circ(b'\ b)$:
$$
\begin{array}[b]{@{\hspace{-1ex}}r@{\ }l@{\quad}l}
p\in X[a\sm u][b\sm v]\liff &\New{a',b'}p[v\ms b'][(b'\ b)\act u\ms a']\in \pi\act X &\text{Proposition~\ref{prop.amgis.iff}} %,\ a\#u,b\#v
\\
\liff &\New{a',b'}p[((b'\ b)\act u)[b'\sm v]\ms a'][v\ms b']\in \pi\act X &\rulefont{\amgis\sigma}\ a'\#v
\\
\liff &\New{a',b'}p[u[b\sm v]\ms a'][v\ms b']\in \pi\act X &\rulefont{\sigma\alpha}\ b'\#u
\\
\liff &\New{a'}p[u[b\sm v]\ms a']\in ((a'\ a)\act X)[b\sm v] &\text{Proposition~\ref{prop.amgis.iff}}
\\
\liff &\New{a'}p[u[b\sm v]\ms a']\in ((a'\ a)\act X)[b\sm (a'\ a)\act v] &\text{Cor~\ref{corr.stuff}}\ a',a\#v
\\
\liff &\New{a'}p[u[b\sm v]\ms a']\in (a'\ a)\act (X[b\sm v]) &\text{Theorem~\ref{thrm.equivar}}
\\
\liff &p\in X[b\sm v][a\sm u[b\sm v]] &\text{Proposition~\ref{prop.amgis.iff}}
\end{array}
\qedhere$$
\end{proof}

%%%%%%%%%%%%%%%%%%%%%%%%
\subsubsection{The $\sigma$-powerset $\powsigma(\ns P)$}

Recall from Subsection~\ref{subsect.finsupp.pow} the \emph{small-supported powerset} $\nompow(\ns X)$ of a nominal set $\ns X$.
\begin{defn}
\label{defn.powsigma}
Suppose $\ns P$ is an $\amgis$-algebra.
Define the \deffont{$\sigma$-powerset} algebra $\powsigma(\ns P)$
by setting:
\begin{itemize*}
\item
$|\powsigma(\ns P)|$ is those $X\in|\nompow(\ns P)|$ with the actions $\pi\act X$ and $X[a\sm u]$ from Definition~\ref{defn.sub.sets}, satisfying conditions \ref{item.fresh.powsigma} and~\ref{item.alpha.powsigma} below.
\item
$\powsigma(\ns P)^\prg=\ns P^\prg$.
\end{itemize*}
The $X\in|\nompow(\ns P)|$ above are restricted with conditions as follows, where $u\in|\ns P^\prg|$ and $p\in|\ns P|$:
\begin{frametxt}
\begin{enumerate*}
\item
\label{item.fresh.powsigma}
$\Forall{u}\New{a}\Forall{p}(p[u\ms a]\in X\liff p\in X)$.

In other words: if $a\#X,u$ then $p[u\ms a]\in X\liff p\in X$.
\item
\label{item.alpha.powsigma}
$\Forall{a}\New{b}\Forall{p}(p[b\ms a]\in X\liff (b\ a)\act p\in X)$.

In other words: if $b\#X$ then $p[b\ms a]\in X\liff (b\ a)\act p\in X$.
\end{enumerate*}
\end{frametxt}
\end{defn}

Lemma~\ref{lemm.X.sub.fresh.alpha} reformulates the conditions of Definition~\ref{defn.powsigma} by packaging up some complexity using the $\sigma$-action on subsets of an $\amgis$-algebra from Definition~\ref{defn.sub.sets}:
\begin{lemm}
\label{lemm.X.sub.fresh.alpha}
Continuing the notation of Definition~\ref{defn.powsigma}, if $X\in|\powsigma(\ns P)|$ then
\begin{enumerate*}
\item\label{X.sub.fresh.alpha.1}
If $a\#X$ then $X[a\sm u]=X$.
\item
If $b\#X$ then $X[a\sm b]=(b\ a)\act X$.
\end{enumerate*}
\end{lemm}
\begin{proof}
\begin{enumerate*}
\item
Suppose $a\#X$.
By part~1 of Proposition~\ref{prop.amgis.iff} $p\in X[a\sm u]$ if and only if $\New{c}p[u\ms c]\in (c\ a)\act X$.
By Corollary~\ref{corr.stuff}(\ref{stuff.fixsupp.fixelt}) $(c\ a)\act X=X$ and
by condition~\ref{item.fresh.powsigma} of Definition~\ref{defn.powsigma} $p[u\ms c]\in X$ if and only if $p\in X$,
so this is if and only if $\New{c}(p\in X)$, that is $p\in X$.
\item
We combine Proposition~\ref{prop.amgis.iff} with condition~\ref{item.alpha.powsigma} of Definition~\ref{defn.powsigma}, since $a\#b$.
\qedhere\end{enumerate*}
\end{proof}

We continue the discussion from Remark~\ref{rmrk.not.just.flipping.over}:
\begin{rmrk}
Condition~\ref{item.fresh.powsigma} of Definition~\ref{defn.powsigma} corresponds to part~\ref{X.sub.fresh.alpha.1} of Lemma~\ref{lemm.X.sub.fresh.alpha} for $i{\in}\{1,2\}$;
Definition~\ref{defn.powsigma} uses the $\amgis$-action and Lemma~\ref{lemm.X.sub.fresh.alpha} uses the corresponding $\sigma$-action for sets from Definition~\ref{defn.sub.sets}.

Lemma~\ref{lemm.X.sub.fresh.alpha} is more than just a rephrasing of Definition~\ref{defn.powsigma},
because of the freshness conditions discussed in Remark~\ref{rmrk.not.just.flipping.over}: condition~\ref{item.fresh.powsigma} of Definition~\ref{defn.powsigma} insists that $a\#u$ (because we have $\forall u.\new a$), whereas
part~1 of Lemma~\ref{lemm.X.sub.fresh.alpha} does not insist on $a\#u$.

Thus, Lemma~\ref{lemm.X.sub.fresh.alpha} is slightly stronger than Definition~\ref{defn.powsigma} (harder to prove; easier to apply).

The conditions of Definition~\ref{defn.powsigma} are optimised for a situation where we want to prove that some $X$ is in $\powsigma(\ns P)$ (because we may assume $a\#u$, and so have a relatively weaker proof-obligation), whereas Lemma~\ref{lemm.X.sub.fresh.alpha} is optimised for a situation where we are told that $X\in|\powsigma(\ns P)|$ and we want to manipulate it (because we do not need to assume $a\#u$ to apply part~\ref{X.sub.fresh.alpha.1} of the Lemma).
We will freely use whichever form is most convenient in context.
\end{rmrk}

\begin{corr}
\label{corr.amgis.id.sub}
Suppose $X\in|\powsigma(\ns P)|$.
Then $X[a\sm a]=X$.
\end{corr}
\begin{proof}
Suppose $b\#X$.
By Lemma~\ref{lemm.sigma.alpha} $X[a\sm a]=((b\ a)\act X)[b\sm a]$.
By Proposition~\ref{prop.pi.supp} $a\#(b\ a)\act X$.
By part~2 of Lemma~\ref{lemm.X.sub.fresh.alpha} $((b\ a)\act X)[b\sm a]=(b\ a)\act((b\ a)\act X)=X$.
\end{proof}

\begin{lemm}
\label{lemm.pow.closed}
If $X\in|\powsigma(\ns P)|$ and $u\in|\ns P^\prg|$ then also $X[a\sm u]\in|\powsigma(\ns P)|$.

As a corollary, in Definition~\ref{defn.powsigma}, $|\powsigma(\ns P)|$ is closed under the $\sigma$-action from Definition~\ref{defn.sub.sets}.
\end{lemm}
\begin{proof}
By construction $X[a\sm u]\subseteq|\ns P|$, so we now check the properties listed in Definition~\ref{defn.powsigma}.

By assumption in Definition~\ref{defn.powsigma}, $X$ is small-supported.
Small support of $X[a\sm u]$ is from Theorem~\ref{thrm.no.increase.of.supp}.

We check the conditions of Definition~\ref{defn.powsigma} for $X[a\sm u]$:
\begin{enumerate*}
\item
\emph{For fresh $b$ (so $b\#u,X$),\ $X[a\sm u][b\sm v]=X[a\sm u]$.}\quad

We use Lemma~\ref{lemm.sigma.alpha} to assume without loss of generality that $a\#u$.
It suffices to reason as follows:
\begin{tab2rrr}
X[a\sm u][b\sm v]=&X[b\sm v][a\sm u[b\sm v]] &\text{Proposition~\ref{prop.sub.sub}},\ a\#v
\\
=&X[b\sm v][a\sm u] &\rulefont{\sigma\#},\ b\#u
\\
=&X[a\sm u] &\text{Lemma~\ref{lemm.X.sub.fresh.alpha}(1)},\ b\#X
\end{tab2rrr}
\item
\emph{For fresh $b'$ (so $b'\#u,v,X$) $X[a\sm u][b\sm b']=(b'\ b)\act(X[a\sm u])$.}\quad

It suffices to reason as follows:
\begin{tab2rrr}
X[a\sm u][b\sm b']
=&X[b\sm b'][a\sm u[b\sm b']]
&\text{Proposition~\ref{prop.sub.sub}},\ a\#b'
\\
=&((b'\ b)\act X)[a\sm (b'\ b)\act u]
&\text{Lemma~\ref{lemm.X.sub.fresh.alpha}},\ b'\#u,X
\\
=&(b'\ b)\act (X[a\sm u])
&\text{Theorem~\ref{thrm.equivar}(2)}
\qedhere\end{tab2rrr}
\end{enumerate*}
\end{proof}

\begin{prop}
\label{prop.pow.sub.algebra}
If $\ns P$ is an $\amgis$-algebra then $\powsigma(\ns P)$ (Definition~\ref{defn.powsigma}) is indeed a $\sigma$-algebra.
\end{prop}
\begin{proof}
By Lemma~\ref{lemm.pow.closed} the $\sigma$-action does indeed map to $|\powsigma(\ns P)|$.
By Theorem~\ref{thrm.equivar} so does the permutation action.
It remains to check validity of the axioms from Definition~\ref{defn.sub.algebra}.
\begin{itemize*}
\item
Axiom \rulefont{\sigma id} is Corollary~\ref{corr.amgis.id.sub}.
\item
Axiom \rulefont{\sigma\#} is part~1 of Lemma~\ref{lemm.X.sub.fresh.alpha}.
\item
Axiom \rulefont{\sigma\alpha} is Lemma~\ref{lemm.sigma.alpha}.
\item
Axiom \rulefont{\sigma\sigma} is Proposition~\ref{prop.sub.sub}.
\qedhere
\end{itemize*}
\end{proof}

%%%%%%%%%%%%%%%%%
\subsubsection{Some further remarks}
\label{subsect.further.remarks}

The $\amgis$-action $p[u\ms a]$ in Definition~\ref{defn.p.action} is the \emph{functional preimage} of an underlying $\sigma$-action, and the $\sigma$-action in $X[a\sm u]$ in Definition~\ref{defn.sub.sets} is obtained by first binding $a$ using a $\new$-quantifier---and then taking the functional preimage of an underlying $\amgis$-action.
This binding is designed to make $\alpha$-equivalence (Lemma~\ref{lemm.sigma.alpha}) a structural fact of the definition---that is, \rulefont{\sigma\alpha} from Figure~\ref{fig.nom.sigma}.

We can suggest the following intuitions for the amgis- and sigma-actions from Definitions~\ref{defn.p.action} and~\ref{defn.sub.sets}:
\begin{itemize*}
\item
$p[u\ms a]$ is intuitively ``$p$, reprogrammed to believe that $a$ is equal to $u$''.
\item
$X[a\sm u]$ is intuitively ``$X$, reprogrammed to believe that $a$ is equal to $u$, then bind $a$''.
\end{itemize*}

In the case that $\ns P$ is small-supported, so that every $p\in|\ns P|$ has small support, then Definition~\ref{defn.p.action} can be simplified as described in part~2 of Proposition~\ref{prop.amgis.iff}.
This version appeared as a definition in \cite{gabbay:stodfo,gabbay:semooc}.
We arrived at Definition~\ref{defn.sub.sets} as a modification and generalisation of the first definition to the case where we cannot assume that points have small support (because of the infinitely many choices made in Theorem~\ref{thrm.maxfilt.zorn}).

Interestingly, only \rulefont{\sigma\sigma} comes directly from the structure of the underlying $\amgis$-algebra (from \rulefont{\amgis\sigma}).
Other axioms are forced---\rulefont{\sigma\alpha} from the definition (Lemma~\ref{lemm.sigma.alpha}), and \rulefont{\sigma\#} and \rulefont{\sigma\id} from conditions~\ref{item.fresh.powsigma} and~\ref{item.alpha.powsigma} in Definition~\ref{defn.powsigma}.

%%%%%%%%%%%%%%%%%%%%%%%%%%%%%%%%%%%%%%%
\section{Nominal posets}
\label{sect.nom.pow}

%%%%%%%%%%%%%%%%%%%%%%%%%%%%%%%%%%%%%%%%%%%%%%%%%
\subsection{Nominal posets and fresh-finite limits}
\label{subsect.fresh-finite.limit}

\begin{defn}
\label{defn.nom.poset}
A \deffont{nominal poset} is a tuple $\mathcal L=(|\mathcal L|,\act,\leq)$ where
\begin{enumerate*}
\item
$(|\mathcal L|,\act)$ is a nominal set, and
\item\label{nom.poset.leq.equivar}
The relation $\leq\ \subseteq|\mathcal L|{\times}|\mathcal L|$ is an equivariant partial order.\footnote{So $\leq$ is transitive, reflexive, and antisymmetric, and $x\leq y$ if and only if $\pi\act x\leq\pi\act y$.}
\end{enumerate*}
\end{defn}

\begin{defn}
\label{defn.fresh.finite.limit}
Say a nominal poset $\mathcal L$ is \deffont{finitely fresh-complete} or has \deffont{fresh-finite limits} when:
\begin{itemize*}
\item
$\mathcal L$ has a top element $\ttop$.
\item
$\mathcal L$ has conjunctions $x\tand y$ (a greatest lower bound for $x$ and $y$).
\item
$\mathcal L$ has $a$-fresh limits $\freshwedge{a}x$, where $\freshwedge{a}x$ is greatest amongst elements $x'$ such that $x'\leq x$ and $a\#x'$.
\end{itemize*}
Say $\mathcal L$ is \deffont{finitely cocomplete}\footnote{There is also a notion of finitely fresh-cocomplete, but we will not need it.}
or say it has \deffont{finite colimits} when:
\begin{itemize*}
\item
$\mathcal L$ has a bottom element $\tbot$.
\item
$\mathcal L$ has disjunctions $x\tor y$ (a least upper bound for $x$ and $y$).
\end{itemize*}
\end{defn}

Lemmas~\ref{lemm.freshwedge.alpha}, \ref{lemm.comp.unique}, and~\ref{lemm.tand.tall} will be useful later:
\begin{lemm}
\label{lemm.freshwedge.alpha}
If $b\#x$ then $\freshwedge{a}x=\freshwedge{b}(b\ a)\act x$.
\end{lemm}
\begin{proof}
By assumption $a\#\freshwedge{a}x$ and by Theorem~\ref{thrm.no.increase.of.supp} also $b\#\freshwedge{a}x$, so by Corollary~\ref{corr.stuff}(\ref{stuff.fixsupp.fixelt}) $\freshwedge{a}x=(b\ a)\act\freshwedge{a}x$.
By Theorem~\ref{thrm.equivar}(\ref{equivar.func}) $(b\ a)\act\freshwedge{a}x=\freshwedge{b}(b\ a)\act x$.
\end{proof}

\begin{lemm}
\label{lemm.comp.unique}
$\ttop$,
$\tbot$, $x\tand y$, $x\tor y$, and $\freshwedge{a}x$ are unique if they exist.
\end{lemm}
\begin{proof}
Since for a partial order, $x\leq y$ and $y\leq x$ imply $x=y$.
\end{proof}

\begin{lemm}
\label{lemm.tand.tall}
$\tall a.(x_1\tand\dots\tand x_n)=(\tall a.x_1)\tand\dots\tand(\tall a.x_n)$.
\end{lemm}
\begin{proof}
Both the left-hand and right-hand sides specify a greatest element $z$ such that $a\#z$ and $z\leq x_i$ for $1{\leq}i{\leq}n$.
\end{proof}

\begin{lemm}
$\tall a.\tall b.x=\tall b.\tall a.x$.
\end{lemm}
\begin{proof}
Both the left-hand and right-hand sides specify a greatest element $z$ such that $a\#z$ and $b\#z$ and $z\leq x$.
\end{proof}

\begin{lemm}
\label{lemm.tall.monotone}
Suppose $\mathcal L$ is a finitely fresh-complete nominal poset. %with a monotone $\sigma$-action.
Then $x\leq x'$ implies $\tall a.x\leq \tall a.x'$.
\end{lemm}
\begin{proof}
Suppose $x\leq x'$.
By assumption $\tall a.x\leq x$ and $a\#\tall a.x$.
But then also $\tall a.x\leq x'$ and $a\#\tall a.x$, so $\tall a.x\leq\tall a.x'$.
\end{proof}

\begin{nttn}
We may write $\tall a_1,\dots,a_n.x$ for $\tall a_1.\dots.\tall a_n.x$.
\end{nttn}

%%%%%%%%%%%%%%%%%%%%%%%%%%%%%%%%%%%%%%%%%%%%%%%%%%%
\subsection{Compatible $\sigma$-structure}

\begin{defn}
\label{defn.fresh.continuous}
Say that a finitely fresh-complete and finitely cocomplete nominal poset $\mathcal L=(|\mathcal L|,\act,\leq)$ has a \deffont{compatible} $\sigma$-algebra structure when it is also a $\sigma$-algebra $(|\mathcal L|,\act,\nslprg,\tf{sub}_{\ns L})$ and in addition
\begin{frameqn}
\begin{array}{l@{\ }r@{\ }l@{\qquad}r@{\ }l}
&(x\tand y)[a\sm u]=&(x[a\sm u])\tand(y[a\sm u])
\\
&(x\tor y)[a\sm u]=&(x[a\sm u])\tor (y[a\sm u])
\\
b\#u\limp&(\freshwedge{b}x)[a\sm u]=&\freshwedge{b}(x[a\sm u])
\end{array}
\end{frameqn}
where $x,y\in|\mathcal L|$ and $u\in|\nslprg|$, where $\tand$, $\tor$, and $\tall b$ exist.

Call the $\sigma$-action \deffont{monotone} when
\begin{frameqn}
x\leq y \quad\text{implies}\quad x[a\sm u]\leq y[a\sm u] .
\end{frameqn}
\end{defn}

\begin{lemm}
\label{lemm.sigma.monotone}
Continuing Definition~\ref{defn.fresh.continuous}, if the $\sigma$-structure is compatible then it is monotone.
\end{lemm}
\begin{proof}
It is a fact that $x\leq y$ if and only if $x\land y=x$.
The result follows.
\end{proof}

Our main source of nominal posets with a compatible or monotone $\sigma$-action, in this paper, will be the nominal distributive lattices with $\tall$ of the later Definition~\ref{defn.FOLeq}.

\begin{lemm}
\label{lemm.fresh.glb.sub}
Suppose $\mathcal L$ is a nominal poset with a monotone $\sigma$-action,
and suppose $a\#z$ for $z\in|\mathcal L|$.

Then if $z\leq x$ then $z$ is a lower bound for $\{x[a\sm u]\mid u\in|\nslprg|\}$, and as a particular corollary,
$$
\tall a.x\leq x[a\sm u]
$$
for every $x\in|\mathcal L|$ and $u\in|\nslprg|$.
\end{lemm}
\begin{proof}
By monotonicity $z[a\sm u]\leq x[a\sm u]$ for every $u\in|\nslprg|$.
By \rulefont{\sigma\#} also $z=z[a\sm u]$.
The corollary follows just noting that by the definition of fresh-finite limit in Definition~\ref{defn.fresh.finite.limit}, $\tall a.x\leq x$ and $a\#\tall a.x$.
\end{proof}

\begin{lemm}
\label{lemm.a.fresh.bigset}
$a\#\{x[a\sm u]\mid u\in|\nslprg|\}$ and $a\#\{x[a\sm n_{\nslprg}]\mid n\in\mathbb A\}$.
\end{lemm}
\begin{proof}
We use Corollary~\ref{corr.stuff}(\ref{stuff.freshness.criterion}).
Choose fresh $b$ (so $b\#x$).
Then we reason as follows:
$$
\begin{array}{r@{\ }l@{\qquad}l}
(b\ a)\act \{x[a\sm u]\mid u\in|\nslprg|\}
=&
\{(b\ a)\act (x[a\sm u])\mid u\in|\nslprg|\}
&\text{Pointwise action}
\\
=&
\{((b\ a)\act x)[b\sm (b\ a)\act u]\mid u\in|\nslprg|\}
&\text{Theorem~\ref{thrm.equivar}}
\\
=&
\{((b\ a)\act x)[b\sm u]\mid u\in|\nslprg|\}
&\pi\act|\nslprg|=|\nslprg|
\\
=&
\{x[a\sm u]\mid u\in|\nslprg|\}
&\rulefont{\sigma\alpha},\ b\#x
\end{array}
$$
The reasoning for $a\#\{x[a\sm n_{\nslprg}]\mid n\in\mathbb A\}$ is similar.
\end{proof}

In Definition~\ref{defn.fresh.finite.limit} we characterised universal quantification as a fresh-finite limit.
However, in the presence of a $\sigma$-action we also have an intuition that universal quantification is an infinite intersection (so $\freshwedge{a}x$ should mean `$x[a\sm u]$ for every $u$').
Proposition~\ref{prop.char.freshwedge} makes this intuition formal:
\begin{prop}
\label{prop.char.freshwedge}
Suppose $\mathcal L$ is a nominal poset with a monotone $\sigma$-action (Definition~\ref{defn.fresh.continuous}), and suppose $x\in|\mathcal L|$.
Then:
\begin{enumerate*}
\item
If $\freshwedge{a}x$ exists in $\mathcal L$ then so does $\bigwedge_{u{\in}|\nslprg|} x[a\sm u]$ the limit for $\{x[a\sm u]\mid u\in|\nslprg|\}$, and they are equal.
In symbols:
$$
\freshwedge{a}x = \bigwedge_{u{\in}|\nslprg|} x[a\sm u].
$$
\item
If $\bigwedge_u x[a\sm u]$ exists in $\mathcal L$ then so does $\freshwedge{a}x$, and they are equal.
\end{enumerate*}
\end{prop}
\begin{proof}
By Lemma~\ref{lemm.fresh.glb.sub} $\freshwedge{a}x$ is a lower bound for $\{x[a\sm u]\mid u\in|\nslprg|\}$.

Now suppose $z$ is any other lower bound, that is: $z\leq x[a\sm u]$ for every $u\in|\nslprg|$.
Note that we do not know \emph{a priori} that $a\#z$.

Choose $b$ fresh (so $b\#z,x$) and take $u=b$.
Then $z\leq x[a\sm b]\stackrel{\text{Lem~\ref{lemm.sub.alpha}}}{=} (b\ a)\act x$.
Since $b\#z$ it follows that $z\leq\freshwedge{b}(b\ a)\act x\stackrel{\text{Lem.~\ref{lemm.freshwedge.alpha}}}{=}\freshwedge{a}x$.
So $\freshwedge{a}x = \bigwedge_u x[a\sm u]$.

Now suppose that $\bigwedge_u x[a\sm u]$ exists.
By Lemma~\ref{lemm.a.fresh.bigset} and part~2 of Theorem~\ref{thrm.no.increase.of.supp} we have that $a\#\bigwedge_u x[a\sm u]$.
Also by assumption $\bigwedge_u x[a\sm u]\leq x[a\sm a]\stackrel{\rulefont{\sigma id}}{=} x$.
Thus $\bigwedge_u x[a\sm u]$ is an $a\#$lower bound for $x$.

Now suppose $z\leq x$ and $a\#z$; we need to show that $z\leq \bigwedge_u x[a\sm u]$.
This is direct from Lemma~\ref{lemm.fresh.glb.sub}.
\end{proof}

We also mention a characterisation of $\freshwedge{a}x$ using a `smaller' conjunction which does not depend on most of $\nslprg$ (see also Remark~\ref{rmrk.prop.ii.remarkable}):
\begin{prop}
\label{prop.char.freshwedge.names}
Suppose $\mathcal L$ is a nominal poset with a monotone $\sigma$-action (Definition~\ref{defn.fresh.continuous}) and $x\in|\mathcal L|$.
Then:
\begin{enumerate*}
\item
If $\freshwedge{a}x$ exists in $\mathcal L$ then so does $\bigwedge_{n} x[a\sm n]$ where $n$ ranges over all atoms, and they are equal.
\item
If $\bigwedge_n x[a\sm n]$ exists in $\mathcal L$ then so does $\freshwedge{a}x$, and they are equal.\footnote{More fully this is $\bigwedge_n x[a\sm \tf{atm}_{\nslprg}(n)]$.  See the notation in Definition~\ref{defn.term.sub.alg}.}
\end{enumerate*}
\end{prop}
\begin{proof}
The proof is just like the proof of Proposition~\ref{prop.char.freshwedge}, but we are careful and still give full details.
The important point is that by Lemma~\ref{lemm.a.fresh.bigset} and part~2 of Theorem~\ref{thrm.no.increase.of.supp} we have $a\#\bigwedge_n x[a\sm n]$.

Using Lemma~\ref{lemm.fresh.glb.sub} we see that $\freshwedge{a}x$ is a lower bound for $\{x[a\sm n]\mid n{\in}\mathbb A\}$.

Now suppose $z$ is any other lower bound, that is: $z\leq x[a\sm n]$ for every $n\in\mathbb A$.
Choose $b$ fresh (so $b\#z,x$) and take $n=b$.
Then $z\leq x[a\sm b]\stackrel{\text{Lem~\ref{lemm.sub.alpha}}}{=} (b\ a)\act x$.
Since $b\#z$ it follows that $z\leq\freshwedge{b}(b\ a)\act x\stackrel{\text{Lem.~\ref{lemm.freshwedge.alpha}}}{=}\freshwedge{a}x$.
So $\freshwedge{a}x = \bigwedge_n x[a\sm n]$.

Now suppose that $\bigwedge_n x[a\sm n]$ exists.
By Lemma~\ref{lemm.a.fresh.bigset} and part~2 of Theorem~\ref{thrm.no.increase.of.supp} we have that $a\#\bigwedge_n x[a\sm n]$.
Also by assumption $\bigwedge_n x[a\sm n]\leq x[a\sm a]\stackrel{\rulefont{\sigma id}}{=} x$.
Thus $\bigwedge_n x[a\sm n]$ is an $a\#$lower bound for $x$.

Now suppose $z\leq x$ and $a\#z$; we need to show that $z\leq \bigwedge_n x[a\sm n]$.
This is direct from Lemma~\ref{lemm.fresh.glb.sub}.
\end{proof}

\begin{rmrk}
\label{rmrk.prop.ii.remarkable}
Proposition~\ref{prop.char.freshwedge.names} is important:
later on when we consider morphisms, we will need to know that a map $g^\mone$ in Proposition~\ref{prop.G.funct} commutes with quantification.
But $g^\mone$ can change the underlying termlike $\sigma$-algebra---that is, the domain of substitution and of quantification $\nslprg$ can change.

Could this interfere with the universal quantifier, since extra elements might make what was a true universal quantification, into a false one?
If we look at Proposition~\ref{prop.char.freshwedge}, it seems that this might be the case; Proposition~\ref{prop.char.freshwedge.names} says that \emph{this cannot happen}.

It asserts that the notion of morphism used in Proposition~\ref{prop.G.funct} must take any extra elements into account.
In the language of \cite{selinger:lamca}, Proposition~\ref{prop.char.freshwedge.names} implies that our models are \emph{well-pointed}.

A dedicated discussion of such issues in the context of nominal models of the $\lambda$-calculus is in \cite[Subsection~3.4]{gabbay:nomhss}.
More on this in Remark~\ref{rmrk.the.structure}.
\end{rmrk}

%%%%%%%%%%%%%%%%%%%%%%%%%%%%%%%%%%%%%%%%
\subsection{Definition of a nominal distributive lattice with $\tall$}

\begin{defn}
\label{defn.distrib}
Suppose $\mathcal L$ is a fresh-finitely complete and finitely cocomplete nominal poset.
Call $\mathcal L$ \deffont{distributive} when
\begin{frameqn}
\begin{array}{l@{\qquad\quad}l@{\ }r@{\ }l@{\qquad}l}
\rulefont{distrib\tand}& &x\tor (y\tand z)=&(x\tor y)\land(x\tor z)
\\
\rulefont{distrib\tall}&a\#x\limp&x\tor \freshwedge{a}y=&\freshwedge{a}(x\tor y)
&\text{for every }x,y,z\in|\mathcal L| .
\end{array}
\end{frameqn}
\end{defn}

\begin{rmrk}
Definition~\ref{defn.distrib} generalises the usual notion of distributivity; $\tor$ distributes over $\tand$ and also over $\tall a$ (subject to a typical nominal algebra freshness side-condition), which we have seen exhibited as an infinite intersection in Proposition~\ref{prop.char.freshwedge}.\footnote{A dual version of part~1 of Definition~\ref{defn.distrib} is
$x\land (y\tor z)=(x\tand y)\tor(x\tand z)$ and by a standard argument \cite[Lemma~4.3]{priestley:intlo} the two are equivalent.}

An elegant, though arguably less readable, version of Definition~\ref{defn.distrib} unifies \rulefont{distrib\tand} and \rulefont{distrib\tall} to a single axiom which we could write as $S\#x\limp x\tor\freshwedge{S}Y=\freshwedge{S}\{x\tor y\mid y{\in} Y\}$, where $S\subseteq\mathbb A$ and $Y\subseteq|\mathcal L|$ are finite.
Or in English: $\tor$ distributes over fresh-finite limits.

Note in passing that in Subsection~\ref{subsect.alg}, \rulefont{distrib\tand} and \rulefont{distrib\tall} will feature as part of a purely nominal algebraic axiomatisation of fresh-finite limits.
\end{rmrk}

\begin{frametxt}
\begin{defn}
\label{defn.FOLeq}
A \deffont{nominal distributive lattice with $\tall$} is a tuple $\ns D=(|\ns D|,\act,\leq,\ns D^\prg,\tf{sub}_{\ns D})$ such that:
\begin{itemize*}
\item
$(|\ns D|,\act,\leq)$ is a nominal poset (Definition~\ref{defn.nom.poset}).
\item
$\ns D$ has fresh-finite limits and finite colimits (Definition~\ref{defn.fresh.finite.limit}), and is distributive (Definition~\ref{defn.distrib}).
\item
$(|\ns D|,\act,\ns D^\prg,\tf{sub}_{\ns D})$ is a $\sigma$-algebra (Definition~\ref{defn.sub.algebra}), and the $\sigma$-algebra structure is compatible (Definition~\ref{defn.fresh.continuous}).
\end{itemize*}
\end{defn}
\end{frametxt}

Lemma~\ref{lemm.b.bigger} is a technical lemma which we use later in Lemma~\ref{lemm.uparrow.filter} and Proposition~\ref{prop.even.stronger}.
We mention it now as an example, since it illustrates the extra structure that nominal distributive lattices with $\tall$ have, compared with `ordinary' distributive lattices:
\begin{lemm}
\label{lemm.b.bigger}
Suppose $\ns D$ is a nominal distributive lattice with $\tall$.
Suppose $x,y,z{\in}|\ns D|$.
Then:
\begin{enumerate*}
\item
If $a\#z$ and $z\leq x$ then $z\leq \tall a.x$.

\emph{Think: ``if $\Gamma\cent\phi$ and $a$ is not free in $\Gamma$ then $\Gamma\cent\forall a.\phi$''.}
\item\label{b.bigger.2}
If $b\#z,x$ and $z\leq (b\ a)\act x$ then $z\leq \tall a.x$.

\emph{Think: ``if $\Gamma\cent\phi[b/a]$ for some $b$ not free in $\Gamma$ or $\phi$ then $\Gamma\cent\forall a.\phi$''.}
\item\label{b.bigger.3}
If $b\#z,y,x$ and $z\leq y{\tor}((b\ a)\act x)$ then $z\leq y{\tor}\tall a.x$.

\emph{Think: ``if $\Gamma\cent\phi[b/a],\Psi$ and $b$ is not free in $\Gamma$, $\phi$, or $\Psi$ then $\Gamma\cent\forall a.\phi,\Psi$''.}
\end{enumerate*}
\end{lemm}
\begin{proof}
We consider each part in turn:
\begin{enumerate*}
\item
Suppose $a\#z$ and $z\leq x$.
By Lemma~\ref{lemm.tall.monotone} $\tall a.z\leq\tall a.x$ and since $a\#z$ we have $\tall a.z=z$.
\item
Suppose $b\#z,x$ and $z\leq(b\ a)\act x$.
By part~1 of this result $z\leq\tall b.(b\ a)\act x$.
We use Lemma~\ref{lemm.freshwedge.alpha}.
\item
From part~2 of this result and condition~\ref{filter.up} of distributivity (Definition~\ref{defn.distrib}).
\qedhere\end{enumerate*}
\end{proof}

Recall the notions of $\sigma$-algebra and termlike $\sigma$-algebra from Definitions~\ref{defn.term.sub.alg} and~\ref{defn.sub.algebra}:
\begin{defn}
\label{defn.morphism.sigma.alg}
Suppose $\ns X$ and $\ns X'$ are $\sigma$-algebras.
Call a pair of functions $f=(f_{\ns X},f_{\ns X}^\prg)$ where $f_{\ns X}\in|\ns X|\to|\ns X'|$ and $f_{\ns X}^\prg\in|\ns X^\prg|\to|{\ns X'}^\prg|$ a \deffont{($\sigma$-algebra) morphism} from $\ns X$ to $\ns X'$ when:
\begin{enumerate*}
\item
\label{item.f.prg.maps.atoms.to.atoms}
$f_{\ns X}^\prg(a_{\ns X^\prg})=a_{{\ns X'}^\prg}$ (so $f$ maps atoms to atoms).
\item
$f_{\ns X}^\prg(\pi\act u)=\pi\act f_{\ns X}^\prg(u)$ and
$f_{\ns X}(\pi\act x)=\pi\act f_{\ns X}(x)$
(so $f$ is equivariant).
\item\label{item.f.commutes.with.sigma}
$f_{\ns X}^\prg(u'[a\sm u])=f_{\ns X}^\prg(u')[a\sm f_{\ns X}^\prg(u)]$ and
$f_{\ns X}(x[a\sm u])=f_{\ns X}(x)[a\sm f_{\ns X}^\prg(u)]$
(so $f$ commutes with the $\sigma$-action).
\end{enumerate*}
If $\ns X$ is termlike then we insist $\ns X=\ns X^\prg$ and we insist that $f_{\ns X}=f_{\ns X}^\prg$.

We may omit the subscripts, writing for instance $f$ and $f^\prg$, or even $f=(f,f^\prg)$, where the meaning is clear.
\end{defn}

\begin{defn}
\label{defn.hom.nba}
Suppose $\ns D$ and $\ns D'$ are nominal distributive lattices with $\tall$.
Call a morphism $f=(f_{\ns D},f_{\ns D}^\prg): \ns D\to\ns D'$ of underlying $\sigma$-algebras (Definition~\ref{defn.morphism.sigma.alg}) a \deffont{morphism} of nominal distributive lattices with $\tall$ when $f$ commutes with fresh-finite limits and with finite colimits:
\begin{enumerate*}
\item
$f(\ttop)=\ttop$, and
$f(x\tand y)=f(x)\tand f(y)$ and $f(\tall a.x)=\tall a.f(x)$, and
\item
$f(\tbot)=\tbot$ and $f(x\tor y)=f(x)\tor f(y)$.
\end{enumerate*}
Write $\ndia$ for the category of nominal distributive lattices with $\tall$ and morphisms between them.
\end{defn}

%%%%%%%%%%%%%%%%%%%%%%%%%%%%%%%%%%%%%%%%%%%%%%%%%%%%
\subsection{Impredicative nominal distributive lattices}
\label{subsect.impredicative.nom.dist.lat}

We are interested in modelling the $\lambda$-calculus, so we care about lattices where the substitution action is over \emph{itself}.
Therefore we introduce \emph{impredicative} nominal distributive lattices with $\tall$: this is Definition~\ref{defn.D.impredicative}.

Recall the notion of termlike $\sigma$-algebra from Definition~\ref{defn.term.sub.alg}, and the notion of a nominal distributive lattice with $\tall$ from Definition~\ref{defn.FOLeq}.

\begin{frametxt}
\begin{defn}
\label{defn.D.impredicative}
An \deffont{impredicative} nominal distributive lattice with $\tall$ is a tuple $(\ns D,\prg_{\ns D})$ where:
\begin{enumerate*}
\item
$\ns D\in\ndia$ is a nominal distributive lattice with $\tall$ (Definition~\ref{defn.FOLeq}).
\item
$(\prg_{\ns D},\f{id}):\ns D^\prg\to\ns D$ is a morphism of $\sigma$-algebras (Definition~\ref{defn.morphism.sigma.alg}).
\item\label{size.limit}
$|\ns D|$ has cardinality no greater than $\size(\mathbb A)$ (Definition~\ref{defn.atoms})---in other words, there are no more programs than there are names.
\end{enumerate*}
\end{defn}
\end{frametxt}

The interested reader can find Definition~\ref{defn.FOLeq} extended with further structure in Definition~\ref{defn.FOLeq.pp}.

\begin{rmrk}
\label{rmrk.size.issues}
The dual definition to Definition~\ref{defn.D.impredicative} is Definition~\ref{defn.impredicative.top}.

So $\ns D$ is impredicative when the $u\in |\ns D^\prg|$ can be viewed as a subset of the $x\in |\ns D|$.
We use an explicit casting function $\prg_{\ns D}$ to do this.\footnote{This costs notation; casting functions always do. Would it be simpler to take $\prg\ns U\subseteq\ns U$ as a literal subset inclusion?  At this stage it probably would---but when we consider $\amgis$-algebras, and then dualities, an explicit casting function gives cleaner results, precisely because our constructions do not need to maintain a literal subset inclusion.}

Thus given $u\in|\ns D^\prg|$, we can obtain $\prg_{\ns D} u\in|\ns D|$ and so write (for instance) $(\prg_{\ns D} u)[a\sm u]$.
This is not quite $\lambda$-calculus self-application, but we are moving in that direction.

We use the size limit (condition~\ref{size.limit} of Definition~\ref{defn.D.impredicative}) in Theorem~\ref{thrm.maxfilt.zorn}.
The intuition for why is that when we come to build prime filters we will need to `name' every element of $\ns D$ with an atom; so we need to make sure that we will not run out.
This is not precisely true, but it captures the spirit of the proof.

More on this in Subsection~\ref{subsect.cardinality}.
\end{rmrk}

\begin{nttn}
\label{nttn.impredicative.D}
We introduce some notation for Definition~\ref{defn.D.impredicative}:
\begin{itemize*}
\item
We may write $\prg_{\ns D}$ for $(\prg_{\ns D},\f{id})$.
\item
We may drop subscripts and write $\prg u$ for $\prg_{\ns D} u$ where $u\in|\ns D^\prg|$.
\item
We may write $\prg a$ for $\prg_{\ns D}(a_{\ns D^\prg})$ where $a_{\ns D^\prg}$ is itself shorthand for $\tf{atm}_{\ns D^\prg}(a)$ from Definition~\ref{defn.term.sub.alg}.\footnote{Atoms get mapped into $\ns D^\prg$ by $\tf{atm}_{\ns D^\prg}$, and $\ns D^\prg$ gets mapped into $\ns D$ by $\prg_\ns D$ \dots so atoms get mapped into $\ns D$.}
\item
We may write $\prg\ns D$ for $\{\prg u\mid u\in|\ns D^\prg|\}\subseteq|\ns D|$ and call this set the \deffont{programs} of $\ns D$.
\end{itemize*}
\end{nttn}

\begin{rmrk}
\label{rmrk.explain.prg.a}
It might help to break down the notation a little:
\begin{itemize*}
\item
$\mathbb A$ injects into $\ns D^\prg$ via an injection $\tf{atm}_{\ns D^\prg}$ (this is the equivariant injection specified in Definition~\ref{defn.term.sub.alg}).
\item
$\ns D^\prg$ maps to $\ns D$ via $\prg_{\ns D}$.
\item
Thus we obtain $\prg a\in\prg{\ns D}$---an atom-as-a-program---living in a sub-$\sigma$-algebra of $\ns D$ which is an image of $\ns D^\prg$, and which we call the \emph{programs} of $\ns D$.
\end{itemize*}
\end{rmrk}

Lemma~\ref{lemm.impredicative.sigma.a} is a routine sanity check that the definitions match up sensibly.
It will be useful later:
\begin{lemm}
\label{lemm.impredicative.sigma.a}
Suppose $\ns D$ is impredicative and $u\in|\ns D^\prg|$.
Then $(\prg_{\ns D} a_{\ns D^\prg}) [a\sm u]=\prg_{\ns D} u$.
\end{lemm}
\begin{proof}
By assumption $\prg_{\ns D}$ is a morphism of $\sigma$-algebras from $D^\prg$ to $\ns D$.
By Definition~\ref{defn.morphism.sigma.alg} $(\prg_{\ns D} a_{\ns D^\prg})[a\sm u]=\prg_{\ns D}(a_{\ns D^\prg}[a\sm u])$ (using the third condition, noting that $u=\f{id}(u)$).
The result follows by \rulefont{\sigma a} from Figure~\ref{fig.nom.sigma} for $\ns D^\prg$.
\end{proof}

\begin{rmrk}
\label{rmrk.D.extra}
Definition~\ref{defn.D.impredicative} can be looked at in some interesting ways:
\begin{itemize*}
\item
$\ns D$ is impredicative when it has substitution $x[a\sm u]$ over a substructure of itself.
\item
$\ns D$ is impredicative when its quantifier $\tall a.x$ quantifies over a sub-$\sigma$-structure of $\ns D$.
Thus, $|\prg\ns D|\subseteq |\ns D|$ is the set of things we quantify over when we write $\tall a.x$, if $\ns D$ is impredicative.
\end{itemize*}
\end{rmrk}

\begin{rmrk}
The programs of $\ns D$ need not be closed under logical structure like $\tand$, $\tor$, and $\tall$.

So for instance $x,x'\in\prg{\ns D}$ does not imply $x\tor x'\in\prg{\ns D}$ and it is not necessarily the case that $\tbot\in\prg{\ns D}$, and so on.
We do not forbid this either.
\end{rmrk}

\begin{defn}
\label{defn.hom.impredicative}
Suppose $\ns D$ and $\ns D'$ are impredicative nominal distributive lattices with $\tall$.

Call $f=(f_{\ns D},f_{\ns D}^\prg):\ns D\to\ns D'$ a \deffont{morphism} in $\india$ when it is a morphism in $\ndia$ (Definition~\ref{defn.hom.nba}) and when in addition:
\begin{enumerate*}
\setcounter{enumi}{2}
\item
\label{item.third.condition}
$f_{\ns D}\circ\prg_{\ns D}=\prg_{\ns D'}\circ f_{\ns D}^\prg$.
That is,
$$
f_{\ns D}(\prg_{\ns D} u)=\prg_{\ns D'}(f_{\ns D}^\prg(u))
\ \text{ for every }\ u\in|\ns D^\prg|.
$$
\end{enumerate*}
\end{defn}

Definition~\ref{defn.india} extends Definition~\ref{defn.hom.nba}:
\begin{frametxt}
\begin{defn}
\label{defn.india}
Write $\india$ for the category of \deffont{impredicative nominal distributive lattices with $\tall$}, and morphisms between them.

As standard write $\ns D\in\india$ for ``$\ns D$ is an impredicative nominal distributive lattice with $\tall$'' and $f:\ns D\equivarto\ns D'\in\india$ for ``$\ns D,\ns D'\in\india$ and $f$ is a morphism in $\india$ from $\ns D$ to $\ns D'$''.
\end{defn}
\end{frametxt}

\begin{rmrk}
\label{rmrk.explain.prg.b}
We continue the notation of Definition~\ref{defn.india} and the discussion of Remark~\ref{rmrk.explain.prg.a}.
Suppose $f:\ns D\equivarto\ns D'\in\india$ is a morphism.

Note that $f_{\ns D}(\prg_{\ns D} a_{\ns D^\prg})=\prg_{\ns D'}(a_{{\ns D'}^\prg})$.
Informally we can say that $f$ maps atoms-as-programs (Remark~\ref{rmrk.explain.prg.a}) in $\ns D$ to themselves in $\ns D'$.
In symbols we can be even more brief:
$$
f(\prg a)=\prg a.
$$
We informally trace through how this happens.
By condition~\ref{item.f.prg.maps.atoms.to.atoms} of Definition~\ref{defn.morphism.sigma.alg} $f$ maps an atom in $\ns D^\prg$ to its incarnation in ${\ns D'}^\prg$.
By condition~\ref{item.third.condition} of Definition~\ref{defn.hom.impredicative} these are mapped to atoms-as-programs in $\ns D^\prg$ and ${\ns D'}^\prg$ respectively.
\end{rmrk}

%%%%%%%%%%%%%%%%%%%%%%%%%%%%%%%%%%%%%%%%%%%%%%%%%%%%%%%%%%%%%
\section{The $\sigma$-powerset as a nominal distributive lattice with $\tall$}
\label{sect.sigma.foleq}

We saw in Proposition~\ref{prop.pow.sub.algebra} how the nominal powerset of an $\amgis$-algebra $\ns P$ generates a $\sigma$-algebra $\powsigma(\ns P)$ (Definition~\ref{defn.powsigma}).
But powersets are also a lattice under subset inclusion, so perhaps $\powsigma(\ns P)$ has more structure?

In fact, $\powsigma(\ns P)$ is a nominal distributive lattice with $\tall$ (Definition~\ref{defn.FOLeq}).
This is Theorem~\ref{thrm.powerset}.

%%%%%%%%%%%%%%%%%%%%%%%%%%%%%%%%%%%%%%%%%%%%%%%%%%%%%%%%%%%%%%%%%
\subsection{Basic sets operations}

Suppose $\ns P=(|\ns P|,\act,\ns P^\prg,\tf{amgis}_{\ns P})$ is an $\amgis$-algebra.
Recall the nominal powerset $\nompow(\ns P)$ from Subsection~\ref{subsect.finsupp.pow}.

\begin{lemm}
\label{lemm.sub.bigcap}
Suppose $\mathcal X,\mathcal Y\subseteq|\nompow(\ns P)|$ and $X,Y\subseteq|\ns P|$.
Then:
\begin{enumerate*}
\item\label{sub.bigcap.sfs}
If $\mathcal X$ is strictly small-supported (Definition~\ref{defn.strictpow}) then
$$(\bigcap_{X{\in}\mathcal X} X)[a\sm u]=\bigcap_{X{\in}\mathcal X}(X[a\sm u]).
$$
In words: $\sigma$ commutes with strictly small-supported sets intersections.

Note by Lemma~\ref{lemm.finite.strict} that this holds in particular if $\mathcal X$ is finite.
\item\label{sub.bigcup.sfs}
If $\mathcal X$ is strictly small-supported then
$$
(\bigcup_{X{\in}\mathcal X} X)[a\sm u]=\bigcup_{X{\in}\mathcal X}(X[a\sm u]).
$$
In words: $\sigma$ commutes with strictly small-supported sets unions.

Note by Lemma~\ref{lemm.finite.strict} that this holds in particular if $\mathcal X$ is finite.
\item
For any $\mathcal X\subseteq|\nompow(\ns P)|$,\
$$
\pi\act\bigcap_{X{\in}\mathcal X} X=\bigcap_{X{\in}\mathcal X} \pi\act X
\quad\text{and}\quad
\pi\act\bigcup_{X{\in}\mathcal X} X=\bigcup_{X{\in}\mathcal X} \pi\act X.
$$
In words: intersections and unions are equivariant.
\item\label{sub.bigcap.monotone}
If $X\subseteq Y$ then $X[a\sm u]\subseteq Y[a\sm u]$.
In words: $\sigma$ is monotone (Definition~\ref{defn.fresh.continuous}).
\end{enumerate*}
\end{lemm}
\begin{proof}
For part~1 we reason as follows:
$$
\begin{array}{r@{\ }l@{\qquad}l}
p\in (\bigcap_{X{\in}\mathcal X} X)[a\sm u]
\liff&
\New{c}p[u\ms c]\in \bigcap_{X{\in}\mathcal X}(c\ a)\act X
&\text{Prop~\ref{prop.amgis.iff}, Thm~\ref{thrm.equivar}}
\\
\liff&
\New{c}\Forall{X{\in}\mathcal X}p[u\ms c]\in (c\ a)\act X
&\text{Fact}
\\[1.5ex]
p\in \bigcap_{X{\in}\mathcal X}(X[a\sm u])
\liff&
\Forall{X{\in}\mathcal X}p\in X[a\sm u]
&\text{Fact}
\\
\liff&
\Forall{X{\in}\mathcal X}\New{c}p[u\ms c]\in (c\ a)\act X
&\text{Proposition~\ref{prop.amgis.iff}}
\end{array}
$$
We note that by Lemma~\ref{lemm.strict.support}(\ref{strict.union}),\ $c\#\mathcal X$ if and only if $c\#X$ for every $X\in\mathcal X$.
This allows us to swap the $\forall$ and the $\new$ quantifiers, and the result follows.

The second and third parts are similar.
Part~4 follows from part~1 as in the proof of Lemma~\ref{lemm.sigma.monotone}.
\end{proof}

Recall the definition of $\powsigma(\ns P)$ from Definition~\ref{defn.powsigma}.
\begin{corr}
\label{corr.powsigma.nompowset.monotone.sigma}
$\powsigma(\ns P)$ ordered under subset inclusion is a nominal poset with a monotone $\sigma$-action.
\end{corr}
\begin{proof}
Subset inclusion partially orders $\powsigma(\ns P)$, which is a $\sigma$-algebra by Proposition~\ref{prop.pow.sub.algebra}.
This action is monotone by Lemma~\ref{lemm.sub.bigcap}(\ref{sub.bigcap.monotone}).
\end{proof}

\begin{lemm}
\label{lemm.pow.nu.closed}
\begin{itemize*}
\item
$\varnothing$ and $|\ns P|$ are in $|\powsigma(\ns P)|$ and these are least and greatest elements in the subset inclusion ordering.
\item
If $X$ and $Y$ are in $|\powsigma(\ns P)|$ then so are $X\cap Y$ and $X\cup Y$ and these are greatest lower bounds and least upper bounds in the subset inclusion ordering.
\end{itemize*}
\end{lemm}
\begin{proof}
We check the properties listed in Definition~\ref{defn.powsigma} for $X\cap Y$; the case of $X\cup Y$ is similar and the cases of $\varnothing$ and $|\ns P|$ are even easier.
We check that $X\cap Y$ is a greatest lower bound for $\{X,Y\}$ just as for ordinary sets.
$X\cap Y$ has small support by Theorem~\ref{thrm.no.increase.of.supp}.
\begin{enumerate*}
\item
\emph{If $a$ is fresh (so $a\#X,Y,u$) then $(X\cap Y)[a\sm u]=X\cap Y$.}
By Lemmas~\ref{lemm.sub.bigcap}(\ref{sub.bigcap.sfs}) and~\ref{lemm.X.sub.fresh.alpha} $(X\cap Y)[a{\sm}u]=(X[a{\sm}u])\cap(Y[a{\sm}u])=X\cap Y$.
\item
\emph{If $b$ is fresh (so $b\#X,Y$) then $(X\cap Y)[a\sm b]=(b\ a)\act (X\cap Y)$.}
We reason as follows:
$$
\begin{array}[b]{r@{\ }l@{\qquad}l}
(X\cap Y)[a{\sm}b]=&(X[a{\sm}b])\cap(Y[a{\sm}b])
&\text{Lemma~\ref{lemm.sub.bigcap}(\ref{sub.bigcap.sfs})}
\\
=&((b\ a)\act X)\cap((b\ a)\act Y))
&\text{Part~2 of Lemma~\ref{lemm.X.sub.fresh.alpha}}
\\
=&(b\ a)\act(X\cap Y)
&\text{Theorem~\ref{thrm.equivar}}
\end{array}
\qedhere$$
\end{enumerate*}
\end{proof}

%%%%%%%%%%%%%%%%%%%%%%%%%%%%%%%%%%%%%%%%%%%%%%%%%%%%%%
\subsection{Sets quantification}
\label{subsect.powsigma.quant}

We now explore quantification.
This is where we part company from Boolean algebras.

Suppose $\ns P=(|\ns P|,\act,\ns P^\prg,\tf{amgis}_{\ns P})$ is an $\amgis$-algebra.
Recall the definitions of $\nompow(\ns P)$ from Subsection~\ref{subsect.finsupp.pow} and of $\powsigma(\ns P)$ from Definition~\ref{defn.powsigma}.

\begin{defn}
\label{defn.nu.U}
If $X\in|\nompow(\ns P)|$ then define\footnote{We do not use the dual $\freshcup{a}X=\bigcup\{X[a\sm u]\mid u{\in}|\ns P^\prg|\}$ but investigating it would be interesting future work.  See also Subsection~\ref{subsect.fine.structure} and Appendix~\ref{subsect.existential}.}
\begin{frameqn}
\freshcap{a}X=\bigcap\{ X[a\sm u]\mid u{\in}|\ns P^\prg|\} .
\end{frameqn}
\end{defn}

We need the key technical Lemma~\ref{lemm.technical} for Proposition~\ref{prop.all.sub.commute}:
\begin{lemm}
\label{lemm.technical}
Suppose $X{\in}\nompow(\ns P)$ and $v{\in}|\ns P^\prg|$ and suppose $a\#v$.
Suppose $p{\in}|\ns P|$.
Then
$$
\begin{array}{l}
\New{b'}\Forall{u{\in}|\ns P^\prg|}\New{a'} p[v\ms b'][u\ms a']\in (b'\,b)\act(a'\,a)\act X
\quad\text{if and only if}
\\
\Forall{u{\in}|\ns P^\prg|}\New{b'}\New{a'} p[u\ms a'][v\ms b']\in (b'\,b)\act(a'\,a)\act X .
\end{array}
$$
\end{lemm}
\begin{proof}
We prove two implications:
\begin{itemize}
\item
\emph{The up-down implication.}
Assume
$$
\New{b'}\Forall{u{\in}|\ns P^\prg|}\New{a'} p[v\ms b'][u\ms a']\in (b'\,b)\act(a'\,a)\act X.
$$
Choose $u{\in}|\ns P^\prg|$.
Choose fresh $b'$ and $a'$ (so $b',a'\#X,v,u$).
Then by assumption (since $b'\#X,v$ and $a'\#X,v,u$) $p[v\ms b'][u\ms a']\in (b'\,b)\act(a'\,a)\act X$
so that by \rulefont{\amgis\sigma} of Figure~\ref{fig.amgis} (since $a'\#v$) $p[u[b'\sm v]\ms a'][v\ms b']\in (b'\,b)\act(a'\,a)\act X$.
Now by \rulefont{\sigma\#} $u[b'\sm v]=u$ (since $b'\#u$).
Therefore $p[u\ms a'][v\ms b']\in (b'\,b)\act(a'\,a)\act X$.
\item
\emph{The down-up implication.}
Assume
$$\Forall{u{\in}|\ns P^\prg|}\New{b'}\New{a'} p[u\ms a'][v\ms b']\in (b'\,b)\act(a'\,a)\act X.
$$
Choose fresh $b'$ (so $b'\#X,v$), choose $u{\in}|\ns P^\prg|$ (for which $b$ need not necessarily be fresh), and choose fresh $a'$ (so $a'\#X,v,u$).
By Lemma~\ref{lemm.fresh.sub} $b'\#u[b'\sm v]$ (since $b'\#v$).
Therefore $p[u[b'\sm v]\ms a'][v\ms b']\in (b'\,b)\act(a'\,a)\act X$ and by \rulefont{\amgis\sigma} of Figure~\ref{fig.amgis} (since $a'\#v$) $p[v\ms b'][u\ms a']\in (b'\,b)\act(a'\,a)\act X$.
\qedhere
\end{itemize}
\end{proof}

We cannot use Lemma~\ref{lemm.sub.bigcap}(\ref{sub.bigcap.sfs}) to derive Proposition~\ref{prop.all.sub.commute} because $\{X[a\sm u]\mid u{\in}|\ns P^\prg|\}$ is not necessarily strictly small-supported.
The result still holds, by a proof using Lemma~\ref{lemm.technical}:
\begin{prop}
\label{prop.all.sub.commute}
Suppose $X{\in}|\nompow(\ns P)|$ and $v{\in}|\ns P^\prg|$ and $a\#v$.
Then
$$
(\freshcap{a}X)[b\sm v]=\freshcap{a}(X[b\sm v]) .
$$
(Note by our permutative convention in Definition~\ref{defn.atoms} that $a$ and $b$ are assumed distinct.)
\end{prop}
\begin{proof}
Consider $p{\in}|\ns P|$.
We reason as follows:
$$
\begin{array}[b]{@{\hspace{-0em}}r@{\ }l@{\quad}l}
p\in (\freshcap{a}&X)[b\sm v]
\\
\liff&
\New{b'}p[v\ms b']\in \freshcap{a}(b'\,b)\act X
&\text{Prop~\ref{prop.amgis.iff},\ Thm~\ref{thrm.equivar}}
\\
\liff&
\New{b'}p[v\ms b']\in \bigcap_{u{\in}|\ns P^\prg|} ((b'\ b)\act X)[a\sm u]
&\text{Definition~\ref{defn.nu.U}}
\\
\liff&
\New{b'}\Forall{u{\in}|\ns P^\prg|}\New{a'} p[v\ms b'][u\ms a']\in (b'\,b)\act(a'\,a)\act X
&\text{Proposition~\ref{prop.amgis.iff}}
\\
\liff&
\Forall{u{\in}|\ns P^\prg|}\New{b'}\New{a'} p[u\ms a'][v\ms b']\in (b'\,b)\act(a'\,a)\act X
&\text{Lemma~\ref{lemm.technical}}
\\
\liff&
\Forall{u{\in}|\ns P^\prg|}\New{a'} p[u\ms a']\in ((a'\,a)\act X)[b\sm v]
&\text{Proposition~\ref{prop.amgis.iff}}
\\
\liff&
\Forall{u{\in}|\ns P^\prg|}\New{a'} p[u\ms a']\in (a'\,a)\act (X[b\sm v])
&\text{T\ref{thrm.equivar},\ C\ref{corr.stuff}},\,a',a\#v
\\
\liff&
\Forall{u{\in}|\ns P^\prg|} p\in X[b\sm v][a\sm u]
&\text{Proposition~\ref{prop.amgis.iff}}
\\
\liff&
p\in \freshcap{a}(X[b\sm v])
&\text{Definition~\ref{defn.nu.U}}
\end{array}
\qedhere$$
\end{proof}

\begin{lemm}
\label{lemm.all.alpha}
Suppose $X\in|\nompow(\ns P)|$.
Then
$$
b\#X
\quad\text{implies}\quad
\freshcap{a}X=\freshcap{b}(b\ a)\act X .
$$
As a corollary,
$a\#\freshcap{a}X$ and
$\supp(\freshcap{a}X)\subseteq\supp(X){\setminus}\{a\}$.
\end{lemm}
\begin{proof}
The corollary follows by Corollary~\ref{corr.stuff}(\ref{stuff.freshness.criterion}) and by Theorem~\ref{thrm.no.increase.of.supp}.
For the first part, we reason as follows:
$$
\begin{array}[b]{r@{\ }l@{\quad}l}
\freshcap{a}X=&\bigcap\{X[a{\sm}u]\mid u{\in}|\ns P^\prg|\}
&\text{Definition~\ref{defn.nu.U}}
\\
=&\bigcap\{((b\ a)\act X)[b{\sm}u]\mid u{\in}|\ns P^\prg|\}
&\text{Lemma~\ref{lemm.sigma.alpha}}
\\
=&\freshcap{b}(b\ a)\act X
&\text{Definition~\ref{defn.nu.U}}
\end{array}
\qedhere$$
\end{proof}

\maketab{tab3}{@{\hspace{-2em}}L{2em}R{10em}@{\ }L{6em}L{20em}}

\begin{thrm}
\label{thrm.all.closed}
If $X\in|\powsigma(\ns P)|$ then
\begin{enumerate*}
\item
$\freshcap{a}X\in|\powsigma(\ns P)|$, and as a corollary
\item
$\freshcap{a}X$ is equal to $\freshwedge{a}X$ (the $a$-fresh limit of $X$) in $\powsigma(\ns P)$ considered as a nominal poset with a monotone $\sigma$-action.
\end{enumerate*}
\end{thrm}
\begin{proof}
The corollary is from Corollary~\ref{corr.powsigma.nompowset.monotone.sigma} and Proposition~\ref{prop.char.freshwedge}(2).
We now prove part~1 of this Theorem.

Small support is from Theorem~\ref{thrm.no.increase.of.supp}.
It remains to check conditions~\ref{item.fresh.powsigma} and~\ref{item.alpha.powsigma} of Definition~\ref{defn.powsigma}:
\begin{enumerate*}
\item
Suppose $b$ is fresh (so $b\#X$) and suppose $v\in|\ns P^\prg|$.
Using Lemma~\ref{lemm.all.alpha} suppose without loss of generality that $a\#v$.
Then we reason as follows:
\begin{tab2rr}
(\freshcap{a}X)[b{\sm}v]=&\freshcap{a}(X[b{\sm}v])
&\text{Proposition~\ref{prop.all.sub.commute}}
\\
=&\freshcap{a}X
&\text{C\ref{item.fresh.powsigma} of Def~\ref{defn.powsigma}},\ b\#X
\end{tab2rr}
\item
Suppose $b'$ is fresh (so $b'\#X$).
Then we reason as follows:
\begin{tab2rr}
(\freshcap{a}X)[b{\sm}b']=&\freshcap{a}(X[b{\sm}b'])
&\text{Proposition~\ref{prop.all.sub.commute}}
\\
=&\freshcap{a}((b'\ b)\act X)
&\text{C\ref{item.alpha.powsigma} of Def~\ref{defn.powsigma}},\ b'\#X
\\
=&(b'\ b)\act (\freshcap{a}X)
&\text{Theorem~\ref{thrm.equivar}}
\qedhere\end{tab2rr}
\end{enumerate*}
\end{proof}

\maketab{tab9}{@{\hspace{-3em}}L{15em}@{\quad}L{25em}}

\begin{thrm}
\label{thrm.powerset}
Suppose $\ns P$ is an $\amgis$-algebra.
Then the $\sigma$-algebra $\powsigma(\ns P)$ from Definition~\ref{defn.powsigma} naturally becomes a nominal distributive lattice with $\tall$ where $\ttop$,
$\tand$, $\tbot$, $\tor$, and $\tall$ are interpreted as $|\ns P|$, set intersection $\cap$, the empty set $\varnothing$, set union $\cup$, and $\freshcap{a}$.
\end{thrm}
\begin{proof}
By Lemma~\ref{lemm.pow.nu.closed} for most of the connectives, and by Theorem~\ref{thrm.all.closed}
for $\tall$.
\end{proof}

This completes Part~\ref{part.indias}.
So far, we have defined nominal distributive lattices with $\tall$ and seen how to build them using nominal powersets.
In Part~\ref{part.inspectas} we show how to go from topologies (i.e. subsets of powersets subject to various sanity conditions, since powersets are usually very large) back to nominal distributive lattices with $\tall$.

%%%%%%%%%%%%%%%%%%%%%%%%%%%%%%%%%%%%%%%%%%%%
\jamiepart{Nominal spectral space representation}
\label{part.inspectas}

%%%%%%%%%%%%%%%%%%%%%%%%%%%%%%%%%%%%%%%%%%%%%
\section{Completeness}
\label{sect.completeness}

The key definition of this section is that of \emph{filter} in Definition~\ref{defn.filter}. 
The main result is a representation result, Theorem~\ref{thrm.pp.iso}, which represents an impredicative nominal distributive lattice with $\tall$ as a set of sets of prime filters.

The key technical results are in the sequence Proposition~\ref{prop.technical.contradiction}, Proposition~\ref{prop.max.2}, and Theorem~\ref{thrm.maxfilt.zorn}, which use Zorn-style arguments to exhibit every filter as a subset of some prime filter.

This story is familiar: it is standard to represent a lattice using sets of prime filters.
However, making this work is not trivial.
It is interesting to highlight three reasons for this:
\begin{enumerate*}
\item
We must account for $\tall$, of course.

This is extra structure and our treatment uses properties of nominal sets and the $\new$-quantifier in subtle ways.
See the discussion opening Subsection~\ref{subsect.filters}.
\item
Zorn's Lemma is related to the Axiom of Choice, which can cause difficulties with nominal sets because it may lead to non-small support (our definition of a nominal set from Definition~\ref{defn.nominal.set} requires small support).

We will find ourselves using elements which do have a permutation action---we are still within nominal techniques---but the elements do not necessarily have small support.
This is unusual.
See Remark~\ref{rmrk.p.not.finite.support}.
\item
Once these difficulties are navigated, we must still give points (prime filters) an $\amgis$-algebra structure.

There is no reason to expect prime filters to behave well and support an $\amgis$-algebra structure.
`By magic', it just works: see Lemma~\ref{lemm.bus.filter}.
\end{enumerate*}
A final technical hurdle is generated by our intended application of giving semantics to the untyped $\lambda$-calculus.
In effect, this means that we want to consider lattices with a substitution action over themselves, in a suitable sense, whence the notion of \emph{impredicativity} developed in Subsection~\ref{subsect.impredicative.nom.dist.lat}.
As usual for impredicative definitions, care is needed.
Yet, once these definitions and results are in place, the main result Theorem~\ref{thrm.pp.iso} becomes quite natural.

%%%%%%%%%%%%%%%%%%%%%%%%%%%%%%%%%%%%%%%%%%%
\subsection{Filters and prime filters}
\label{subsect.filters}

For this subsection, fix $\ns D\in\ndia$ a nominal distributive lattice with $\tall$ (Definition~\ref{defn.FOLeq}).
Recall from Definition~\ref{defn.FOLeq} that $\ns D$ is a set with a small-supported permutation action, a $\sigma$-action (like a substitution action but abstractly specified as a nominal algebra), finite joins, fresh-finite meets, and satisfying a generalisation of the usual distributivity properties for lattices.

We start by defining our notion of (prime) filter, and proving that every filter is included in some prime filter.

The main definition is Definition~\ref{defn.filter} and the main result is Theorem~\ref{thrm.maxfilt.zorn}.
The main technical result is Proposition~\ref{prop.technical.contradiction}.

\begin{rmrk}[Discussion of condition~\ref{filter.new}]
\label{rmrk.interpret.new.filter}
Condition~\ref{filter.new} of Definition~\ref{defn.filter} is specific to the \emph{nominal} filters.
See also its verification in Proposition~\ref{prop.technical.contradiction}, in which $\tall$ is decomposed into $\new$ and the permutation action $\pi$---echoing Proposition~\ref{prop.char.freshwedge.names} and Proposition~\ref{prop.these.are.equivalent}.
This decomposition of $\tall$ into $\new$ and $\pi$ is important for two reasons:
\begin{itemize*}
\item
it converts an infinite conjunction over the entire domain into a $\new$-quantified assertion---the $\new$-quantifier has some excellent properties, such as commuting with conjunction \emph{and} disjunction---and
\item
it does not depend on $\ns D^\prg$.
\end{itemize*}
So condition~\ref{filter.new} of Definition~\ref{defn.filter} means that to check the universal quantifier $\tall a.x$ we do not need to know about all of the programs of $\ns D$.
We just need to know about the atoms, and in particular, we just need to know about the fresh atoms.

This is familiar from proof-theory.
To prove $\Gamma\cent\forall x.\phi$ we do not need to check $\phi[a\sm t]$ for every term $t$;
we just check $\phi[x\sm y]$ for fresh $y$.

More on this in Remark~\ref{rmrk.p.not.finite.support}.
\end{rmrk}

%%%%%%%%%%%%%%%%%%%%%%%%%%%%
\subsubsection{Filters}

\begin{frametxt}
\begin{defn}
\label{defn.filter}
A \deffont{filter} in $\ns D$ is a nonempty subset $p\subseteq|\ns D|$ (which need not have small support) such that:
\begin{enumerate*}
\item
\label{filter.proper}
$\tbot\not\in p$ (we say $p$ is \deffont{consistent}).
\item
\label{filter.up}
If $x\in p$ and $x\leq x'$ then $x'\in p$ (we call $p$ \deffont{up-closed}).
\item
\label{filter.and}
If $x\in p$ and $x'\in p$ then $x\tand x'\in p$.
\item
\label{filter.new}
If $\New{b}(b\ a)\act x\in p$ then $\tall a.x\in p$.
\end{enumerate*}
\end{defn}
\end{frametxt}
The notion of \emph{prime filter} is in Definition~\ref{defn.prime.filter}, and has no further surprises.

Ideals are dual to filters; Definition~\ref{defn.ideal} is standard:
\begin{defn}
\label{defn.ideal}
An \deffont{ideal} in $\ns D$ is a nonempty subset $Z\subseteq|\ns D|$ (which need not have small support) such that:
\begin{enumerate*} 
\item\label{ideal.top}
$\ttop\not\in Z$.
\item\label{ideal.down}
If $x\in Z$ and $x'\leq x$ then $x'\in Z$ (we call $Z$ \deffont{down-closed}).
\item\label{ideal.or}
If $x\in Z$ and $x'\in Z$ then $x\tor x'\in Z$.
\end{enumerate*}
\end{defn}

\begin{rmrk}
Definition~\ref{defn.ideal} is not a perfect dual to Definition~\ref{defn.filter}: we do not have $\tall$.
(Correspondingly, we assume that a universal quantifier exists in $\ns D$, but not an existential.)
This will not be a problem.\footnote{We need $\tand$ and $\tor$ to build filters.  Later on when we model the untyped $\lambda$-calculus in Subsection~\ref{subsect.beta.eta}, we will need $\tall$, $\ppa$, and $\app$.  We will not need an existential $\texi$.
The existential may still exist; see Appendix~\ref{subsect.existential}.
}
Indeed, the lack of a fourth condition in Definition~\ref{defn.ideal} will be convenient in Lemma~\ref{lemm.ascending.chain.ideals}.
\end{rmrk}

\begin{rmrk}
Condition~\ref{filter.new} of Definition~\ref{defn.filter} seems odd
in view of, say, Proposition~\ref{prop.char.freshwedge} or Definition~\ref{defn.nu.U}.
Should it not be
\begin{quote}
$\Forall{u{\in}\ns D^\prg}x[a\sm u]\in p$ implies $\tall a.x\in p$?
\end{quote}
Or, in view of Proposition~\ref{prop.char.freshwedge.names}
should it not be at least
\begin{quote}
$\Forall{n{\in}\mathbb A}x[a\sm n_{\ns D^\prg}]\in p$ implies $\tall a.x\in p$?
\end{quote}
In fact, we shall see in Proposition~\ref{prop.these.are.equivalent} that condition~\ref{filter.new} as written, implies both of these.
\end{rmrk}

\begin{rmrk}
\label{rmrk.p.not.finite.support}
We continue Remark~\ref{rmrk.interpret.new.filter}.
We do not assume that $p$ has small support, so the $b$ bound by the $\new$-quantifier in condition~\ref{filter.new} of Definition~\ref{defn.filter} need not necessarily be fresh for $p$.
So the reader familiar with nominal techniques should note that our use of $\new$ is atypical.
The `standard' decomposition of $\new$ into `$\forall$+freshness' and `$\exists$+freshness' familiar from e.g. Theorem~2.17 of \cite{gabbay:stodfo}, Theorem~6.5 of \cite{gabbay:fountl}, or Theorem~9.4.6 of \cite{gabbay:thesis} will not necessarily work unless $p$ has small support, which in general is not the case.
Nevertheless, we have enough structure to obtain the results we need.

What is the case, is that $x\in|\ns D|$ is assumed to have small support, and $b$ will be fresh for $x$.

It will be important for the proof of Theorem~\ref{thrm.maxfilt.zorn} that we allow $p$ to have non-small support.
\end{rmrk}

\begin{rmrk}
\label{rmrk.why.filter.new}
Condition~\ref{filter.new} of Definition~\ref{defn.filter} seems not only elegant but necessary:
\begin{enumerate*}
\item
The design choice for condition~\ref{filter.new} of Definition~\ref{defn.filter} to assume $\Forall{n{\in}\mathbb A}x[a\sm n_{\ns D^\prg}]\in p$ causes failure in case~\ref{bus.filter.new} of Lemma~\ref{lemm.bus.filter} where (in the notation of that case) $n{=}a$.
\item
The design choice $\Forall{u{\in}\ns D^\prg}x[a\sm u]\in p$ causes failure in the final part of Proposition~\ref{prop.G.funct} (morphisms commute with quantification; see also Remark~\ref{rmrk.the.structure}).
\end{enumerate*}
We will see a dual version of the condition in condition~\ref{sober.all} of Definition~\ref{defn.sober}.
See also Remark~\ref{rmrk.sober.second.condition} and Lemma~\ref{lemm.tall.unions}.
\end{rmrk}

Definition~\ref{defn.uparrow} and Lemma~\ref{lemm.uparrow.filter} give examples of filters and ideals.
The definitions and proofs are standard, except we must verify condition~\ref{filter.new} of Definition~\ref{defn.filter} in Lemma~\ref{lemm.uparrow.filter}:
\begin{defn}
\label{defn.uparrow}
If $x\in|\ns D|$ then define $x{\uparrow}$ and $x{\downarrow}$ by
$$
x{\uparrow} = \{y \mid x\leq y\}
\qquad\text{and}\qquad
x{\downarrow} = \{y \mid y\leq x\}
.
$$
\end{defn}

\begin{lemm}
\label{lemm.uparrow.filter}
\begin{itemize*}
\item
If $x\neq\tbot$ then $x{\uparrow}$ from Definition~\ref{defn.uparrow} is a small-supported filter.
\item
If $x\neq\ttop$ then
$x{\downarrow}$ is a small-supported ideal.
\end{itemize*}
\end{lemm}
\begin{proof}
It is routine to verify conditions~1 to~3 of Definition~\ref{defn.filter}.
Suppose $x\leq (b\ a)\act y$ for cosmall many $b$.
We take one particular $b\#x,y$ and use part~\ref{b.bigger.2} of Lemma~\ref{lemm.b.bigger}.
Small support is direct from Theorem~\ref{thrm.no.increase.of.supp}.
The case of $x{\downarrow}$ is no harder.
\end{proof}

\begin{prop}
\label{prop.these.are.equivalent}
Suppose $p$ is a filter in $\ns D$ and $a$ is an atom.
Then:
\begin{enumerate}
\item\label{these.are.equivalent.1}
The following conditions are equivalent (below, $n$ ranges over all atoms, including $a$):
$$
\begin{array}{r@{\ }l}
\tall a.x\in p
\ \liff\quad&
\Forall{u{\in}|\ns D^\prg|}x[a\sm u]\in p
\\
\quad\liff\quad&
\Forall{n{\in}\mathbb A}x[a\sm n]\in p
\\
\quad\liff\quad&
\New{b}(b\ a)\act x\in p
\end{array}
$$
\item\label{these.are.equivalent.2}
If furthermore $p$ is small-supported and $a\#p$ then the following conditions are equivalent:
$$
\begin{array}{r@{\ }l@{\hspace{6.8em}}}
\tall a.\phi\in p
\ \liff\quad&
\phi\in p
\end{array}
$$
\end{enumerate}
\end{prop}
\begin{proof}
If $\tall a.x\in p$ then by Lemma~\ref{lemm.fresh.glb.sub} and condition~\ref{filter.up} of Definition~\ref{defn.filter} also $x[a\sm u]\in p$ for every $u{\in}|\ns D^\prg|$.
It follows in particular that $x[a\sm n_{\ns D^\prg}]\in p$ for every $n\in\mathbb A$, so $x[a\sm b]\in p$ for all $b\#x$ so by Lemma~\ref{lemm.sub.alpha} also $\New{b}(b\ a)\act x\in p$.
We complete the cycle of implications using condition~\ref{filter.new} of Definition~\ref{defn.filter}.

Part~2 follows from part~1 using Theorem~\ref{thrm.new.equiv}.
\end{proof}

\begin{rmrk}
Recall $\freshcap{a}$ from Definition~\ref{defn.nu.U}.
If we are willing to borrow the notation $\pp x=\{p\mid x{\in}p\}$ from the future Definition~\ref{defn.pp}, then we can rewrite Proposition~\ref{prop.these.are.equivalent}(\ref{these.are.equivalent.1}) as follows:
$$
\pp{(\tall a.x)} = \freshcap{a}\pp x = \bigcap_{n{\in}\mathbb A}\pp x[a\sm n]
= \{p{\in}\pp x\mid \New{b}(b\ a)\act p\in \pp x\} .
$$
\end{rmrk}

\subsubsection{Growing filters}

We consider one useful way to build new (larger) filters out of old filters:
\begin{defn}
\label{defn.tand.bar}
Suppose $p\subseteq|\ns D|$ and $y\in|\ns D|$.
Then define $p{+}y\subseteq|\ns D|$ by
\begin{frameqn}
p{+}y=\{x'\mid \Exists{x{\in}p}x\tand y\leq x'\} .
\end{frameqn}
\end{defn}

\begin{prop}
\label{prop.technical.contradiction}
Suppose $p$ is a small-supported filter in $\ns D$ and suppose $y\in|\ns D|$.
Then:
\begin{itemize*}
\item
$p\subseteq p{+} y$.
\item
$y\in p{+} y$.
\item
$p{+}y$ is closed under conditions~\ref{filter.up} to~\ref{filter.new} of Definition~\ref{defn.filter} (so if $p{+}y$ is consistent then it is a filter).
\end{itemize*}
As a corollary, if $Z$ is an ideal and $(p{+}y)\cap Z=\varnothing$ then $p{+}y$ is a filter.
\end{prop}
\begin{proof}
The corollary follows from the body of this result because since from condition~\ref{ideal.down} of Definition~\ref{defn.ideal} $\tbot\in Z$, so $\tbot\not\in p{+}y$.

We now consider the body of this result.
It is clear from the construction that $p\subseteq p{+}y$ and $y\in p{+}y$.
We now check that $p{+}y$ satisfies conditions~\ref{filter.up} to~\ref{filter.new} of Definition~\ref{defn.filter}:
\begin{enumerate*}
\setcounter{enumi}{1}
\item
\emph{If $z\in p{+}y$ and $z\leq z'$ then $z'\in p{+}y$.}\quad
By construction.
\item
\emph{If $z\in p{+}y$ and $z'\in p{+}y$ then $z{\tand}z'\in p{+}y$.}\quad
Suppose $z\geq x\tand y$ and $z'\geq x'\tand y$ for $x,x'\in p$.
Then by condition~\ref{filter.and} of Definition~\ref{defn.filter} $x\tand x'\in p$, and it is a fact that $z{\tand}z'\geq (x{\tand}x')\tand y$.
\item
\emph{If $\New{b}((b\ a)\act z\in p{+}y)$ then $\tall a.z\in p{+}y$.}\quad
Suppose for cosmall many $b$ there exists an $x_b\in p$ such that $(b\ a)\act z\geq x_b{\tand}y$.
Then certainly there exists some $b$ such that $b\#y,z,p$ and $(b\ a)\act z\geq x_b{\tand}y$.
We apply $\tall b$ to both sides and use distributivity (Definition~\ref{defn.distrib}), and we conclude that
$$
\tall b.(b\ a)\act z\geq (\tall b.x_b){\tand}y .
$$
We assumed $p$ has small support so by Proposition~\ref{prop.these.are.equivalent}(\ref{these.are.equivalent.2}) (since $x_b{\in}p$ and $b\#p$) $\tall b.x_b\in p$.
By Lemma~\ref{lemm.freshwedge.alpha} and condition~\ref{filter.up} of Definition~\ref{defn.filter} we conclude that $\tall a.z\in p{+}y$ as required.
\qedhere\end{enumerate*}
\end{proof}

%%%%%%%%%%%%%%%%%%%%%%%%%%%%%%%
\subsubsection{Growing ideals}

We consider one useful way to build new (larger) ideals out of old ideals:
\begin{defn}
\label{defn.tor.bar}
Suppose $Z,Y\subseteq|\ns D|$.
Then define $Z{+}Y\subseteq|\ns D|$ by
\begin{frameqn}
Z{+}Y=\{z'\mid \Exists{z{\in}Z,n{\in}\mathbb N,y_1,\dots,y_n{\in}Y}z'\leq z\tor y_1\tor\dots\tor y_n\} .
\end{frameqn}
\end{defn}

Lemma~\ref{lemm.ideal.plus} is a version of Proposition~\ref{prop.technical.contradiction} for ideals.
It is the simpler result, because Definition~\ref{defn.ideal} has nothing corresponding to condition~\ref{filter.up} of Definition~\ref{defn.filter}:
\begin{lemm}
\label{lemm.ideal.plus}
Suppose $Z\subseteq|\ns D|$ is an ideal and $Y\subseteq|\ns D|$.
Then:
\begin{itemize*}
\item
$Z\subseteq Z{+}Y$ and $Y\subseteq Z{+}Y$.
\item
$Z{+}Y$ is closed under conditions~\ref{ideal.down} and~\ref{ideal.or} of Definition~\ref{defn.ideal} (so that if $\ttop\not\in Z{+}Y$ then it is an ideal).
\end{itemize*}
\end{lemm}
\begin{proof}
By routine calculations.
\end{proof}

%%%%%%%%%%%%%%%%%%%%%%%%%%%%%%
\subsubsection{Maximal and prime filters}

The results of this subsection are Lemma~\ref{lemm.can.extend.Z} and Proposition~\ref{prop.max.2}.
From a great distance they follow a familiar pattern:
\begin{itemize*}
\item
Lemma~\ref{lemm.can.extend.Z} expresses ``if a filter is not maximal, then we can extend it''.
\item
Proposition~\ref{prop.max.2} expresses ``if a filter is maximal, then it is prime''.
\end{itemize*}
But in Lemma~\ref{lemm.can.extend.Z} we work not with filters but with filter-ideal pairs, and we express something quite subtle:
\begin{quote}
``If a filter-ideal pair is not maximal but is small-supported, then we can extend \emph{either} the filter with a universal quantification $\tall a.y$, \emph{or} we can extend the ideal with an equivalence class---really an $\alpha$-equivalence class---of $(b\ a)\act y$ for fresh atoms $b$.''
\end{quote}
To see why this is so, see Remark~\ref{rmrk.hand.rolled.zorn} and Theorem~\ref{thrm.maxfilt.zorn}.

%[mjg where used?]
%\begin{lemm}
%Suppose $p\subseteq|\ns D|$ is a filter and $Z\subseteq|\ns D|$ is an ideal, and suppose $y{\in}|\ns D|$.
%Then
%$$
%\tbot\in p{+}y\text{\ \ and\ \ }\ttop\in Z{+}y
%\quad\text{imply}\quad p\cap Z\neq\varnothing.
%$$
%\end{lemm}
%\begin{proof}
%Suppose $\tbot\in p{+}y$ and $\ttop\in Z{+}y$ and $p{\cap} Z=\varnothing$.
%We will derive a contradiction.
%Since $\tbot\in p{+}y$ there exists $x{\in}p$ such that $x\tand y=\tbot$.
%Similarly since $\ttop\in Z{+}y$ there exists $z{\in}Z$ such that $z\tor y=\ttop$.
%Then we reason as follows using standard properties of distributive lattices:
%$$
%x=x\tand (z\tor y)=(x\tand z)\tor(x\tand y)=x\tand z=\tbot .
%$$
%Thus $\tbot{\in}p$, contradicting condition~\ref{filter.proper} of Definition~\ref{defn.filter}.
%\end{proof}

\begin{frametxt}
\begin{defn}
\label{defn.prime.filter}
\begin{itemize}
\item
Call a filter $p{\subseteq}|\ns D|$ \deffont{prime} when $x_1{\tor}x_2\in p$ implies either $x_1\in p$ or $x_2\in p$.
\item
Suppose $p$ is a filter and $Z\subseteq|\ns D|$ is an ideal.
Call $p$ \deffont{maximal with respect to $Z$} when $p{\cap}Z=\varnothing$ and for every filter $p'$ with $p'\cap Z=\varnothing$, if $p\subseteq p'$ then $p=p'$.
\item
Call $p$ \deffont{maximal} when it is maximal with respect to the ideal $\{\tbot\}$.
\end{itemize}
\end{defn}
\end{frametxt}

\begin{lemm}
\label{lemm.tall.cent.char}
Suppose $p$ is a filter in $\ns D$ (Definition~\ref{defn.filter}).
Then:
\begin{enumerate*}
\item
$\tbot\not\in p$ and $\ttop\in p$.
\item
$x{\tand}y\in p$ if and only if $x\in p$ and $y\in p$.
\item
$\tall a.x\in p$ if and only if $\New{b}(b\ a)\act x\in p$.\footnote{We assume that $x\in|\ns D|$ is small-supported because in Definition~\ref{defn.nom.poset} we assume $(|\ns D|,\act)$ is a nominal set.
We do not assume that a filter $p$ is small-supported in Definition~\ref{defn.filter}.

This means that we may assume that $b$ is fresh for $x$ under the $\new b$ quantifier, but we do not know that $b$ is fresh for $p$.
This will not be a problem.
}
\item
If $p$ is prime (Definition~\ref{defn.prime.filter}) then
$x{\tor}y\in p$ if and only if $x\in p$ or $y\in p$.
\end{enumerate*}
\end{lemm}
\begin{proof}
\begin{enumerate*}
\item
The first part is condition~1 of Definition~\ref{defn.filter}.
For the second part, by assumption in Definition~\ref{defn.filter} $p$ is nonempty, so there exists some $x\in p$.
Now $\ns D$ has a top element $\ttop$, so $x\leq \ttop$ and by condition~\ref{filter.up} of Definition~\ref{defn.filter} $\ttop\in p$.
\item
From conditions~\ref{filter.up} and~\ref{filter.and} of Definition~\ref{defn.filter}.
\item
This is Proposition~\ref{prop.these.are.equivalent}(\ref{these.are.equivalent.1}).
\item
Routine, again using condition~\ref{filter.up} of Definition~\ref{defn.filter}.
\qedhere\end{enumerate*}
\end{proof}

We use Lemma~\ref{lemm.can.extend.Z} and Proposition~\ref{prop.max.2} to prove Theorem~\ref{thrm.maxfilt.zorn}.
We are most interested in Lemma~\ref{lemm.can.extend.Z} for the case that $Z$ is small-supported (as well as $p$), but the proof does not depend on it.
\begin{lemm}
\label{lemm.can.extend.Z}
Suppose
\begin{itemize*}
\item
$p{\subseteq}|\ns D|$ is a small-supported filter and $Z{\subseteq}|\ns D|$ is an ideal, and suppose
\item
$\tbot\in p{+}\tall a.y$ and $p{\cap}Z=\varnothing$.
\end{itemize*}
Write $Y=\{(b\ a)\act y\mid b\in \mathbb A\setminus(\supp(p){\cup}\supp(y)\cup\{a\})\}$.
Then:
\begin{enumerate*}
\item
$p\cap (Z{+}Y)=\varnothing$.
\item
As a corollary, $Z{+}Y$ is an ideal (and by part~1 is disjoint from $p$).
\end{enumerate*}
\end{lemm}
\begin{proof}
Suppose $\tbot\in p{+}\tall a.y$ and $p{\cap}Z=\varnothing$ and $p\cap (Z{+}Y)\neq \varnothing$.
We note the following:
\begin{itemize*}
\item
Since $p{\cap}(Z{+}Y)\neq\varnothing$, there exist $b_1,\dots,b_n\#p,y$ and $z{\in}Z$
with
$$
z\tor(b_1\ a)\act y\tor\dots\tor(b_n\ a)\act y\in p.
$$
\item
Since $\tbot\in p{+}\tall a.y$ there exists $x{\in}p$ with $x\tand\tall a.y=\tbot$.

We assumed $p$ has small support so by Proposition~\ref{prop.these.are.equivalent}(\ref{these.are.equivalent.2}) $\tall b_1\dots\tall b_n.x\in p$, and we see that we may assume without loss of generality that $b_1,\dots,b_n\#x$.
\item
Since $p{\cap}Z=\varnothing$ we have $\Forall{x'{\in}p,z'{\in}Z}x'\tand z'=\tbot$.
\end{itemize*}
By condition~\ref{filter.and} of Definition~\ref{defn.filter}
$$
x\tand(z\tor(b_1\ a)\act y\tor\dots\tor(b_n\ a)\act y)\in p
$$
and therefore by distributivity (Definition~\ref{defn.distrib})
\begin{multline*}
(x\tand z)\tor(x\tand (b_1\ a)\act y)\tor\dots\tor(x\tand (b_n\ a)\act y)
\stackrel{x\tand z=\tbot}=
\\
(x\tand (b_1\ a)\act y)\tor\dots\tor(x\tand (b_n\ a)\act y)
\in p .
\end{multline*}
By assumption $b_1,\dots,b_n\#p$ so by Proposition~\ref{prop.these.are.equivalent}(\ref{these.are.equivalent.2}) we deduce that
$$
\tall b_1\dots b_n.\bigl((x\tand (b_1\ a)\act y)\tor\dots\tor(x\tand (b_n\ a)\act y)\bigr)
\in p .
$$
Recall that by assumption $b_1,\dots,b_n\#x,y$; by calculations using Proposition~\ref{prop.pi.supp}, \rulefont{distrib\tall} from Definition~\ref{defn.distrib}, and Lemma~\ref{lemm.freshwedge.alpha} we conclude that
$$
x\tand \tall a.y=\tbot\in p .
$$
This contradicts condition~\ref{filter.proper} of Definition~\ref{defn.filter}.

For the corollary, we note by Lemma~\ref{lemm.tall.cent.char}(1) that $\ttop\in p$ so that $\ttop\not\in Z{+}Y$.
We use Lemma~\ref{lemm.ideal.plus}.
\end{proof}

\begin{prop}
\label{prop.max.2}
Suppose $p\subseteq|\ns D|$ is a filter and $Z\subseteq|\ns D|$ is an ideal, and suppose $p\cap Z=\varnothing$.
If $p$ is a maximal filter with respect to $Z$ then it is prime.
\end{prop}
\begin{proof}
Suppose $y_1\tor y_2\in p$ and $y_1,y_2\not\in p$.
By Proposition~\ref{prop.technical.contradiction} and maximality we have that $(p{+}y_1)\cap Z\neq\varnothing$ and $(p{+}y_2)\cap Z\neq\varnothing$.
It follows that there exist $x_1,x_2\in p$ with $x_1{\tand}y_1,x_2{\tand}y_2\in Z$.
Since $Z$ is an ideal, by condition~\ref{ideal.or} of Definition~\ref{defn.ideal}
$$
(x_1{\tand}y_1)\tor(x_2{\tand}y_2)\in Z .
$$
Now we rearrange the left-hand-side to deduce that
\begin{equation}\label{eq.mess}
u=(x_1{\tor}x_2)\tand (x_1{\tor}y_2) \tand (y_1{\tor}x_2) \tand (y_1{\tor}y_2) \in Z .
\end{equation}
We now note that $x_1{\tor}x_2\in p$ (since $x_1\in p$, and indeed also $x_2\in p$) and $x_1{\tor}y_2\in p$ (since $x_1\in p$) and $y_1{\tor}x_2\in p$ (since $x_2\in p$) and $y_1{\tor}y_2\in p$ by assumption.
But then $u\in p$, contradicting our assumption that $p\cap Z=\varnothing$.
\end{proof}

\begin{rmrk}
Experts who have read \cite{gabbay:stodnb} might wish to compare Proposition~\ref{prop.max.2} with Lemma~6.13 of \cite{gabbay:stodnb}.
We cannot limit the support of $p$ by wrapping $y_1$ and $y_2$ in universal quantifiers---like we did for the duality result for the $\new$-quantifier in Lemma~6.13 and indirectly in Definition~6.5 of \cite{gabbay:stodnb}---because $\forall$ does not distribute over $\lor$, whereas $\new$ does.

This property of $\new$ is one precise technical reason that the proofs of \cite{gabbay:stodnb} can be simpler than the proofs here.
\end{rmrk}

\subsubsection{The Zorn argument}

We start off with an easy technical result:
\begin{lemm}
\label{lemm.ascending.chain.ideals}
If $(Z_\alpha)_{\alpha\leq\beta}$ is an ascending chain of ideals (for some ordinal $\beta$) then $\bigcup_{\alpha\leq\beta} Z_\alpha$ is an ideal.
\end{lemm}
\begin{proof}
By standard calculations on the three conditions of Definition~\ref{defn.ideal}.
\end{proof}

\begin{rmrk}
\label{rmrk.hand.rolled.zorn}
The reader might now expect us to prove a version of Lemma~\ref{lemm.ascending.chain.ideals} for filters, and using Lemma~\ref{lemm.can.extend.Z} and Proposition~\ref{prop.max.2} along with Zorn's Lemma \cite[page~153]{enderton:elest} deduce the existence of prime filters.

This does not work because condition~\ref{filter.new} of Definition~\ref{defn.filter} is not closed under ascending chains of filters.\footnote{Many thanks to an anonymous referee for pointing this out.}
We might `accidentally' insert elements $(b\ a)\act y$ for cosmall many $b$
yet `forget' to insert $\tall a.y$ (or worse, $\tall a.y$ might be in $Z$).
Thus, we cannot use Zorn's lemma directly because the sets union of a countably ascending chain of filters is not necessarily a filter.

However, we do not really need closure under all chains; we only need the existence of some chain.
So we build a `bespoke' chain---unusually, this is not just a chain of filters but a chain of filter-ideal pairs $(p,Z)$---in such a way as to preserve condition~\ref{filter.new} of Definition~\ref{defn.filter}.
Intuitively, Lemma~\ref{lemm.can.extend.Z} has the form that it does because we grow the pair $(p,Z)$ at each stage such that
\begin{itemize*}
\item
\emph{either} we put $\tall a.y$ into $p$,
\item
\emph{or} we make sure that $\New{b}(b\ a)\act y\in p$ can never hold---not even after infinitely many stages---by putting $(b\ a)\act y$ into the ideal $Z$ immediately for most $b$; recall that $p$ must remain disjoint from $Z$.
\end{itemize*}
\end{rmrk}

\begin{thrm}
\label{thrm.maxfilt.zorn}
Suppose $p\subseteq|\ns D|$ is a small-supported filter and $Z\subseteq|\ns D|$ is an %small-supported
ideal and suppose $p\cap Z=\varnothing$.
Then there exists a prime filter $q$ with $p\subseteq q$ and $q\cap Z=\varnothing$.

As a corollary, if $p$ is a filter then there exists a prime filter $q$ containing $p$.
\end{thrm}
\begin{proof}
The corollary follows taking $Z$ to be $\{\tbot\}$, which we can easily verify is an ideal.
We now consider the main result.

Let $\alpha$ be the least ordinal of cardinality $\f{size}(\mathbb A)$ from Definition~\ref{defn.atoms}.
Note that $\mathbb A$ is assumed at least countably large, so $\alpha\times\alpha$ has the same cardinality as $\alpha$.
Note also that $\alpha$ is necessarily a limit ordinal, because it is the least ordinal with cardinality $\f{size}(\mathbb A)$.\footnote{There are two natural foundational definitions of `cardinal':
(1) A cardinal is an ordinal that cannot be bijected with any lesser ordinal; or
(2) A cardinal is an equivalence class of bijectable sets.
It is routine to biject an infinite ordinal $\alpha'$ with $\alpha'{+}1$, so in either (1) or (2) we see that the least ordinal with cardinality $\f{size}(\mathbb A)$ must be a limit.  Furthermore, in case (1) $\f{size}(\mathbb A)$ \emph{is} that limit; in case (2) $\alpha\in\f{size}(\mathbb A)$ is the least ordinal in $\f{size}(\mathbb A)$.}

We enumerate $\mathbb A\times\ns D$ as a list of pairs
\begin{equation}
\label{eq.succ}
(a_{i{+}1},x_{i{+}1})_{i{<}\alpha}.
\end{equation}
We write $(a_{i{+}1},x_{i{+}1})_{i{<}\alpha}$ instead of $(a_i,x_i)_{i{<}\alpha}$ to ensure that every pair $(a,x)$ is indexed by a successor ordinal and not a limit ordinal (so for instance we have $(a_1,x_1)$ but not $(a_0,x_0)$).
This is just for convenience in the base and limit cases of the inductive construction which now follows.

Recall the notations `$p{+}y$' from Definition~\ref{defn.tand.bar} and `$Z{+}Y$' from Definition~\ref{defn.tor.bar}.
We define a sequence of (by Theorem~\ref{thrm.no.increase.of.supp}) small-supported disjoint filter-ideal pairs $(p_i,Z_i)_{i{\leq}\alpha}$ inductively as follows:
\begin{enumerate}
\item
\emph{Base case.}\quad
We take $(p_0,Z_0)=(p,Z)$.
By assumption $p\cap Z=\varnothing$.
\item
\emph{Successor ordinal (first case).}\quad
If $p_{i}{+}\tall a_{i{+}1}.x_{i{+}1}\cap Z_{i}=\varnothing$
then take
$$
(p_{i{+}1},Z_{i{+}1})=(p_{i}{+}\tall a_{i{+}1}.x_{i{+}1},Z_{i}).
$$
It follows from Proposition~\ref{prop.technical.contradiction} (since $p_i$ is small-supported) that $p_{i{+}1}$ is a filter and from Theorem~\ref{thrm.no.increase.of.supp} that $p_{i{+}1}$ is small-supported.
\item
\emph{Successor ordinal (second case).}\quad
If $p_{i}{+}\tall a_{i{+}1}.x_{i{+}1}\cap Z_{i}\neq\varnothing$
then take
$$
\begin{array}{@{\hspace{-1em}}r@{\ }l}
Y=&\{(b\ a_{i{+}1})\act x_{i{+}1}\mid b\in\mathbb A{\setminus}(\supp(p_{i}){\cup}\supp(x_{i{+}1}){\cup}\{a_{i{+}1}\})\}
\quad\text{and}
\\
(p_{i{+}1},Z_{i{+}1})=&(p_{i},Z_{i}{+}Y).
\end{array}
$$
It follows from Lemma~\ref{lemm.can.extend.Z} that $Z_{i{+}1}$ is an ideal and $p_{i{+}1}\cap Z_{i{+}1}=\varnothing$.
\item
\emph{Nonzero limit ordinal.}\quad
If $i$ is a nonzero limit ordinal no greater than $\alpha$ then we define
$$
(p_i,Z_i)= (\bigcup_{j{<}i}p_j,\bigcup_{j{<}i}Z_j).
$$
By construction $p_i{\cap}Z_i=\varnothing$.
By Lemma~\ref{lemm.ascending.chain.ideals} $Z_i$ is an ideal.

We note that $p_i$ is a filter:
Conditions~\ref{filter.proper} to~\ref{filter.and} of Definition~\ref{defn.filter} are routine.
To check condition~\ref{filter.new}, suppose $\New{b}(b\ a)\act x\in p_i$ meaning that $(b\ a)\act x\in p_i$ for cosmall many $b$; we need to prove $\tall a.x\in p_i$.

So let $j{<}i$ be that index, which by construction in equation~\ref{eq.succ} above must be a successor ordinal, such that $(a,x)=(a_j,x_j)$.
At stage $j$, when we built $(p_j,Z_j)$, there were two possibilities:
\begin{itemize*}
\item
If $p_{j\minus 1}{+}\tall a.x\cap Z=\varnothing$ then we must have put $\tall a.x$ into $p_j$, so that $\tall a.x\in\bigcup_{j{<}i} p_j$ and we are done.
\item
If $p_{j\minus 1}{+}\tall a.x\cap Z\neq\varnothing$ then we must have put $(b\ a)\act x$ into $Z_j$ for cosmall many $b$, so that $\New{b}(b\ a)\act x\in\bigcup_{j{<}i} Z_j$.
But this is impossible because we assumed that cosmall many $(b\ a)\act x$ were in $p_i$, which is disjoint from $Z_i$.
\end{itemize*}
If $i\lneq\alpha$ then it follows using Theorem~\ref{thrm.no.increase.of.supp}(\ref{item.conservation.of.support}) that $p_i$ is small-supported.\footnote{If $i=\alpha$ then $p_\alpha$ need not have small support, but this will not be a problem because if $i=\alpha$ then we are finished with the induction.}
\end{enumerate}
We note that $p_\alpha$ is a maximal filter disjoint from the ideal $Z_\alpha$:
For consider any $x$ and choose some fresh $a$ (so $a\#x$).
Let $j$ be that index in the enumeration above such that $(a,x)=(a_j,x_j)$.
Note by Definition~\ref{defn.fresh.finite.limit} that $\tall a.x=x$, since $a\#x$.
It follows from the structure of the algorithm above that precisely one of $\tall a.x=x\in p_j$ or $x\in Z_j$ will hold.
Maximality follows.

By Proposition~\ref{prop.max.2} $p_\alpha$ is prime, and by construction $p\subseteq p_\alpha$ and $p_\alpha\cap Z=\varnothing$.
\end{proof}

\begin{rmrk}
It might be helpful to trace how Theorem~\ref{thrm.maxfilt.zorn} will be used in the rest of this paper:
\begin{itemize*}
\item
we use it in Lemma~\ref{lemm.completeness} to prove Corollary~\ref{corr.pp.injective} (injectivity of $x\mapsto \pp x$);
\item
we use it in Lemma~\ref{lemm.p.q.inj};
\item
we use it in Propositions~\ref{prop.extra.filter} and~\ref{prop.even.stronger}; and 
\item
we use it in Lemmas~\ref{lemm.hard.1} and~\ref{lemm.hard.2}, which are needed for Theorem~\ref{thrm.pp.app} (compositionality of $x\mapsto\pp x$ with respect to $\app$ and $\ppa$).
\end{itemize*}
\end{rmrk}

\begin{rmrk}
Theorem~\ref{thrm.maxfilt.zorn} assumes a filter $p$ and an ideal $Z$.
This is more asymmetric than it might seem:
\begin{enumerate}
\item
The notion of filter in Definitions~\ref{defn.filter} has a condition for $\tall$ (condition~\ref{filter.new}) whereas the notion of ideal in Definition~\ref{defn.ideal} does not.
This is important; it gives us a closure property Lemma~\ref{lemm.ascending.chain.ideals} for ideals, which we use to prove Theorem~\ref{thrm.maxfilt.zorn}.

As noted in Remark~\ref{rmrk.hand.rolled.zorn} we cannot expect a corresponding closure property for filters; perhaps we should call the ideals we need for this paper \emph{proposition-like}, because ideals only interact with propositional structure, and the filters \emph{predicate-like}, because filters must also interact well with the predicate quantifier $\tall$.
\item
We assume $p$ is small-supported and we do not assume $Z$ is small-supported.
We use small support of $p$ in the proof of Theorem~\ref{thrm.maxfilt.zorn} and, later, we will require the full generality of allowing non-small-supported $Z$ to prove Lemmas~\ref{lemm.hard.1} and~\ref{lemm.hard.2}, which handle applicative structure $\app$.
See Remark~\ref{rmrk.r.not.ssp}.
\end{enumerate}
Asymmetry~2 is related to asymmetry~1: we require small support for $p$ because this helps handle its interaction with $\tall$; ideals are `proposition-like', so for ideals small support does not matter.
\end{rmrk}

\begin{rmrk}
\label{rmrk.size.set.up}
Theorem~\ref{thrm.maxfilt.zorn} is set up to work for $\mathbb A$ and $\ns D$ of equal cardinality---as per condition~\ref{size.limit} of Definition~\ref{defn.D.impredicative} (see also condition~\ref{size.limit.top} of Definition~\ref{defn.impredicative.top}).
We need the cardinalities to be equal so that we do not run out of fresh atoms in the construction of the chain of small-supported filters in the proof of Theorem~\ref{thrm.maxfilt.zorn}.

Note that the canonical model (see Example~\ref{xmpl.canonical.idiom} and subsequent proofs) is countable by construction and the soundness proof in Theorem~\ref{thrm.lambda.soundness} holds for any size $\ns D$ regardless of the size of $\mathbb A$.
So the reader interested only in duality for countable models,\footnote{This is a reasonable special case for this particular paper, since one could argue that $\lambda$-calculus models should be countable since they are supposed to represent computation. However, the fully general version of Theorem~\ref{thrm.maxfilt.zorn} is no harder to write out---the mathematics does not really care.}
 or only interested in soundness and completeness but not in duality, can ignore the generality in Theorem~\ref{thrm.maxfilt.zorn} and take atoms and $\ns D$ to be countable.

If we want a full duality result also for \emph{uncountable} $\ns D$, then we require the potent generality of Theorem~\ref{thrm.maxfilt.zorn} and the size conditions in Definitions~\ref{defn.D.impredicative} and~\ref{defn.impredicative.top}.
More on this in Subsection~\ref{subsect.cardinality}.
\end{rmrk}

%%%%%%%%%%%%%%%%%%%%%%%%%%%%%%%%%%
\subsection{The amgis-action on (prime) filters}

For this subsection, fix $\ns D\in\ndia$ a nominal distributive lattice with $\tall$ (Definition~\ref{defn.FOLeq}).

Recall from Definition~\ref{defn.p.action} the pointwise $\amgis$-action $p[u\ms a]=\{x \mid x[a\sm u]\in p\}$ where $u\in|\ns D^\prg|$.
In this subsection we check that this action preserves the property of being a (prime) filter (Definitions~\ref{defn.filter} and~\ref{defn.prime.filter}).

The work happens in the key technical result Lemma~\ref{lemm.bus.filter}; Proposition~\ref{prop.filters.amgis} then puts the result in a some nice packaging.

\begin{lemm}
\label{lemm.bus.filter}
If $p$ is a filter in $\ns D$ then so is $p[u\ms a]$.
Furthermore, if $p$ is prime then so is $p[u\ms a]$.
\end{lemm}
\begin{proof}
We check the conditions of Definition~\ref{defn.filter}
We use Proposition~\ref{prop.sigma.iff} without comment:
\begin{enumerate*}
\item
\emph{$\tbot\not\in p[u\ms a]$.}\quad
Since by \rulefont{\sigma\#} $\tbot[a\sm u]=\tbot$.
\item
\emph{If $x{\in}p[u\ms a]$ and $x{\leq}x'$ then $x'{\in}p[u\ms a]$.}\quad
From Lemma~\ref{lemm.sigma.monotone}.
\item
\emph{If $x{\in}p[u\ms a]$ and $x'{\in}p[u\ms a]$ then $x{\tand}x'\in p[u\ms a]$.}\quad
Since by assumption
the $\sigma$-action is compatible, so $(x{\tand}x')[a\sm u]=x[a\sm u]{\tand}(x'[a\sm u])$ (Definition~\ref{defn.fresh.continuous}).
\item\label{bus.filter.new}
\emph{If $\New{b'}((b'\ b)\act x\in p[u\ms a])$ then $\tall b.x\in p[u\ms a]$.}\quad

Choose some fresh $c$ (so $c\#x,u$).
It is a fact that $x=(c\ b)\act ((c\ b)\act x)$ and by \rulefont{\tall\alpha} also $\tall b.x=\tall c.(c\ b)\act x$.
Thus, we may rename to assume without loss of generality that $b\#u$.

Now suppose $((b'\ b)\act x)[a\sm u]\in p$ for cosmall many $b'$; so suppose $b'\#u$.
By Corollary~\ref{corr.stuff} $(b'\ b)\act u=u$.
Thus by Remark~\ref{rmrk.what.is.equivar.for.term} and Corollary~\ref{corr.stuff} we have that $(b'\ b)\act(x[a\sm u])\in p$ for cosmall many $b'$.
Therefore $\tall b.(x[a\sm u])\in p$ and by compatibility (Definition~\ref{defn.fresh.continuous}) $(\tall b.x)[a\sm u]\in p$.
\end{enumerate*}
Now suppose $p$ is prime and suppose $(y_1{\tor}y_2)[a\sm u]\in p$.
Then by compatibility (Definition~\ref{defn.fresh.continuous}) $y_1[a\sm u]{\tor}(y_2[a\sm u])\in p$.
Therefore either $y_1[a\sm u]\in p$ or $y_2[a\sm u]\in p$.
\end{proof}

\begin{defn}
\label{defn.points}
If $\ns D\in\ndia$ write $\points(\ns D)$ for the $\amgis$-algebra determined by prime filters and the pointwise actions from Definition~\ref{defn.p.action}.
That is:
\begin{frametxt}
\begin{itemize*}
\item
$|\points(\ns D)|=\{p\subseteq|\ns D| \mid p\text{ is a prime filter}\}$.
\item
$\points(\ns D)^\prg=\ns D^\prg$.
\item
$\pi\act p=\{\pi\act x\mid x\in p\}$ and $p[u\ms a]=\{x \mid x[a\sm u]\in p\}$ for $u\in|\points(\ns D)^\prg|$.
\end{itemize*}
\end{frametxt}
\end{defn}

\begin{nttn}
\label{nttn.points.prime.filters}
We will use \deffont{points} and \emph{prime filters} synonymously henceforth.
\end{nttn}

\begin{prop}
\label{prop.filters.amgis}
$\points(\ns D)$ is indeed an $\amgis$-algebra.
\end{prop}
\begin{proof}
This is just Lemma~\ref{lemm.bus.filter} combined with Proposition~\ref{prop.amgis.2}.
\end{proof}

%%%%%%%%%%%%%%%%%%%%%%%%%%%%%%%%%%%%%%%%%%%%%%%%
\subsection{Injecting $\ns D$ into the set of sets of prime filters}

Recall from Definition~\ref{defn.powsigma} the notion of \emph{$\sigma$-powerset algebra} $\powsigma$, and from Definition~\ref{defn.prime.filter} the notion of a \emph{prime filter}.

In this subsection we consider how to embed $\ns D$ a nominal distributive lattice with $\tall$ in the $\sigma$-powerset of its prime filters.
The main definition is Definition~\ref{defn.pp}.
The main results are Lemma~\ref{lemm.bullet.commute}, Corollary~\ref{corr.pp.injective}, and Lemma~\ref{lemm.p.q.inj}.

It is standard to embed a lattice into sets of prime filters, and Definition~\ref{defn.pp} has the form one would expect.
There is extra structure; for instance Lemma~\ref{lemm.dup} and part~\ref{item.bullet.commute.tall} of Lemma~\ref{lemm.bullet.commute}.
With the results we have proven so far, we can deal with this extra structure.

Recall from Definition~\ref{defn.prime.filter} the notion of prime filter:
\begin{defn}
\label{defn.pp}
Suppose $\ns D$ is a nominal distributive nominal lattice with $\tall$.
Define
$$
\pp x = \{p\text{ a prime filter}\mid x\in p\}
\subseteq|\points(\ns D)| .
$$
\end{defn}

\maketab{tab5}{@{\hspace{-1em}}R{10em}@{\ }L{18em}L{12em}}

\begin{lemm}
\label{lemm.completeness}
$x\leq y$ if and only if $\pp x\subseteq\pp y$.
\end{lemm}
\begin{proof}
Suppose $x\leq y$.
By condition~\ref{filter.up} of Definition~\ref{defn.filter} if $x\in p$ then $y\in p$.
It follows that $\pp x\subseteq\pp y$.

Suppose $x\not\leq y$.
There are three easy cases:
\begin{itemize}
\item
If $x=\tbot$ then using Lemma~\ref{lemm.tall.cent.char}(1) $\pp x$ is the empty set and $\pp x\subseteq\pp y$.
\item
If $y=\ttop$ then using Lemma~\ref{lemm.tall.cent.char}(1) $\pp y$ is the set of all prime filters and $\pp x\subseteq\pp y$.
\item
Otherwise, by Lemma~\ref{lemm.uparrow.filter} $x{\downarrow}$ is a (small-supported) ideal and $y{\uparrow}$ is a small-supported filter.
Also since $x\not\leq y$ we have that $x{\downarrow}\cap y{\uparrow}=\varnothing$ and so by Theorem~\ref{thrm.maxfilt.zorn} there exists a prime filter $p$ containing $y{\uparrow}$ (so containing $y$) and disjoint from $x{\downarrow}$ thus not containing $x$.
Then $p\in \pp y$ and $p\not\in \pp x$.
\qedhere\end{itemize}
\end{proof}

\begin{corr}
\label{corr.pp.injective}
The assignment $x\mapsto \pp x$ is injective.
\end{corr}
\begin{proof}
Direct from Lemma~\ref{lemm.completeness}.
\end{proof}

\begin{corr}
\label{corr.supp.pp}
$\supp(\pp x)=\supp(x)$.
\end{corr}
\begin{proof}
Using part~3 of Theorem~\ref{thrm.no.increase.of.supp} and Corollary~\ref{corr.pp.injective}; see \cite[Theorem~4.7]{gabbay:fountl}.
\end{proof}

We never use Lemma~\ref{lemm.p.q.inj} or Corollary~\ref{corr.family.resemblance} later but we mention them anyway because they
express that the set of all filters is (intuitively) a conservative extension of the set of all prime filters:
\begin{lemm}
\label{lemm.p.q.inj}
Suppose $p\subseteq|\ns D|$ is a (not necessarily prime) filter and $x\in|\ns D|$, and suppose for every prime filter $r$, if $p\subseteq r$ then $x\in r$.
Then $x\in p$.
\end{lemm}
\begin{proof}
Suppose $x\not\in p$.
By Lemma~\ref{lemm.uparrow.filter} $x{\downarrow}$ is an ideal and by Theorem~\ref{thrm.maxfilt.zorn} there exists a prime filter $r$ such that $p\subseteq r$ and $x\not\in r$.
The result follows.
\end{proof}

\begin{corr}
\label{corr.family.resemblance}
A nice rephrasing of Lemma~\ref{lemm.p.q.inj} is possible using Definition~\ref{defn.pp}.
If $p$ is a (not necessarily prime) filter then
$$
p=\bigcap\{\pp x\mid x\in p\}.
$$
\end{corr}

%%%%%%%%%%%%%%%%%%%%%%%%%%%%%%%%%%
\subsection{The map from $x$, to prime filters containing $x$, as a morphism}

For this subsection, fix $\ns D\in\ndia$ a nominal distributive lattice with $\tall$ (Definition~\ref{defn.FOLeq}).
Recall from Definition~\ref{defn.pp} that if $x\in|\ns D|$ then $\pp x$ is the set of prime filters in $\ns D$ that contain $x$.

By Proposition~\ref{prop.filters.amgis} $\points(\ns D)$ is an $\amgis$-algebra.
So following Definition~\ref{defn.sub.sets}, sets of points $X\subseteq|\points(\ns D)|$ inherit an action $X[a\sm u]=\{p\mid p[u\ms a]\in X\}$.
With this action we have the following:
\begin{lemm}
\label{lemm.dup}
Suppose $x\in|\ns D|$ and $u\in|\ns D^\prg|$.
Then:
\begin{enumerate*}
\item
$\pi\act(\pp x)=\pp{(\pi\act x)}$
\item
$\pp x[a\sm u]=\pp{(x[a\sm u])}$
\end{enumerate*}
\end{lemm}
\begin{proof}
The case of $\pi\act(\pp x)$ is direct from Theorem~\ref{thrm.equivar} (a proof by concrete calculations similar to the case of $\pp x[a\sm u]$ is also possible).
For the case of $\pp x[a\sm u]$, we reason as follows:
$$
\begin{array}[b]{r@{\ }l@{\qquad}l}
p\in (\pp x)[a\sm u]
\liff&
\New{c} p[u\ms c]\in \pp{((c\ a)\act x)}
&\text{Prop~\ref{prop.amgis.iff}, pt~1 this result}
\\
\liff&
\New{c}(c\ a)\act x\in p[u\ms c]
&\text{Definition~\ref{defn.pp}}
\\
\liff&
\New{c} ((c\ a)\act x)[c\sm u]\in p
&\text{Proposition~\ref{prop.sigma.iff}}
\\
\liff&
p\in \pp{(x[a\sm u])}
&\rulefont{\sigma\alpha},\ \text{Definition~\ref{defn.pp}}
\end{array}
\qedhere$$
\end{proof}

\begin{lemm}
\label{lemm.bullet.commute}
\begin{enumerate*}
\item
\label{item.bullet.commute.tbot}
$\pp \tbot=\varnothing$ and $\pp\ttop=|\points(\ns D)|$
\item
\label{item.bullet.commute.tand}
$\pp{(x\tand y)}=\pp x\cap \pp y$
\item
\label{item.bullet.commute.tall}
$\pp{(\tall a.x)}=\freshcap{a}(\pp x)$ (Definition~\ref{defn.nu.U})
\item
\label{item.bullet.commute.tor}
$\pp{(x\tor y)}=\pp x\cup \pp y$
\end{enumerate*}
\end{lemm}
\begin{proof}
Parts~1, 2, and~4 just reformulate parts~1, 2, and~4 of Lemma~\ref{lemm.tall.cent.char}.

For part~3, suppose $p\in\pp{(\tall a.x)}$.
By Definition~\ref{defn.pp} this is if and only if $\tall a.x\in p$.
By Proposition~\ref{prop.these.are.equivalent}(\ref{these.are.equivalent.1}) this is if and only if $x[a\sm u]\in p$ for every $u{\in}|\ns D^\prg|$.
Using Definition~\ref{defn.pp} and Lemma~\ref{lemm.dup} this is if and only if $p\in\pp x[a\sm u]$ for every $u{\in}|\ns D^\prg|$.
The result follows from Definition~\ref{defn.nu.U}.
\end{proof}

\begin{defn}
\label{defn.pp.B}
Define $\pp{\ns D}\in\ndia$ a nominal distributive lattice with $\tall$ by the following data:
\begin{enumerate*}
\item
$|\pp{\ns D}|=\{\pp x\mid x\in|\ns D|\}$.
\item
$(\pp{\ns D})^\prg=\ns D^\prg$.
\item
$\pp{\ns D}$ has permutation and $\amgis$-actions following Definition~\ref{defn.sub.sets} (so $\pi\act\pp x=\{\pi\act p\mid p\in\pp x\}$ and $\pp x[a\sm u]=\{p\mid p[u\ms a]\in\pp x\}$).
\end{enumerate*}
If $\ns D$ is impredicative (Definition~\ref{defn.D.impredicative}), so we assume a %injective
$\sigma$-algebra morphism $\prg:\ns D^\prg\to\ns D$, then $\pp{\ns D}$ naturally becomes impredicative where:
\begin{itemize*}
\item
$\prg_{\pp{\ns D}} u=\pp{(\prg_{\ns D} u)}$ for $u\in|\ns D^\prg|=|(\pp{\ns D})^\prg|$.
\end{itemize*}
\end{defn}

It is now easy to state and prove a nominal sets representation theorem, representing an abstract $\ns D\in\india$ concretely as the nominal sets-based structure $\pp{\ns D}$:
\begin{frametxt}
\begin{thrm}[First representation theorem]
\label{thrm.pp.iso}
If $\ns D$ is in $\ndia$ then so is $\pp{\ns D}$, and the pair of maps $(x\mapsto \pp x,u\mapsto u)$ is an isomorphism from $\ns D$ to $\pp{\ns D}$ in $\ndia$.

If furthermore $\ns D$ is in $\india$ (is impredicative) then so is $\pp{\ns D}$ and $x\mapsto\pp x$ is an isomorphism from $\ns D$ to $\pp{\ns D}$.
\end{thrm}
\end{frametxt}
\begin{proof}
We unpack Definition~\ref{defn.pp.B} and use Lemmas~\ref{lemm.dup} and~\ref{lemm.bullet.commute} to check the conditions on morphisms from Definition~\ref{defn.hom.nba}.
Surjectivity is by construction and injectivity by Corollary~\ref{corr.pp.injective}.
\end{proof}

Theorem~\ref{thrm.pp.iso.pp} extends Theorem~\ref{thrm.pp.iso} with $\app$ and $\ppa$.

%%%%%%%%%%%%%%%%%%%%%%%%%%%%%%%%%%%%%%%%%%%%%%%%
\section{Nominal $\sigma$-topological spaces}
\label{sect.stone}

\subsection{The basic definition}

\begin{defn}
\label{defn.nom.top}
A \deffont{nominal $\sigma$-topological space} $\ns T$ is a tuple
$(|\ns T|,\act,\ns T^\prg,\tf{amgis},\otop{\ns T})$ where
\begin{itemize*}
\item
$(|\ns T|,\act,\ns T^\prg,\tf{amgis})$ forms an $\amgis$-algebra (Definition~\ref{defn.bus.algebra}) and
\item
$\otop{\ns T}{\subseteq}|\nompow(\ns T)|$ (i.e. a set of small-supported sets; see Subsection~\ref{subsect.finsupp.pow}) is an equivariant (see Lemma~\ref{lemm.when.set.equivar}) set of \deffont{open sets}.
\end{itemize*}
Furthermore we impose the following conditions on $\otop{\ns T}$:
\begin{enumerate*}
\item
$\emptyset\in \otop{\ns T}$ and $|\ns T|\in\otop{\ns T}$
\item
If $X\in\otop{\ns T}$ and $Y\in\otop{\ns T}$ then $X\cap Y\in\otop{\ns T}$.
\item
\label{item.top.strict}
If $\mathcal X\subseteq\otop{\ns T}$ is small-supported then $\bigcup \mathcal X\in\otop{\ns T}$; we call this a \deffont{small-supported union} of open sets.
\end{enumerate*}
\end{defn}

\begin{rmrk}
\label{rmrk.discuss.topologies}
Topological spaces over Zermelo-Fraenkel (\deffont{ZF}) sets---that is, over `ordinary' sets---are such that an arbitrary union of open sets is open.

Condition~\ref{item.top.strict} of Definition~\ref{defn.nom.top} generalises that condition, because any ZF set is also naturally an FM set with the trivial permutation action, and with empty support; so arbitrary sets of ZF sets are already small-supported, by $\varnothing$.

As a design choice in Definition~\ref{defn.nom.top} we do not insist that $U\in\otop{\ns T}$ implies $\freshcap{a}U\in\otop{\ns T}$.
See the later notion of \emph{coherence} in Definition~\ref{defn.coherent}.
% this postpones complexity and increases generality
\end{rmrk}

%%%%%%%%%%%%%%%%%%%%%%%%
\subsection{The map $F$ from distributive lattices to nominal $\sigma$-topological spaces}

\begin{defn}
\label{defn.F}
Suppose $\ns D\in\ndia$ is a nominal distributive lattice with $\tall$.
\begin{frametxt}
Define $F(\ns D)$ a nominal $\sigma$-topological space (Definition~\ref{defn.nom.top}) by (technical references follow):
\begin{enumerate*}
\item
$F(\ns D)$ has as underlying $\amgis$-algebra $\points(\ns D)$ from Definition~\ref{defn.points}, so that
\begin{itemize*}
\item
$|F(\ns D)|=|\points(\ns D)|$ and $F(\ns D)^\prg=\ns D^\prg$, and
\item
$\pi\act p=\{\pi\act x\mid x\in p\}$ and $p[u\ms a]=\{x\mid x[a\sm u]\in p\}$ for $u\in|\ns D^\prg|$.
\end{itemize*}
\item\label{F.strict}
The topology $\otop{F(\ns D)}$ is generated under small-supported unions by $\{\pp x\mid x\in|\ns D|\}$.
\end{enumerate*}
\end{frametxt}
\end{defn}

\begin{rmrk}
For the reader's convenience we give references for technical definitions above:
\begin{itemize*}
\item
$\points(\ns D)$ is from Definition~\ref{defn.points}.
\item
The pointwise actions are from Definition~\ref{defn.sub.sets}.
\item
Small support is from Definition~\ref{defn.fin.supp}.
\item
$\pp x$ is from Definition~\ref{defn.pp}.
\end{itemize*}
So $X\in\otop{F(\ns D)}$ when $X=\bigcup\mathcal X$ where $\mathcal X=\{\pp x_i\mid i\in I\}$ is a small-supported set of sets of points of the form $\pp x$.
\end{rmrk}

\begin{thrm}
\label{thrm.FB.amgis}
If $\ns D\in\ndia$ is a nominal distributive lattice (Definition~\ref{defn.D.impredicative})
then $F(\ns D)$ is a $\sigma$-topological space.
\end{thrm}
\begin{proof}
By Theorem~\ref{thrm.no.increase.of.supp} $\otop{F(\ns D)}$ is equivariant.
We consider the conditions in Definition~\ref{defn.nom.top} in turn:
\begin{enumerate*}
\item
$\varnothing$ is open by construction and $\pp\ttop=|\points(\ns D)|$.
\item
Suppose $X,Y\in\otop{F(\ns D)}$.
So $X=\bigcup\pp x_i$ and $Y=\bigcup\pp y_j$ for some small-supported sets $\{x_i\mid i\in I\},\{y_j\mid j\in J\}\subseteq|\ns D|$.
Thus using part~\ref{item.bullet.commute.tand} of Lemma~\ref{lemm.bullet.commute} and some elementary sets calculations
$X\cap Y=\bigcup_{i{\in}I,j{\in}J}\pp{(x_i\tand y_j)}$ and by Theorem~\ref{thrm.no.increase.of.supp} this is a small-supported union.
\item
Suppose $\mathcal X\subseteq\otop{F(\ns D)}$ is small-supported.
Using Theorem~\ref{thrm.no.increase.of.supp} so is $\bigcup\mathcal X\in\otop{F(\ns D)}$.
\qedhere\end{enumerate*}
\end{proof}

%%%%%%%%%%%%%%%%%%%%%%%%%%%%%%%%%
\subsection{Technical interlude: two important propositions}

For this subsection, fix $\ns D\in\ndia$ a nominal distributive lattice with $\tall$.

This subsection proves Proposition~\ref{prop.extra.filter}, and uses that to prove Proposition~\ref{prop.even.stronger}, which is an important technical lemma for Theorem~\ref{thrm.pp.x.clopen}.

Intuitively Proposition~\ref{prop.even.stronger} says that the set of all prime filters containing some $z{\in}\ns D$ is compact (so any cover of $\pp z$ has a subcover with a very strict bound on its size; see Theorem~\ref{thrm.pp.x.clopen} for the full result).
The reader familiar with duality results should recognise the overall argument, and we just need to do some extra work to account for the extra `nominal' structure.

Proposition~\ref{prop.extra.filter} is new, though the statement and proof have a similar outline; intuitively it says that any cover of $\pp z$ has a subcover with a very strict bound on its support.

Some notation will be useful:
\begin{defn}
\label{defn.ppX}
\label{defn.clostandup}
Given $\nomathcal X\subseteq|\ns D|$,
we define $\pp{\nomathcal X}$, ${\nomathcal X}_\tand$, ${\nomathcal X}_\tor$, ${\nomathcal X}{\downarrow}$, ${\nomathcal X}{\uparrow}$, $\clostordown{\nomathcal X}$, and $\clostandup{\nomathcal X}$ as in Figure~\ref{fig.lotsofstuff}.
\end{defn}

\begin{figure}
$$
\begin{array}{@{\hspace{-.0ex}}r@{\ }l@{\ }r@{\ }l}
\pp{{\nomathcal X}}=&\{\pp x\mid x{\in}{\nomathcal X}\}
\\
{\nomathcal X}_\tand =& \{\ttop\}\cup\{ x_1\tand\dots\tand x_n \mid \{x_1,\dots,x_n\}{\subseteq}{\nomathcal X} \}
\\
{\nomathcal X}_\tor =& \{\tbot\}\cup\{ x_1\tor\dots\tor x_n \mid \{x_1,\dots,x_n\}{\subseteq}{\nomathcal X} \}
\\
{\nomathcal X}{\downarrow} =& \{ x' \mid \Exists{x{\in}{\nomathcal X}}x'\leq x\}
\\
{\nomathcal X}{\uparrow} =& \{ x' \mid \Exists{x{\in}{\nomathcal X}}x\leq x'\}
\\
\clostordown{\nomathcal X} =& ({\nomathcal X}_\tor){\downarrow}
\\
\clostandup{\nomathcal X} =& ({\nomathcal X}_\tand){\uparrow}
\end{array}
$$
\caption{Some useful notation}
\label{fig.lotsofstuff}
\end{figure}

\begin{rmrk}
\begin{itemize}
\item
The notation ${\nomathcal X}{\downarrow}$ and ${\nomathcal X}{\uparrow}$ echoes $x{\downarrow}$ and $x{\uparrow}$ from Definition~\ref{defn.uparrow}, and indeed $x{\downarrow}$ from Definition~\ref{defn.uparrow} is equal to $\{x\}{\downarrow}$ from Definition~\ref{defn.clostandup}.
\item
Later, we will encounter Definition~\ref{defn.ppmone}, which is a kind of inverse to $\pp{{\nomathcal X}}$.
\item
By convention, we take $\ttop\in{\nomathcal X}_\tand$ and $\tbot\in{\nomathcal X}_\tor$.
Think of this as ``the case that $n{=}0$''.
\end{itemize}
\end{rmrk}

\begin{rmrk}
\label{rmrk.costandup.filter}
Intuitively, $\clostordown{\nomathcal X}$ is trying to be the least ideal (Definition~\ref{defn.ideal}) containing ${\nomathcal X}$; it may fail to be an ideal if $\ttop\in\clostordown{\nomathcal X}$.

Similarly, $\clostandup{\nomathcal X}$ is moving in the direction of being a filter (Definition~\ref{defn.filter}) containing ${\nomathcal X}$, though it may fail to be a filter if either $\tbot\in\clostandup{\nomathcal X}$, or $\New{b}(b\ a)\act x\in \clostandup{\nomathcal X}$ and $\tall a.x\not\in \clostandup{\nomathcal X}$.
\end{rmrk}

If we look ahead to Definition~\ref{defn.n.closed} then we can state Proposition~\ref{prop.extra.filter} intuitively as follows: a small-supported cover of $\pp z$ naturally generates a \emph{strictly} small-supported cover.
\begin{prop}
\label{prop.extra.filter}
Suppose $z{\in}|\ns D|$ and $\nomathcal X\subseteq|\ns D|$ is small-supported, and suppose
$\pp z\subseteq\bigcup\pp{\nomathcal X}$.
Then there exists $\nomathcal X'\subseteq\nomathcal X_\tor$ (Figure~\ref{fig.lotsofstuff}) such that:
\begin{enumerate*}
\item
$\nomathcal X'\subseteq\nomathcal X_\tor$ is strictly small-supported, and
\item
$\pp z\subseteq\bigcup\pp{\nomathcal X'}$.
\end{enumerate*}
\end{prop}
\begin{proof}
If $z=\tbot$ then by Lemma~\ref{lemm.bullet.commute}(\ref{item.bullet.commute.tbot}) $\pp z=\varnothing$ and we take $\nomathcal X'=\varnothing$ (by convention, $\bigcup\varnothing=\varnothing$).
So we assume henceforth that $z\neq\tbot$.

If for every $x{\in}\nomathcal X$ there exists $x'{\in}\nomathcal X_\tor$ with $\supp(x\tor x')\subseteq\supp(z)$,
then we build $\nomathcal X'$ to contain $x\tor x'$ for each $x\in\nomathcal X$ and some corresponding choice of $x'\in\nomathcal X_\tor$.
By construction $\nomathcal X'$ is strictly supported by $\supp(z)$ and $\pp z\subseteq\bigcup\pp{\nomathcal X'}=\bigcup\pp{\nomathcal X}$.

So now suppose there exists $x{\in}\nomathcal X$ such that for any $x'\in\nomathcal X_\tor$ it is the case that $\supp(x\tor x')\not\subseteq\supp(z)$.
Recall $\text{-}{\downarrow}$ from Figure~\ref{fig.lotsofstuff} and define $\nomathcal Y_x$ by
$$
\nomathcal Y_x=\{x\tor x' \mid x'\in\nomathcal X_\tor \}{\downarrow} .
$$
We note some properties of $\nomathcal Y_x$:
\begin{enumerate*}
\item
$x{\in}\nomathcal Y_x$ by construction (note that $\tbot\in\nomathcal X_\tor$).
\item
$\ttop\not\in\nomathcal Y_x$, since $\supp(\ttop)=\varnothing\subseteq\supp(z)$.
\item
$\nomathcal Y_x$ is down-closed, by the use of ${\downarrow}$ (Definition~\ref{defn.uparrow}).
\item
If $y\in\nomathcal Y_x$ and $y'\in\nomathcal Y_x$ then $y\tor y'\in \nomathcal Y_x$, by the use of $\nomathcal X_\tor$ and ${\downarrow}$ in the construction of $\nomathcal Y_x$.
\item
Thus $\nomathcal Y_x$ is an ideal (Definition~\ref{defn.ideal}).\footnote{$\nomathcal Y_x$ is also small-supported by Theorem~\ref{thrm.no.increase.of.supp}, since $\nomathcal X$ and $x$ are small-supported.}
\item
$\bigcup\pp{\nomathcal Y_x}=\bigcup\pp{\nomathcal X}$, by construction.
\end{enumerate*}
There are now two possibilities: $z\in\nomathcal Y_x$ or $z\not\in\nomathcal Y_x$.
We treat each in turn:
\begin{itemize}
\item
\emph{Suppose $z\in\nomathcal Y_x$.}\quad
Then $z\leq x\tor x'$ for some $x'{\in}\nomathcal X_\tor$ and we can take
$\nomathcal X'=\{x\tor x'\}\subseteq\nomathcal X_\tor$.\footnote{It may not be that $\supp(x\tor x')\subseteq\supp(z)$, but we do not care; we only need $\nomathcal X'$ to be strictly supported.}
\item
\emph{Suppose $z\not\in\nomathcal Y_x$.}\quad
We noted above that $\nomathcal Y_x$ is a (small-supported) ideal, and by Lemma~\ref{lemm.uparrow.filter} $z{\uparrow}$ is a small-supported filter (recall we assumed above that $z{\neq}\tbot$).
By Theorem~\ref{thrm.maxfilt.zorn} there exists a prime filter $p$ such that $z\in p$ and $p\cap\nomathcal Y_x=\varnothing$.

By Definition~\ref{defn.pp} $p\in\pp z$ and by assumption $\pp z\subseteq\bigcup\pp{\nomathcal X}$, so there exists $x'\in\nomathcal X$ with $p\in\pp{x'}$, thus $x'\in p$, thus by condition~\ref{filter.up} of Definition~\ref{defn.filter} $x\tor x'\in p$.
However by assumption $x\tor x'\in \nomathcal Y_x$, a contradiction.
So this case is impossible.
\qedhere\end{itemize}
\end{proof}

\begin{prop}
\label{prop.even.stronger}
Suppose $z\in|\ns D|$ and suppose $\nomathcal Y{\subseteq}|\ns D|$ is small-supported.
Then if $\pp z\subseteq\bigcup\pp{\nomathcal Y}$ then $z\leq y$ for some $y\in\nomathcal Y_\tor$.
\end{prop}
\begin{proof}
If $z=\tbot$ then by Lemma~\ref{lemm.bullet.commute}(\ref{item.bullet.commute.tbot}) $\pp z=\varnothing$ and we may take $y=\tbot\in\nomathcal Y_\tor$.
So suppose $z\neq\tbot$; it follows using Corollary~\ref{corr.pp.injective} that $\pp z\neq\varnothing$ and therefore that $\nomathcal Y$ is nonempty.

Using Proposition~\ref{prop.extra.filter} we may assume without loss of generality
that $\nomathcal Y$ is \emph{strictly} small-supported.
Write
\begin{equation}
\label{eq.X}
\nomathcal X=\clostandup{\{x \mid \Exists{y{\in}\nomathcal Y}(z\leq y{\tor} x)\}}\subseteq|\ns D| .
\end{equation}
If $\tbot\in \nomathcal X$ then $\Exists{y{\in}\nomathcal Y}z\leq y$ and we are done.
So suppose $\tbot\not\in \nomathcal X$ so that $\nomathcal X$ satisfies condition~\ref{filter.proper} of Definition~\ref{defn.filter}.

Now $\nomathcal Y$ is nonempty, and since $z\leq y\tor\ttop$ for any $y{\in}\nomathcal Y$, also $\ttop\in \nomathcal X$ so that $\nomathcal X$ is nonempty.
We now observe that $\nomathcal X$ also satisfies conditions~\ref{filter.up}, \ref{filter.and}, and~\ref{filter.new} of Definition~\ref{defn.filter} and so is a filter:
\begin{enumerate*}
\setcounter{enumi}{1}
\item
\emph{$\nomathcal X$ is up-closed.}\quad
By the use of $\clostandup{\text{-}}$.
\item
\emph{$x\in \nomathcal X$ and $x'\in \nomathcal X$ imply $x\tand x'\in \nomathcal X$.}\quad
By the use of $\clostandup{\text{-}}$.
\item\label{item.check.filter.new}
\emph{If $\New{b}((b\ a)\act x\in \nomathcal X)$ then $\tall a.x\in \nomathcal X$.}\quad
Choose fresh $b$ (so $b\#z,x,\nomathcal X,\nomathcal Y$), so there exist $y_i{\in}\nomathcal Y$ for $1{\leq}i{\leq}n$ and $x_1,\dots,x_n$ such that $z\leq y_i{\tor}x_i$ for $1{\leq}i{\leq}n$, and $(b\ a)\act x\geq x_1{\tand}\dots{\tand}x_n$.

Now $b\#\nomathcal Y$ and by assumption above 
$\nomathcal Y$ is strictly small-supported, so that by Lemma~\ref{lemm.strict.support}(\ref{strict.to.supp.elements}) and Corollary~\ref{corr.supp.pp} also $b\#y_i$ for $1{\leq}i{\leq}n$.
We then have by part~\ref{b.bigger.3} of Lemma~\ref{lemm.b.bigger} that $z\leq y_i \tor\,\tall b.x_i$ for $1{\leq}i{\leq}n$.

Also, using Lemmas~\ref{lemm.tand.tall} and~\ref{lemm.tall.monotone} $\tall b.(b\ a)\act x\geq (\tall b.x_1){\tand}\dots{\tand}\tall b.x_n$.

It follows using Lemma~\ref{lemm.freshwedge.alpha} that $\tall a.x\in \nomathcal X$.
\end{enumerate*}
It is now useful to consider two distinct possibilities: $\nomathcal X\cap\clostordown{\nomathcal Y}\neq \varnothing$ or $\nomathcal X\cap\clostordown{\nomathcal Y}=\varnothing$.
We treat each in turn:
\begin{itemize*}
\item
\emph{Suppose $\nomathcal X\cap\clostordown{\nomathcal Y}\neq\varnothing$.}\quad
So take any $x\in \nomathcal X\cap\clostordown{\nomathcal Y}$.

Since $x\in\clostordown{\nomathcal Y}$, there exist $y_1',\dots,y_m'\in\nomathcal Y$ with $x\leq y_1'{\tor}\dots{\tor}y_m'$.

Since $x\in \nomathcal X$, there exist $y_1,\dots,y_n\in\nomathcal Y$ and $x_1,\dots,x_n\in \nomathcal X$ such that $z\leq y_i{\tor}x_i$ for $1{\leq}i{\leq}n$ and $x\geq x_1{\tand}\dots{\tand}x_n$.

We note that $z\leq (y_1{\tor}x_1)\tand\dots\tand(y_n{\tor}x_n)$ and by some easy calculations that
$$
z\leq x{\tor}y_1{\tor}\dots{\tor}y_n
\leq y_1'{\tor}\dots{\tor}y_m'\tor y_1{\tor}\dots{\tor}y_n ,
$$
and we are done.
\item
\emph{Now suppose $\nomathcal X\cap\clostordown{\nomathcal Y}=\varnothing$.}\quad
We noted above that $\ttop\in \nomathcal X$, so $\ttop\not\in\clostordown{\nomathcal Y}$ and it follows that $\clostordown{\nomathcal Y}$ is an ideal.

Also, $\nomathcal Y$ is is small-supported, and since $z$ is also small-supported we can apply Theorem~\ref{thrm.no.increase.of.supp} to equation~\ref{eq.X} to see that $\nomathcal X$ is small-supported.

Therefore by Theorem~\ref{thrm.maxfilt.zorn} $\nomathcal X\subseteq p$ for some prime filter $p$ such that $p\cap \clostordown{\nomathcal Y}=\varnothing$.
It is easy to check that $z\in \nomathcal X$, so that $z\in p$ and by Definition~\ref{defn.pp} $p\in\pp z$.
By assumption $\pp z\subseteq\bigcup\pp{\nomathcal Y}$ so there exists some $y\in\nomathcal Y$ with $p\in \pp y$, that is, $y\in p$.
But our assumption that $p\cap\clostordown{\nomathcal Y}=\varnothing$ implies that $p\cap\nomathcal Y=\varnothing$, contradicting this.
So this case is impossible.
\qedhere\end{itemize*}
\end{proof}

%%%%%%%%%%%%%%%%%%%%%%%%%%%%%%
\subsection{Compactness}

Fix $\ns T$ a nominal $\sigma$-topological space.

\begin{defn}
\label{defn.n.closed}
Suppose $\mathcal U\subseteq\otop{\ns T}$ and $U\in\otop{\ns T}$.
\begin{itemize*}
\item
Say $\mathcal U$ \deffont{covers} $U$ when $\mathcal U$ is small-supported and $U\subseteq\bigcup \mathcal U$.
Call $\mathcal U$ a \deffont{cover} when it covers $|\ns T|$.
\item
Call $U$ \deffont{compact} when every cover of $U$ has a finite subcover.
Write $\ctop{\ns T}$ for the set of compact open sets of $\ns T$:
$$
\ctop{\ns T} = \{ U\in\otop{\ns T} \mid U\text{ is compact}\}
$$
\end{itemize*}
\end{defn}

Proposition~\ref{prop.really.do} looks familiar enough, but we are in a nominal context so we have to check facts about support.
It all works:
\begin{prop}
\label{prop.really.do}
Suppose $U,V\in\ctop{\ns T}$.
Then:
\begin{enumerate*}
\item
$\varnothing$ is compact.
\item
$U{\cup} V$ is compact.
\item
$\pi\act U$ is compact.
\end{enumerate*}
\end{prop}
\begin{proof}
\begin{enumerate*}
\item
There are two covers of $\varnothing$: the empty set of open sets and the set containing the empty set of points.
Both are finite.
\item
Suppose $U$ and $V$ are compact and $\mathcal W$ covers $U{\cup} V$.

By Theorem~\ref{thrm.no.increase.of.supp} $\supp(\{W{\cap} U \mid W{\in}\mathcal W\})\subseteq\supp(\mathcal W){\cup}\supp(U)$ so it is small-supported, and it follows that
$\{W{\cap} U \mid W{\in}\mathcal W\}$ covers $U$, and we obtain a finite subcover of $U$.

Reasoning similarly for $\{W{\cap} V \mid W{\in}\mathcal W\}$ we obtain a finite subcover of $V$.
Putting these two finite subcovers together, we obtain one of $U\cup V$.
\item
Using Theorem~\ref{thrm.equivar} and Proposition~\ref{prop.pi.supp}.
\qedhere\end{enumerate*}
\end{proof}

In a sense, Definition~\ref{defn.coherent} is a continuation of Proposition~\ref{prop.really.do}.

\begin{thrm}
\label{thrm.pp.x.clopen}
Suppose $\ns D\in\ndia$ and recall $F(\ns D)$ from Definition~\ref{defn.F}.
Then:
\begin{enumerate*}
\item
If $x\in|\ns D|$ then $\pp x$ is open and compact in $F(\ns D)$.
\item
If $U$ is open and compact in $F(\ns D)$ then $U=\pp x$ for some unique $x\in |\ns D|$.
\end{enumerate*}
\end{thrm}
\begin{proof}
We consider each part in turn:
\begin{enumerate}
\item
$\pp x$ is open by construction in Definition~\ref{defn.F}.
Now consider a cover $\mathcal U$ of $\pp x$.
By Definitions~\ref{defn.n.closed} and~\ref{defn.F}(\ref{F.strict}),\ $\mathcal U$ is a small-supported set of unions of small-supported sets of elements, and all these elements have the form $\pp y$ for $y{\in}|\ns D|$.

So we can write $\bigcup\mathcal U=\bigcup\pp{\mathcal Y}$ for some (by Theorem~\ref{thrm.no.increase.of.supp}) small-supported $\mathcal Y{\subseteq}|\ns D|$.
To find a finite subcover of $\mathcal U$, it would suffice to find a finite subset $\mathcal X'\subseteq\mathcal Y_\tor$ such that $\pp x\subseteq\bigcup\pp{\mathcal X'}$.

We use Proposition~\ref{prop.even.stronger}.
\item
By construction in Definition~\ref{defn.F} the open sets of $F(\ns D)$ are unions of small-supported sets of sets the form $\pp x$.
We assumed that $U$ is compact so it has a finite subcover $\pp x_1\cup\dots\cup\pp x_n$.
We use part~\ref{item.bullet.commute.tor} of Lemma~\ref{lemm.bullet.commute}.
Uniqueness is Corollary~\ref{corr.pp.injective}.
\qedhere\end{enumerate}
\end{proof}

%%%%%%%%%%%%%%%%%%%%%%%%%%%%%%%%%%%%%%%%%%%%%%%%%%%%%%%%
\subsection{Coherent spaces: closure under $\sigma$, $\cap$ and $\protect\freshcap{a}$}
\label{subsect.coherence}

\emph{Coherence} usually means that the compact open sets are closed under lattice operations and generate all open sets via sets unions.
Our lattices have more structure, notably: a $\sigma$-action, and $\freshcap{a}$ from Definition~\ref{defn.nu.U}.
Also, our notion of `generating' open sets has nominal aspects to it; see condition~\ref{item.top.strict} of Definition~\ref{defn.nom.top}.

Definition~\ref{defn.coherent} is how we extend the notion of coherence to account for this structure.
Proposition~\ref{prop.FB.new} then checks that $F$ from Definition~\ref{defn.F} does indeed generate coherent spaces.

\begin{frametxt}
\begin{defn}
\label{defn.coherent}
Call a nominal $\sigma$-topological space $\ns T$ \deffont{coherent} when:
\begin{enumerate*}
\item
\label{item.coherent.sub}
If $U$ is open and compact then so is $U[a\sm u]$ for every $u\in|\ns T^\prg|$.
\item
$|\ns T|$ is (open and) compact, and
if $U$ and $V$ are open and compact then so is $U\cap V$.
\item\label{coherent.freshcap}
If $U$ is open and compact then so is $\freshcap{a}U$.
\item
Every open $U\in\otop{\ns T}$ is equal to $\bigcup\mathcal U$ for some small-supported $\mathcal U\subseteq\ctop{\ns T}$.
\end{enumerate*}
\end{defn}
\end{frametxt}

\begin{rmrk}
\label{rmrk.where.pointwise}
The $U[a\sm u]$ mentioned in Definition~\ref{defn.coherent} is the \emph{pointwise} action from Definition~\ref{defn.sub.sets}
$$
U[a\sm u]=\{p \mid \New{c} p[u\ms c]\in (c\ a)\act X\} .
$$
It is applicable to any set of points (not just a compact set of points) and is inherited from the $\amgis$-action $p[u\ms c]$ which we assumed on the underlying points.
\end{rmrk}

\begin{rmrk}
We rewrite Definition~\ref{defn.coherent} in less precise, but more intuitive language:
\begin{enumerate*}
\item
Compactness is closed under the $\sigma$-action, so the compact sets form a $\sigma$-algebra over $\ns T^\prg$.
\item
Compactness is closed under (possibly empty) sets intersection.
\item
Compactness is closed under universal quantification.
\item
Compact open sets are a small-supported (Subsection~\ref{subsect.strict.pow}) basis for all open sets.
\end{enumerate*}
\end{rmrk}

\begin{prop}
\label{prop.FB.new}
Suppose $\ns D\in\ndia$.
Then $F(\ns D)$ (Definition~\ref{defn.F}) is coherent.
\end{prop}
\begin{proof}
By Theorem~\ref{thrm.pp.x.clopen} we can identify the compact open sets of $F(\ns D)$ with sets of the form $\pp x$ for $x\in|\ns D|$.
We now reason as follows:
\begin{enumerate*}
\item
By Lemma~\ref{lemm.dup} $\pp x[a\sm u]=\pp{(x[a\sm u])}$.
\item
By part~1 of Lemma~\ref{lemm.tall.cent.char} $\pp\ttop=|F(\ns D)|$ (every point contains $\ttop$), and by part~\ref{item.bullet.commute.tand} of Lemma~\ref{lemm.bullet.commute} $\pp x\cap\pp y=\pp{(x\tand y)}$.
\item
By part~\ref{item.bullet.commute.tall} of Lemma~\ref{lemm.bullet.commute} $\pp{(\tall a.x)}=\freshcap{a}\pp x$.
\item
By construction in Definition~\ref{defn.F} open sets are small-supported unions of the $\pp x$.
\qedhere\end{enumerate*}
\end{proof}

%%%%%%%%%%%%%%%%%%%%%%%%%%%%%%%%%%%%%%%
\subsection{Completely prime filters in a coherent space}

Recall the notion of \emph{prime} filter from Definition~\ref{defn.prime.filter}.
A stronger property will also be of interest; this definition is standard:
\begin{defn}
Suppose $\ns T$ is a nominal $\sigma$-topological space and suppose $\mathcal U\subseteq\otop{\ns T}$ is a filter (Definition~\ref{defn.filter}) of open sets in $\ns T$.

Call a filter $\mathcal U\subseteq\otop{\ns T}$ \deffont{completely prime} when:
\begin{itemize*}
\item
if $\mathcal V\subseteq\otop{\ns T}$ is small-supported and $\bigcup\mathcal V\in \mathcal U$,
\item
then $V\in\mathcal U$ for some $V\in \mathcal V$.\footnote{A \emph{prime} filter satisfies this property---for \emph{finite} $\mathcal V$.}
\end{itemize*}
\end{defn}

For coherent spaces, a more economical characterisation will be useful:
\begin{lemm}
\label{lemm.natural.bij}
If $\ns T$ is coherent then the completely prime filters of open sets are in a natural bijection with the prime filters of compact open sets, with the bijection given by:
\begin{itemize*}
\item
A completely prime filter $\mathcal U'\subseteq\otop{\ns T}$ corresponds to $\mathcal U=\mathcal U'\cap\ctop{\ns T}$.
\item
A prime filter $\mathcal U\subseteq\ctop{\ns T}$ corresponds to $\mathcal U'=\{U'{\in}\otop{\ns T}\mid \Exists{U{\in}\mathcal U}U\subseteq U'\}$, the \deffont{up-closure} of $\mathcal U$ in $\otop{\ns T}$.
\end{itemize*}
\end{lemm}
\begin{proof}
Suppose $\mathcal U'$ is a completely prime filter in $\otop{\ns T}$.
We will show that $\mathcal U=\mathcal U'\cap\ctop{\ns T}$ is a filter (Definition~\ref{defn.filter}) and is prime (Definition~\ref{defn.prime.filter}).
\begin{enumerate}
\item
$\varnothing\not\in\mathcal U$ since $\varnothing\not\in\mathcal U'$.
\item
If $U\in\mathcal U$ and $U'\in\ctop{\ns T}$ and $U\subseteq U'$ then $U'\in\mathcal U'$ so $U'\in\mathcal U$.
\item
It follows similarly that $U,U'\in\mathcal U$ implies $U\cap U'\in\mathcal U$.
\item
Suppose $U\in\ctop{\ns T}$ and suppose $\New{b}(b\ a)\act U\in\mathcal U$.
Then also $\New{b}(b\ a)\act U\in\mathcal U'$ so by condition~\ref{filter.new} of Definition~\ref{defn.filter} (since $\mathcal U'$ is a filter) the filter condition on $\mathcal U'$,\ $\freshcap{a}U\in\mathcal U'$.
By coherence (Definition~\ref{defn.coherent} condition~\ref{coherent.freshcap}) $\freshcap{a}U\in\ctop{\ns T}$, therefore $\freshcap{a}U\in\mathcal U$ as required.
\end{enumerate}
By part~2 of Proposition~\ref{prop.really.do} compactness is closed under finite unions and it follows that $\mathcal U$ is prime.

Conversely suppose $\mathcal U\subseteq\ctop{\ns T}$ is a prime filter.
First we will show that $\mathcal U'=\{U'{\in}\otop{\ns T}\mid \Exists{U{\in}\mathcal U}U\subseteq U'\}$ is a filter, then we will show it is completely prime.
\begin{enumerate}
\item
$\varnothing\not\in\mathcal U'$ since $\varnothing\not\in\mathcal U$.
\item
If $U\in\mathcal U'$ and $U'\in\otop{\ns T}$ and $U\subseteq U'$ then by construction $U'\in\mathcal U'$.
\item
It follows similarly that $U,U'\in\mathcal U'$ implies $U\cap U'\in\mathcal U$.
\item
Suppose $U'\in\otop{\ns T}$ and suppose $\New{b}(b\ a)\act U'\in\mathcal U'$.
We need to show that $\freshcap{a}U'\in\mathcal U'$.

Choose fresh $b$ (so $b\#U',\mathcal U,\mathcal U'$).
By Theorem~\ref{thrm.new.equiv} $(b\ a)\act U'\in\mathcal U'$ and this means that there exists $U\in\mathcal U$ such that $U\subseteq(b\ a)\act U'$.
It follows by Theorem~\ref{thrm.equivar} and Corollary~\ref{corr.stuff} that $\New{b'}(b'\ b)\act U\in\mathcal U$.
By condition~\ref{filter.new} of Definition~\ref{defn.filter} (since $\mathcal U$ is a filter) $\freshcap{b'}(b'\ b)\act U\in\mathcal U$ so that by Lemma~\ref{lemm.all.alpha} (since $b'\#U$) $\freshcap{b}U\in\mathcal U$.

By assumption $U\subseteq (b\ a)\act U'$ so by Lemma~\ref{lemm.tall.monotone} and Theorem~\ref{thrm.all.closed} also $\freshcap{b}(b\ a)\act U'\in\mathcal U'$.
By Lemma~\ref{lemm.all.alpha} again (since $b\#U'$) $\freshcap{a}U'\in\mathcal U'$, as required.
\end{enumerate}

Now suppose $\mathcal V\subseteq\otop{\ns T}$ is small-supported and suppose $\bigcup\mathcal V\in\mathcal U'$, so that $U\subseteq\bigcup\mathcal V$ for some $U\in\mathcal U$.
But this just states that $\mathcal V$ covers $U$, and by compactness $\mathcal V$ has a finite subcover $\{V_1,\dots,V_n\}\subseteq\mathcal V$.
It follows that $\bigcup\{V_1\cap U,\dots,V_n\cap U\}=U\in\mathcal U$.
Since $\mathcal U$ is prime it follows that $V_i\cap U\in\mathcal U$ for some $i$, and therefore that $V_i\in\mathcal U'$.

It is routine to verify that
the correspondences between $\mathcal U$ and $\mathcal U'$ defined above are bijective.
\end{proof}

%%%%%%%%%%%%%%%%%%%%%%%%%%%%%%
\subsection{Impredicativity}

We saw in Subsection~\ref{subsect.impredicative.nom.dist.lat} and Definition~\ref{defn.D.impredicative} a notion of \emph{impredicativity}, based on the idea that the things we substitute for should map to the things we substitute in.
In the context of a topological space, this means that $\ns T^\prg$ should map to open sets.
This is Definition~\ref{defn.impredicative.top}, and Theorem~\ref{thrm.T.to.G.obj} shows how $F(\ns D)$ inherits any impredicative structure of $\ns D$.

We can think of Definition~\ref{defn.impredicative.top} as a dual to Definition~\ref{defn.D.impredicative}, for the nominal $\sigma$-topological spaces from Definition~\ref{defn.nom.top}:
\begin{frametxt}
\begin{defn}
\label{defn.impredicative.top}
An \deffont{impredicative nominal $\sigma$-topological space} is a pair
$(\ns T,\prg_{\ns T})$ where:
\begin{enumerate*}
\item
$\ns T$ is a nominal $\sigma$-topological space (Definition~\ref{defn.nom.top}).
\item
$\prg_{\ns T}:\ns T^\prg\to\ctop{\ns T}$ is a morphism of $\sigma$-algebras (Definition~\ref{defn.morphism.sigma.alg}).
\item\label{size.limit.top}
$\ctop{\ns T}$ has cardinality no greater than $\size(\mathbb A)$ (Definition~\ref{defn.atoms}).
\end{enumerate*}
\end{defn}
\end{frametxt}

\begin{nttn}
\label{nttn.impredicative.T}
Following Notation~\ref{nttn.impredicative.D} we introduce some notation for Definition~\ref{defn.impredicative.top}:
\begin{itemize*}
\item
We may drop subscripts and write $\prg u$ for $\prg_{\ns T} u$ where $u\in|\ns T^\prg|$.
\item
We may write $\prg a$ for $\prg_{\ns T}(a_{\ns T^\prg})$ where $a_{\ns T^\prg}=\tf{atm}_{\ns D^\prg}(a)$ (Definition~\ref{defn.term.sub.alg}).
\item
We may write $\prg{\ns T}$ for $\{\prg u\mid u\in|\ns T^\prg|\}\subseteq\ctop{\ns T}$ and call this set the \deffont{programs} of $\ns T$.
\end{itemize*}
The exposition in and following Notation~\ref{nttn.impredicative.D} is also valid here, so we do not repeat it.
\end{nttn}

\begin{thrm}
\label{thrm.FB.totsep.comp}
If $\ns D\in\india$ (so $\ns D$ is impredicative) then $F(\ns D)$ is also naturally impredicative.
\end{thrm}
\begin{proof}
We take $F(\ns D)^\prg=\ns D^\prg$ and $\prg_{F(\ns D)} u=\pp{(\prg_{\ns D} u)}$.
(In fact, this is an injection by Theorem~\ref{thrm.pp.iso}.)
\end{proof}

%%%%%%%%%%%%%%%%%%%%%%%%%%%%%%%%%%%%%%%%%%%%%%%%
\subsection{The map $G$ from coherent spaces to distributive lattices}

\begin{defn}
\label{defn.G}
Suppose $\ns T$ is a coherent (Definition~\ref{defn.coherent}) nominal $\sigma$-topological space.
\begin{frametxt}
Define $G(\ns T)$
as follows:
\begin{itemize*}
\item
$|G(\ns T)|=\ctop{\ns T}$ (compact opens) and $G(\ns T)^\prg=\ns T^\prg$.
\item
Given $U\in|G(\ns T)|$ and $u\in|G(\ns T)^\prg|{=}|\ns T^\prg|$ define $\pi\act U$ and $U[a\sm u]$ following Definition~\ref{defn.sub.sets}.
\item
$\ttop$,
$\tand$, $\tbot$, $\tor$, and $\tall$ are interpreted as the whole underlying set $\ns T$,
set intersection $\cap$, the empty set $\varnothing$, set union $\cup$, and $\freshcap{a}$ from Definition~\ref{defn.nu.U}.
\end{itemize*}
\end{frametxt}
\end{defn}

\begin{rmrk}
We unpack Definition~\ref{defn.sub.sets} for $U{\in}\ctop{\ns T}$ for the reader's convenience:
$$
\pi\act U=\{\pi\act \topel\mid \topel\in U\}
\qquad
U[a\sm u]=\{\topel\mid \New{c}\topel[u\ms c]\in (c\ a)\act U\} .
$$
\end{rmrk}

\begin{thrm}
\label{thrm.T.to.G.obj}
Continuing Definition~\ref{defn.G}, if $\ns T$ is coherent then
$G(\ns T)$ is a nominal distributive lattice with $\tall$.

Furthermore, if $\ns T$ is impredicative then so naturally is $G(\ns T)$.
\end{thrm}
\begin{proof}
By Proposition~\ref{prop.really.do} $\tbot$, $\ttop$,
$\tor$, and the permutation action give results in $|G(\ns T)|$.
By our assumption that $\ns T$ is coherent, so do $\tand$, the $\sigma$-action and $\tall$.
We use Theorem~\ref{thrm.powerset}.

Now suppose $\ns T$ is impredicative, so it is equipped with a %n injective
$\sigma$-algebra morphism $\prg_{\ns T}:\ns T^\prg\to\ctop{\ns T}$ (Definition~\ref{defn.morphism.sigma.alg}).
We take $\prg_{G(\ns T)}=\prg_{\ns T}$.
\end{proof}

\begin{prop}
\label{prop.ess.surj}
If $\ns D\in\ndia/\india$ then $GF(\ns D)$ is equal to $\pp{\ns D}$ from Definition~\ref{defn.pp.B}, and the map $x\mapsto \pp x$ is an isomorphism in $\ndia/\india$.
\end{prop}
\begin{proof}
By Theorem~\ref{thrm.pp.x.clopen} $|GF(\ns D)|=|\pp{\ns D}|$.
We use Theorem~\ref{thrm.pp.iso}.
\end{proof}

%%%%%%%%%%%%%%%%%%%%%%%%%%%%%%
\subsection{Sober spaces}

\begin{frametxt}
\begin{defn}
\label{defn.sober}
Call a nominal $\sigma$-topological space $\ns T$ \deffont{sober} when:
\begin{enumerate}
\item\label{sober.standard}
If $\mathcal U\subseteq\otop{\ns T}$ is a completely prime filter then there exists a unique $\topel_{\mathcal U}\in|\ns T|$ such that
$$
\Forall{U{\in}\otop{\ns T}}U\in\mathcal U
\quad\text{if and only if}\quad
\topel_{\mathcal U}\in U.
$$
\item\label{sober.all}
If $\topel\in|\ns T|$ and $U\in\ctop{\ns T}$ then
$$
\New{b}\topel\in (b\ a)\act U
\quad\text{implies}\quad
\topel\in\freshcap{a}U.
$$
\end{enumerate}
\end{defn}
\end{frametxt}

\begin{rmrk}
\label{rmrk.sober.second.condition}
Condition~\ref{sober.standard} of Definition~\ref{defn.sober} is the standard notion of sobriety.
Intuitively, it states that completely prime filters characterise the underlying points of the space.

Condition~\ref{sober.all} is clearly related to condition~\ref{filter.new} of Definition~\ref{defn.filter}, and carries the intuition that universal quantification is determined by the behaviour at a fresh atom (see Remark~\ref{rmrk.interpret.new.filter} and the discussion following Definition~\ref{defn.filter}).
We use it for just that in Proposition~\ref{prop.tast.filter}.
More on this in Remark~\ref{rmrk.the.structure}.
\end{rmrk}

Lemma~\ref{lemm.topel.sober.equiv} is related to Proposition~\ref{prop.these.are.equivalent}. 
Technically, it will be useful later in Proposition~\ref{prop.G.funct}.
Intuitively, it corresponds to the quantifier intro- and elim-rules in logic that if $a$ is not free in $\Gamma$ then $\Gamma\cent\phi$ is derivable if and only if $\Gamma\cent\tall a.\phi$ is derivable (see the discussion in Remark~\ref{rmrk.interpret.new.filter}):
\begin{lemm}
\label{lemm.topel.sober.equiv}
Suppose $\ns T$ is a sober nominal $\sigma$-topological space and $U\in\ctop{\ns T}$.
Then $\topel\in\freshcap{a}U$ if and only if $\New{b}\topel\in(b\ a)\act U$.
\end{lemm}
\begin{proof}
The right-to-left implication is condition~\ref{sober.all} of Definition~\ref{defn.sober}.
For the left-to-right implication, suppose $\topel\in\freshcap{a}U$.
By Definition~\ref{defn.nu.U} $\topel\in U[a\sm u]$ for all $u{\in}|\ns T^\prg|$ and (just as in the proof of Proposition~\ref{prop.these.are.equivalent}(\ref{these.are.equivalent.1})) $\topel\in U[a\sm b]$ for fresh $b$ so using Lemma~\ref{lemm.sub.alpha} $\New{b}\topel\in (b\ a)\act U$.
\end{proof}

\begin{defn}
\label{defn.qq}
Suppose $\ns T$ is a nominal $\sigma$-topological space %$\ns T\in\ntop$
and $\topel\in|\ns T|$.
Define
\begin{frameqn}
\qq \topel=\{U\in\ctop{\ns T} \mid \topel\in U\},
\end{frameqn}
so that
$$
U\in \qq\topel \liff \topel\in U.
$$
\end{defn}

Recall from Definition~\ref{defn.F} that if $\ns D\in\india$ then $|F(\ns D)|$ is the set of prime filters in $\ns D$.
\begin{prop}
\label{prop.tast.filter}
Suppose a nominal $\sigma$-topological space $\ns T$ is coherent (Definition~\ref{defn.coherent}) and suppose $\topel\in|\ns T|$.
Then $\qq \topel$ is a prime filter in $G(\ns T)$, so that $\qq\topel$ is an element of $|FG(\ns T)|$.
\end{prop}
\begin{proof}
First we check that $\qq\topel$ is a filter (Definition~\ref{defn.filter}).
Conditions~1, 2, and~3 of Definition~\ref{defn.filter} are easy to check.
For condition~\ref{filter.new} it suffices to show that if $U$ is open and compact in $\ns T$ and $\New{b}\topel\in (b\ a)\act U$ then $\topel\in \freshcap{a}U$.
This is condition~\ref{sober.all} of Definition~\ref{defn.sober}.

It is a fact that $\qq \topel$ is prime, since if $\topel\in U\cup V$ then $\topel\in U$ or $\topel\in V$.
\end{proof}

\begin{corr}
\label{corr.qq.max}
Suppose a nominal $\sigma$-topological space $\ns T$ is coherent and sober.
Then $\mathcal U$ is a prime filter in $G(\ns T)$ if and only if $\mathcal U=\qq \topel$ for some $\topel\in|\ns T|$, and that $\topel$ is unique.

As a corollary, the map $\topel\mapsto\qq\topel$ is a bijection between $|\ns T|$ and $|FG(\ns T)|$.
\end{corr}
\begin{proof}
Suppose we are given $\topel\in|\ns T|$.
By Proposition~\ref{prop.tast.filter} $\qq\topel$ is a prime filter in $G(\ns T)$, so we map $\topel$ to $\qq\topel\in|FG(\ns T)|$.

Conversely suppose we are given $\mathcal U{\in}|FG(\ns T)|$, so $\mathcal U$ is a prime filter in $G(\ns T)$.
By Lemma~\ref{lemm.natural.bij} $\mathcal U'=\{U'{\in}\otop{\ns T}\mid \Exists{U{\in}\mathcal U}U\subseteq U'\}$ is a completely prime filter.
Since $\ns T$ is sober we can uniquely map $\mathcal U'$ to an element $\topel_{\mathcal U'}\in\ns T$ such that $\mathcal U'=\qq{\topel_{\mathcal U'}}$.

The bijection follows.
\end{proof}

\begin{lemm}
\label{lemm.coherent.iff}
Suppose a nominal $\sigma$-topological space $\ns T$ is coherent and assume:
\begin{enumerate}
\item
The map $\topel\in|\ns T|\longmapsto\qq \topel\in|FG(\ns T)|$ (a point maps to the prime filter of compact open sets containing it) is a bijection.
\item
If $\topel\in |\ns T|$ and $U\in\ctop{\ns T}$ then $\New{b}\topel\in (b\ a)\act U$ implies $\topel\in\freshcap{a}U$.
\end{enumerate}
Then $\ns T$ is sober.
\end{lemm}
\begin{proof}
We examine Definition~\ref{defn.sober} and see that we must verify two conditions:
\begin{enumerate}
\item
Suppose $\mathcal U'\subseteq\otop{\ns T}$ is a completely prime filter in $\ns T$.
We use the correspondence of Lemma~\ref{lemm.natural.bij} to map to $\mathcal U=\mathcal U'{\cap}\ctop{\ns T}$ and then assumption~1 above to obtain a unique $\topel_{\mathcal U}$ such that $U\in\mathcal U$ if and only if $\topel_{\mathcal U}\in U$.
It follows by construction of $\mathcal U$ that $U'\in\mathcal U'$ if and only if $\topel_{\mathcal U}\in U'$.
\item
Condition~2 is immediate.
\qedhere\end{enumerate}
\end{proof}

Definition~\ref{defn.ppmone} is in some sense dual to Definition~\ref{defn.qq}, and will be useful:
\begin{defn}
\label{defn.ppmone}
If $\mathcal Y\subseteq\ctop{F(\ns D)}$ then define $\ppmone{\mathcal Y}\subseteq|\ns D|$ by
\begin{frameqn}
\ppmone{\mathcal Y}=\{y\in|\ns D| \mid \pp y\in \mathcal Y\} ,
\end{frameqn}
so that $y\in\ppmone{\mathcal Y}\liff \pp y\in\mathcal Y$.
\end{defn}

\begin{rmrk}
\label{rmrk.strip}
By Theorem~\ref{thrm.pp.x.clopen} an element of $\mathcal Y$ in Definition~\ref{defn.ppmone} has the form $\pp y$ for some $y{\in}|\ns D|$.
So intuitively, Definition~\ref{defn.ppmone} just `strips the $\bullet$s' from inside $\mathcal Y$.
\end{rmrk}

Recall the definitions of $\pp x$ and $\qq p$ from Definitions~\ref{defn.pp} and \ref{defn.qq} respectively.
\begin{corr}
\label{corr.FD.sober}
Suppose $\ns D\in\ndia$.
Then:
\begin{enumerate*}
\item
If $p\in |F(\ns D)|$ then $\ppmone{(\qq p)}=p$.
\item
If $\mathcal U\in |FGF(\ns D)|$ then $\qq{(\ppmone{\mathcal U})}=\mathcal U$.
\item
If $p\in |F(\ns D)|$ and $U\in\ctop{F(\ns D)}$ then $\New{b}p\in (b\ a)\act U$ implies $p\in\freshcap{a}U$.
\end{enumerate*}
As a corollary, $F(\ns D)$ is sober (Definition~\ref{defn.sober}).
\end{corr}
\begin{proof}
We unravel definitions and see that:
\begin{enumerate*}
\item
$x\in p$ if and only if $p\in\pp x$ if and only if $\pp x\in\qq p$ if and only if $x\in\ppmone{(\qq p)}$.
\item
$\pp x\in\mathcal U$ if and only if $x\in\ppmone{\mathcal U}$ if and only if $\ppmone{\mathcal U}\in \pp x$ if and only if $\pp x\in\qq{(\ppmone{\mathcal U})}$.
\item
Suppose $p\in|F(\ns D)|$ (so $p$ is a prime filter in $\ns D$) and suppose $U\in\ctop{F(\ns D)}$ (so $U$ is a compact set of points).
By Theorem~\ref{thrm.pp.x.clopen} $U=\pp x$ for some $x\in|\ns D|$.
We reason as follows:
$$
\begin{array}{r@{\ }l@{\qquad}l}
\New{b}p\in(b\ a)\act\pp x
\liff&
p\in\pp{((b\ a)\act x)}
&\text{Lemma~\ref{lemm.dup}}
\\
\liff&
\New{b}(b\ a)\act x\in p
&\text{Definition~\ref{defn.pp}}
\\
\limp&
\tall a.x\in p
&\text{Definition~\ref{defn.filter} condition~\ref{filter.new}}
\\
\liff&
p\in\pp{(\tall a.x)}
&\text{Definition~\ref{defn.pp}}
\\
\liff&
p\in\freshcap{a}\pp x
&\text{Lemma~\ref{lemm.bullet.commute}(\ref{item.bullet.commute.tall})}
\end{array}
$$
\end{enumerate*}
The corollary follows by combining parts~1 to~3 of this result with Lemma~\ref{lemm.coherent.iff}.
\end{proof}

%%%%%%%%%%%%%%%%%%%%%%%%%%%%%%%%%%%%%%%%%%
\subsection{Nominal spectral spaces}

\begin{frametxt}
\begin{defn}
\label{defn.nom.top.spectral}
A \deffont{nominal spectral space} is a coherent (Definition~\ref{defn.coherent}) sober (Definition~\ref{defn.sober}) impredicative nominal $\sigma$-topological space $\ns T$ (Definition~\ref{defn.impredicative.top}).
\end{defn}
\end{frametxt}

\begin{prop}
\label{prop.FD.nomspecspace}
If $\ns D\in\india$ then $F(\ns D)$ (Definition~\ref{defn.F}) is a nominal spectral space.
\end{prop}
\begin{proof}
By Theorem~\ref{thrm.FB.amgis} $F(\ns D)$ is a $\sigma$-topological space.
By Theorem~\ref{thrm.FB.totsep.comp} it is impredicative.
By Proposition~\ref{prop.FB.new} it is coherent.
By Corollary~\ref{corr.FD.sober} it is sober.
\end{proof}

%%%%%%%%%%%%%%%%%%%%%%%%%%%%%%%%%%
\section{Morphisms of nominal spectral spaces}
\label{sect.spectral.morphisms}

\subsection{The definition of $\inspecta$, and $F$ viewed as a functor to it}
\label{subsect.inspecta}

We see from Definition~\ref{defn.G} that we obtain a nominal distributive lattice with $\tall$ from an impredicative nominal spectral space by taking the lattice of compact open sets.

A \emph{spectral} morphism is usually taken to be a map of points whose inverse preserves the property of being compact.
Our compact sets have permutation and $\sigma$-actions (and our points have permutation and $\amgis$-actions) so we need morphisms to interact appropriately with this extra structure.
This is Definition~\ref{defn.morphism}.

Then, we extend $F$ from Definition~\ref{defn.F} to act on morphisms, and check that this does indeed yield a functor.
This is Definition~\ref{defn.Ff} and Proposition~\ref{prop.inverse.filter}.
Theorem~\ref{thrm.F.funct} packages this all up into a theorem.

\begin{defn}
\label{defn.morphism}
Suppose $\ns S$ and $\ns T$ are nominal spectral spaces (Definition~\ref{defn.nom.top.spectral}).

Suppose $g_{\ns T}\in |\ns T|{\to}|\ns S|$.
Then:
\begin{itemize*}
\item
Call $g_{\ns T}$ \deffont{continuous} when $X\in\otop{\ns S}$ implies $(g_{\ns T})^\mone(X)\in\otop{\ns T}$ (inverse image of an open is open).
\item
Call $g_{\ns T}$ \deffont{spectral} when $X\in\ctop{\ns S}$ implies $(g_{\ns T})^\mone(X)\in\ctop{\ns T}$ (inverse image of a compact open is compact open).
\end{itemize*}
Call
$g=(g_{\ns T},g_{\ns S}^\prg)$
a \deffont{morphism} from $\ns T$ to $\ns S$ when:
\begin{enumerate*}
\item
\label{item.g.equivariant}
$g_{\ns T}$ is equivariant from $|\ns T|$ to $|\ns S|$, meaning that $\pi\act g_{\ns T}(\topel)=g_{\ns T}(\pi\act \topel)$.
\item
\label{item.g.braided}
$g_{\ns T}\in |\ns T|{\to}|\ns S|$ is continuous and spectral, and $g_{\ns S}^\prg$ is a $\sigma$-algebra morphism (Definition~\ref{defn.morphism.sigma.alg}) from $\ns S^\prg$ to $\ns T^\prg$.

There is a `braided' structure here, that $g_{\ns T}$ goes from $\ns T$ to $\ns S$ but $g_{\ns S}^\prg$ goes from $\ns S^\prg$ to $\ns T^\prg$.
Of course, the inverse image $(g_{\ns T})^\mone$ maps from $\otop{\ns S}$ to $\otop{\ns T}$ and so points in the same direction as $g_{\ns S}^\prg$.
This is used in Lemma~\ref{lemm.g.commute.amgis}.
\item
\label{item.g.prog.to.prog}
The inverse image $(g_{\ns T})^\mone$ maps atoms-as-programs in $\ns S$ to atoms-as-programs in $\ns T$, meaning that $(g_{\ns T})^\mone(\prg_{\ns S}a)=\prg_{\ns T}a$ for every atom $a$.
\item
\label{item.g.commute.amgis}
$g_{\ns T}$ commutes with the $\amgis$-action,
meaning that for $u\in|\ns S^\prg|$ and $p\in|\ns T|$
$$
g_{\ns T}(p)[u\ms a] = g_{\ns T}(p[g_{\ns S}^\prg(u)\ms a]) .
$$
See Lemma~\ref{lemm.g.commute.amgis} for a view of this as ``the inverse image $g^\mone$ commutes with the $\sigma$-action''.
\end{enumerate*}
Write $\inspecta$ for the category of nominal spectral spaces, and morphisms between them.

We may drop subscripts and write $g=(g,g^\prg):\ns T\equivarto \ns S$, so that for example the equation above becomes
$$
g(p)[u\ms a] = g(p[g^\prg(u)\ms a]) .
$$
\end{defn}

Lemmas~\ref{lemm.g.commute.amgis.atom} and~\ref{lemm.g.commute.amgis} are direct derivatives of condition~\ref{item.g.commute.amgis} of Definition~\ref{defn.morphism}:

\begin{lemm}
\label{lemm.g.commute.amgis.atom}
If $g:\ns T\equivarto\ns S$ then for $p\in|\ns T|$ and $n\in\mathbb A$,
$g(p)[n_{\ns S^\prg}\ms a]=g(p[n_{\ns T^\prg}\ms a])$.
\end{lemm}
\begin{proof}
The proof is just a special case of condition~\ref{item.g.commute.amgis} of Definition~\ref{defn.morphism}, noting that $g^\prg$ is assumed to be a $\sigma$-algebra morphism, and by condition~\ref{item.f.prg.maps.atoms.to.atoms} of Definition~\ref{defn.morphism.sigma.alg} $g^\prg(n_{\ns S^\prg})=n_{\ns T^\prg}$.
\end{proof}

\begin{lemm}
\label{lemm.g.commute.amgis}
Suppose $g:\ns T\equivarto\ns S$ and suppose $U\in\ctop{\ns S}$ and $u\in|\ns S^\prg|$.
Then
$$
g^\mone(U[a\sm u]) = g^\mone(U)[a\sm g^\prg(u)].
$$
(The $\sigma$-action $[a\sm u]$ and $[a\sm g^\prg(u)]$ is from Definition~\ref{defn.sub.sets}.)
\end{lemm}
\begin{proof}
Suppose $p\in|\ns T|$.
Then:
$$
\begin{array}[b]{r@{\ }l@{\quad}l}
p\in g^\mone(U[a\sm u])\liff& g(p)\in U[a\sm u]
&\text{Pointwise action}
\\
\liff&\New{c}g(p)[u\ms c]\in (c\ a)\act U
&\text{Proposition~\ref{prop.amgis.iff}}
\\
\liff&\New{c}g(p[g^\prg(u)\ms c])\in (c\ a)\act U
&\text{C\ref{item.g.commute.amgis} of Def~\ref{defn.morphism}}
\\
\liff&\New{c}p[g^\prg(u)\ms c]\in (c\ a)\act g^\mone(U)
&\text{Pointwise action}
\\
\liff&p\in g^\mone(U)[a\sm g^\prg(u)]
&\text{Proposition~\ref{prop.amgis.iff}}
\end{array}
\qedhere$$
\end{proof}

We take a moment to note an interaction of Lemma~\ref{lemm.g.commute.amgis} with condition~\ref{item.g.prog.to.prog} of Definition~\ref{defn.morphism}:
\begin{corr}
$g^\mone$ maps programs to programs, meaning that $g^\mone(\prg_{\ns S}u)=\prg_{\ns T}u$ for every $u\in|\ns S^\prg|$.
\end{corr}
\begin{proof}
By condition~\ref{item.g.prog.to.prog} of Definition~\ref{defn.morphism} $g^\mone(\prg_{\ns S} a)=\prg_{\ns T} a$.
We apply $[a\sm u]$ to both sides and use Lemma~\ref{lemm.g.commute.amgis} to deduce that $g^\mone(\prg_{\ns S} u)=\prg_{\ns T} u$.
\end{proof}

Definition~\ref{defn.Ff} extends Definition~\ref{defn.F} from objects to morphisms (and is extended further to $\app$ and $\ppa$ in Definition~\ref{defn.F.pp}):
\begin{defn}
\label{defn.Ff}
Given a morphism $f=(f_{\ns D},f_{\ns D}^\prg):\ns D\equivarto \ns D'$ in $\india$ (Definitions~\ref{defn.hom.impredicative},  \ref{defn.hom.nba}, and~\ref{defn.morphism.sigma.alg})
define $F(f):F(\ns D')\equivarto F(\ns D)$ by
\begin{itemize*}
\item
$F(f)_{\ns D'}(p')=(f_{\ns D})^\mone(p')$ where $p'\in|F(\ns D')|$
\\
\ \ (so $p'$ is a point---a prime filter---in $\ns D'$), and
\item
$F(f)_{\ns D}^\prg(u)=f_{\ns D^\prg}(u)$ where $u\in|\ns D^\prg|$.
\end{itemize*}
That is, without subscripts:
\begin{frameqn}
x\in F(f)(p) \liff f(x)\in p
\qquad
\text{and}\qquad
F(f)^\prg(u)=f^\prg(u)
\end{frameqn}
\end{defn}

\begin{rmrk}[A word on subscripts]
Our convention is that a subscript on a function/morphism indicates domain/source.
This information may elided, but the authors find some explicit bookkeeping helpful.

So consider $f=(f_{\ns D},f^\prg_{\ns D}):\ns D\to\ns D'$.
Both components are subscripted with $\ns D$.
So far, so normal.

Now consider Theorem~\ref{thrm.F.funct}, the culminating result of this Subsection.
It proves that $F$ is a functor from $\india$ to $\inspecta^\f{op}$---that is, it is a contravariant functor from $\india$ to $\inspecta$.
Thus
$$
\text{if }f:\ns D\equivarto\ns D\in\india
\quad\text{then}\quad
F(f):F(\ns D')\equivarto F(\ns D)\in\inspecta.
$$
Unpacking Definition~\ref{defn.Ff} we see that $F(f)$ has two components called $F(f)_{\ns D'}$ and $F(f)_{\ns D}^\prg$, where $F(f)_{\ns D'}$ is subscripted with $\ns D'$ and $F(f)_{\ns D}^\prg$ is subscripted with $\ns D$.
\begin{itemize*}
\item
The first component $F(f)_{\ns D'}$ maps $|F(\ns D')|$ to $|F(\ns D)|$ (by functional preimage $(f_{\ns D})^\mone$, as standard),
whereas
\item
the second component $F(f)_{\ns D}^\prg$ maps $|\ns D^\prg|$ to $|{\ns D'}^\prg|$.
\end{itemize*}
$F$ is contravariant and the `braided' structure of $\inspecta$ (clause~\ref{item.g.braided} of Definition~\ref{defn.morphism}) causes subscripts on $F(f)$ to be $\ns D'$ for the first component \dots and $\ns D$ for the second.
\end{rmrk}

We now work towards proving that $F(f)$ maps points to points in Proposition~\ref{prop.inverse.filter}.

\begin{lemm}
\label{lemm.FfmoneU}
Suppose $f=(f_{\ns D},f_{\ns D}^\prg):\ns D\equivarto\ns D'$ is a morphism in $\india$ and suppose $x\in|\ns D|$.
Then $F(f)^\mone(\pp x)=\pp{(f(x))}$.
\end{lemm}
\begin{proof}
We reason as follows, where $p\in|\points(\ns D')|$:
$$
\begin{array}[b]{r@{\ }l@{\qquad}l}
p\in F(f)^\mone(\pp x)
\liff&
F(f)(p)\in\pp x
&\text{Inverse image}
\\
\liff&
x\in F(f)(p)
&\text{Definition~\ref{defn.pp}}
\\
\liff&
f(x)\in p
&\text{Definition~\ref{defn.Ff}}
\\
\liff&
p\in\pp{(f(x))}
&\text{Definition~\ref{defn.pp}}
\end{array}
\qedhere$$
\end{proof}

\begin{lemm}
\label{lemm.inverse.filter.filter}
Suppose $f:\ns D\equivarto \ns D'$ is a morphism (Definition~\ref{defn.hom.impredicative}).
Then if $p\subseteq|\ns D'|$ is a filter then so is $F(f)(p)=f^\mone(p)$, and if $p$ is prime then so is $f^\mone(p)$.
\end{lemm}
\begin{proof}
Suppose $p$ is a filter.
We check the conditions of Definition~\ref{defn.filter}
for $f^\mone(p)$, freely using the fact that $x\in f^\mone(p)$ if and only if $f(x)\in p$:
\begin{enumerate*}
\item
\emph{$\tbot\not\in f^\mone(p)$.}\quad
By assumption in Definition~\ref{defn.hom.nba} $f(\tbot)=\tbot$, and by condition~1 of Definition~\ref{defn.filter} $\tbot\not\in p$.
\item
\emph{$x\in f^\mone(p)$ and $x\leq y$ implies $y\in f^\mone(p)$.}\quad
Immediate.
\item
\emph{$x\in f^\mone(p)$ and $y\in f^\mone(p)$ implies $x\tand y\in f^\mone(p)$.}\quad
Since $f(x\tand y)=f(x)\tand f(y)$.
\item
\emph{If $\New{b}((b\ a)\act x\in f^\mone(p))$ then $\tall a.x\in f^\mone(p)$.}\quad
Suppose $\New{b}(b\ a)\act x{\in} f^\mone(p)$.
It follows from Definition~\ref{defn.hom.nba} that $\New{b}(b\ a)\act f(x){\in} p$, so by condition~3 of Definition~\ref{defn.filter} $\tall a.f(x)\in p$.
By Definition~\ref{defn.hom.nba} again $\tall a.f(x)=f(\tall a.x)$.
The result follows.
\end{enumerate*}
Now suppose $p$ is prime and suppose $x{\tor}y\in f^\mone(p)$, so that $f(x{\tor}y)\in p$.
From Definition~\ref{defn.hom.nba} we have $f(x){\tor}f(y)\in p$.
Since $p$ is prime, either $f(x)\in p$ or $f(y)\in p$.
\end{proof}

\begin{prop}
\label{prop.inverse.filter}
If $f:\ns D\equivarto \ns D'$ is a morphism (Definition~\ref{defn.hom.impredicative}) then $F(f)$ from Definition~\ref{defn.Ff}
is a morphism from $F(\ns D')$ to $F(\ns D)$ (Definition~\ref{defn.morphism}).
\end{prop}
\begin{proof}
By Lemma~\ref{lemm.inverse.filter.filter} $F(f)$ maps prime filters of $\ns D'$ to prime filters of $\ns D$---that is, $F(f)_{\ns D'}\in|F(\ns D')|\to|F(\ns D)|$.

Now we show that $F(f)$ is a morphism.
We verify the properties of Definition~\ref{defn.morphism}.
By definition $F(f)^\prg=f^\prg$ which by construction is a $\sigma$-algebra morphism from ${\ns D}^\prg$ to ${\ns D'}^\prg$.
Also:
\begin{enumerate*}
\item
\emph{$F(f)$ is equivariant.}\quad
We briefly sketch the reasoning; in step \rulefont{\ast} we use condition~1 of Definition~\ref{defn.morphism} for $f$:
\begin{multline*}
x\in \pi\act (F(f)(p)) \liff \pi^\mone\act x\in F(f)(p) \liff f(\pi^\mone\act x)\in p \stackrel{\rulefont{\ast}}{\liff} \pi^\mone\act f(x)\in p
\\
\liff f(x)\in \pi\act p \liff x\in F(f)(\pi\act p)
\end{multline*}
\item
\emph{$F(f)^\mone$ is continuous and spectral.}\quad
We must prove two things:
\begin{itemize*}
\item
\emph{$F(f)^\mone$ maps open sets to open sets.}\quad
By construction $F(f)^\mone$ preserves unions, and by construction in Definition~\ref{defn.F}(\ref{F.strict}) every $X\in\otop{F(\ns D)}$ has the form of some small-supported sets union $\bigcup\mathcal X$ for some small-supported set $\mathcal X=\{\pp x_i\mid i{\in}I\}$.

Note by Corollary~\ref{corr.pp.injective} that the assignment $x\mapsto\pp x$ is injective.
Therefore the map $Q$ mapping $\mathcal X=\{\pp x_i\mid i{\in}I\}$ to $Q\mathcal X=\{\pp{f(x_i)}\mid i{\in}I\}$ is well-defined, and by Theorem~\ref{thrm.no.increase.of.supp} $Q\mathcal X$ has small support if $\mathcal X$ does.

By Lemma~\ref{lemm.FfmoneU} $F(f)^\mone (\pp x_i)=\pp{f(x_i)}$.
It follows that $F(f)^\mone(X)=\bigcup_{i{\in}I}\pp{f(x_i)}$ which by the note in the previous paragraph 
is a small-supported union and so is open in $F(\ns D')$.
\item
\emph{$F(f)^\mone$ maps compact sets to compact sets.}\quad
By Lemma~\ref{lemm.FfmoneU} $F(f)^\mone (\pp x)=\pp{f(x)}$.
We use Theorem~\ref{thrm.pp.x.clopen}.
\end{itemize*}
\item
\emph{$F(f)^\mone$ maps atoms-as-programs to themselves.}\quad
By Lemma~\ref{lemm.FfmoneU} $F(f)^\mone(\pp{(\prg_{\ns D} a)})=\pp{(f(\prg_{\ns D} a))}$ and
by assumption in Definition~\ref{defn.hom.impredicative} $f(\prg_{\ns D} a)=\prg_{\ns D'}a$.
\item
\emph{$F(f)$ commutes with the $\amgis$-action.}
Suppose $p'\in |F(\ns D')|$ and $u\in|F(\ns D)^\prg|$ and $x\in|\ns D|$.
Following Theorem~\ref{thrm.FB.totsep.comp} $|F(\ns D)^\prg|=|\ns D^\prg|$ so $u\in|\ns D^\prg|$.
We reason as follows:
$$
\begin{array}[b]{r@{\ }l@{\qquad}l}
x\in F(f)(p'[f^\prg(u)\ms a])
\liff&
f(x)\in p'[f^\prg(u)\ms a]
&\text{Definition~\ref{defn.Ff}}
\\
\liff&
f(x)[a\sm f^\prg(u)]\in p'
&\text{Proposition~\ref{prop.sigma.iff}}
\\
\liff&
f(x[a\sm u])\in p'
&\text{C\ref{item.f.commutes.with.sigma} of Def~\ref{defn.morphism.sigma.alg}}
\\
\liff&
x[a\sm u]\in F(f)(p')
&\text{Definition~\ref{defn.Ff}}
\\
\liff&
x\in (F(f)(p'))[u\ms a]
&\text{Proposition~\ref{prop.sigma.iff}}
\end{array}
\qedhere$$
\end{enumerate*}
\end{proof}

\begin{thrm}
\label{thrm.F.funct}
$F$ from Definitions~\ref{defn.F} and~\ref{defn.Ff} is a functor from $\india$ to $\inspecta^\f{op}$.
\end{thrm}
\begin{proof}
By Proposition~\ref{prop.FD.nomspecspace} $F$ maps objects of $\india$ to objects of $\inspecta^\f{op}$.
Furthermore if $f:\ns D\equivarto \ns D'$ in $\india$ then by Proposition~\ref{prop.inverse.filter} $F(f)$ is a morphism from $F(\ns D')$ to $F(\ns D)$.
The result follows by some easy calculations.
\end{proof}

%%%%%%%%%%%%%%%%%%%%%%%%%%%%%%%%%%%%
\subsection{The action of $G$ on morphisms in $\inspecta$}
\label{subsect.G.action}

In Subsection~\ref{subsect.inspecta} we went from $\india$ (Definition~\ref{defn.hom.impredicative})
to $\inspecta$ (Definition~\ref{defn.morphism}).
Now we go back.

So Definition~\ref{defn.Gg} mapping $\inspecta$ to $\india$ is the dual to Definition~\ref{defn.Ff} mapping $\india$ to $\inspecta$:
\begin{defn}
\label{defn.Gg}
Given $g:\ns T'\equivarto \ns T$ in $\inspecta$ define $G(g):G(\ns T)\equivarto G(\ns T')$ by $G(g)_{\ns T}(U)=(g_{\ns T'})^\mone(U)$ and $G(g)_{\ns T^\prg}(u)=u$, that is (without subscripts):
\begin{frameqn}
\topel'\in G(g)(U) \liff g(\topel')\in U
\quad\text{and}\quad
G(g)^\prg(u)=g^\prg(u)
\end{frameqn}
\end{defn}

\begin{prop}
\label{prop.G.funct}
$G$ from Definitions~\ref{defn.G} and~\ref{defn.Gg} is a functor from $\inspecta^\f{op}$ to $\india$.
\end{prop}
\begin{proof}
By Theorem~\ref{thrm.T.to.G.obj} $G$ maps objects of $\inspecta^\f{op}$ to objects of $\india$.

Now consider a morphism $g:\ns T'\equivarto \ns T$ in the sense of Definition~\ref{defn.morphism}; the interesting part is to check that $G(g)_{\ns T}$---that is, $(g_{\ns T'})^\mone$---is a morphism in the sense of Definition~\ref{defn.hom.impredicative}.

We may drop subscripts henceforth.
If $U\in\ctop{\ns T}$, so $U$ is compact, then since $g$ is assumed spectral in Definition~\ref{defn.morphism} also $g^\mone(U)$ is compact.
It is routine to check that $g^\mone$ preserves the
top $\ttop$ and
bottom elements ($|\ns T|$ and $\varnothing$ respectively) and interacts correctly with intersections and unions.

It remains to show that $g^\mone$ is equivariant, commutes with the $\sigma$-action, and commutes with $\freshcap{a}$.

Equivariance, meaning that $\pi\act g^\mone(U)=g^\mone(\pi\act U)$, is immediate by Theorem~\ref{thrm.equivar} (a proof by concrete calculations using Proposition~\ref{prop.amgis.iff} and condition~1 of Definition~\ref{defn.morphism} is also possible).
Commuting with the $\sigma$-action, meaning that $(g_{\ns T'})^\mone(U)[a\sm g_{\ns T}^\prg(u)]=(g_{\ns T'})^\mone(U[a\sm u])$, is Lemma~\ref{lemm.g.commute.amgis}.

Finally we check that $g^\mone$ commutes with $\freshcap{a}$.
First, note that by assumption $\ns T,\ns T'\in\inspecta$ so by Definition~\ref{defn.morphism} they are sober and so satisfy Lemma~\ref{lemm.topel.sober.equiv}.
Now suppose $\topel'{\in}|\ns T'|$ and $U{\in}|G(\ns T)|{=}\ctop{\ns T}$.
We reason as follows:
\begin{tab2c}
\topel\in g^\mone(\freshcap{a}U)
\liff &g(\topel)\in\freshcap{a}U
&\text{Pointwise action}
\\
\liff &\New{b} g(\topel)\in (b\ a)\act U
&\text{Lemma~\ref{lemm.topel.sober.equiv}}
\\
\liff &\New{b} \topel\in (b\ a)\act g^\mone(U)
&\text{Pointwise action, Thm~\ref{thrm.equivar}}
\\
\liff &\topel\in \freshcap{a} g^\mone(U)
&\text{Lemma~\ref{lemm.topel.sober.equiv}}
\end{tab2c}
Thus $G(g)$ is a morphism in $\india$.
\end{proof}

\begin{rmrk}
\label{rmrk.the.structure}
We continue Remark~\ref{rmrk.sober.second.condition}.
Note in the proof of Lemma~\ref{lemm.topel.sober.equiv} that we prove a property of an infinite intersection using its characterisation using the $\new$-quantifier.
We can do this thanks to sobriety; specifically, thanks to condition~\ref{sober.all} of Definition~\ref{defn.sober} (which reflects condition~\ref{filter.new} of Definition~\ref{defn.filter}).

We want $g^\mone$ to commute with $\freshcap{a}$ because in the duality $g^\mone$ should correspond to a morphism of nominal distributive lattices \emph{with $\tall$} (the formal statement and proof are in Subsection~\ref{subsect.the.equivalence}, below).

When we designed Definition~\ref{defn.morphism} we could imagine just insisting on $g^\mone(\freshcap{a}U)=\freshcap{a}g^\mone(U)$.
But this would have been unsatisfactory: we did not insist on $g^\mone(U\cap V)=g^\mone(U)\cap g^\mone(V)$ or $g^\mone(|\ns T|{\setminus}U)=|\ns T|{\setminus}g^\mone(U)$ because these emerge from the structure of $\ns T$, by properties of sets of points and the inverse image.
Imposing $g^\mone(\freshcap{a}U)=\freshcap{a}g^\mone(U)$ as a condition \emph{on $g$} is not the best or most informative design choice; better to impose conditions \emph{on $\ns T$}.

Thus, one way to read condition~\ref{sober.all} of Definition~\ref{defn.sober} is that it gives $\ns T$ the structure to ensure that $g^\mone$ will commute with $\freshcap{a}$ as well.
\end{rmrk}

%%%%%%%%%%%%%%%%%%%%%%%%%%%%%%%%%%%%%%%%%%%%%%%%%
\subsection{The equivalence}
\label{subsect.the.equivalence}

In Subsections~\ref{subsect.inspecta} and~\ref{subsect.G.action} we considered two functors $F:\india\longrightarrow\inspecta$ and $G:\inspecta\longrightarrow\india$.
They are dual; the key is to observe that $FG(\ns T)$ is isomorphic to $\ns T$.
This is Lemma~\ref{lemm.qq.commute} and Proposition~\ref{prop.qq.iso}.
Theorem~\ref{thrm.equivalence} puts it all together.

\begin{lemm}
\label{lemm.qq.commute}
Suppose $\ns T\in\inspecta$ and $\topel\in|\ns T|$.
If we give $\qq \topel\in|FG(\ns T)|$ from Definition~\ref{defn.qq} the pointwise $\amgis$-action from Definition~\ref{defn.p.action} then for $U\in\ctop{\ns T}$
$$
U\in \qq{\topel}[u{\ms}a] \liff \New{c}(c\ a)\act U\in \qq{(\topel[u{\ms}c])}
.
$$
\end{lemm}
\begin{proof}
We reason as follows:
$$
\begin{array}[b]{r@{\ }l@{\qquad}l}
U\in\qq{\topel}[u{\ms}a]
\liff&
U[a\sm u]\in\qq\topel
&\text{Proposition~\ref{prop.sigma.iff}}
\\
\liff&
\topel\in U[a\sm u]
&\text{Definition~\ref{defn.qq}}
\\
\liff&
\New{c}\topel[u\ms c]\in (c\ a)\act U
&\text{Proposition~\ref{prop.amgis.iff}}
\\
\liff&
\New{c}(c\ a)\act U\in\qq{(\topel[u\ms c])}
&\text{Definition~\ref{defn.qq}}
\end{array}
\qedhere$$
\end{proof}

Recall from Definitions~\ref{defn.F} and~\ref{defn.G} that $FG(\ns T)$ is a topological space whose points are prime filters of compact opens in $\ns T$.
\begin{prop}
\label{prop.qq.iso}
If $\ns T\in\inspecta$ then $\alpha_{\ns T}$ mapping $\topel\in|\ns T|$ to $\qq{\topel}\in |FG(\ns T)|$ defines an isomorphism in $\inspecta$ between $\ns T$ and $FG(\ns T)$.
\end{prop}
\begin{proof}
Injectivity and surjectivity are Corollary~\ref{corr.qq.max}.
Commutativity with the $\amgis$-action is Lemma~\ref{lemm.qq.commute}, as can be checked by unravelling definitions.

We also need to show that $\alpha$ is continuous.
The reasoning is standard \cite[Section~4]{burris:couua} so we just sketch it.
First, if $U\in G(\ns T)$ (so $U$ is a compact open set of $\ns T$) consider the inverse image under $\qq{\text{-}}$ of $\pp{U}$.\footnote{Unpacking Definitions~\ref{defn.pp} and~\ref{defn.G}, $\pp{U}$ is the set of prime filters of compact opens of $\ns T$ of which $U$ is an element.}
$$
\begin{array}{r@{\ }l@{\qquad}l}
\topel\in \alpha_{\ns T}^\mone(\pp U) \liff& \alpha_{\ns T}(\topel)\in\pp U
\\
\liff& \{U'\mid \topel\in U'\} \in \pp U
\\
\liff& U\in \{U'\mid \topel\in U'\}
\\
\liff&t\in U
\end{array}
$$
Thus, $\alpha_{\ns T}^\mone(\pp U)=U$.

Now by construction any open set in $FG(\ns T)$ is a union of $\pp U$, and it is a fact that the inverse image function $\alpha_{\ns T}^\mone$ preserves these unions.
It follows that the inverse image of an open set is open.
\end{proof}

\begin{thrm}
\label{thrm.equivalence}
$G:\inspecta^\tf{op}\to\india$ defines an equivalence between $\india$ and $\inspecta^{\f{op}}$.
\end{thrm}
\begin{proof}
We use \cite[Theorem~1, Chapter~IV, Section~4]{CWM71}.
\begin{itemize*}
\item
\emph{$G$ is essentially surjective on objects.}\quad
This is Proposition~\ref{prop.ess.surj}.

\item\emph{$G$ is faithful.}\quad
Suppose $g_1,g_2:\ns T\equivarto \ns S\in\inspecta$ and $g_1\neq g_2$; the interesting case here is then that there exists $p\in|\ns T|$ such that $g_1(p)\neq g_2(p)$.
($F$ and $G$ leave programs unchanged, so we can elide $g_1^\prg$ and $g_2^\prg$.)

By assumption $\ns S$ is coherent and sober, so that by Corollary~\ref{corr.qq.max}
$\qq{g_1(p)}$ and $\qq{g_2(p)}$---these are the sets of compact open sets in $\ns S$ containing $g_1(p)$ and $g_2(p)$ respectively---are distinct.

Thus there exists a compact open set $U\in\otop{\ns S}$ with $g_1(p)\in U$ and $g_2(p)\not\in U$.
Examining Definition~\ref{defn.Gg} we see that $p\in G(g_1)(U)$ and $p\not\in G(g_2)(U)$.
Thus, $G(g_1)\neq G(g_2)$.

\item\emph{$G$ is full.}\quad
Given $\ns S,\ns T$ in $\inspecta$ and
$f:G(\ns S)\equivarto G(\ns T)$ in $\india$ we construct a morphism $g:\ns T\equivarto \ns S$ in $\inspecta$ such that $G(g)=f$.

By Proposition~\ref{prop.qq.iso} $\alpha_{\ns T}:\ns T\equivarto FG(\ns T)$ mapping $t$ to $\qq t$ is an isomorphism in $\inspecta$.
Set $g=\alpha_{\ns S}^\mone\circ F(f)\circ \alpha_{\ns T}$.
By routine calculations we can check that $G(g)(U')=f(U')$ for every $U'\in|G(\ns S)|$.
\qedhere\end{itemize*}
\end{proof}

\jamiepart{Adding application and its topological dual the combination operator}
\label{part.application}

So far we have seen $\india$ and $\inspecta$, and Theorem~\ref{thrm.equivalence} is a topological duality theorem relating them.
This is in itself an interesting result: duality for an impredicative propositional logic (propositional logic with quantifiers over propositions; the reader might be familiar with this kind of logical system in the form of the type system of System F \cite{girard:prot}).

However, to model the $\lambda$-calculus we need more structure.

This is developed in Section~\ref{sect.operators}, and our results so far are extended accordingly---culminating in Subsection~\ref{subsect.app.equivalence} with Theorem~\ref{thrm.equivalence.pp}.

%%%%%%%%%%%%%%%%%%%%%%%%%%%%%%
\section{$\indiapp$ and $\inspectapp$}
\label{sect.operators}

\subsection{Adding $\app$ and $\ppa$ to $\india$ to get $\indiapp$}
\label{subsect.indiapp}

\begin{figure}[b]
$$
\begin{array}{l@{\qquad}r@{\ }l}
\rulefont{\app\epsilon}
&
(u\ppa x)\app u\leq & x
\\
\rulefont{\app\eta}
&
x\leq& u\ppa(x\app u)
\end{array}
$$
\caption{Adjointness properties for $\app$ and $\ppa$}
\label{fig.app.adjoint}
\end{figure}

\begin{figure}
$$
\begin{array}{l@{\qquad}l@{\hspace{-1.5em}}r@{\ }l@{\quad}l@{\ }r@{\ }l}
\rulefont{\sigma\app}&&
(x\app y)[b\sm v]=&(x[b\sm v])\app (y[b\sm v])
\\
\rulefont{\sigma\ppa}&&
(\prg b\ppa x)[a\sm v]= &\prg b\ppa (x[a\sm v])
\\[1.5ex]
\rulefont{\app\tbot}&&\tbot\app u=&\tbot
&&x\app \tbot=&\tbot
\\
\rulefont{\app{\tand}}&&(x\tand y)\app u\leq &(x\app u)\tand(y\app u)
&&x\app(u\tand v)\leq& (x\app u)\tand (x\app v)
\\
\rulefont{\app{\tor}}&&(x\tor y)\app u=&(x\app u)\tor (y\app u)
&&x\app(u\tor v)=&(x\app u)\tor(x\app v)
\\
\rulefont{\app\tall}&b\#u\limp&(\freshwedge{b}x)\app u\leq &\freshwedge{b}(x\app u)
\\[1.5ex]
\rulefont{\ppa{\tand}}&&u\ppa (x\tand y)=&(u\ppa x)\tand(u\ppa y)
\\
\rulefont{\ppa{\tor}}&&u\ppa (x\tor y)\geq &(u\ppa x)\tor (u\ppa y)
\\
\rulefont{\ppa\tall}&b\#u\limp&\freshwedge{b}(u\ppa x)\leq& u\ppa (\freshwedge{b}x)
\end{array}
$$
\caption{Compatibility properties for $\app$ and $\ppa$}
\label{fig.app.compatible}
\end{figure}

\begin{frametxt}
\begin{defn}
\label{defn.FOLeq.pp}
We extend the notion of an impredicative nominal distributive lattice with $\tall$ from Definition~\ref{defn.D.impredicative} with two equivariant operators $\app:(\ns X\times\ns X)\Rightarrow\ns X$ and $\ppa:(\ns X\times\ns X)\Rightarrow\ns X$, written infix as $x\app y$ and $y\ppa x$.

They must be \deffont{adjoint} as described in Figure~\ref{fig.app.adjoint},
and they must be \deffont{compatible} as described in Figure~\ref{fig.app.compatible} (the notation $\prg b$ is from Notation~\ref{nttn.impredicative.D}).
\end{defn}
\end{frametxt}

\begin{rmrk}
In categorical terminology, axiom \rulefont{\app\epsilon} is a \emph{counit} and \rulefont{\app\eta} is a \emph{unit}.
In Proposition~\ref{prop.lambda.beta.eta} we will derive $\beta$-reduction from \rulefont{\app\epsilon} and $\eta$-expansion from \rulefont{\app\eta}.
Later on in Remark~\ref{rmrk.ppa} we will examine how $\ppa$ behaves in a concrete model.
\end{rmrk}

\begin{rmrk}
\label{rmrk.ppa.stronger}
In the presence of \rulefont{\sigma\#} there is redundancy in \rulefont{\sigma\ppa}; we could take \rulefont{\sigma\ppa} to be $\prg b\ppa(x[a\sm v])\leq (\prg b\ppa x)[a\sm v]$ and get the reverse inequality from Lemma~\ref{lemm.sigma.ppa} (below).
The form given in Definition~\ref{defn.FOLeq.pp} is slightly more convenient to work with.
\end{rmrk}

\begin{rmrk}
\rulefont{\ppa\tall} is an inequality, not an equality.
This seems odd, given that \rulefont{\ppa\tand} is an equality; is not $\tall$ intuitively an infinite conjunction or a fresh-finite limit?
The reason is that this reflects the \emph{in}equality in Lemma~\ref{lemm.sigma.ppa} below.
To see how this works, we refer the interested reader to the case of \rulefont{\ppa\tall} in Lemma~\ref{lemm.GT.app}.
The even more interested reader is referred to Proposition~\ref{prop.qappx.filter}, where \rulefont{\ppa\tall} is just what we need for the final stages of the proof.
\end{rmrk}

\begin{rmrk}
\rulefont{\sigma{\ppa}} might seem odd: why $\prg b\ppa x$ and not $y\ppa x$?
The precise technical reason is in Proposition~\ref{prop.lsm.special.distrib}; the $\prg b\ppa x$ form is what we prove of our canonical syntactic model $\points_\Pi$.
\end{rmrk}

We package our definitions up as a category:
\begin{defn}
\label{defn.indiapp}
Continuing Definition~\ref{defn.FOLeq.pp}, write $\indiapp$ for the category with objects impredicative nominal distributive lattices with $\tall$, $\app$, and $\ppa$,
and morphisms are morphisms in $\india$ (Definition~\ref{defn.india}) that commute with $\app$ and $\ppa$ in the following sense:
$$
\begin{array}{r@{\ }l}
f(x\app y)=&f(x)\app f(y)
\quad\text{and}
\\
f(\prg b\ppa x)=&\prg b\ppa f(x) .
\end{array}
$$
For more on $\prg b$ see Remarks~\ref{rmrk.explain.prg.a} and~\ref{rmrk.explain.prg.b}.
\end{defn}

We conclude with technical lemmas concerning the interaction of $\app$ and $\ppa$ with $\leq$ and the $\sigma$-action.
For the rest of this subsection we fix $\ns D\in\indiapp$ and $x,y,x',y'\in|\ns D|$ and $u\in|\ns D^\prg|$.

\begin{lemm}
\label{lemm.app.leq}
\begin{enumerate*}
\item\label{app.leq.left}
If $x\leq x'$ then $x\app y\leq x'\app y$ and $y\app x\leq y\app x'$.
\item
If $x\leq x'$ then $y\ppa x\leq y\ppa x'$.
\end{enumerate*}
\end{lemm}
\begin{proof}
It is a fact that $x\leq x'$ if and only if $x\tor x'=x'$, if and only if $x=x\tand x'$.
We reason as follows:
\begin{enumerate*}
\item
$x'\app y = (x\tor x')\app y\stackrel{\rulefont{\app{\tor}}}{=} (x\app y)\tor (x'\app y)$.
The proof that $y\app x\leq y\app x'$ is similar.
\item
$y\ppa x = y\ppa (x\tand x') \stackrel{\rulefont{\ppa{\tand}}}{=} (y\ppa x)\tand (y\ppa x') \leq y\ppa x'$
\qedhere\end{enumerate*}
\end{proof}

\begin{lemm}
\label{lemm.india.ppa.adjoint}
$x\app y\leq z$ if and only if $x\leq y\ppa z$.
\end{lemm}
\begin{proof}
The reasoning is standard:
\begin{itemize*}
\item
If $x\app y\leq z$ then $x\stackrel{\rulefont{\app\eta}}{\leq} y\ppa(x\app y) \stackrel{\text{L\ref{lemm.app.leq}(2)}}{\leq} y\ppa z$.
\item
If $x\leq y\ppa z$ then $x\app y\stackrel{\text{L\ref{lemm.app.leq}(1)}}{\leq} (y\ppa z)\app y \stackrel{\rulefont{\app\epsilon}}{\leq} z$.
\qedhere\end{itemize*}
\end{proof}

\begin{lemm}
\label{lemm.sigma.ppa}
$(y\ppa x)[a\sm u]\leq y[a\sm u]\ppa x[a\sm u]$.
\end{lemm}
\begin{proof}
We reason as follows:
$$
\begin{array}[b]{r@{\ }l@{\qquad}l}
&\hspace{-4em}x[a\sm u]\leq x[a\sm u]
\\
\limp^{\rulefont{\dagger}}&
((y\ppa x)\app y)[a\sm u]\leq x[a\sm u]
&\text{Lem~\ref{lemm.sigma.monotone}, \rulefont{\app\eta}}
\\
\limp^{\rulefont{\ast}}&
(y\ppa x)[a\sm u]\app y[a\sm u]\leq x[a\sm u]
&\text{From \rulefont{\sigma\app}}
\\
\liff&
(y\ppa x)[a\sm u]\leq y[a\sm u]\ppa x[a\sm u]
&\text{Lemma~\ref{lemm.india.ppa.adjoint}}
\end{array}
\qedhere$$
\end{proof}

\begin{rmrk}[An aside on explicit substitutions]
\label{rmrk.explicit.substitutions}
Note the implications $\limp^{\rulefont{\dagger}}$ and $\limp^{\rulefont{\ast}}$ in the proof of Lemma~\ref{lemm.sigma.ppa} above.
The $\limp^{\rulefont{\ast}}$ could be replaced by a stronger if-and-only-if $\liff$.
This would make no difference to the final statement of the Lemma, because of the $\limp^{\rulefont{\dagger}}$ immediately above.
We have a reason for writing the proof above as we have done, which we now describe:

Recall from Figure~\ref{fig.app.compatible} that \rulefont{\sigma\app} is an equality (that $(x\app y)[b\sm v]=(x[b\sm v])\app (y[b\sm v])$ for all $x$, $y$, and $v$) and so can be viewed as two inequalities.
The precise inequality required to derive $\limp^{\rulefont{\ast}}$ in the proof of Lemma~\ref{lemm.sigma.ppa} above is that
$$
(y\ppa x)[a\sm u]\app y[a\sm u] \leq ((y\ppa x)\app y)[a\sm u] .
$$
This paper is based on $\lambda$-calculus without explicit substitutions (see the absence of an explicit substitution \cite{abadi:exps} in Definitions~\ref{defn.lamtrm} or~\ref{defn.idiom} below).
So substitution is not a reduction step, and this is why \rulefont{\sigma\app} is an equality.

To model explicit substitutions we would naturally weaken \rulefont{\sigma\app} from Figure~\ref{fig.app.compatible} to an inequality that $(x\app y)[b\sm v]\leq (x[b\sm v])\app (y[b\sm v])$ and condition~\ref{item.nom.top.app.sub} of Definition~\ref{defn.coherent.bpp} to a subset inclusion, reflecting that substitution becomes a distinct reduction step.\footnote{This is nice also semantically: it is easier to build models of the inequality and subset inclusion than of the equality.}
We already noted in Remark~\ref{rmrk.ppa.stronger} that \rulefont{\sigma\ppa} could already be taken as an inequality at some small cost in complexity.
So this seems very natural.

Therefore, it seems worth noting that the inequality we need for the proof above is 
$((y\ppa x)\app y)[a\sm u]\leq(y\ppa x)[a\sm u]\app y[a\sm u]$,
which is not the inequality we would expect to get from the reduction rule for an explicit substitution distributing over an application. 
We would lose Lemma~\ref{lemm.sigma.ppa} in its current form.

This is probably not fatal: the uses of Lemma~\ref{lemm.sigma.ppa} required for this paper to work are for the special case that $a\#y$
(Proposition~\ref{prop.unpack.lambda.for.G} uses the full result but is not itself needed for the rest of the paper).
This suggests that a suitable generalisation would include an inequality $(x\app y)[b\sm v]\leq (x[b\sm v])\app (y[b\sm v])$, along with an equality $x[b\sm v]\app y=(x\app y)[b\sm v]$ if $a\#y$ (which is the special case for which we need Lemma~\ref{lemm.sigma.ppa}), and this for a calculus for explicit substitutions which includes a reduction $(st)[a\sm u]\to (s[a\sm u])(t[a\sm u])$, along with a syntactic equivalence $s[a\sm u]t= (st)[a\sm u]$ (or just a reduction $s[a\sm u]t\to (st)[a\sm u]$) if $a\#t$. 

Investigating further which explicit substitution calculi this paper would naturally generalise to, if any, is future research.
\end{rmrk}

%%%%%%%%%%%%%%%%%%%%%%%%%%%%%%%%
\subsection{The combination operator $\bpp$: a topological dual to $\app$ and $\ppa$}
\label{subsect.bpp}

In Subsection~\ref{subsect.indiapp} we extended $\india$ with extra structure $\app$ and $\ppa$.

We can expect this to be reflected in the topologies by some kind of extension with a structure dual to $\app$ and $\ppa$.\footnote{For instance, $\ns D\in\india$ has a $\sigma$-action, and this is reflected dually as an $\amgis$-action on prime filters.

Similarly, $\ns D$ has a permutation action; permutations are invertible, so the dual structure on prime filters is \dots another permutation action.}

What should dually correspond to $\app$ and $\ppa$?
Remarkably, this requires only a little more structure on points: see the combination operator $\bpp$ in Definition~\ref{defn.bpp}.

%%%%%%%%%%%%%%%%%%%%%%%%%%%%%%%%
\subsubsection{The basic definition}

Recall the small-supported powerset $\nompow(\text{-})$ from Subsection~\ref{subsect.finsupp.pow}.
\begin{defn}
\label{defn.bpp}
We define an \deffont{$\amgis$-algebra with $\bpp$} or \deffont{$\amgis$-algebra with combination} by extending the notion of $\amgis$-algebra $\ns P$ from Definition~\ref{defn.bus.algebra} with an equivariant \deffont{combination} operator
$$
\bpp:(\ns P\times\ns P)\Rightarrow\nompow(\ns P),
$$
written infix as $p\bpp q$.
\end{defn}

An $\amgis$-algebra with combination is just an $\amgis$-algebra with an equivariant \emph{combination} function $\bpp$ mapping pairs of points to sets of points.
We now outline how the notion of spectral space is enriched by assuming $\bpp$ (Definition~\ref{defn.bpp}) on its points.

We extend Definition~\ref{defn.nom.top}:
\begin{frametxt}
\begin{defn}
\label{defn.nom.top.bpp}
A \deffont{nominal $\sigma$-topological space with $\bpp$} $\ns T$ (or just `$\sigma\bpp$-topological space') is a nominal $\sigma$-topological space (Definition~\ref{defn.nom.top}) whose points have the additional structure of a combination operator $\bpp$ (Definition~\ref{defn.bpp}).
\end{defn}
\end{frametxt}

\begin{rmrk}
\label{rmrk.where.from}
What concrete models of Definition~\ref{defn.nom.top.bpp} look like is a very interesting question---especially what the \emph{combination operator} $\bpp$ looks like.

Concrete models of $\bpp$ arise from Kripke-style models with a \emph{ternary} accessiblity relation.
We refer the interested reader to \cite[Definition~2.1]{gabbay:simcks} and the discussion in \cite[Subsection~5.2]{gabbay:simcks}, which notes the similarity (and differences) of the Kripke model of $\bpp$ compared with multiplicative conjunction, phase spaces, and certain models of relevance logic.

Combination $\bpp$ will be dual to $\lambda$-calculus application $\app$.
Since the dual to sigma is amgis, perhaps $\bpp$ should be called \emph{ppa}, pronounced `pah' to rhyme with `bah'.
\emph{Combination} is a more serious name for it.
\end{rmrk}

%%%%%%%%%%%%%%%%%%%%%%%%%%%%%%%%
\subsubsection{$\app$ and $\ppa$, and spectral spaces}

Extend Definition~\ref{defn.sub.sets} as follows:
\begin{defn}
\label{defn.sub.sets.pp}
\label{defn.YppaX}
Suppose $\ns P=(|\ns P|,\act,\ns P^\prg,\tf{amgis},\bpp)$ is an $\amgis$-algebra with $\bpp$.
Suppose $X,Y\subseteq|\ns P|$ and suppose $p{\in}|\ns P|$.
Then the following notation will be useful: 
\begin{frameqn}
\begin{array}{r@{\ }l}
X\app Y=&\bigcup\{p\bpp q\mid p\in X,\ q\in Y\}
\\
Y\ppa X=&\{p\mid \Forall{q{\in}Y}p\bpp q\subseteq X\}
\\
p\bpp Y =& \bigcup\{p\bpp q\mid q\in Y\} .
\end{array}
\end{frameqn}
\end{defn}

\begin{lemm}
\label{lemm.rewrite.ppa}
$Y\ppa X$ can be conveniently rewritten as 
$$
Y\ppa X = \{p\mid p\bpp Y\subseteq X\} .
$$
\end{lemm}
\begin{proof}
By routine calculations.
\end{proof}

We can extend Proposition~\ref{prop.amgis.iff} to reflect the structure created by the combination action $\bpp$:
\begin{prop}
\label{prop.amgis.iff.pp}
We have the following:
$$
\begin{array}{r@{\ }l@{\qquad}r@{\ }l}
r\in X\app Y \liff& \Exists{p{\in}X,q{\in}Y}r\in p\bpp q
&
r\in Y\ppa X \liff& r\bpp Y\subseteq X
\\
&& \liff &\Forall{q{\in}Y} r\bpp q\subseteq X
\end{array}
$$
\end{prop}

We extend Definition~\ref{defn.coherent} to account for the extra structure:
\begin{defn}
\label{defn.coherent.bpp}
Call a nominal $\sigma$-topological space $\ns T$ with $\bpp$ \deffont{coherent} when it is coherent in the sense of Definition~\ref{defn.coherent} and in addition, for all $X$ and $Y$ open and compact (so $X,Y\in\ctop{\ns T}$):
\begin{frametxt}
\begin{enumerate*}
\setcounter{enumi}{4}
\item
$X\app Y$ and $Y\ppa X$ are open and compact.
\item
\label{item.nom.top.app.sub}
$(X\app Y)[a\sm u]=X[a\sm u]\app Y[a\sm u]$ (so $\ctop{\ns T}$ satisfies \rulefont{\sigma\app}).
\item
\label{item.nom.top.ppa.sub}
$(\prg_{\ns T} b\ppa X)[a\sm u]=\prg_{\ns T} b\ppa (X[a\sm u])$ (so $\ctop{\ns T}$ satisfies \rulefont{\sigma\ppa}).
\end{enumerate*}
\end{frametxt}
\end{defn}

\begin{lemm}
\label{lemm.coherent.sigma.ppa}
Suppose $\ns T$ is a nominal $\sigma$-topological space with $\bpp$ (Definition~\ref{defn.nom.top.bpp}) and $X,Y\subseteq|\ns T|$.
Then
$$
\Forall{W{\in}\otop{\ns T}}\bigl( W\subseteq Y\ppa X\liff W\app Y\subseteq X \bigr).
$$
If furthermore $\ns T$ is coherent and $X,Y\in\ctop{\ns T}$, then
$$
Y\ppa X=\bigcup\{W\in\ctop{\ns T}\mid W\app Y\subseteq X\}.
$$
\end{lemm}
\begin{proof}
Consider $W\in\otop{\ns T}$ and $X,Y\subseteq|\ns T|$.

By Lemma~\ref{lemm.rewrite.ppa} $q\bpp Y\subseteq X$ for any $q\in Y$ ($q\bpp Y$ is from Definition~\ref{defn.YppaX}).
It follows from Definition~\ref{defn.sub.sets.pp} that $W\app Y\subseteq X$.
Conversely if $W\app Y\subseteq X$ then by Definition~\ref{defn.sub.sets.pp},\ $p\bpp Y\subseteq X$ for every $p\in W$, and thus by Lemma~\ref{lemm.rewrite.ppa} also $p\in Y\ppa X$ for every $p\in W$.
So $W\subseteq Y\ppa X$.

By the first part, $\bigcup\{W\in\otop{\ns T}\mid W\app Y\subseteq X\}\subseteq Y\ppa X$.
If in addition $\ns T$ is coherent then if $X,Y\in\ctop{\ns T}$ then $Y\ppa X\in\ctop{\ns T}$ and so $Y\ppa X$ is one of the $W$ such that $W\app Y\subseteq X$, so $Y\ppa X\subseteq\bigcup\{W\in\otop{\ns T}\mid W\app Y\subseteq X\}$.
\end{proof}

Definition~\ref{defn.nom.top.spectral.bpp} extends Definition~\ref{defn.nom.top.spectral}:
\begin{defn}
\label{defn.nom.top.spectral.bpp}
A \deffont{nominal spectral space with $\bpp$} is a nominal $\sigma\bpp$-topological space $\ns T$ (Definition~\ref{defn.nom.top.bpp}) that is impredicative (\ref{defn.impredicative.top}), coherent (Definition~\ref{defn.coherent.bpp}), and sober (Definition~\ref{defn.sober}).
\end{defn}

\begin{defn}
\label{defn.g.extend}
We extend the notion of \deffont{morphism} $g:\ns T'\equivarto\ns T$ of nominal spectral spaces from Definition~\ref{defn.morphism} to insist that the inverse image $g^\mone$ should commute with $\app$ and $\ppa$ in the following sense:
\begin{itemize*}
\item
If $X,Y\in\ctop{\ns T}$ then
$g^\mone(X\app Y)=g^\mone(X)\app g^\mone(Y)$.
\item
If $X\in\ctop{\ns T}$ then $g^\mone(\prg b\ppa X)=\prg b\ppa g^\mone(X)$.
\end{itemize*}
\begin{frametxt}
Write $\inspectapp$ for the category of nominal spectral spaces with $\bpp$ (Definition~\ref{defn.nom.top.spectral.bpp}), and morphisms between them whose inverse image functions commute with $\app$ and $\ppa$, as described above.
\end{frametxt}
\end{defn}

%%%%%%%%%%%%%%%%%%%%%%%%%%%%%%%%
\subsubsection{Useful technical lemmas}

We conclude with a pair of useful technical lemmas:
\begin{lemm}
\label{lemm.bigcup.app}
\begin{itemize*}
\item
$(\bigcup_i X_i)\app (\bigcup_j Y_j)=\bigcup_{ij}(X_i\app Y_j)$.
\item
$\varnothing\app Y=\varnothing=X\app\varnothing$.
\item
If $X\subseteq X'$ and $Y\subseteq Y'$ then $X\app Y\subseteq X'\app Y'$.
\end{itemize*}
\end{lemm}
\begin{proof}
By elementary sets calculations on Definition~\ref{defn.sub.sets.pp}.
\end{proof}

\begin{lemm}
\label{lemm.bruce}
The adjoint axioms \rulefont{\app\epsilon} and \rulefont{\app\eta} from Definition~\ref{defn.FOLeq.pp} hold.
That is, for $X,Y\subseteq|\ns P|$:
$$
(U\ppa X)\app U\subseteq X
\quad\text{and}\quad
X\subseteq U\ppa(X\app U)
$$
\end{lemm}
\begin{proof}
Suppose $r\in (U\ppa X)\app U$.
By Proposition~\ref{prop.amgis.iff.pp} there exist $p\in U\ppa X$ and $q\in U$ such that $r\in p\bpp q$, and by Proposition~\ref{prop.amgis.iff.pp} again, $p\bpp U\subseteq X$.
Therefore by Definition~\ref{defn.YppaX} $p\bpp q\subseteq X$ and in particular $r\in X$ as required.

Now consider $p\in X$.
By Proposition~\ref{prop.amgis.iff.pp} it suffices to show that $p\bpp U\subseteq X\app U$.
But this is immediate from Definition~\ref{defn.sub.sets.pp}. 
\end{proof}

%%%%%%%%%%%%%%%%%%%%%%%%%%%%%%%%
\subsection{Filters in the presence of $\app$ and $\ppa$}

The notions of filter and ideal from Definitions~\ref{defn.filter} and~\ref{defn.ideal} do not change with the addition of $\app$ and $\ppa$.
This is very convenient, because it leaves unaffected the `logical' structure studied previously to Section~\ref{sect.operators}, and the theorems we proved so far still hold.

However, the addition of $\app$ and $\ppa$ adds structure, and this gives us three useful ways to build new filters out of old ones, which we now consider.
Fix some $\ns D\in\indiapp$ (Definition~\ref{defn.indiapp}).

\subsubsection{Combining filters I ($q\bpp y$)}

\begin{defn}
\label{defn.qappy}
Suppose $q$ is a filter in $\ns D$ and $y\in|\ns D|$.
Define $q\bpp y\subseteq|\ns D|$ by
\begin{frameqn}
\begin{array}{r@{\ }l}
q\bpp y =&\{ x{\in}|\ns D| \mid y\ppa x\in q\}
\end{array}
\end{frameqn}
\end{defn}

A justification for the notation $q\bpp y$ is Lemma~\ref{lemm.something}, which exhibits $q\bpp y$ as a kind of dual to $y\ppa x$:
\begin{lemm}
\label{lemm.something}
$x\in q\bpp y$ if and only if $y\ppa x\in q$.
\end{lemm}
\begin{proof}
Routine from Definition~\ref{defn.qappy}.
\end{proof}

\begin{prop}
\label{prop.qappx.filter}
Suppose $q$ is a filter in $\ns D$ and $y\in|\ns D|$.
Then:
\begin{enumerate*}
\item
$q\bpp y$ satisfies conditions~\ref{filter.up}, \ref{filter.and}, and~\ref{filter.new} of Definition~\ref{defn.filter}.
\item
If $q$ is prime then $q\bpp y$ in addition satisfies the primeness condition (Definition~\ref{defn.prime.filter}).
\item\label{qappx.prime.to.prime}
If $q$ is a prime filter in $\ns D$ and $\tbot\not\in q\bpp y$ then $q\bpp y$ is a prime filter.
\end{enumerate*}
\end{prop}
\begin{proof}
Part~\ref{qappx.prime.to.prime} follows from parts~1 and~2 since the only remaining condition from the conditions of Definitions~\ref{defn.filter} and~\ref{defn.prime.filter} that $q\bpp y$ can fail to satisfy, in order to avoid being a prime filter, is condition~\ref{filter.proper} of Definition~\ref{defn.filter}.

We now prove parts~1 and~2 of this result by checking the conditions of Definitions~\ref{defn.filter} and~\ref{defn.prime.filter}
for $q\bpp x$, freely using Lemma~\ref{lemm.something}:
\begin{enumerate}
\setcounter{enumi}{1}
\item
\emph{If $x\in q\bpp y$ and $x\leq x'$ then $x'\in q\bpp y$.}\quad
If $x\in q\bpp y$ then $y\ppa x\in q$.
By part~2 of Lemma~\ref{lemm.app.leq} and condition~\ref{filter.up} of Definition~\ref{defn.filter} for $q$ we have $y\ppa x'\in q$, and so $x'\in q\bpp y$.
\item
\emph{If $x\in q\bpp y$ and $x'\in q\bpp y$ then $x{\tand}x'\in q\bpp y$.}\quad
Suppose $x,x'\in q\bpp y$, so that $y\ppa x,y\ppa x'\in q$.
Therefore $(y\ppa x)\tand(y\ppa x')\in q$ and by \rulefont{\ppa{\tand}} $y\ppa(x{\tand}x')\in q$.
Thus $x{\tand}x'\in q\bpp y$.
\item
\emph{If $\New{b}(b\ a)\act x\in q\bpp y$ then $\tall a.x\in q\bpp y$.}
\quad
We note that for fresh $a'$ (so $a'\#x,y$)\footnote{We do not assume that $q$ has small support so we do not know that $a'\#q$, but that will not be a problem.}
if $b\#x$ then by Corollary~\ref{corr.stuff} $(b\ a)\act x=(b\ a')\act((a'\ a)\act x)$ and by Lemma~\ref{lemm.freshwedge.alpha} $\tall a.x=\tall a'.(a'\ a)\act x$.
So using Proposition~\ref{prop.pi.supp} we may assume without loss of generality that $a\#y$.

Now suppose $\New{b}(b\ a)\act x\in q\bpp y$.
By Lemma~\ref{lemm.something} $\New{b}(y\ppa (b\ a)\act x\in q)$, and by Lemma~\ref{lemm.F.magic} (for the equivariant function $\ppa$) $\New{b}(b\ a)\act(y\ppa x)\in q$.
By assumption $q$ is prime, so that by condition~\ref{filter.new} of Definition~\ref{defn.filter}
$\tall a.(y\ppa x)\in q$, thus by \rulefont{\ppa\tall} and condition~\ref{filter.up} of Definition~\ref{defn.filter} $y\ppa \tall a.x\in q$ and so by Lemma~\ref{lemm.something} $\tall a.x\in q\bpp y$.
\item
\emph{If $x\tor x'\in q\bpp y$ then either $x\in q\bpp y$ or $x'\in q\bpp y$.}
\quad
Suppose $x\tor x'\in q\bpp y$.
By Lemma~\ref{lemm.something} $y\ppa (x\tor x')\in q$, so that (since $q$ is prime; Definition~\ref{defn.prime.filter}) $y\ppa x\in q$ or $y\ppa x'\in q$.
We use Lemma~\ref{lemm.something} again.
\qedhere\end{enumerate}
\end{proof}

%%%%%%%%%%%%%%%%%%%%%%%%%%
\subsubsection{Combining filters II ($q\app x$)}

We prove Lemmas~\ref{lemm.hard.1} and~\ref{lemm.hard.2}, which will help us prove Theorem~\ref{thrm.pp.app}.
These lemmas are versions of (and corollaries of) Theorem~\ref{thrm.maxfilt.zorn}, but for the applicative structure.

Definition~\ref{defn.qappx} is essentially a special case of Definition~\ref{defn.pappq} for the case of $q\app(x{\uparrow})$:\footnote{It is \emph{not} literally true that $q\app x=q\app(x{\uparrow})$, just because $q\app x$ is not up-closed (neither is $q\app(x{\uparrow})$).  This will not matter.}
\begin{defn}
\label{defn.qappx}
If $q$ is a filter in $\ns D$ and $x\in|\ns D|$ then define $q\app x\subseteq|\ns D|$ by
\begin{frameqn}
q\app x = \{y\app x \mid y\in q\} .
\end{frameqn}
\end{defn}
$q\app x$ is not necessarily a filter, but the notation will be useful.

\begin{lemm}
\label{lemm.hard.1}
Suppose $r{\subseteq}|\ns D|$ is a prime filter and $y{\in}|\ns D|$.
\begin{enumerate*}
\item
If $\ttop\app y\in r$
then
$$
I=\{x{\in}|\ns D| \mid x\app y\not\in r\}
$$
is an ideal (Definition~\ref{defn.ideal}).
\item
For any $x{\in}|\ns D|$, if $x\app y\in r$ then there exists a prime filter $p$ such that $x{\in}p$ and $p\app y\subseteq r$.
\end{enumerate*}
\end{lemm}
\begin{proof}
\begin{enumerate}
\item
By \rulefont{\app\tbot} of Figure~\ref{fig.app.compatible} $\tbot\app y=\tbot$ and by condition~\ref{filter.proper} of Definition~\ref{defn.filter} $\tbot\not\in r$.
Thus $\tbot\in I$ and $I$ is nonempty.

We now verify that $I$ satisfies the other conditions of Definition~\ref{defn.ideal}:
\begin{enumerate}[(a)]
\item
We assumed $\ttop\app y\in r$, so $\ttop\not\in I$.
\item
By condition~\ref{filter.up} of Definition~\ref{defn.filter} $r$ is up-closed so that by Lemma~\ref{lemm.app.leq}(\ref{app.leq.left}) 
$I$ is down-closed.
\item
We use \rulefont{\app{\tor}} of Figure~\ref{fig.app.compatible}.
If $x\app y\not\in r$ and $x'\app y\not\in r$ then by primeness of $r$ (Definition~\ref{defn.prime.filter}) it follows that $(x\app y)\tor(x'\app y)\stackrel{\rulefont{\app{\tor}}}=(x\tor x')\app y\not\in r$.
\end{enumerate}
\item
Suppose $x\app y\in r$; then by Lemma~\ref{lemm.app.leq}(\ref{app.leq.left}) $\ttop\app y\in r$ and we use Theorem~\ref{thrm.maxfilt.zorn} for the (by Lemma~\ref{lemm.uparrow.filter}(1)) small-supported filter $x{\uparrow}$ and the (by part~1 of this result) ideal $I$.
\qedhere\end{enumerate}
\end{proof}

\begin{lemm}
\label{lemm.hard.2}
Suppose $p,r{\subseteq}|\ns D|$ are prime filters and $y{\in}|\ns D|$.
Then:
\begin{enumerate*}
\item
If $p\app \ttop\subseteq r$ then
$$
J=\{y{\in}|\ns D| \mid p\app y\not\subseteq r\}
$$
is an ideal (Definition~\ref{defn.ideal}).
\item
For any $y{\in}|\ns D|$, if $p\app y\subseteq r$ then there exists a prime filter $q$ such that $y{\in}q$ and $p\app q\subseteq r$.
\end{enumerate*}
\end{lemm}
\begin{proof}
\begin{enumerate}
\item
Using \rulefont{\app\tbot} of Figure~\ref{fig.app.compatible} $p\app\tbot=\{\tbot\}$ and by condition~\ref{filter.proper} of Definition~\ref{defn.filter} $\{\tbot\}\not\subseteq r$.
Thus $\tbot\in J$ and $J$ is nonempty.

We now verify that $J$ satisfies the other conditions of Definition~\ref{defn.ideal}:
\begin{enumerate}[(a)]
\item
We assumed $p\app\ttop\subseteq r$, so $\ttop\not\in J$.
\item
By condition~\ref{filter.up} of Definition~\ref{defn.filter} $r$ is up-closed so that by Lemma~\ref{lemm.app.leq}(\ref{app.leq.left}) 
$J$ is down-closed.
\item
Suppose $p\app y\not\subseteq r$ and $p\app y'\not\subseteq r$; then there exist $x{\in}p$ and $x'{\in}p$ such that $x\app y\not\in r$ and $x'\app y'\not\in r$.
By up-closure of $r$ (condition~\ref{filter.up} of Definition~\ref{defn.filter}) and Lemma~\ref{lemm.app.leq}(\ref{app.leq.left}) $(x\tand x')\app y\not\in r$ and $(x\tand x')\app y'\not\in r$ so that by primeness of $r$ (Definition~\ref{defn.prime.filter}) $(((x\tand x')\app y)\tor((x\tand x')\app y)\stackrel{\rulefont{\app{\tor}}}=(x\tand x')\app (y\tor y')\not\in r$.
Thus $p\app(y\tor y')\not\subseteq r$.
\end{enumerate}
\item
Suppose $p\app y\subseteq r$; then by Lemma~\ref{lemm.app.leq}(\ref{app.leq.left}) $p\app\ttop\subseteq r$ and we use Theorem~\ref{thrm.maxfilt.zorn} for the (by Lemma~\ref{lemm.uparrow.filter}(1)) small-supported filter $y{\uparrow}$ and the (by part~1 of this result) ideal $J$.
\qedhere\end{enumerate}
\end{proof}

\begin{rmrk}
\label{rmrk.r.not.ssp}
$I$ and $J$ from Lemmas~\ref{lemm.hard.1} and~\ref{lemm.hard.2} are ideals, but they need not be small-supported, because the prime filter $r$ need not be small-supported. 
\end{rmrk}

%%%%%%%%%%%%%%%%%%%%%%%%%%%%%%%%%%
\subsection{The second representation theorem}

For the rest of this subsection fix some $\ns D\in\indiapp$ (Definition~\ref{defn.indiapp}).
We now show how to build an $\amgis$-algebra with $\bpp$ using $\ns D$.
This is reasonable; recall from Subsection~\ref{subsect.bpp} that $\bpp$ is the topological dual to application $\app$.

\subsubsection{The basic actions}

The first two clauses of Definition~\ref{defn.p.action.pp} echo Definition~\ref{defn.p.action}; the third gives prime filters a $\bpp$ action in the sense of Definition~\ref{defn.bpp}.
Recall that $\pp x$ is from Definition~\ref{defn.pp}:
\begin{defn}
\label{defn.p.action.pp}
Give prime filters $p$ and $q$ in $\ns D$ actions as follows:
\begin{frameqn}
\begin{array}{r@{\ }l@{\qquad}r@{\ }l}
\pi\act p=&\{\pi\act x\mid x\in p\}
&
p[u\ms a]=&\{x\mid x[a\sm u]\in p\}
\\
p\bpp q=&\bigcap\{\pp{(x\app y)}\mid x\in p, y\in q\}
\end{array}
\end{frameqn}
\end{defn}

\begin{prop}
\label{prop.filters.amgis.pp}
Prime filters of $\ns D$ form an $\amgis$-algebra with $\bpp$ in the sense of Definition~\ref{defn.bpp}.
\end{prop}
\begin{proof}
Equivariance of $\bpp$ is immediate from Theorem~\ref{thrm.equivar}.
We just use Proposition~\ref{prop.filters.amgis}.
\end{proof}

We now work towards Proposition~\ref{prop.bpp.subset}, which will be needed for Theorem~\ref{thrm.pp.app}.

%%%%%%%%%%%%%%%%%%%%%%%%%%
\subsubsection{Combining filters III ($p\app q$)}

We can extend the applicative structure $x\app y$ from Definition~\ref{defn.FOLeq.pp}, to points $p\app q$:
\begin{defn}
\label{defn.pappq}
Suppose $p$ and $q$ are prime filters in $\ns D$ and $y\in|\ns D|$.
Define $p\app q$ by
\begin{frameqn}
p\app q =\{ x\app y \mid x\in p,\ y\in q\}
\end{frameqn}
\end{defn}
$p\app q$ is not necessarily a filter; for instance, it is not necessarily up-closed (condition~\ref{filter.up} of Definition~\ref{defn.filter}).
However, we have Lemmas~\ref{lemm.sigma.iff.pp} and~\ref{lemm.sigma.iff.pp.supplement}:
\begin{lemm}
\label{lemm.sigma.iff.pp}
Suppose $p$, $q$, and $r$ are prime filters in $\ns D$.
Then
$r\in p\bpp q$ if and only if $p\app q\subseteq r$.
\end{lemm}
\begin{proof}
By construction in Definition~\ref{defn.p.action.pp}. 
\end{proof}

\begin{lemm}
\label{lemm.sigma.iff.pp.supplement}
Suppose $p$ and $r$ are prime filters in $\ns D$ and suppose $y{\in}|\ns D|$.
Then
$r\in p\bpp (\pp y)$
if and only if
there exists a prime filter $q$ in $\ns D$ such that $y{\in}q\land p\app q\subseteq r$.
\end{lemm}
\begin{proof}
Recall from Definition~\ref{defn.pp} that if $y{\in}|\ns D|$ then $\pp y$ is the set of prime filters in $\ns D$ containing $y$.
Recall from Definition~\ref{defn.YppaX} that $r\in p\bpp \pp y$ if and only if $p\app q\subseteq r$ for some prime filter $q\in\pp y$.
The result follows.
\end{proof}

\subsubsection{The representation theorem}

\begin{prop}
\label{prop.bpp.subset}
Suppose $r$ is a prime filter in $\ns D$ and $x,y{\in}|\ns D|$.
Then $r\bpp (\pp y)\subseteq\pp x$ if and only if $y\ppa x\in r$.
\end{prop}
\begin{proof}
We prove two implications.

\emph{The right-to-left implication.}\quad
Suppose $y\ppa x\in r$ and suppose $p\in r\bpp \pp y$.
We need to show that $p\in\pp x$.

By Lemma~\ref{lemm.sigma.iff.pp.supplement} $p\in r\bpp \pp y$ means that for some $q$ with $y\in q$ it is the case that $r\app q\subseteq p$.
So given some such $q$, since $y\ppa x\in r$ we have that $(y\ppa x)\app y\in p$.
By \rulefont{\app\epsilon} $(y\ppa x)\app y\leq x$ and since $p$ is up-closed (condition~\ref{filter.up} of Definition~\ref{defn.filter}) we have $x\in p$, thus by Definition~\ref{defn.pp} $p\in\pp x$ as required.

\emph{The left-to-right implication.}\quad
Suppose $r\bpp \pp y\subseteq\pp x$; so if $p$ is a prime filter and $p\in r\bpp \pp y$, then $x\in p$.
We need to show that $y\ppa x\in r$.

Consider $r\bpp y$ from Definition~\ref{defn.qappy}.
There are two cases:
\begin{itemize}
\item
\emph{Suppose $\tbot{\in} r\bpp y$.}\quad
Then by Lemma~\ref{lemm.something} $y\ppa \tbot\in r$ and by Lemma~\ref{lemm.app.leq} $y\ppa x\in r$ as required.
\item
\emph{Suppose $\tbot{\not\in} r\bpp y$.}\quad
Then by Proposition~\ref{prop.qappx.filter}(\ref{qappx.prime.to.prime}) $r\bpp y$ is a prime filter.

By Lemma~\ref{lemm.uparrow.filter} $y{\uparrow}$ is a filter, and by Lemma~\ref{lemm.hard.2}(2) there exists a prime filter $q$ with $y{\uparrow}\subseteq q$ and $r\app q\subseteq r\bpp y$.
Now $y\in y{\uparrow}\subseteq q$ so we have a prime filter $q$ with $y\in q$ and $r\app q\subseteq r\bpp y$.
By Lemma~\ref{lemm.sigma.iff.pp.supplement} $p\in r\bpp\pp y$.
Thus by assumption $x\in r\bpp y$ and by Lemma~\ref{lemm.something} $y\ppa x\in r$ as required.
\qedhere\end{itemize}
\end{proof}

Theorem~\ref{thrm.pp.app} extends Lemma~\ref{lemm.bullet.commute} for the additional structure of $\app$ and $\ppa$:
\begin{thrm}
\label{thrm.pp.app}
Suppose $\ns D\in\indiapp$ and $x,y\in|\ns D|$.
Then:
\begin{enumerate*}
\item
$\pp x\app \pp y=\pp{(x\app y)}$.
\item
$\pp y\ppa \pp x=\pp{(y\ppa x)}$.
\end{enumerate*}
\end{thrm}
\begin{proof}
We consider each part in turn.
Below, $p$, $q$, and $r$ range over prime filters in $\ns D$.
\begin{enumerate}
\item
By Definition~\ref{defn.pp} and Proposition~\ref{prop.amgis.iff.pp} $r\in \pp x\app \pp y$ if and only if
$\Exists{p,q}(x{\in}p\land y{\in}q)\land r\in p\bpp q$,
and by
Lemma~\ref{lemm.sigma.iff.pp}
this is if and only if
$\Exists{p,q}(x{\in}p\land y{\in}q)\land p\app q{\subseteq} r$.

So suppose there exist prime filters $p$ and $q$ with $x{\in}p$ and $y{\in}q$ and
$p\app q\subseteq r$.
Then clearly $x\app y\in r$.

Conversely suppose $x\app y\in r$ for some prime filter $r$.
By Lemma~\ref{lemm.hard.1}(2) there exists a prime filter $p$ with $x\in p$ and $p\app y\subseteq r$, and
by Lemma~\ref{lemm.hard.2}(2) there exists a prime filter $q$ with $y\in q$ and $p\app q\subseteq r$.
\item
We reason as follows, using Propositions~\ref{prop.amgis.iff.pp} and~\ref{prop.bpp.subset}:
$$
r\in \pp y\ppa \pp x
\ \stackrel{\text{Prop~\ref{prop.amgis.iff.pp}}}{\liff}\
r\bpp \pp y\subseteq\pp x
\ \stackrel{\text{Prop~\ref{prop.bpp.subset}}}{\liff}\
r\in \pp{(y\ppa x)}
\qedhere$$
\end{enumerate}
\end{proof}

\subsubsection{Properties of the representation}

We can now easily extend Definition~\ref{defn.pp.B} (the definition of $\pp{\ns D}$):
\begin{defn}
\label{defn.pp.B.app}
As in Definition~\ref{defn.pp.B} we take $\pp{\ns D}$ to have:
\begin{itemize*}
\item
$|\pp{\ns D}|=\{\pp x\mid x\in|\ns D|\}$
\item
$(\pp{\ns D})^\prg=\ns D^\prg$ and $\prg_{\pp{\ns D}} u=\pp{(\prg_{\ns D}u)}$
\item
$\pi\act (\pp x)=\pp{(\pi\act x)}$
\item
$\pp x[a\sm u]=\pp{(x[a\sm u])}$
\end{itemize*}
We give $\pp{\ns D}$ the actions $X\app Y$ and $Y\ppa X$ from Definition~\ref{defn.sub.sets.pp}.
\end{defn}
So Definitions~\ref{defn.pp.B} and~\ref{defn.pp.B.app} overload the notation $\pp{\ns D}$ for ``the object in $\india$ composed of prime filters of $\ns D\in\india$'' and ``the object of $\indiapp$ composed of prime filters of $\ns D\in\indiapp$''.
The meaning will always be clear.

\begin{lemm}
\label{lemm.filter.pp.app.commute}
Suppose $x,y\in|\ns D|$ and $u\in|\ns D^\prg|$.
Then \rulefont{\sigma\app} and \rulefont{\sigma\ppa} are valid in $\pp{\ns D}$:
$$
\begin{array}{r@{\ }l}
(\pp x\app\pp y)[a\sm u]=&
\pp x[a\sm u]\app (\pp y[a\sm u])
\\
(\prg b\ppa \pp x)[a\sm u]=&
\prg b\ppa (\pp x[a\sm u])
\end{array}
$$
\end{lemm}
\begin{proof}
We reason as follows:
\begin{tab2r2}
(\pp x\app\pp y)[a\sm u]
=&
\pp{(x\app y)}[a\sm u]
&\text{Theorem~\ref{thrm.pp.app}}
\\
=&
\pp{((x\app y)[a\sm u])}
&\text{Lemma~\ref{lemm.dup}}
\\
=&
\pp{(x[a\sm u]\app y[a\sm u])}
&\rulefont{\sigma\app}
\\
=&
\pp{(x[a\sm u])}\app \pp{(y[a\sm u])}
&\text{Theorem~\ref{thrm.pp.app}}
\\
=&
\pp x[a\sm u]\app \pp y[a\sm u]
&\text{Lemma~\ref{lemm.dup}}
\\[1.5ex]
(\prg_{\pp{\ns D}} b\ppa \pp x)[a\sm u]=&
\pp{(\prg_{\ns D} b\ppa x)}[a\sm u]
&\text{Theorem~\ref{thrm.pp.app}}
\\
=&
\pp{(\prg_{\ns D} b\ppa x)[a\sm u]}
&\text{Lemma~\ref{lemm.dup}}
\\
=&
\pp{(\prg_{\ns D} b\ppa x[a\sm u])}
&\rulefont{\sigma\ppa}
\\
=&
\prg_{\pp{\ns D}} b\ppa \pp{(x[a\sm u])}
&\text{Theorem~\ref{thrm.pp.app}}
\qedhere\end{tab2r2}
\end{proof}

\begin{lemm}
\label{lemm.filter.adjoint.compat}
The adjoint and compatibility axioms from Definition~\ref{defn.FOLeq.pp} are valid in $\pp{\ns D}$.
That is:
$$
\begin{array}{@{\hspace{-0ex}}l@{\ \ }r@{\ }r@{\ }l@{\quad}r@{\ }l}
\rulefont{\app\epsilon} && (\pp u\ppa \pp x)\app \pp x\subseteq&\pp x
\\
\rulefont{\app\eta} && \pp x\subseteq&\pp u\ppa(\pp x\app \pp u)
\\
\rulefont{\app\tbot} && \varnothing\app \pp u=&\varnothing
& \pp x\app\varnothing =&\varnothing
\\
\rulefont{\app\tand} && (\pp x\cap\pp y)\app \pp u \subseteq & (\pp x\app \pp u)\cap (\pp y\app \pp u)
&
\pp x\app(\pp u\cap \pp v)\subseteq& (\pp x\app \pp u)\cap(\pp x\app\pp v)
\\
\rulefont{\app\tor} && (\pp x\cup\pp y)\app u=&(\pp x\app\pp u)\cup(\pp y\app\pp u)
&
\pp x\app(\pp u\cup\pp v)=&(\pp x\app \pp u)\cup(\pp x\app \pp v)
\\
&&
\dots
\\
\rulefont{\ppa\tall} &
b\#\pp u\limp & \freshcap{b}(\pp u\ppa \pp x)\subseteq&\pp u\ppa(\freshcap{b}\pp x)
\end{array}
$$
\end{lemm}
\begin{proof}
The easiest proof
is to combine Theorem~\ref{thrm.pp.app} and Lemma~\ref{lemm.bullet.commute} with Lemma~\ref{lemm.completeness} and with the relevant axiom for $\ns D$.

In the cases of \rulefont{\app\tall} and \rulefont{\ppa\tall} which have a freshness condition, we use Lemma~\ref{lemm.some.fresh.s} to choose a suitably fresh representative of $\pp u$ (i.e. a $u'$ such that $\pp{(u')}=\pp u$ and $b\#u'$).\footnote{Lemma~\ref{lemm.some.fresh.s} is applied here to the function $u\mapsto \pp u$.  Equivariance is defined in Definition~\ref{defn.equivariant}; Definition~\ref{defn.pp.B.app} states (amongst other things) that this function is equivariant.}
\end{proof}

Recall the definition of $\indiapp$ from Definition~\ref{defn.indiapp}.
\begin{frametxt}
\begin{thrm}[Second representation theorem]
\label{thrm.pp.iso.pp}
If $\ns D\in\indiapp$ then $\pp{\ns D}$ from Definition~\ref{defn.pp.B.app} is in $\indiapp$, and the assignment $x\mapsto \pp x$ is an isomorphism in $\indiapp$ between $\ns D$ and $\pp{\ns D}$.
\end{thrm}
\end{frametxt}
\begin{proof}
Theorem~\ref{thrm.pp.iso} handles the purely logical structure ($\cap$, $\varnothing$, $\cup$, and $\freshcap{a}$).
Lemmas~\ref{lemm.filter.pp.app.commute} and~\ref{lemm.filter.adjoint.compat} validate the axioms for $\app$ and $\ppa$.
\end{proof}

%%%%%%%%%%%%%%%%%%%%%%%%%%%%%%%%%%%%%%%
\subsection{Construction of the topological space $F(\ns D)$, with $\app$ and $\ppa$}

Recall the definitions of $\indiapp$ and $\inspectapp$ from Definitions~\ref{defn.indiapp} and~\ref{defn.g.extend}.
We will now extend Definitions~\ref{defn.F} and~\ref{defn.Ff}:
\begin{defn}
\label{defn.F.pp}
Suppose $\ns D\in\indiapp$.
Define $F(\ns D)\in\inspectapp$ by:
\begin{enumerate*}
\item
$|F(\ns D)|=|\points(\ns D)|$ and $F(\ns D)^\prg=\ns D^\prg$.
\item
$\pi\act p=\{\pi\act x\mid x\in p\}$.
\item
$p[u\ms a]=\{x\mid x[a\sm u]\in p\}$.
\item
$p\bpp q=\bigcap\{\pp{(x\app y)}\mid x\in p, y\in q\}$, following Definition~\ref{defn.p.action.pp}.
\item
$\otop{F(\ns D)}$ is the closure of $\{\pp x\mid x\in|\ns D|\}$ under small-supported unions.
\item
$\prg_{F(\ns D)}$ maps $u\in|\ns D^\prg|$ to $\pp{(\prg_{\ns D} u)}$.\footnote{It might be helpful to unwind the definitions for this final clause.
This is not complicated---it just takes in a lot of definitions!

$u\in|\ns D^\prg|$ is an element of the termlike $\sigma$-algebra over which $F(\ns D)$ has an $\amgis$-action; $\prg_{\ns D}u\in|\ns D|$ is an element of $\ns D$ the (impredicative) nominal distributive lattice with $\tall$ and $\app$; $\pp{(\prg_{\ns D} u)}$ is the set of prime filters in $\points(\ns D)$ that contain $\prg_{\ns D} u$.
By Theorem~\ref{thrm.pp.x.clopen}(1) this set of prime filters is compact, that is it is in $\ctop{\ns T}$, as required in Definition~\ref{defn.impredicative.top}.}
\end{enumerate*}
\end{defn}
In Definition~\ref{defn.F.pp} we claim that $F(\ns D)$ is a nominal spectral space with $\bpp$ (Definition~\ref{defn.nom.top.spectral.bpp}).
This needs to be proved: Theorem~\ref{thrm.F.funct.with.app} assembles the various verifications.

Recall from Definition~\ref{defn.Ff} the map from $f:\ns D\equivarto\ns D'$ to $F(f):F(\ns D')\equivarto F(\ns D)$.
\begin{prop}
\label{prop.Ff.app}
If $f:\ns D\equivarto\ns D'$ is in $\indiapp$ then $F(f)$ commutes with $\app$ and $\ppa$ as specified in Definition~\ref{defn.F.pp}.
That is, using Theorem~\ref{thrm.pp.x.clopen}:
\begin{enumerate*}
\item
$F(f)^\mone(\pp x\app \pp y)=F(f)^\mone(\pp x)\app F(f)^\mone(\pp y)$
\item
$F(f)^\mone(\prg b\ppa \pp x)=\prg b\ppa F(f)^\mone(\pp x)$.
\end{enumerate*}
As a corollary, if $f:\ns D\equivarto\ns D'$ then $F(f):F(\ns D')\equivarto F(\ns D)$ is a morphism in the sense of Definition~\ref{defn.g.extend}.
\end{prop}
\begin{proof}
The corollary follows direct from Definition~\ref{defn.g.extend} using Theorem~\ref{thrm.F.funct}.

For the first part, we reason as follows:
\begin{tab2r2}
r\in F(f)^\mone(\pp x\app \pp y)\liff& F(f)(r)\in \pp x\app \pp y
&\text{Inverse image}
\\
\liff& F(f)(r)\in \pp{(x\app y)}
&\text{Theorem~\ref{thrm.pp.app}}
\\
\liff&
x\app y\in F(f)(r)
&\text{Definition~\ref{defn.pp}}
\\
\liff&
f(x\app y)\in r
&\text{Definition~\ref{defn.Ff}}
\\
\liff&
f(x)\app f(y)\in r
&\text{Definition~\ref{defn.indiapp}}
\\
\liff&
r\in \pp{(f(x)\app f(y))}
&\text{Definition~\ref{defn.pp}}
\\
\liff&
r\in \pp{f(x)}\app \pp{f(y)}
&\text{Theorem~\ref{thrm.pp.app}}
\\
\liff&
r\in F(f)^\mone(\pp x)\app F(f)^\mone(y)
&\text{Lemma~\ref{lemm.FfmoneU}}
\end{tab2r2}
The reasoning for the second part, for $\ppa$, is similar:
\begin{tab2r2}
r\in F(f)^\mone(\prg_{F(\ns D)} b\ppa \pp x)\liff& F(f)(r)\in \prg_{F(\ns D)}b\ppa \pp x
&\text{Inverse image}
\\
\liff& F(f)(r)\in \pp{(\prg_{\ns D}b\ppa x)}
&\text{Theorem~\ref{thrm.pp.app}}
\\
\liff&
\prg_{\ns D}b\ppa x\in F(f)(r)
&\text{Definition~\ref{defn.pp}}
\\
\liff&
f(\prg_{\ns D}b \ppa x)\in r
&\text{Definition~\ref{defn.Ff}}
\\
\liff&
\prg_{\ns D'} b\ppa f(x)\in r
&\text{Definition~\ref{defn.indiapp}}
\\
\liff&
r\in \pp{(\prg_{\ns D'}b\ppa f(x))}
&\text{Definition~\ref{defn.pp}}
\\
\liff&
r\in \pp{(\prg_{\ns D'}b)}\ppa \pp{f(x)}
&\text{Theorem~\ref{thrm.pp.app}}
\\
\liff&
r\in \prg_{F(\ns D')}b\ppa F(f)^\mone(x)
&\text{Lemma~\ref{lemm.FfmoneU}}
\qedhere\end{tab2r2}
\end{proof}

Theorem~\ref{thrm.F.funct.with.app} extends Theorem~\ref{thrm.F.funct} from $\india$ and $\inspecta$ to $\indiapp$ and $\inspectapp$:
\begin{thrm}
\label{thrm.F.funct.with.app}
$F$ is a functor from $\indiapp$ to $\inspectapp$.
\end{thrm}
\begin{proof}
This is mostly Theorem~\ref{thrm.F.funct} combined with Proposition~\ref{prop.Ff.app}.
We do also need to check the extra conditions on coherence from Definition~\ref{defn.coherent.bpp};
these follow easily from Theorem~\ref{thrm.pp.app}, Lemma~\ref{lemm.filter.pp.app.commute}, and Theorem~\ref{thrm.pp.x.clopen}.
\end{proof}

Part~1 of Proposition~\ref{prop.Ff.app} is stated only for sets of the form $\pp x$ and $\pp y$.
In fact, we note that it can be extended to all $X,Y\in\otop{F(\ns D)}$:
\begin{corr}
\label{corr.extend.Ff.app}
$F(f)^\mone(X\app Y)=F(f)^\mone(X)\app F(f)^\mone(Y)$ for all $X,Y\in\otop{F(\ns D)}$ (and not only $X,Y\in\ctop{F(\ns D)}$).
\end{corr}
\begin{proof}
By construction in Definition~\ref{defn.F} every open set in $F(\ns D)$ is a small-supported union of compact opens.
By Theorem~\ref{thrm.pp.x.clopen} compact opens have in $F(\ns D)$ the form $\pp x$ and $\pp y$ for $x,y\in|\ns D|$.
We use Lemma~\ref{lemm.bigcup.app}(1) and Proposition~\ref{prop.Ff.app}(1).
\end{proof}
Corollary~\ref{corr.extend.Ff.app} would not work for $\ppa$, because a result corresponding to Lemma~\ref{lemm.bigcup.app} does not hold for it.

\subsection{The duality, in the presence of $\app$ and $\ppa$}
\label{subsect.app.equivalence}

It is routine to extend Definition~\ref{defn.G}---which sends a spectral space $\ns T$ to the lattice of its compact open sets $G(\ns T)$---and Definition~\ref{defn.Gg}---which sends a spectral map $g$ to its inverse image function $G(g)=g^\mone$---to the case where we also assume $\app$ and $\ppa$.
We write it out, just to be clear:
\begin{defn}
\label{defn.G.bpp}
\label{defn.Gg.bpp}
If $\ns T\in\inspectapp$ (Definition~\ref{defn.nom.top.spectral.bpp}) define $G(\ns T)\in\indiapp$ (Definition~\ref{defn.indiapp}) by:
\begin{itemize*}
\item
$|G(\ns T)|=\ctop{\ns T}$ and $G(\ns T)^\prg=\ns T^\prg$.
\item
$\pi\act U=\{\pi\act p\mid p\in U\}$ and $U[a\sm u]=\{p\mid p[u\ms a]\in U\}$, where $U\in|G(\ns T)|$ and $u\in|G(\ns T)^\prg|$.
\item
$\ttop$,
$\tand$, $\tbot$, $\tor$, and $\tall$ are interpreted as the whole underlying set,
set intersection, the empty set, set union, and $\freshcap{a}$.
\item
$X\app Y$ and $Y\ppa X$ are interpreted as specified in Definition~\ref{defn.sub.sets.pp}.
\end{itemize*}
Given $g:\ns T\equivarto \ns T'\in \inspectapp$ from Definition~\ref{defn.g.extend}, define $G(g):G(\ns T')\equivarto G(\ns T)$ by $G(g)(U)=g^\mone(U)$.
\end{defn}

\begin{lemm}
\label{lemm.ppa.bigcap}
Continuing the notation of Definition~\ref{defn.G.bpp}, suppose $U,V\in|G(\ns T)|$ and suppose $\mathcal U\subseteq|G(\ns T)|$.
Then
$$
\bigcap_{U{\in}\mathcal U}(V\ppa U)= V\ppa\bigcap_{U{\in}\mathcal U}\mathcal U.
$$
\end{lemm}
\begin{proof}
Using Proposition~\ref{prop.amgis.iff.pp} $r\in V\ppa\bigcap\mathcal U$ if and only if $r\bpp V\subseteq U$ for every $U\in\mathcal U$, and by Proposition~\ref{prop.amgis.iff.pp} again this is if and only if $r\in V\ppa U$ for every $U\in\mathcal U$, which is if and only if $r\in\bigcap_{U{\in}\mathcal U} V\ppa U$.
\end{proof}

\begin{lemm}
\label{lemm.GT.app}
If $\ns T\in\inspectapp$ then $G(\ns T)$ validates the axioms from Figure~\ref{fig.app.compatible} (Definition~\ref{defn.FOLeq.pp}).
\end{lemm}
\begin{proof}
We consider each axiom in turn.
We take $X,Y,X',U\in\ctop{\ns T}$ (open compacts in $\ns T$):
\begin{itemize*}
\item
\emph{Axioms \rulefont{\sigma\app} and \rulefont{\sigma\ppa}.}
\quad
By assumption in Definition~\ref{defn.coherent.bpp}.
\item
\emph{The adjoint axioms \rulefont{\app\epsilon} and \rulefont{\app\eta}.}
\quad
Direct from Lemma~\ref{lemm.bruce}.
\item
\emph{Axioms \rulefont{\app\bot} and \rulefont{\app\tor}}
\quad
\dots are Lemma~\ref{lemm.bigcup.app}.
\item
\emph{Axioms \rulefont{\app\tand}.}\quad
If $r\in (X\cap X')\app Y$ then $r\in p\bpp q$ for some $p\in X\cap X'$ and $q\in Y$.
It follows that $r\in p\bpp q$ for $p\in X$ and $q\in Y$ and $r\in p\bpp q$ for $p\in X'$ and $q\in Y$, and therefore $r\in (X\app Y)\cap(X'\app Y)$.
The second \rulefont{\app\tand} axiom follows similarly.
\item
\emph{Axiom \rulefont{\app\tall}.}\quad
Suppose $b\#U$ and $r\in (\freshcap{a}X)\app U$.
It follows by Lemma~\ref{lemm.bigcup.app} that $r\in X[a\sm w]\app U$ for every $w\in|\ns T^\prg|$.
Now by part~1 of Lemma~\ref{lemm.X.sub.fresh.alpha}
and condition~\ref{item.nom.top.app.sub} of Definition~\ref{defn.coherent.bpp}, $X[a\sm w]\app U=(X\app U)[a\sm w]$.

So $r\in (X\app U)[a\sm w]$ for every $w\in|\ns T^\prg|$, and by Definition~\ref{defn.nu.U}
we have $r\in\freshcap{a}(X\app U)$.
\item
\emph{Axiom \rulefont{\ppa{\tand}}.}\quad
From Lemma~\ref{lemm.ppa.bigcap}.
\item
\emph{Axiom \rulefont{\ppa{\tor}}.}\quad
$r\in (U\ppa X)\cup (U\ppa X')$ means $r\bpp U\subseteq X$ or $r\bpp U\subseteq X'$.
In either case, $r\bpp U\subseteq X\cup X'$ and this means $r\in U\ppa(X\cup X')$.
\item
\emph{Axiom \rulefont{\ppa\tall}.}\quad
Suppose $b\#U$.
We reason as follows:
$$
\begin{array}[b]{r@{\ }l@{\qquad}l}
\freshcap{b}(U\ppa X)
=&\bigcap_{u\in|\ns T^\prg|} (U\ppa X)[b\sm u]
&\text{Definition~\ref{defn.nu.U}}
\\
\subseteq& \bigcap_{u\in|\ns T^\prg|} U[b\sm u]\ppa X[b\sm u]
&\text{Lemma~\ref{lemm.sigma.ppa}}
\\
=&\bigcap_{u\in|\ns T^\prg|} U\ppa X[b\sm u]
&\rulefont{\sigma\#},\ b\#U
\\
=& U\ppa \bigcap_{u\in|\ns T^\prg|}X[b\sm u]
&\text{Lemma~\ref{lemm.ppa.bigcap}}
\\
=& U\ppa\hspace{1.5pt} \freshcap{b}X
&\text{Definition~\ref{defn.nu.U}}
\end{array}
\qedhere$$
\end{itemize*}
\end{proof}

\begin{prop}
\label{prop.G.funct.bpp}
$G$ from Definition~\ref{defn.G.bpp} is a functor from $\inspectapp^\f{op}$ to $\indiapp$.
\end{prop}
\begin{proof}
The action on objects is handled by Theorem~\ref{thrm.T.to.G.obj} and Lemma~\ref{lemm.GT.app}.
The action on morphisms is handled by Proposition~\ref{prop.G.funct} and by the two conditions on $g^\mone$ in Definition~\ref{defn.g.extend}.
\end{proof}

We checked in Proposition~\ref{prop.Ff.app} that this is true for $g=F(f)^\mone$, so we can extend Proposition~\ref{prop.ess.surj}:
\begin{prop}
\label{prop.ess.surj.pp}
If $\ns D\in\indiapp$ then $GF(\ns D)$ is equal to $\pp{\ns D}$ from Definition~\ref{defn.pp.B.app}, and the map $x\mapsto \pp x$ is an isomorphism in $\indiapp$.
\end{prop}
\begin{proof}
Just as the proof of Proposition~\ref{prop.ess.surj}; the extra structure of $\app$ and $\ppa$ has no effect.
We use Theorem~\ref{thrm.pp.iso.pp}.
\end{proof}

It is routine to check that the $\app$ and $\ppa$ structure is orthogonal to the material of Subsections~\ref{subsect.G.action} and~\ref{subsect.the.equivalence}, and so we obtain Theorem~\ref{thrm.equivalence.pp}:

\begin{frametxt}
\begin{thrm}[The duality theorem]
\label{thrm.equivalence.pp}
$G:\inspectapp^\tf{op}\to\indiapp$ defines an equivalence between $\indiapp$ and $\inspectapp^{\f{op}}$.
\end{thrm}
\end{frametxt}

Theorem~\ref{thrm.equivalence.pp} exhibits $\indiapp$ and $\inspectapp$ (Definitions~\ref{defn.indiapp} and~\ref{defn.g.extend}) as dual to one another.
This is a general result---the abstract nominal algebra structures in $\indiapp$ correspond dually to concrete topological spaces in $\inspectapp$.

It remains to show how $\indiapp$ and $\inspectapp$ relate specifically to the untyped $\lambda$-calculus.

%%%%%%%%%%%%%%%%%%%%%%%%%%%%%%%%%%%%%%%%%%%%
\jamiepart{Application to the $\lambda$-calculus}

%%%%%%%%%%%%%%%%%%%%%%%%%%%%%%%%%%%%%
\section{The $\lambda$-calculus}
\label{sect.lambda.calculus}

In this section we sketch the untyped $\lambda$-calculus and show how it has been living inside $\indiapp$ all along: this is $\tlam a.x$ in Notation~\ref{nttn.lambda}.
We make formal that Notation~\ref{nttn.lambda} is `a right thing to do' with Proposition~\ref{prop.lambda.beta.eta}, Definition~\ref{defn.ddenot}, and Theorem~\ref{thrm.lambda.soundness}.

We also briefly unpack what $\tlam a.X$ is when $X$ is an open set in the topological representations in $\inspectapp$.
This is Proposition~\ref{prop.unpack.lambda.for.G}.

Thus, we leverage our topological duality to give both abstract and concrete (i.e. nominal poset flavoured and nominal sets flavoured) semantics for the $\lambda$ of the untyped $\lambda$-calculus.

\subsection{Syntax of the $\lambda$-calculus}

\begin{defn}
\label{defn.lamtrm}
Define \deffont{$\lambda$-terms} as usual by
$$
s::= a \mid \lam{a}s \mid s's
$$
where $a$ ranges over atoms (so we use atoms as variable symbols, in nominal style).\footnote{We could allow constants $\tf c$ too, if we wished.}
\begin{itemize*}
\item
We treat $\lambda$-terms as equal up to $\alpha$-equivalence.\footnote{\dots using nominal abstract syntax \cite{gabbay:thesis,gabbay:newaas-jv} or by taking equivalence classes or by whatever other method the reader prefers.}
\item
We assume capture-avoiding substitution $s[a\ssm u]$.
\item
We write $\fa(s)$ for the free atoms (variables) of $s$.
\end{itemize*}
\end{defn}

\begin{defn}
\label{defn.lamtrm.pi}
Consider $\lambda$-terms as a nominal set (Definition~\ref{defn.nominal.set})
by giving them the natural permutation action:
$$
\pi\act a=\pi(a)\qquad
\pi\act(\lam{a}s)=\lam{\pi(a)}\pi\act s\qquad
\pi\act (s's)=(\pi\act s')(\pi\act s) 
$$
Write $\lamtrm$ for the nominal set of $\lambda$-terms with this permutation action.
\end{defn}

It is a fact that with the permutation action above, $\fa(s)$ the free atoms of $s$ and $\supp(s)$ the atoms in the support of $s$, coincide.

\begin{defn}
\label{defn.lamtrm.sigma}
Consider $\lambda$-terms as a termlike $\sigma$-algebra (Definition~\ref{defn.term.sub.alg})
by setting
$$
s[a\sm u] = s[a\ssm u] ,
$$
so that $[a\sm u]$ acting on $s$ is `$s$ with $a$ substituted for $u$'.
\end{defn}

It is a fact that this does indeed determine a termlike $\sigma$-algebra.
The nominal algebra axioms of Figure~\ref{fig.nom.sigma} reflect valid properties of capture-avoiding substitution on $\lambda$-terms.

%%%%%%%%%%%%%%%%%%%%%%%%%%%%%%%%%%%%
\subsection{$\lambda$, $\beta$, and $\eta$ using adjoints}
\label{subsect.beta.eta}

In objects of $\indiapp$, $\lambda$-abstraction arises naturally by combining the `logical' structure $\tall$ and $\leq$ with the `combinational' structure of $\app$ and $\ppa$; this is Notation~\ref{nttn.lambda}.
We shall see that $\beta$-reduction and $\eta$-expansion arise as natural corollaries of the adjoint properties of $\app$ and $\ppa$; this is Proposition~\ref{prop.lambda.beta.eta}.

Proposition~\ref{prop.unpack.lambda.for.G} unpacks what this means in $F(\ns D)$, and Definition~\ref{defn.ddenot} and Theorem~\ref{thrm.lambda.soundness} show how we can interpret the full untyped $\lambda$-calculus.

%%%%%%%%%%%%%%%%%%%%%%%%%%%%%%%
\subsubsection{$\lambda$ using $\tall$ and $\ppa$}

\begin{frametxt}
\begin{nttn}
\label{nttn.lambda}
Suppose $\ns D\in\indiapp$ and $x\in|\ns D|$.
Write $\tlam a.x$ for $\tall a.(\prg a\ppa x)$.
\end{nttn}
\end{frametxt}

\begin{rmrk}
\label{rmrk.unpack.tlam}
We unpack some of Notation~\ref{nttn.lambda}.
The notation $\prg a$ is explained in detail in Notation~\ref{nttn.impredicative.D}.
In full,
$$
\tlam a.x\ \ \text{and}\ \ \tall a.(\prg a\ppa x)\quad\text{mean}\quad
\tall_{\ns D} a.(\prg_{\ns D}(a_{\ns D^\prg})\ppa_{\ns D} x) .
$$
Here $a_{\ns D^\prg}$ is the copy of $a$ in the termlike $\sigma$-algebra $\ns D^\prg$ and $\prg_{\ns D}$ maps this to $|\ns D|$.

$\tall a$ is from Definition~\ref{defn.fresh.finite.limit}.
$\ppa$ is a right adjoint to application and is from Definition~\ref{defn.FOLeq.pp}.
\end{rmrk}

\begin{lemm}
\label{lemm.tlam.sigma}
If $b\#u$ then $(\tlam b.x)[a\sm u]=\tlam b.(x[a\sm u])$.
\end{lemm}
\begin{proof}
We unpack Notation~\ref{nttn.lambda}.
By assumption in Definition~\ref{defn.FOLeq} the $\sigma$-action is compatible (Definition~\ref{defn.fresh.continuous}) so $(\tall b.\prg b\ppa x)[a\sm u]=\tall b.((\prg b\ppa x)[a\sm u])$.
We use \rulefont{\sigma\ppa} from Figure~\ref{fig.app.compatible} (Definition~\ref{defn.FOLeq.pp}) and \rulefont{\sigma\#} from Figure~\ref{fig.nom.sigma} (since $b\#a$).
\end{proof}

We now derive $\beta$-reduction and $\eta$-expansion from the counit and unit axioms \rulefont{\app\epsilon} and \rulefont{\app\eta} respectively:
\begin{frametxt}
\begin{prop}
\label{prop.lambda.beta.eta}
Suppose $\ns D\in\indiapp$, $x\in|\ns D|$, $u\in|\ns D^\prg|$, and $a$ is an atom.
Then:
\begin{enumerate*}
\item
$(\tlam a.x)\app \prg u\leq x[a\sm u]$.
\item
If $a\#x$ then $x\leq \tlam a.(x\app \prg a)$ ($\prg u$ and $\prg a$ from Notation~\ref{nttn.impredicative.D}).
\end{enumerate*}
\end{prop}
\end{frametxt}
\begin{proof}
We consider each part in turn.
\begin{enumerate*}
\item
Unfolding Notation~\ref{nttn.lambda} we have $(\tlam a.x)\app \prg u=(\tall a.(\prg a\ppa x))\app \prg u$.
Renaming using Lemma~\ref{lemm.freshwedge.alpha} if necessary, assume $a\#u$ so that by Theorem~\ref{thrm.no.increase.of.supp} also $a\#\prg u$.
By \rulefont{\app\tall}
$(\tall a.(\prg a\ppa x))\app \prg u\leq \tall a.((\prg a\ppa x)\app \prg u)$.
By Lemma~\ref{lemm.fresh.glb.sub} $\tall a.((\prg a\ppa x)\app \prg u)\leq ((\prg a\ppa x)\app \prg u)[a\sm u]$.
By Lemma~\ref{lemm.sigma.ppa} and \rulefont{\sigma\app} and \rulefont{\sigma\#},\ $((\prg a\ppa x)\app \prg u)[a\sm u]\leq (\prg u\ppa (x[a\sm u]))\app \prg u$.
By \rulefont{\app\epsilon} $(\prg u\ppa (x[a\sm u]))\app \prg u\leq x[a\sm u]$.
\item
Suppose $a\#x$; note of Definition~\ref{defn.fresh.finite.limit} that $x$ is its own $a$-fresh limit; that is, $\tall a.x=x$.
Unfolding Notation~\ref{nttn.lambda} $\tlam a.(x\app \prg a)=\tall a.(\prg a\ppa (x\app \prg a))$.
By \rulefont{\app\eta} $x\leq \prg a\ppa (x\app \prg a)$, so by Lemma~\ref{lemm.tall.monotone} $x=\tall a.x\leq \tall a.(\prg a\ppa (x\app \prg a))$ as required.
\qedhere\end{enumerate*}
\end{proof}

\begin{rmrk}
$\ns D\in\indiapp$ gives a model of $\beta$-reduction and $\eta$-expansion.
The reverse inclusions do not follow, but they are not forbidden:
\begin{itemize*}
\item
There exist models such that $x[a\sm u]\not\leq(\lam{a}x)\app \prg u$ (so that we do not have $\beta$-equality) and $a\#x$ for some $x$ and yet $\lam{a}(x\app \prg a)\not\leq x$.
\item
There also exist models such that $x[a\sm u]=(\lam{a}x)\app \prg u$ and for all $x$,\ \, $a\#x$ implies $\lam{a}(x\app \prg a)= x$.
To obtain one, choose any $\lambda$-equality theory $\Pi$ (Definition~\ref{defn.lambda.reduction.theory}; $\beta\eta$-equality would do) and construct $\points_\Pi$ from Definition~\ref{defn.pi.point} and the subsequent constructions.
\end{itemize*}
\end{rmrk}

%%%%%%%%%%%%%%%%%%%%%%%%%%%%%%%
\subsubsection{$\lambda$ as a sets operation in $F(\ns D)$}

We take a moment to perform a sanity check by examining $\tlam a$ (Notation~\ref{nttn.lambda})
for the specific case of the sets representation $\ns T=F(\ns D)$ (Definition~\ref{defn.F.pp}) of $\ns D\in\indiapp$.

\begin{prop}
\label{prop.unpack.lambda.for.G}
Suppose $\ns D\in\indiapp$, $X\in\otop{F(\ns D)}$, and $u\in|F(\ns D)^\prg|{=}|\ns D^\prg|$, and let $p$ range over elements of $|F(\ns D)|$, which are prime filters in $\ns D$.
Then for every $u\in|\ns D^\prg|$,
$$
p\in\tlam a.X
\quad\text{implies}\quad p\bpp \prg u\subseteq X[a\sm u].
$$
\end{prop}
\begin{proof}
Recall the unpacking of $\tlam a$ from Remark~\ref{rmrk.unpack.tlam}.
We reason as follows:
$$
\begin{array}[b]{@{\hspace{-3.5em}}r@{\ }l@{\quad}l}
p\in\tlam a.X\liff&
p\in \tall a.(\prg{a} \ppa X)
&\text{Notation~\ref{nttn.lambda}}
\\
\liff&
\Forall{u{\in}|F(\ns D)^\prg|} p\in (\prg{a}\ppa X)[a\sm u]
&\text{Proposition~\ref{prop.char.freshwedge}(1)}
\\
\liff&
\Forall{u{\in}|\ns D^\prg|} p\in (\prg{a}\ppa X)[a\sm u]
&F(\ns D)^\prg{=}\ns D^\prg\ \makebox[0pt][l]{\text{by Thm~\ref{thrm.FB.totsep.comp}}}
\\
\limp&
\Forall{u{\in}|\ns D^\prg|} p\in (\prg{a})[a\sm u]\ppa X[a\sm u]
&\text{Lemma~\ref{lemm.sigma.ppa}} 
\\
\liff&
\Forall{u{\in}|\ns D^\prg|} p\in  \prg{u}\ppa X[a\sm u]
&\rulefont{\sigma a},\ \prg\text{ morphism \makebox[0pt][l]{(Def~\ref{defn.morphism.sigma.alg})}}
\\
\liff&
\Forall{u{\in}|\ns D^\prg|} p\bpp \prg{u}\subseteq X[a\sm u]
&\text{Proposition~\ref{prop.amgis.iff.pp}}
\end{array}
\qedhere$$
\end{proof}

\begin{rmrk}
Continuing the notation of Proposition~\ref{prop.unpack.lambda.for.G}, one might expect $p\in\tlam a.X$ to also be \emph{equivalent} to $\Forall{u{\in}|\ns D^\prg|}p\bpp\prg u\subseteq X[a\sm u]$.
This seems to not be the case, because Lemma~\ref{lemm.sigma.ppa} used in the proof above is an inequality and not an equality.
\end{rmrk}

%%%%%%%%%%%%%%%%%%%%%%%%%%%%%%%%%%
\subsection{Idioms}

It is convenient to generalise $\lambda$-syntax a little.
Recall from Definition~\ref{defn.term.sub.alg} that a \emph{termlike $\sigma$-algebra} expresses in nominal algebra the property of `having a substitution action over itself'.

\begin{defn}
\label{defn.idiom}
A \deffont{($\lambda$-)idiom} is a termlike $\sigma$-algebra $\idiom$ 
equipped with equivariant functions
$$
\begin{array}{r@{\ }l}
\app_\idiom:&\idiom\times\idiom\equivarto \idiom\quad\text{and}
\\
\lambda_\idiom:&\mathbb A\times\idiom\equivarto\idiom
\end{array}
$$
such that for all $a\in\mathbb A$ and $x,y\in|\idiom|$ and $u\in|\idiomprg|$:
\begin{enumerate*}
\item\label{item.idiom.fresh}
$a\#\lambda_\idiom a.x$ (this justifies quantifier notation: $\lambda$ abstracts the atoms argument $a$).
\item
If $b\#u$ then $(\lambda_\idiom b.x)[a\ssm u]_\idiom=\lambda_\idiom b.(x[a\ssm u]_\idiom)$.
\item\label{item.idiom.app}
$(x\app_\idiom y)[a\ssm u]_\idiom=(x[a\ssm u]_\idiom)\app_\idiom (y[a\ssm u]_\idiom)$.
\item
If $x,y\in|\idiomprg|$ then $x\app_\idiom y\in|\idiomprg|$, and $\lambda_\idiom a.x\in|\idiomprg|$.
\end{enumerate*}
Above, we use the fact that because $\idiom$ is a termlike $\sigma$-algebra, it interprets atoms as $\tf{atm}_\idiom(a)$ and as a $\sigma$-action $x[a\ssm u]_\idiom$.
\end{defn}

\begin{xmpl}
\label{xmpl.canonical.idiom}
The canonical example of a $\lambda$-idiom is the \deffont{syntactic idiom} $\lamtrm$; 
$\lambda$-terms up to $\alpha$-equivalence, with their natural substitution, application, and $\lambda$-actions.
\end{xmpl}

\begin{nttn}
\begin{itemize*}
\item
We write $a_\idiom$ for $\tf{atm}_\idiom(a)$, or just $a$.
\item
We write $x[a\ssm u]$ for $x[a\ssm u]_\idiom$.
\item
We write $xy$ for $x\app_\idiom y$.
\item
We write $\lam{a}x$ for $\lambda_\idiom a.x$.
\end{itemize*}
The conditions of Definition~\ref{defn.idiom} can now be written using Corollary~\ref{corr.stuff}, thus:
$$
\begin{array}{r@{\ }r@{\ }l}
b\#x\limp& \lam{a}x=&\lam{b}(b\ a)\act x
\\
b\#u\limp& (\lam{b}x)[a\ssm u]=&\lam{a}(x[a\ssm u])
\\
&(xy)[a\ssm u]=&x[a\ssm u]\,y[a\ssm u]
\end{array}
$$
Note that it follows already from axiom \rulefont{\sigma a} in Figure~\ref{fig.nom.sigma} that $a_\idiom[a\ssm u]=u$.
\end{nttn}

\begin{nttn}
In what follows, what the variables, substitution, application, and $\lambda$ of an idiom $\idiom$ have to be, will always be clear.
For the rest of this section fix some $\lambda$-idiom $\idiom$.
\end{nttn}

\begin{nttn}
We call elements of $|\idiom|$ \deffont{phrases}.
We let $s$ and $t$ range over phrases in $|\idiom|$, and also $u$ and $v$ range over phrases in $|\idiomprg|$.
\end{nttn}

\begin{rmrk}
Phrases of a $\lambda$-idiom `look like' terms of $\lambda$-syntax up to $\alpha$-equivalence, inasmuch as they must support variables, substitution, a binary operator which we suggestively call application, and a variable-abstractor which we suggestively call $\lambda$.
We noted in Example~\ref{xmpl.canonical.idiom} that $\idiom$ might be the syntactic idiom, which is precisely $\lambda$-syntax.

However, we do not insist that phrases \emph{be} $\lambda$-terms; they need not even be syntax.
They just have to support nominal algebraic models of variables, substitution, application and a $\lambda$-abstraction.
Nothing about the constructions that follow immediately below depends on $\idiom$ being syntactic.
\end{rmrk}

\begin{defn}
\label{defn.lambda.reduction.theory}
Suppose $\idiom$ is a $\lambda$-idiom.
Call a preorder $\somerel$ on phrases \deffont{compatible} when
for all $s,s',t,t'\in|\idiom|$ and $u\in|\idiomprg|$:\footnote{$\somerel$ being a preorder means precisely that it is transitive and reflexive.
(A partial order is an antisymmetric preorder.)}
\begin{enumerate*}
\item
\label{item.somerel.app}
If $s\somerel s'$ and $t\somerel t'$ then $st\somerel s't'$.
\item
\label{item.somerel.lam}
If $s\somerel s'$ then $\lam{a}s\somerel\lam{a}s'$.
\item
\label{item.somerel.ssm}
If $s\somerel s'$ then $s[a\ssm u]\somerel s'[a\ssm u]$.
\item
\label{item.somerel.beta}
$(\lam{a}s)u\somerel s[a\ssm u]$.
This is \deffont{$\beta$-reduction}.
\item
\label{item.somerel.eta}
If $a$ is not free in $s$ then $s\somerel \lam{a}(sa)$.
This is \deffont{$\eta$-expansion}.
\item
\label{item.somerel.equivar}
If $s\somerel s'$ then $\pi\act s\somerel\pi\act s'$, or following Definition~\ref{defn.equivariant}: $\somerel$ is \emph{equivariant}.\footnote{See the discussion in Subsection~\ref{subsect.cartesian.product}.
Another description of this condition is that $\mathcal R$ is an element of $\nompow(\lamtrm\times\lamtrm)$ (Subsection~\ref{subsect.finsupp.pow}) with $\supp(\mathcal R)=\varnothing$.}
\end{enumerate*}
Then we write:
\begin{itemize*}
\item
A \deffont{$\lambda$-reduction theory} is a compatible preorder on an idiom $\idiom$.
\item
A \deffont{$\lambda$-equality theory} is a compatible equivalence relation on $\idiom$.\footnote{An equivalence relation is a symmetric preorder, so a $\lambda$-equality theory is, as we expect, a symmetric $\lambda$-reduction theory.}
\end{itemize*}
$\Pi$ will range over $\lambda$-reduction theories.
\end{defn}

\begin{defn}
\label{defn.rew.equ}
Call a pair $s\to t$ a \deffont{($\lambda$-)reduction axiom}.
We let $\mathcal T$ range over sets of reduction axioms.
\begin{itemize*}
\item
Write $\f{rew}(\mathcal T)$ for the least $\lambda$-reduction theory $\Pi$ such that $\mathcal T\subseteq\Pi$.
\item
Write $\f{equ}(\mathcal T)$ for the least $\lambda$-equality theory $\Pi$ such that $\mathcal T\subseteq\Pi$.
\end{itemize*}
\end{defn}

\begin{nttn}
\label{nttn.Pi.cent}
Suppose $\mathcal T$ is a set of reduction axioms.
Then:
\begin{itemize*}
\item
We may write $s\arrowp{\mathcal T} t$ or $\mathcal T\cent  s{\to} t$ 
for $(s\to t)\in\f{rew}(\mathcal T)$.
\item
We may write $s\eqarrowp{\mathcal T} t$ or $\mathcal T\cent s=t$ 
for $(s\to t)\in\f{equ}(\mathcal T)$.
\end{itemize*}
\end{nttn}

We conclude with an easy technical lemma:
\begin{lemm}
\label{lemm.T.Pi}
Suppose $\mathcal T$ is a set of reduction axioms and write $\Pi{=}\f{rew}(\mathcal T)$.
Then:
\begin{itemize*}
\item
$\f{rew}(\mathcal T)=\f{rew}(\Pi)$.
\item
$\mathcal T\cent s{\to}t$ if and only if $\Pi\cent s{\to}t$.
\end{itemize*}
\end{lemm}
\begin{proof}
Direct from Definition~\ref{defn.rew.equ} and Notation~\ref{nttn.Pi.cent}.
\end{proof}

%%%%%%%%%%%%%%%%%%%%%%%%%%%%%%%
\subsection{A sound denotation for the $\lambda$-calculus}

\emph{Any} $\ns D\in\indiapp$ has the structure of $\tall$, $\app$, and $\ppa$, so we can immediately interpret the $\lambda$-calculus in $\ns D$.
Lo and behold, the interpretation is sound.
This is Definition~\ref{defn.ddenot} and Theorem~\ref{thrm.lambda.soundness}.

The denotation we obtain is \emph{absolute}, meaning that a variable/atom $a$ is interpreted `as itself'---there is no valuation.
Slightly more formally, a denotation is absolute when variable symbols in the syntax map to fixed entities in the denotation.
In the case of this paper, $a$ (more precisely: $a_\idiom$) is interpreted as $\prg a$ (more precisely: $\prg_{\ns D} a_{\ns D^\prg}$, see Notation~\ref{nttn.impredicative.D}).

The role of a valuation is played by the $\sigma$-action.
If we have some $x\in|\ns D|$ and want to `evaluate' any $a$ in it to become $u$, then we just apply $[a\sm u]$.
This nominal approach to valuations using $\sigma$-algebras is \emph{more general} than the usual Tarski denotation based on valuations; to see why, see the discussion in \cite[Remark~8.18]{gabbay:semooc}.

\begin{defn}
\label{defn.ddenot}
Suppose $\ns D\in\indiapp$.
Define a \deffont{denotation} of $\lambda$-terms by the rules in Figure~\ref{fig.ddenot} ($\tlam$ is from Notation~\ref{nttn.lambda}; $\app$ is from Definition~\ref{defn.sub.sets.pp}),\footnote{In the case for $\ddenot{a}$, $a_{\ns D^\prg}$ is the copy of $a$ in the termlike $\sigma$-algebra $\ns D^\prg$ (Definition~\ref{defn.term.sub.alg}) and $\prg_{\ns D}$ is the function
mapping $\ns D^\prg$ to $\ns D$ (see Definition~\ref{defn.D.impredicative} and Notation~\ref{nttn.impredicative.D}). We have written $\prg_{\ns D} a_{\ns D^\prg}$ as just $\prg a$, but here we prefer the more careful notation.

Of course, if we wanted to be \emph{really} careful we would also mention that $a_{\ns D^\prg}$ is itself shorthand for $\tf{atm}_{\ns D^\prg}(a)$ from Definition~\ref{defn.term.sub.alg}.
But the reader probably is not interested in that high level of pedantry, and may even be confused by it, so we will not labour the point further.
}
and:
\begin{itemize*}
\item
Write $\ns D\ment s\leq t$ when $\ddenot{s}\leq\ddenot{t}$.
\item
Write $\ns D\ment\mathcal T$ when $\ns D\ment s\leq t$ for every $(s\to t)\in \mathcal T$. 
\item
Write $\mathcal T\ment s\leq t$ when $\Forall{\ns D{\in}\indiapp}\bigl(\ns D{\ment}\mathcal T\limp \ns D{\ment} s\leq t\bigr)$.
\end{itemize*}
\end{defn}

\begin{figure}
$$
\begin{array}{r@{\ }l}
\ddenot{a}=&\prg_{\ns D} a_{\ns D^\prg}
\\
\ddenot{\lam{a}s}=&\tlam a.\ddenot{s}
\\
\ddenot{s's}=&\ddenot{s'}\app\ddenot{s}
\end{array}
$$
\caption{Denotation of $\lambda$-terms}
\label{fig.ddenot}
\end{figure}

\begin{rmrk}
Suppose $\ns D\in\indiapp$.
Recall from Definition~\ref{defn.FOLeq} that $\ns D^\prg$ is the termlike $\sigma$-algebra over which substitution in $\ns D$ is defined, and recall that (since $\ns D\in\indiapp$ is \emph{impredicative}; see Definition~\ref{defn.D.impredicative}) we assume a $\sigma$-algebra morphism $\prg_{\ns D}$ from $\ns D^\prg$ to $\ns D$.

Recall from Notation~\ref{nttn.impredicative.D} that we write $\prg\ns D$ for the sets image of $\prg_{\ns D}$, i.e. $\prg\ns D=\{\prg_{\ns D} u\mid u\in\ns D^\prg\}\subseteq|\ns D|$, and recall that we call this image the \deffont{programs} of $\ns D$.
\end{rmrk}

\begin{defn}
\label{defn.replete}
Call $\ns D\in\indiapp$ \deffont{replete} if $\prg\ns D$ is closed under application and $\lambda$.
That is:
\begin{itemize*}
\item
If $x,y\in\prg\ns D$ then $x\app y\in\prg\ns D$.
\item
If $x\in\prg\ns D$ then $\tlam a.x\in\prg\ns D$.
\end{itemize*}
\end{defn}

\begin{rmrk}
Note that $\prg a\in\prg\ns D$ is a fact, where $\prg a$ is shorthand for $\prg_{\ns D}a_{\ns D^\prg}$.
If $\ns D$ is replete then programs are closed under taking variables, application, or $\lambda$-abstraction, and intuitively this tells us the following:
\begin{quote}
If $\ns D$ is replete then its programs include denotations for all $\lambda$-terms.
\end{quote}
This intuition is exactly the notion of repleteness used in \cite{gabbay:simcks} (we called it \emph{faithful} there, but that terminology clashes with faithfulness of functors in category theory).
In this paper we are using nominal techniques, so we can give a name-based semantic treatment of $\lambda$, so that Definition~\ref{defn.replete} can be more abstract than it needed to be in \cite{gabbay:simcks}, and it needs make no explicit mention of $\lambda$-term syntax.
\end{rmrk}

\begin{rmrk}
Definition~\ref{defn.replete} is needed for Lemma~\ref{lemm.ddenot.sigma}.
In any case, we are most interested in $\ns D$ that are replete, since we are interested in models of the $\lambda$-calculus and we would expect $\lambda$-terms to denote programs.
\end{rmrk}

So if $\ns D$ is replete then Definition~\ref{defn.ddenot} generates programs, which can be substituted for in $\ns D$, and we can express Lemma~\ref{lemm.ddenot.sigma}:
\begin{lemm}
\label{lemm.ddenot.sigma}
Suppose $\ns D\in\indiapp$ is replete.
Then $\ddenot{s}[a\sm \ddenot{u}]=\ddenot{s[a\ssm u]}$.

($\ddenot{u}$ always exists, but repleteness ensures that $\ddenot{u}\in\prg\ns D$ so that the substitution $[a\sm\ddenot{u}]$ also exists.)
\end{lemm}
\begin{proof}
By induction on $s$.
\begin{itemize*}
\item
\emph{The case of $a$.}\quad
By Definition~\ref{defn.ddenot} $\ddenot{a}=\prg_{\ns D} a_{\ns D^\prg}$.
We use Lemma~\ref{lemm.impredicative.sigma.a}.
\item
\emph{The case of $\lam{b}s$.}\quad
Renaming if necessary assume $b\#u$.
We reason as follows:
$$
\begin{array}{r@{\ }l@{\qquad}l}
\ddenot{(\lam{b}s)[a\ssm u]}
=&
\ddenot{\lam{b}(s[a\ssm u])}
&\text{Fact of $\lambda$-terms, }b\#u
\\
=&
\tlam b.\ddenot{s[a\ssm u]}
&\text{Definition~\ref{defn.ddenot}}
\\
=&
\tlam b.(\ddenot{s}[a\sm \ddenot{u}])
&\text{ind. hyp.}
\\
=&
(\tlam b.\ddenot{s})[a\sm \ddenot{u}]
&\text{Lemma~\ref{lemm.tlam.sigma}}
\\
=&
\ddenot{\lam{b}s}[a\sm\ddenot{u}]
&\text{Definition~\ref{defn.ddenot}}
\end{array}
$$
In the use of Lemma~\ref{lemm.tlam.sigma} above we know $b\#\ddenot{u}$ by Theorem~\ref{thrm.no.increase.of.supp} since $b\#u$.
\item
\emph{The case of $s's$.}\quad
Routine using the inductive hypothesis and \rulefont{\sigma\app} from Figure~\ref{fig.app.compatible} (Definition~\ref{defn.FOLeq.pp}).
\qedhere\end{itemize*}
\end{proof}

Recall the notation 
$\mathcal T\cent s\to t$ from Notation~\ref{nttn.Pi.cent}, applied here to the idiom $\lamtrm$ ($\lambda$-terms).
Recall the notation 
$\mathcal T\ment s\leq t$ from Definition~\ref{defn.ddenot}.
\begin{thrm}[Soundness]
\label{thrm.lambda.soundness}
Suppose $\mathcal T$ is a set of reduction axioms.
Then
\begin{frameqn}
\mathcal T\cent s\to t
\quad\text{implies}\quad \mathcal T\ment s\leq t.
\end{frameqn}
\end{thrm}
(The reverse implication also holds; see Theorem~\ref{thrm.Pi.completeness}.)
\begin{proof}
Suppose $\ns D\in\indiapp$ is replete and suppose $\ns D\ment\mathcal T$ (Definition~\ref{defn.ddenot}).
We consider the rules defining a compatible relation on $\lambda$-terms (Definition~\ref{defn.lambda.reduction.theory}):
\begin{enumerate*}
\item
\emph{If $s\somerel s'$ and $u\somerel u'$ then $su\somerel s'u'$.}\quad
We use Lemma~\ref{lemm.app.leq}.
\item
\emph{If $s\somerel s'$ then $\lam{a}s\somerel\lam{a}s'$.}\quad
We use Lemmas~\ref{lemm.tall.monotone} and~\ref{lemm.app.leq}.
\item
\emph{If $s\somerel s'$ then $s[a\ssm u]\somerel s'[a\ssm u]$.}\quad
We use Lemmas~\ref{lemm.ddenot.sigma} and~\ref{lemm.sigma.monotone}.
\item
\emph{$(\lam{a}s)t\somerel s[a\ssm t]$.}\quad
We use Lemma~\ref{lemm.ddenot.sigma} and part~1 of Proposition~\ref{prop.lambda.beta.eta}.
\item
\emph{If $a$ is not free in $s$ then $s\somerel \lam{a}(sa)$.}\quad
If $a\#s$ then by Theorem~\ref{thrm.no.increase.of.supp} also $a\#\ddenot{s}$.
We use part~2 of Proposition~\ref{prop.lambda.beta.eta}.
\item
\emph{If $s\somerel s'$ then $\pi\act s\somerel \pi\act s'$.}\quad
From Theorem~\ref{thrm.equivar}.
\qedhere\end{enumerate*}
\end{proof}

%%%%%%%%%%%%%%%%%%%%%%%%%%%%%%%%%%%%%%%%
\subsection{Interlude: axiomatising the $\lambda$-calculus in nominal algebra}
\label{subsect.interlude.lambda}

Some words on where we are and where we are going.

Nominal algebra considers equality over nominal sets.\footnote{It is descended from nominal rewriting, which considers rewriting over nominal terms \cite{gabbay:nomr,gabbay:nomr-jv}.}
It was introduced in two papers \cite{gabbay:capasn,gabbay:oneaah} where it was applied to axiomatise to substitution and first-order logic respectively.\footnote{The papers wrote axioms and proved them sound and complete. So we really did check that the axioms do what one would expect them to do; no more and no less.}
Both applications feature $\alpha$-equivalence and freshness side-conditions, which are of course just what nominal sets were developed to model, so this was natural.

See \cite{gabbay:newaas-jv} or see \cite{gabbay:fountl,gabbay:nomtnl} for surveys.

In \cite{gabbay:lamcna,gabbay:nomalc} nominal algebra was applied to the $\lambda$-calculus, extending an incomplete axiomatisation from \cite{gabbay:nomr-jv}---Henkin style models of such axioms were considered in \cite{gabbay:nomhss} and found to have some interesting properties.
In particular the axiomatisation is sound and complete---so the axioms below \emph{really do} axiomatise the $\lambda$-calculus;
and this proof, in greatly strengthened form, has become the duality, soundness, and completeness results of the current paper.

So an axiomatisation of the $\lambda$-calculus is implicit in this paper.
The reader could extract it by tracing through Notation~\ref{nttn.lambda} and the axioms of $\indiapp$.
We do not have to write out this theory to prove soundness in this paper, because the notion of $\lambda$-calculus we use in this paper is the standard one based on $\lambda$-term syntax and reduction.

Yet the axiomatisation is there in the background, and
for the reader's convenience it might be illuminating to write it out.

Consider a nominal set $\ns X$.

We assume equivariant functions $\tf{atm}:\mathbb A\to\ns X$ and $\tf{sub}:\ns X\times\mathbb A\times\ns X\to\ns X$ and impose the axioms of a termlike $\sigma$-algebra from Figure~\ref{fig.nom.sigma}:
\begin{tab0}
\rulefont{\sigma a} && a[a\sm x]=&x
\\
\rulefont{\sigma id} && x[a \sm a]=&x
\\
\rulefont{\sigma\#} &a\#x\limp& x[a \sm u]=&x
\\
\rulefont{\sigma\alpha}&b\#x\limp&x[a \sm u]=&((b\;a)\act x)[b \sm u]
\\
\rulefont{\sigma\sigma} &a\#v\limp & x[a \sm u][b \sm v]=&x[b \sm v][a \sm u[b \sm v]]
\end{tab0}
Here we sugar $\tf{atm}(a)$ to just $a$ and $\tf{sub}(x,a,u)$ to just $x[a\sm u]$.

Next we assume equivariant functions $\tf{app}:\ns X\times\ns X\to \ns X$ and $\tf{lam}:\mathbb A\times\ns X\to\ns X$, and impose the axioms of $\beta$- and $\eta$-equality:
\begin{tab0}
\rulefont{\lambda\alpha} &b\#x\limp&\tlam a.x=&\tlam b.(b\ a)\act x
\\
\rulefont{\beta{=}} && (\tlam a.x)y =& x[a\sm y]
\\
\rulefont{\eta{=}} &a\#x\limp& \tlam a.(xa)=&x
\end{tab0}
Here we sugar $\tf{app}(x,y)$ to $xy$ and $\tf{lam}(a,x)$ to $\tlam a.x$.

A few notes on this axiomatisation:
\begin{itemize*}
\item
The axiomatisation of \cite{gabbay:lamcna,gabbay:nomalc} identified substitution with a $\beta$-reduct.
The axiomatisation above distinguishes substitution and $\beta$-reducts.
This turns out to be important for making the results in this paper work; for more discussion see the Conclusions.
\item
The body of this paper is based on lattices, so we do not \emph{assume} $\beta$- or $\eta$-equality; we only assume $\beta$-reduction and $\eta$-expansion.
(The equalities might happen to be valid anyway, see for instance $\lambda$-equality theories in Definition~\ref{defn.lambda.reduction.theory}.)
This is also important for this paper.
\end{itemize*}
In summary the axiomatisation above is a special case of a generalisation of \cite{gabbay:lamcna,gabbay:nomalc}, which is itself a complete extension of a rewrite theory from \cite{gabbay:nomr-jv,gabbay:capasn}.

The axiomatisation above is also what we are aiming for, and models of $\lambda$-equality theories constructed in Section~\ref{sect.lambda.representation} are models of the axioms above, though the demands of our main results are such that we do not phrase matters in that specific form.

%%%%%%%%%%%%%%%%%%%%%%%%%%%%%%%%%%%%%%%%
\section{Representation of the $\lambda$-calculus in $\inspectapp$}
\label{sect.lambda.representation}

In Section~\ref{sect.lambda.calculus} we showed how any $\ns D\in\indiapp$ / (dually) any $\ns T\in\inspectapp$, gives a sound abstract / (dually) concrete interpretation of the untyped $\lambda$-calculus.

The next step is to prove completeness.
This is Theorem~\ref{thrm.Pi.completeness}.
The method is to construct a nominal spectral space $\points_\Pi$ out of a $\lambda$-reduction theory $\Pi$, in which only those subset inclusions are valid that are insisted on by $\Pi$.

$\points_\Pi$ is a rich structure.
Notable technical definitions and results are Definition~\ref{defn.afreshp} and Proposition~\ref{prop.fresh.point} ($a\fresh p$ and its equivalence with $a\#p$), completeness under small-supported sets unions and intersections (Proposition~\ref{prop.points.bigcap}), the $\sigma$-action on points (Definition~\ref{defn.qasmu}) and its two characterisation in Subsection~\ref{subsect.two.chars.lsm}---one in terms of the now-ubiquitous $\new$.

For this section, fix the following data:
\begin{itemize*}
\item
Fix a $\lambda$-idiom $\idiom$ (Definition~\ref{defn.idiom}).
\item
Fix a $\lambda$-reduction theory $\Pi$ on $\idiom$ (Definition~\ref{defn.lambda.reduction.theory}).
\end{itemize*}
$s$, $s'$, $s''$, $t$, $u$, and $v$ will range over elements of $|\idiom|$.

%%%%%%%%%%%%%%%%%%%%%%%%%%%%
\subsection{$\Pi$-points and $\sigma$-freshness}
\label{subsect.Pi.points}

Given a subset $p\subseteq|\idiom|$, we can suggest two notions of `$a$ is fresh for $p$':
\begin{itemize*}
\item
One inherited from nominal techniques: $\New{b}(b\ a)\act p=p$.
We write this $a\#p$.
\item
One inherited from our syntactic intuitions: if $s\in p$ then $\Forall{u{\in}|\idiomprg|}s[a\ssm u]\in p$.
We will make this formal in Definition~\ref{defn.afreshp} and write it $a\fresh p$.
\end{itemize*}
A \emph{point} is then defined to be a set of phrases for which these two notions of freshness coincide.
This is Definition~\ref{defn.pi.point}.
The interesting part is condition~\ref{item.point.lam}, which is not obviously just ``$a\#p\liff a\fresh p$''---for this, see %the important and beautiful
Proposition~\ref{prop.fresh.point} and Lemma~\ref{lemm.fresh.point.check},
which work from surprisingly little in the way of assumptions.

We conclude with Proposition~\ref{prop.points.bigcap}, an important result asserting that $\points_\Pi$ is complete for small-supported diagrams (i.e. small-supported sets of points have an intersection that is also a point).
This has some useful consequences: for instance it makes possible the use of $\bigcap$ in Definitions~\ref{defn.some.ops}, and~\ref{defn.qasmu}, and also in Corollary~\ref{corr.lsm.char.1}.

\begin{defn}
\label{defn.afreshp}
Suppose $p\subseteq|\idiom|$.\footnote{$p$ need not be small-supported, but it will turn out that we are most interested in the case where it is.}
Define $a\fresh p$ by:
\begin{frameqn}
a\fresh p\quad\text{when}\quad \Forall{s{\in}|\idiom|}\Forall{u{\in}|\idiomprg|}(s\in p\limp s[a\ssm u]\in p)
\end{frameqn}
If $a\fresh p$ we say that $a$ is \deffont{$\sigma$-fresh} for $p$.
\end{defn}

\begin{rmrk}
\label{rmrk.rewrite.afreshp}
We rewrite Definition~\ref{defn.afreshp} twice:
\begin{enumerate*}
\item
$a\fresh p$ when
$\Forall{s{\in}|\idiom|}(s\in p\limp\Forall{u{\in}|\idiomprg|}s[a\ssm u]\in p)$.
\item
$a\fresh p$ when $\Forall{u{\in}|\idiomprg|}p\subseteq p[u\ms a]$.
\end{enumerate*}
$p[u\ms a]$ is from Definition~\ref{defn.p.action}; but see also Definition~\ref{defn.lam.perm.amgis} and Lemma~\ref{lemm.point.fresh.ms} below.
\end{rmrk}

Recall the notion of a \emph{$\lambda$-reduction theory} from Definition~\ref{defn.lambda.reduction.theory}.
\begin{frametxt}
\begin{defn}
\label{defn.pi.point}
Suppose $\Pi$ is a $\lambda$-reduction theory.
Call a subset $p\subseteq|\idiom|$ a \deffont{$\Pi$-point} when:
\begin{enumerate*}
\item
\label{item.point.to}
$\Forall{s,t{\in}|\idiom|}(s\in p\land s\arrowp{\Pi} t)\limp t\in p$.
We call $p$ \deffont{closed under $\Pi$}.
\item
\label{item.point.lam}
$\New{a} a\fresh p$.
We call $p$ \deffont{$\sigma$-supported}.
\end{enumerate*}
Write $|\points_\Pi|$ for the set of $\Pi$-points.
\end{defn}
\end{frametxt}

\begin{rmrk}
\label{rmrk.varnothing.is.a.point}
We will give $\points_\Pi$ an $\amgis$-algebra structure in Proposition~\ref{prop.pumsa.point} and Corollary~\ref{corr.points.Pi.amgis}.
We consider examples of points, and sets that are not points:
\begin{enumerate}
\item
$\varnothing$ (the empty set) is a point.

This will be useful in Proposition~\ref{prop.points.sigma.bpp} to prove that the set of all points is compact and covered by $\{\pp\varnothing\}$ (the $\pp{\text{-}}$ notation will be defined in Definition~\ref{defn.qasmu}).
\item
$\idiom$ is a point.
\item
Take $\idiom$ to be the syntactic idiom $\lamtrm$ from Example~\ref{xmpl.canonical.idiom} and take $\Pi$ to be $\f{rew}(\varnothing)$ from Definition~\ref{defn.rew.equ}---that is, the minimal $\lambda$-reduction theory (Definition~\ref{defn.lambda.reduction.theory}), containing just $\beta$-reduction and $\eta$-expansion.

Take $p$ to be  $\{a, \lam{b}ab, \lam{c}\lam{b}abc, \dots, b, \lam{b'}bb', \dots\}$
the set of all $\eta$-expansions of variable symbols.

This is not a point, because by Theorem~\ref{thrm.no.increase.of.supp} $a\#p$ and yet $a[a\ssm \lam{a}a]=\lam{a}a\not\in p$.
In the terminology of condition~\ref{item.point.lam} of Definition~\ref{defn.pi.point}, this $p$ is not $\sigma$-supported.
\end{enumerate}
\end{rmrk}

\begin{lemm}
\label{lemm.fresh.point.check}
Suppose $p\subseteq|\idiom|$ and suppose $p$ satisfies condition~\ref{item.point.lam} of Definition~\ref{defn.pi.point}.
Then:
\begin{enumerate*}
\item
$a\fresh p$ implies $\New{b}(b\ a)\act p=p$.
\item\label{point.finsupp}
As a corollary, if $p\in|\points_\Pi|$ then $p$ has small support.
\end{enumerate*}
\end{lemm}
\begin{proof}
Suppose $a\fresh p$ and consider any $b\fresh p$; by condition~\ref{item.point.lam} of Definition~\ref{defn.pi.point} there are cosmall many such $b$.
We will prove $(b\ a)\act p=p$.
The permutation action is pointwise (Definition~\ref{defn.pointwise.action}) so it suffices to show that for any $s\in|\idiom|$, $s\in p$ implies $(b\ a)\act s\in p$.

By Lemma~\ref{lemm.sm.to.pi} $(b\ a)\act s = s[a\ssm c][b\ssm a][c\ssm b]$ for fresh $c$ (so $c\#s$ and $c$ is distinct from $a$ and $b$, and by condition~\ref{item.point.lam} of Definition~\ref{defn.pi.point} $c\fresh p$).
Now $s\in p$ and $a\fresh p$ so $s[a\ssm c]\in p$.
Also $b\fresh p$ so $s[a\ssm c][b\ssm a]\in p$.
Also $c\fresh p$ so $s[a\ssm c][b\ssm a][c\ssm b]\in p$, and we are done.

For the corollary, from part~1 of this result $\New{a}\New{b}(b\ a)\act p=p$.
We use Lemma~\ref{lemm.old.school}.
\end{proof}

We discussed at the introduction to this subsection why Proposition~\ref{prop.fresh.point} is interesting:
\begin{prop}
\label{prop.fresh.point}
Suppose $p\in|\points_\Pi|$.
Then
$$
a\#p\quad\text{if and only if}\quad a\fresh p.
$$
\end{prop}
\begin{proof}
First, suppose $a\#p$.
By condition~\ref{item.point.lam} of Definition~\ref{defn.pi.point} $\New{b}b\fresh p$ and so by Theorem~\ref{thrm.new.equiv} (since by Lemma~\ref{lemm.fresh.point.check}(2) $p$ has small support)
we have $a\fresh p$.

Conversely, if $a\fresh p$ then by Lemma~\ref{lemm.fresh.point.check}(1) $\New{b}(b\ a)\act p=p$ and so using Corollary~\ref{corr.stuff}(\ref{stuff.freshness.criterion}) (since by Lemma~\ref{lemm.fresh.point.check}(2) $p$ has small support) $a\#p$.
\end{proof}

We conclude the subsection with Proposition~\ref{prop.points.bigcap}, a useful result with an attractive proof:
\begin{prop}
\label{prop.points.bigcap}
Suppose $\mathcal P\subseteq|\points_\Pi|$ is small-supported.
Then
$$
\bigcap \mathcal P\in\points_\Pi
\quad\text{and}\quad
\bigcup \mathcal P\in\points_\Pi .
$$
In words: small-supported intersections and unions of points, are points.
\end{prop}
\begin{proof}
We check the conditions of Definition~\ref{defn.pi.point} for $\bigcap\mathcal P\in\points_\Pi$.
\begin{enumerate}
\item
Condition~\ref{item.point.to} is by a routine calculation.
\item
For condition~\ref{item.point.lam}, suppose $a$ is fresh (so $a\#\mathcal P$, because we assumed $\mathcal P$ has small support).
To prove $a\fresh\bigcap\mathcal P$ we must show $s\in\bigcap\mathcal P$ implies $\Forall{u{\in}|\idiomprg|}s[a\ssm u]{\in}\bigcap\mathcal P$.

Consider $s\in\bigcap\mathcal P$ and any $p\in\mathcal P$, so $s\in p$.
We want to prove $\Forall{u{\in}|\idiomprg|}s[a\ssm u]{\in}p$.
We cannot do this directly since we do not necessarily know that $a\#p$.\footnote{We assumed $\mathcal P\subseteq|\points_\Pi|$ is small-supported, not \emph{strictly} small-supported.
See Lemma~\ref{lemm.pow.not.true}.
}

So choose fresh $c$ (so $c\#\mathcal P,p,s$).
Since $p\in\mathcal P$ also $(c\ a)\act p\in (c\ a)\act\mathcal P\stackrel{\text{Cor~\ref{corr.stuff}}}{=}\mathcal P$.
So $s\in (c\ a)\act p$ and therefore $(c\ a)\act s\in p$.
By assumption $c\#p$ so by Proposition~\ref{prop.fresh.point} $c\fresh p$ so $\Forall{u}((c\ a)\act s)[c\ssm u]\in p$.
We $\alpha$-convert, and conclude that $\Forall{u{\in}|\idiomprg|}s[a\ssm u]\in p$ as required.
%mjg beautiful argument (and important)

The reasoning for $\bigcup\mathcal P\in\points_\Pi$ is almost identical.
\qedhere\end{enumerate}
\end{proof}

%%%%%%%%%%%%%%%%%%%%%%%%%%%%
\subsection{Constructing $\Pi$-points, and their amgis-algebra structure}

Recall the notion of idiom $\idiom$ from Definition~\ref{defn.idiom}, the notion of $\lambda$-reduction theory $\Pi$ from Definition~\ref{defn.lambda.reduction.theory}, and the notion of point $p$ from Definition~\ref{defn.pi.point}.

\begin{defn}
\label{defn.suparrpi}
Suppose $s\in|\idiom|$.
Define the \deffont{($\Pi$-)principal filter}
$s\uparrowp{\Pi}$ by
$$
s\uparrowp{\Pi}=\{s'\in|\idiom|\mid s\arrowp{\Pi} s'\} .
$$
\end{defn}

\begin{lemm}
\label{lemm.suparrowp.point}
If $s\in|\idiom|$ then $s\uparrowp{\Pi}$ is a point.
\end{lemm}
\begin{proof}
We check the conditions of Definition~\ref{defn.pi.point}.
\begin{enumerate}
\item
By transitivity if $s''\in s\uparrowp{\Pi}$, meaning $s\arrowp{\Pi}s''$, and $s''\arrowp{\Pi}s'$ then $s\arrowp{\Pi}s'$.
\item
Suppose $a\#s$ and suppose $s'\in s\uparrowp{\Pi}$.
By condition~\ref{item.somerel.ssm} of Definition~\ref{defn.lambda.reduction.theory} also $s[a\ssm u]\arrowp{\Pi} s'[a\ssm u]$ and by \rulefont{\sigma\#} $s[a\ssm u]=s$.
It follows that $s'[a\ssm u]\in s\uparrowp{\Pi}$ for every $u$.
Thus $s\uparrowp{\Pi}$ is $\sigma$-supported.
\qedhere\end{enumerate}
\end{proof}

\begin{defn}
\label{defn.lam.perm.amgis}
Give $p\subseteq|\idiom|$ a permutation action and an $\amgis$-action
following Definitions~\ref{defn.p.action} and~\ref{defn.F}:
\begin{frameqn}
\pi\act p=\{\pi\act r\mid r\in p\}
\qquad
p[u\ms a]=\{s\mid s[a\ssm u]\in p\}
\quad (u\in|\idiomprg|)
\end{frameqn}
Write $\points_\Pi$ for (what will prove will be) the $\amgis$-algebra with underlying set $|\points_\Pi|$ and the permutation and $\amgis$-actions defined above.
\end{defn}

Our notation suggests that $[u\ms a]$ is an $\amgis$-action.
This is true, but we must prove it: this is Proposition~\ref{prop.pumsa.point} and Corollary~\ref{corr.points.Pi.amgis}.

\begin{prop}
\label{prop.pumsa.point}
Suppose $u\in|\idiomprg|$.
Then $p\in|\points_\Pi|$ implies $p[u\ms a]\in|\points_\Pi|$.

In words: if $p$ is a $\Pi$-point then so is $p[u\ms a]$.
\end{prop}
\begin{proof}
We check the conditions of Definition~\ref{defn.pi.point}, freely using Proposition~\ref{prop.sigma.iff}:
\begin{enumerate*}
\item
\emph{Suppose $s[a\ssm u]\in p$ and $s\arrowp{\Pi} s'$.}

By condition~\ref{item.somerel.ssm} of Definition~\ref{defn.lambda.reduction.theory} $s[a\ssm u]\arrowp{\Pi} s'[a\ssm u]$.
By assumption $p$ is closed under $\Pi$, and so $s'[a\ssm u]\in p$.
\item
\emph{Suppose $b$ is fresh (so $b\#p,a,u$) and consider $s[a\ssm u]\in p$.}

By assumption $p$ is $\sigma$-supported and $b\#p$ so by Proposition~\ref{prop.fresh.point} $b\fresh p$.
Therefore
$$
\Forall{v}s[a\ssm u][b\ssm v]\in p
$$
%mjg beautiful argument here
Now $b\#u$, so by \rulefont{\sigma\sigma} for any $v'\in|\idiomprg|$,\ $s[b\ssm v'][a\ssm u]=s[a\ssm u][b\ssm v'[a\ssm u]]$.
It follows (taking all $v$ of the form $v'[a\ssm u]$ above) that $\Forall{v}s[b\ssm v][a\ssm u]\in p$, and so $\Forall{v}s[b\ssm v]\in p[u\ms a]$ as required.
\qedhere\end{enumerate*}
\end{proof}

Recall $\points_\Pi$ from Definition~\ref{defn.pi.point}:
\begin{corr}
\label{corr.points.Pi.amgis}
$\points_\Pi$ is indeed an $\amgis$-algebra.
\end{corr}
\begin{proof}
Proposition~\ref{prop.pumsa.point} proves that $[u\ms a]$ maps points to points.
We use Proposition~\ref{prop.amgis.2}.
\end{proof}

We conclude with a technical result which will be useful for Lemma~\ref{lemm.lam.point.fresh}:
\begin{lemm}
\label{lemm.point.fresh.ms}
If $p\in|\points_\Pi|$ and $a\#p$ and $u\in|\idiomprg|$ then $p\subseteq p[u\ms a]$.
\end{lemm}
\begin{proof}
Suppose $a\#p$.
By Proposition~\ref{prop.fresh.point} $a\fresh p$, so by Definition~\ref{defn.afreshp}
$s\in p$ implies $s[a\ssm u]\in p$.
The result follows by Proposition~\ref{prop.sigma.iff}.
\end{proof}

%%%%%%%%%%%%%%%%%%%%%%%%%%%%%%%%%%%%%%%%%%%%%%%%%%%%%%%%
\subsection{The left adjoint $p[a\lsm u]$ to the amgis-action $p[u\ms a]$}

Recall from the start of the Section that we fixed some $\lambda$-reduction theory $\Pi$.

By Proposition~\ref{prop.points.bigcap} a small-supported intersection of points is a point.
This suggests that we could build a left adjoint to $[u\ms a]$ on points by taking a suitable intersection.
We do this in Definition~\ref{defn.qasmu}.

This left adjoint turns out to be very well-behaved.
It has interesting characterisations (Subsection~\ref{subsect.two.chars.lsm}) which give us strong proof-methods for reasoning on it.
Furthermore it is a $\sigma$-action; this is Proposition~\ref{prop.points.Pi.sigma.algebra}.
So $\points_\Pi$ is both an $\amgis$-algebra and a $\sigma$-algebra.\footnote{This may well be a special case of a general result deserving its own paper.  Here, we are simply grateful and press on.}
We shall see in Definition~\ref{defn.tall.p} how this $\sigma$-action enables us to interpret $\tall$ on points.

Even better, the $\sigma$-action commutes with $\pp{\text{-}}$ from Definition~\ref{defn.some.more.ops}, which is key to how points are used to generate compact sets; this is part~1 of Theorem~\ref{thrm.lam.pp.sm}.
Thus we can study the behaviour of substitution on open sets by understanding the behaviour of the left adjoint to the $\amgis$-algebra action of points.

In short, most of the rest of this section depends on Definition~\ref{defn.qasmu} and the results that follow it in this subsection.

\subsubsection{Basic definition}

Recall from Definition~\ref{defn.lam.perm.amgis} the $\amgis$-action on points $p[u\ms a]$.
We can build a left adjoint for it:
\begin{defn}
\label{defn.qasmu}
Given $p\subseteq|\idiom|$ with small support and $u\in|\idiomprg|$, define $p[a\lsm u]$ by:
\begin{frameqn}
p[a\lsm u]=\bigcap\{q\in|\points_\Pi| \mid \New{c}((c\ a)\act p\subseteq q[u\ms c]) \}
\end{frameqn}
\end{defn}
By Proposition~\ref{prop.points.bigcap}, Definition~\ref{defn.qasmu} does indeed define a point.

\begin{rmrk}
Definition~\ref{defn.qasmu} looks like a repeat of Definition~\ref{defn.sub.sets},
but they are not quite the same because $p$ and $q$ above have the same type (subsets of $|\idiom|$) whereas in Definition~\ref{defn.sub.sets} $p$ and $X$ have different types (a point and a set of points, respectively).
\end{rmrk}

Points here have small support, so the proof of Proposition~\ref{prop.qasmu.iff} is (almost) a replay of the proof of part~2 of Proposition~\ref{prop.amgis.iff} (only for subset inclusion instead of sets membership).
\begin{prop}
\label{prop.qasmu.iff}
If $a\#u,q$ then $p\subseteq q[u\ms a]$ if and only if $p[a\lsm u]\subseteq q$.
\end{prop}
\begin{proof}
Suppose $a\#u,q$.
From Definition~\ref{defn.qasmu}, $p[a\lsm u]\subseteq q$ if and only if
$\New{c}(c\ a)\act p\subseteq q[u\ms c]$.
By Corollary~\ref{corr.stuff} $(c\ a)\act u=u$ and $(c\ a)\act q=q$, so (applying $(c\ a)$ to both sides of the subset inclusion) this is
if and only if
$\New{c}p\subseteq q[u\ms a]$, that is: $p\subseteq q[u\ms a]$.
\end{proof}

\noindent The interested reader can also find two corollaries of Proposition~\ref{prop.qasmu.iff} in Subsection~\ref{subsect.additional.lemmas}.

\begin{lemm}
\label{lemm.lam.point.alpha}
\begin{enumerate*}
\item
If $b\#p$ then $p[a\lsm u]=((b\ a)\act p)[b\lsm u]$.
\item
As a corollary, if $a\#u$ then $a\#p[a\lsm u]$.
\end{enumerate*}
\end{lemm}
\begin{proof}
Suppose $c$ is fresh (so $c\#p$).
By Theorem~\ref{thrm.new.equiv} it suffices to show that $(c\ a)\act p\subseteq q[u\ms c]$ if and only if $(c\ b)\act((b\ a)\act p)\subseteq q[u\ms c]$.
It would suffice to prove that $(c\ a)\act p=(c\ b)\act((b\ a)\act p)$.
This follows from Corollary~\ref{corr.stuff} and our assumption that $b\#p$.

The corollary follows using Corollary~\ref{corr.stuff}.
\end{proof}

%%%%%%%%%%%%%%%%%%%%%%%%%%
\subsubsection{Two characterisations of $p[a\lsm u]$}
\label{subsect.two.chars.lsm}

\begin{defn}
\label{defn.p.ssm}
If $p\subseteq|\idiom|$ is small-supported and $u\in|\idiomprg|$ define
\begin{frameqn}
p[a\ssm u]=\bigcup\{s[a\ssm u]\uparrowp{\Pi}\mid s\in p\} .
\end{frameqn}
\end{defn}
By Proposition~\ref{prop.points.bigcap} (since $p$ is small-supported) and Lemma~\ref{lemm.suparrowp.point}, $p[a\ssm u]$ is a point.

For the rest of this subsection, we assume $p,q\in|\points_\Pi|$.

\begin{lemm}
\label{lemm.lsm.ssm.iff}
If $a\#u,q$ then $p[a\lsm u]\subseteq q$ if and only if $p[a\ssm u]\subseteq q$.

As a corollary, if $a\#u$ then $p[a\ssm u]\subseteq p[a\lsm u]$.
\end{lemm}
\begin{proof}
By Proposition~\ref{prop.qasmu.iff} (since $a\#u,q$) $p[a\lsm u]\in q$ if and only if $p\subseteq q[u\ms a]$.
From Proposition~\ref{prop.sigma.iff} and condition~\ref{item.point.to} of Definition~\ref{defn.pi.point} this happens if and only if $p[a\ssm u]\subseteq q$.

The corollary follows since $p[a\lsm u]\subseteq p[a\lsm u]$ and by Lemma~\ref{lemm.lam.point.alpha}(2) $a\#p[a\lsm u]$.
\end{proof}

\begin{corr}
\label{corr.lsm.char.1}
If $a\#u$ then
\begin{equation}
p[a\lsm u]=\bigcap \{q\mid a\#q\land p[a\ssm u]{\subseteq} q\}.
\tag{\em Characterisation 1}
\end{equation}
\end{corr}
\begin{proof}
If $a\#q$ and $p[a\ssm u]\subseteq q$ then by Lemma~\ref{lemm.lsm.ssm.iff} also $p[a\lsm u]\subseteq q$.
Therefore
$$
p[a\lsm u]\subseteq\bigcap \{q\mid a\#q\land p[a\ssm u] \subseteq q\}.
$$
Furthermore by Lemma~\ref{lemm.lam.point.alpha}(2) $a\#p[a\lsm u]$ and by Lemma~\ref{lemm.lsm.ssm.iff} $p[a\ssm u]\subseteq p[a\lsm u]$.
Therefore
$$
\bigcap \{q\mid a\#q\land p[a\ssm u]\subseteq q\} \subseteq p[a\lsm u].
\qedhere$$
\end{proof}

Recall from Definition~\ref{defn.nua} the notion of $\nw a.p$ (the $\new$-quantifier, for sets).
\begin{lemm}
\label{lemm.nw.on.points}
If $p\in|\points_\Pi|$ then also $\nw a.p\in|\points_\Pi|$.
\end{lemm}
\begin{proof}
We check the conditions of Definition~\ref{defn.pi.point}:
\begin{enumerate*}
\item
Suppose $s\in\nw a.p$ and $s\arrowp{\Pi}t$.
Then $\New{b}(b\ a)\act s\in p$.
By condition~\ref{item.somerel.equivar} of Definition~\ref{defn.lambda.reduction.theory} (equivariance)
$(b\ a)\act s\arrowp{\Pi}(b\ a)\act t$ (for any $b$).
It follows by condition~\ref{item.point.to} of Definition~\ref{defn.pi.point} that $\New{b}(b\ a)\act t\in p$, so that $b\in\nw a.p$.
\item
Consider some fresh $c$ (so $c\#p$) and consider $s\in \nw a.p$ and some fresh $b$ (so $b\#p,s$), so that $(b\ a)\act s\in p$.
By condition~\ref{item.point.lam} of Definition~\ref{defn.pi.point},\ $c\fresh p$, so that $\Forall{u}((b\ a)\act s)[c\ssm u]\in p$.
It follows that $\Forall{u}(b\ a)\act (s[c\ssm u])\in p$, and therefore $\Forall{u}s[c\ssm u]\in \nw a.p$.
\qedhere\end{enumerate*}
\end{proof}

Lemma~\ref{lemm.lsm.id} gives a striking characterisation connecting the left adjoint to the $\amgis$-action, the pointwise substitution action, and the $\new$-quantifier for sets (Definitions~\ref{defn.qasmu}, \ref{defn.p.ssm}, and~\ref{defn.nua}):
\begin{lemm}
\label{lemm.lsm.id}
If $a\#u$ then $p[a\ssm u]\subseteq\nw a.(p[a\ssm u])$.
As a corollary, if $a\#u$ then
\begin{equation}
p[a\lsm u]=\nw a.(p[a\ssm u]).
\tag{\em Characterisation 2}
\end{equation}
\end{lemm}
\begin{proof}
Suppose $s\in p[a\ssm u]$.
This means there is $s'{\in}p$ such that $s'[a\ssm u]\arrowp{\Pi}s$.
By Lemma~\ref{lemm.fresh.sub} $a\#s'[a\ssm u]$ and it follows for any fresh $c$ that $s'[a\ssm u]\arrowp{\Pi}(c\ a)\act s$, so that $(c\ a)\act s\in p[a\ssm u]$.
Thus, $s\in \nw a.(p[a\ssm u])$.

The corollary follows from Lemma~\ref{lemm.supp.nua} and Corollary~\ref{corr.lsm.char.1}.
\end{proof}

%%%%%%%%%%%%%%%%%%%%%%%%%%%%
\subsubsection{Additional lemmas about the $\sigma$-action as an adjoint}
\label{subsect.additional.lemmas}

Lemmas~\ref{lemm.lsm.right.adjoint} and~\ref{lemm.catchya} describe a unit and counit style interaction between $[a\lsm u]$ and $[u\ms a]$ acting on points.
We will not use these lemmas later---we will use the result they come from, Proposition~\ref{prop.qasmu.iff}, directly instead.

The lemmas are still worth looking at, because they are subject to freshness side-conditions and so are not quite exactly what one might assume.
They and the related results in this Section suggest a theory of `nominal adjoints', in the spirit of the theories of nominal abstract syntax, unification, rewriting, and algebra which we have already seen \cite{gabbay:newaas-jv,gabbay:nomu-jv,gabbay:nomr-jv,gabbay:nomuae}.

\begin{lemm}
\label{lemm.lsm.right.adjoint}
If $a\#u$ then $p\subseteq p[a\lsm u][u\ms a]$.
\end{lemm}
\begin{proof}
By Lemma~\ref{lemm.lam.point.alpha}(2) (since $a\#u$) $a\#p[a\lsm u]$.
It is a fact that $p[a\lsm u]\subseteq p[a\lsm u]$, therefore by Proposition~\ref{prop.qasmu.iff} (since $a\#u,p[a\lsm u]$) $p\subseteq p[a\lsm u][u\ms a]$.
\end{proof}

\begin{lemm}
\label{lemm.catchya}
\begin{enumerate*}
\item
If $a\#u,p$ then $p[u\ms a][a\lsm u]\subseteq p$.
\item
If $a\#u,p$ then $p\subseteq p[u\ms a][a\lsm u]$.
\item
As a corollary, if $a\#u,p$ then $p[u\ms a][a\lsm u]=p$.
\end{enumerate*}
\end{lemm}
\begin{proof}
\begin{enumerate}
\item
It is a fact that $p[u\ms a]\subseteq p[u\ms a]$.
We use Proposition~\ref{prop.qasmu.iff} (since $a\#u,p$).
\item
Suppose $s\in p$.
Using Lemma~\ref{lemm.lam.point.alpha}(1) to rename if necessary, we may assume $a\#s$.
So by \rulefont{\sigma\#} from Figure~\ref{fig.nom.sigma} $s[a\ssm u]=s$, so that $s\in p[u\ms a]$ and therefore $s\in p[u\ms a][a\ssm u]\stackrel{\text{L~\ref{lemm.lsm.ssm.iff}}}{\subseteq} p[u\ms a][a\lsm u]$.
\item
From parts~1 and~2 of this result.
\qedhere\end{enumerate}
\end{proof}

%%%%%%%%%%%%%%%%%%%%%%%%%%%%%
\subsection{The left adjoint $p[a\lsm u]$ as a $\sigma$-action on points}

%%%%%%%%%%%%%%%%%%%%%%%%%%%%
\subsubsection{It is indeed a $\sigma$-action}

We prove Proposition~\ref{prop.points.Pi.sigma.algebra}, that $p[a\lsm u]$ is indeed a $\sigma$-action on points.
Fix $p\in|\points_\Pi|$ and $u\in|\idiomprg|$.

\begin{lemm}
\label{lemm.lam.point.fresh}
If $a\#p$ then $p[a\lsm u]=p$.
\end{lemm}
\begin{proof}
Using Lemma~\ref{lemm.lam.point.alpha}(1) and Corollary~\ref{corr.stuff} assume without loss of generality that $a\#u$ as well as $a\#p$.

By Lemma~\ref{lemm.point.fresh.ms} (since $a\#p$) $p\subseteq p[u\ms a]$ so by Proposition~\ref{prop.qasmu.iff} (since $a\#u,p$) $p[a\lsm u]\subseteq p$.

Now suppose $s\not\in p[a\lsm u]$.
We will show that $s\not\in p$.

Unpacking Definition~\ref{defn.qasmu},\ $s\not\in p[a\lsm u]$ implies that there exists $q\in|\points_\Pi|$ such that $s\not\in q$ and for fresh $c$ (so $c\#p,q,u,s$) $(c\ a)\act p\subseteq q[u\ms c]$.
Now by \rulefont{\sigma\#} $s[c\ssm u]=s$.
Thus if $s\not\in q$ then by Proposition~\ref{prop.amgis.iff} also $s\not\in q[u\ms c]$.
Therefore $s\not\in (c\ a)\act p$.
By Corollary~\ref{corr.stuff} since $a\#p$ and $c\#p$ also $(c\ a)\act p=p$, so $s\not\in p$ as required.
%mjg beautiful argument here
\end{proof}

\begin{lemm}
\label{lemm.p.sigma.pi}
If $b\#p$ then $p[a\lsm b]=(b\ a)\act p$.
\end{lemm}
\begin{proof}
By Lemma~\ref{lemm.lsm.id} $s\in p[a\lsm b]$
if and only if
$\New{c}(c\ a)\act s\in p[a\ssm b]$, and by Definitions~\ref{defn.suparrpi} and~\ref{defn.p.ssm}
this is if and only if
$\New{c}\Exists{s'{\in}p}s'[a\ssm b]\arrowp{\Pi}(c\ a)\act s$.

Now by assumption $b\#p$ so if $b\in\supp(s')$ then by condition~\ref{item.point.lam} of Definition~\ref{defn.pi.point} also $s'[b\ssm a]\in p$.
So we may assume without loss of generality of the `$\exists s'{\in}p$' above that the $s'$ chosen satisfies $b\#s'$, so that $s'[a\ssm b]=(b\ a)\act s'$.

Thus this is if and only if
$\New{c}\Exists{s'{\in}p}(b\ a)\act s'\arrowp{\Pi}(c\ a)\act s$.
Rearranging the permutations, this is if and only if
$\New{c}\Exists{s'{\in}p}(b\ a)\act ((c\ b)\act s')\arrowp{\Pi}s$.

Again, since $c,b\#p$, by Corollary~\ref{corr.stuff} $(c\ b)\act p=p$ so that $s'{\in}p$ if and only if $(c\ b)\act s'{\in}p$.

Thus this is if and only if
$\New{c}\Exists{s'{\in}p}(b\ a)\act s'\arrowp{\Pi}s$,
and by condition~\ref{item.point.to} of Definition~\ref{defn.pi.point} this is if and only if
$s\in(b\ a)\act p$.
\end{proof}

\begin{rmrk}
Lemma~\ref{lemm.p.sigma.pi} is remarkable.
There is no reason to expect that $p[a\lsm b]=(b\ a)\act p$ should hold---for contrast, we needed to impose this as condition~\ref{item.alpha.powsigma} when we constructed $\powsigma(\ns P)$ in Definition~\ref{defn.powsigma}.
Here, it works without requiring conditions.
\end{rmrk}

Corollary~\ref{corr.p.sigma.id} is a repeat of Corollary~\ref{corr.amgis.id.sub}, but for points.
We will use it in Proposition~\ref{prop.points.Pi.sigma.algebra}:
\begin{corr}
\label{corr.p.sigma.id}
$p[a\lsm a]=p$.
\end{corr}
\begin{proof}
Choose any $b\#p$.
By Lemma~\ref{lemm.lam.point.alpha}(1) $p[a\lsm a]=((b\ a)\act p)[b\lsm a]$.
By Proposition~\ref{prop.pi.supp} $a\#(b\ a)\act p$, and by Lemma~\ref{lemm.p.sigma.pi} $((b\ a)\act p)[b\lsm a]=(b\ a)\act((b\ a)\act p)=p$.
\end{proof}

\begin{lemm}
\label{lemm.p.sigma.sigma}
If $a\#v$ then $p[a\lsm u][b\lsm v]=p[b\lsm v][a\lsm u[a\sm v]]$.
\end{lemm}
\begin{proof}
Using Lemma~\ref{lemm.lam.point.alpha}(1) assume without loss of generality that $a\#u$ and $b\#v$.
Then the result follows using Proposition~\ref{prop.qasmu.iff} and Corollary~\ref{corr.points.Pi.amgis}, from \rulefont{\amgis\sigma}.
\end{proof}

Proposition~\ref{prop.points.Pi.sigma.algebra} does not hold in the general case of $F(\ns D)$ from Definition~\ref{defn.F}, but it holds specifically for $\points_\Pi$.
The underlying reason this happens is Proposition~\ref{prop.points.bigcap}, which allows us to build (small-supported) intersections of points and so construct $p[a\lsm u]$ in Definition~\ref{defn.qasmu}:
\begin{prop}
\label{prop.points.Pi.sigma.algebra}
$\points_\Pi$ with the action $[a\lsm u]$ from Definition~\ref{defn.qasmu} is indeed a $\sigma$-algebra.
\end{prop}
\begin{proof}
The interesting part is to check the axioms of Figure~\ref{fig.nom.sigma}:
\begin{itemize*}
\item
\rulefont{\sigma id} is Corollary~\ref{corr.p.sigma.id}.
\item
\rulefont{\sigma\#} is Lemma~\ref{lemm.lam.point.fresh}.
\item
\rulefont{\sigma\alpha} is Lemma~\ref{lemm.lam.point.alpha}(1).
\item
\rulefont{\sigma\sigma} is Lemma~\ref{lemm.p.sigma.sigma}. 
\qedhere\end{itemize*}
\end{proof}

%%%%%%%%%%%%%%%%%%%%%%%%%%%%
\subsubsection{The $\sigma$-action distributes over union and subset}

\begin{lemm}
\label{lemm.lsm.P}
Suppose $\mathcal P\subseteq|\points_\Pi|$ is strictly small-supported (Subsection~\ref{subsect.strict.pow}).
Then
$$
(\bigcup\mathcal P)[a\lsm u]= \bigcup\{p[a\lsm u] \mid p\in\mathcal P\} .
$$
By Lemma~\ref{lemm.finite.strict} this holds in particular if $\mathcal P$ is finite.
\end{lemm}
\begin{proof}
We use Lemma~\ref{lemm.lam.point.alpha}(1) to assume without loss of generality that $a\#u$.
Take any $r\in|\points_\Pi|$ such that $a\#r$.
We reason as follows:
$$
\begin{array}{r@{\ }l@{\qquad}l}
(\bigcup\mathcal P)[a\lsm u]\subseteq r
\liff&
\bigcup \mathcal P\subseteq r[u\ms a]
&\text{Proposition~\ref{prop.qasmu.iff}}\ a\#u,r
\\
\liff&
\Forall{p{\in}\mathcal P} p\subseteq r[u\ms a]
&\text{Fact of sets}
\\
\liff&
\Forall{p{\in}\mathcal P} p[a\lsm u]\subseteq r
&\text{Proposition~\ref{prop.qasmu.iff}}\ a\#u,r
\\
\liff&
\bigcup\{p[a\lsm u]\mid p\in\mathcal P\}\subseteq r
&\text{Fact of sets}
\end{array}
$$
By Lemma~\ref{lemm.lam.point.alpha}(2) $a\#(\bigcup\mathcal P)[a\lsm u]$.
Also, $a\#p[a\lsm u]$ for every $p\in\mathcal P$ so that by Lemma~\ref{lemm.strict.support}(\ref{strict.union}) $a\#\bigcup\{p[a\lsm u]\mid p\in\mathcal P\}$.
Taking first $r=(\bigcup\mathcal P)[a\lsm u]$ and then $r=\bigcup\{p[a\lsm u]\mid p\in\mathcal P\}$ we obtain two subset inclusions and thus an equality.
\end{proof}

\begin{rmrk}
\label{rmrk.no.sigma.distrib.points}
Lemma~\ref{lemm.lsm.P} does not contain a second part proving distributivity of substitution over sets intersection of points.
This is why, as we will note in Remark~\ref{rmrk.odd.ops}, we do not consider an operation $p\tor q=p\cap q$ in Definition~\ref{defn.some.ops}; it would \emph{not} satisfy $(p\cap q)[a\lsm u]=p[a\lsm u]\cap q[a\lsm u]$ or equivalently $(p\tor q)[a\lsm u]=p[a\lsm u]\tor q[a\lsm u]$.
To get this kind of property we need the topologies, developed below.
See in particular Corollary~\ref{corr.sm.lsm.cup}.
\end{rmrk}

\begin{lemm}
\label{lemm.lsm.monotone}
If $p\subseteq q$ then $p[a\lsm u]\subseteq q[a\lsm u]$.
\end{lemm}
\begin{proof}
Using Lemma~\ref{lemm.lsm.P}, since $p\subseteq q$ if and only if $p\cup q=q$.
\end{proof}

%%%%%%%%%%%%%%%%%%%%%%%%%%%%%%%%%%%%%%%%
\subsection{Some further operations on points}

\subsubsection{The operations: $\tand$, $\tall$, $\app$, and $\ppa$ on points}

Recall from the start of the Section that we fixed some $\lambda$-reduction theory $\Pi$.
Recall also $p[a\lsm u]$ from Definition~\ref{defn.qasmu}:
\begin{defn}
\label{defn.some.ops}
\label{defn.tall.p}
Suppose $p,q\subseteq|\idiom|$.
Define the following operations:
\begin{frameqn}
\begin{array}{r@{\ }l}
p\tand q=&p\cup q
\\
\tall a.p=&\bigcup_{u{\in}|\idiomprg|} p[a\lsm u]
\\
p\app q=&\bigcup\{ (st)\uparrowp{\Pi} \mid s\in p,\ t\in q\}
\\
q\ppa p=&\bigcap\{r\in|\points_\Pi| \mid p\subseteq r\app q\}
\end{array}
\end{frameqn}
\end{defn}

\begin{rmrk}
\label{rmrk.odd.ops}
Two things about Definition~\ref{defn.some.ops} might seem odd:
\begin{itemize*}
\item
$p\tand q$ and $\tall a.p$ are sets \emph{unions}---$p\cup q$ and $\bigcup_{u} p[a\lsm u]$ respectively---and not sets intersections.
This is a contravariance typical in duality results.
See Proposition~\ref{prop.pp.p.commute}, and see Lemma~\ref{lemm.pp.sub} for a clearer view of the contravariance in this case.
\item
There is no $p\tor q$, \emph{even though} in Proposition~\ref{prop.points.bigcap} we proved that a finite sets intersection of points is a point.
This is because the operation of taking a sets intersection of points does not interact correctly with the $\sigma$-action, see Remark~\ref{rmrk.no.sigma.distrib.points}.
For that, we need to consider sets of points; see Corollary~\ref{corr.sm.lsm.cup}.
\end{itemize*}
\end{rmrk}

\begin{lemm}
\label{lemm.some.ops.make.points}
If $p$ and $q$ are $\Pi$-points then so are $p\tand q$, $\tall a.p$, $p\app q$, and $q\ppa p$.
\end{lemm}
\begin{proof}
All from Proposition~\ref{prop.points.bigcap}, and for $p\app q$ also Lemma~\ref{lemm.suparrowp.point}.
\end{proof}

\begin{lemm}
\label{lemm.ppa.app.lambda}
Suppose $p,q\subseteq|\idiom|$ and suppose $r\in|\points_\Pi|$.
Then:
\begin{enumerate*}
\item
$q\ppa p\subseteq r$ if and only if $p\subseteq r\app q$.
\item
If $p,q\in|\points_\Pi|$ then $p\subseteq (q\ppa p)\app q$ and $q\ppa (p\app q)\subseteq p$.
\item
If $p\subseteq p'$ then $q\ppa p\subseteq q\ppa p'$.
\end{enumerate*}
\end{lemm}
\begin{proof}
Part~1 is from Definition~\ref{defn.some.ops}.
Part~2 follows using Lemma~\ref{lemm.some.ops.make.points} since $q\ppa p\subseteq q\ppa p$ and $p\app q\subseteq p\app q$.
For part~3 we note by part~2 that
$p'\subseteq (q\ppa p')\app q$, deduce that
$p\subseteq (q\ppa p')\app q$, and use part~1 to conclude that $q\ppa p\subseteq q\ppa p'$.
\end{proof}

\begin{lemm}
\label{lemm.tall.p.alpha}
Suppose $p\subseteq|\idiom|$ is small-supported.
Then $a\#\tall a.p$.

As a corollary, $\supp(\tall a_1,\dots,a_n.p)\subseteq\supp(p){\setminus}\{a_1,\dots,a_n\}$.
\end{lemm}
\begin{proof}
The first part is from Definition~\ref{defn.some.ops} using Lemma~\ref{lemm.lam.point.alpha}, along with Theorem~\ref{thrm.equivar} and Corollary~\ref{corr.stuff}(\ref{stuff.freshness.criterion}).
The corollary follows using the first part and Theorem~\ref{thrm.no.increase.of.supp}.
\end{proof}

\begin{lemm}
\label{lemm.tall.include.lam.point}
Suppose $p,q\in\points_\Pi$ and $s\in|\idiom|$.
Then:
\begin{enumerate*}
\item\label{tall.include.1}
$p\subseteq \tall a.p$.
\item
If $a\#p$ then $p=\tall a.p$.
\item
If $p\subseteq q$ then $\tall a.p\subseteq \tall a.q$.
\item\label{tall.include.fresh}
If $a\#q$ then $(\tall a.p)\subseteq q$ if and only if $p\subseteq q$.
\end{enumerate*}
\end{lemm}
\begin{proof}
We consider each part in turn:
\begin{enumerate*}
\item
We take $u{=}a$ and use Corollary~\ref{corr.p.sigma.id}.
\item
Using Lemma~\ref{lemm.lam.point.fresh}.
\item
From Lemma~\ref{lemm.lsm.monotone}.
\item
If $\tall a.p\subseteq q$ then $p\subseteq q$ using part~1 of this result, and if $p\subseteq q$ then
$\tall a.p\stackrel{\text{pt~3}}{\subseteq} \tall a.q\stackrel{\text{pt~2}}{=}q$.
\qedhere\end{enumerate*}
\end{proof}

%%%%%%%%%%%%%%%%%%%%%%%%
\subsubsection{$\ppa$ and $\tall$ make $\lambda$}

\begin{lemm}
\label{lemm.ppa.lam.char}
$q\ppa p=\{s'\in|\idiom| \mid \Forall{r}(p\subseteq r\app q\limp s'\in r)\}$.

As a corollary, $s'\in t\uparrowp{\Pi}\ppa s\uparrowp{\Pi}$ if and only if $s't\arrowp{\Pi}s$.
\end{lemm}
\begin{proof}
The first part just unpacks Definition~\ref{defn.some.ops}.

Now suppose $s'\in t\uparrowp{\Pi}\ppa s\uparrowp{\Pi}$.
By the first part, there exists some $t'$ with $t\arrowp{\Pi}t'$ and $s't'\arrowp{\Pi}s$.
By condition~\ref{item.somerel.app} of Definition~\ref{defn.lambda.reduction.theory} $s't\arrowp{\Pi}s't'$.
It follows that $s't\arrowp{\Pi}s$.

Conversely suppose $s't\arrowp{\Pi}s$ and choose any $s''$ with $s\arrowp{\Pi}s''$.
It follows that $s't\arrowp{\Pi} s''$ and since $t\in t\uparrowp{\Pi}$, we are done.
\end{proof}

\begin{rmrk}
\label{rmrk.ppa}
For the reader's convenience we apply Lemma~\ref{lemm.ppa.lam.char} to some concrete cases.
Suppose $\idiom=\lamtrm$.
\begin{itemize*}
\item
Take $q=a\uparrowp{\Pi}=p$.
Then
$s'\in a\uparrowp{\Pi}\ppa a\uparrowp{\Pi}$ if and only if $s'a\arrowp{\Pi}a$.
We can calculate that $\lam{a}a\in q\ppa p$ and also $\lam{b}a\in q\ppa p$.
\item
Assume some implementation of ordered pairs $(s,t)$ and $\pi_1$ and $\pi_2$ for first and second projection, and take $q=(a,b)\uparrowp{\Pi}$ and $p=a\uparrowp{\Pi}$.
Then
$s'\in (a,b)\uparrowp{\Pi}\ppa a\uparrowp{\Pi}$ if and only if $s'(a,b)\arrowp{\Pi}a$.
We can calculate that $\lam{b}a\in q\ppa p$ and $\pi_1\in q\ppa p$.
\end{itemize*}
So we can think of $\ppa$ as a kind of pattern-matching.
We refine this to model $\lambda$ in Proposition~\ref{prop.tall.lam.uparrow}.
\end{rmrk}

Recall $s\uparrowp{\Pi}$ from Definition~\ref{defn.suparrpi}, which is a point by Lemma~\ref{lemm.suparrowp.point}.

\begin{lemm}
\label{lemm.uparrow.app}
$s\uparrowp{\Pi}\app t\uparrowp{\Pi}=(st)\uparrowp{\Pi}$.
\end{lemm}
\begin{proof}
Unpacking Definitions~\ref{defn.suparrpi} and~\ref{defn.some.ops}, $u\in s\uparrowp{\Pi}\app t\uparrowp{\Pi}$ when $s\arrowp{\Pi}s'$ and $t\arrowp{\Pi}t'$ and $s't'\arrowp{\Pi}u$.
Also, $u\in (st)\uparrowp{\Pi}$ when $st\arrowp{\Pi}u$.
It is a fact that these two conditions are equivalent.
\end{proof}

Proposition~\ref{prop.tall.lam.uparrow} connects $\tall a$ and $\ppa$ on points, with $\lambda a$ on $\idiom$.
It will also be useful later in Corollary~\ref{corr.pp.commute.lam}.
We suggested in Remark~\ref{rmrk.ppa} that $\ppa$ is a kind of pattern-matching; by that view, what we do now is pattern-matching on a universally quantified atom:
\begin{prop}
\label{prop.tall.lam.uparrow}
$\tall a.(a\uparrowp{\Pi}\ppa s\uparrowp{\Pi})=(\lam{a}s)\uparrowp{\Pi}$.
\end{prop}
\begin{proof}
We prove two subset inclusions:
\begin{itemize*}
\item
\emph{Proof that $\tall a.(a\uparrowp{\Pi}\ppa s\uparrowp{\Pi})\subseteq(\lam{a}s)\uparrowp{\Pi}$.}\quad
By condition~\ref{item.somerel.beta} of Definition~\ref{defn.lambda.reduction.theory}
$(\lam{a}s)a\arrowp{\Pi}s$, so from Definition~\ref{defn.suparrpi}\footnote{See also the later and more comprehensive Lemma~\ref{lemm.lam.pp.inj}.} $s\uparrowp{\Pi}\subseteq ((\lam{a}s)a)\uparrowp{\Pi}\stackrel{\text{L\ref{lemm.uparrow.app}}}{=}(\lam{a}s)\uparrowp{\Pi}\app a\uparrowp{\Pi}$.
It follows by Lemma~\ref{lemm.ppa.app.lambda} that $a\uparrowp{\Pi}\ppa s\uparrowp{\Pi}\subseteq (\lam{a}s)\uparrowp{\Pi}$.
By condition~\ref{item.idiom.fresh} of Definition~\ref{defn.idiom} $a\#\lam{a}s$ so by Theorem~\ref{thrm.no.increase.of.supp} $a\#(\lam{a}s)\uparrowp{\Pi}$, and therefore 
by part~\ref{tall.include.fresh} of Lemma~\ref{lemm.tall.include.lam.point} $\tall a.(a\uparrowp{\Pi}\ppa s\uparrowp{\Pi})\subseteq (\lam{a}s)\uparrowp{\Pi}$.
\item
\emph{Proof that $(\lam{a}s)\uparrowp{\Pi}\subseteq\tall a.(a\uparrowp{\Pi}\ppa s\uparrowp{\Pi})$.}\quad
By condition~\ref{item.somerel.beta} of Definition~\ref{defn.lambda.reduction.theory} $(\lam{a}s)a\arrowp{\Pi} s$,
so by Lemma~\ref{lemm.ppa.lam.char} $\lam{a}s\in a\uparrowp{\Pi}\ppa s\uparrowp{\Pi}$.
By part~\ref{tall.include.1} of Lemma~\ref{lemm.tall.include.lam.point} we have that $\lam{a}s\in\tall a.(a\uparrowp{\Pi}\ppa s\uparrowp{\Pi})$. 
By condition~\ref{item.point.to} of Definition~\ref{defn.pi.point} we conclude $(\lam{a}s)\uparrowp{\Pi}\subseteq\tall a.(a\uparrowp{\Pi}\ppa s\uparrowp{\Pi})$ as required.
\qedhere\end{itemize*}
\end{proof}

\subsection{How the $\sigma$-action on points commutes}
\label{subsect.commutes}

The set of points $\points_\Pi$ has plenty of structure.
It is a nominal set, it has an $\amgis$-action $p[u\ms a]$ (Corollary~\ref{corr.points.Pi.amgis}), a $\sigma$-action $p[a\lsm u]$ (Proposition~\ref{prop.points.Pi.sigma.algebra}) and a subsidiary pointwise version $p[a\ssm u]$ (Definition~\ref{defn.p.ssm}).
It is a fresh semi-lattice (a top element, $\tand$, and $\tall$; see Remark~\ref{rmrk.varnothing.is.a.point} and Definition~\ref{defn.some.ops}) and has $\app$ and $\ppa$ (Definition~\ref{defn.some.ops}) and even $\nw a.p$ a sets version of the $\new$-quantifier (Lemma~\ref{lemm.nw.on.points}).
There is also a map from $|\idiom|$ to points given by $s\in|\idiom|$ maps to $s\uparrowp{\Pi}$ (Definition~\ref{defn.suparrpi}).

In this subsection we consider useful ways in which the $\sigma$-action commutes with some of this structure.
These commutation results will later be useful in proving that sets of points have the structure of an impredicative lattice with $\tall$ and $\app$.

\begin{lemm}
\label{lemm.lsm.uparrow}
$s\uparrowp{\Pi}[a\lsm u]=s[a\ssm u]\uparrowp{\Pi}$.
\end{lemm}
\begin{proof}
Using
\rulefont{\sigma\alpha} and Lemma~\ref{lemm.lam.point.alpha}(1) assume without loss of generality that $a\#u$.
We prove two subset inclusions.
\begin{itemize*}
\item
\emph{Proof that $s\uparrowp{\Pi}[a\lsm u]\subseteq s[a\ssm u]\uparrowp{\Pi}$.}\quad
By Lemma~\ref{lemm.fresh.sub} $a\#s[a\ssm u]$ and by Theorem~\ref{thrm.no.increase.of.supp} also $a\#s[a\ssm u]\uparrowp{\Pi}$.
Thus by Lemma~\ref{lemm.lsm.ssm.iff} to it suffices to prove $s\uparrowp{\Pi}[a\ssm u]\subseteq s[a\ssm u]\uparrowp{\Pi}$.
Suppose $s\arrowp{\Pi}t$.
By condition~\ref{item.somerel.ssm} of Definition~\ref{defn.lambda.reduction.theory} $s[a\ssm u]\arrowp{\Pi} t[a\ssm u]$.
The result follows.
\item
\emph{Proof that $s[a\ssm u]\uparrowp{\Pi}\subseteq s\uparrowp{\Pi}[a\lsm u]$.}\quad

By condition~\ref{item.point.to} of Definition~\ref{defn.pi.point} it suffices to note that $s[a\ssm u]\in s\uparrowp{\Pi}[a\lsm u]$.
\qedhere\end{itemize*}
\end{proof}

\begin{lemm}
\label{lemm.pq.ssm}
\begin{itemize*}
\item
$(p\app q)[a\ssm u]=p[a\ssm u]\app q[a\ssm u]$.
\item
$(\nw a.p)\app(\nw a.q)=\nw a.(p\app q)$ (Definition~\ref{defn.nua}; Lemma~\ref{lemm.nw.on.points}).
\end{itemize*}
\end{lemm}
\begin{proof}
From Definitions~\ref{defn.some.ops} and~\ref{defn.p.ssm}, we have that $r\in (p\app q)[a\ssm u]$ when there exist $s\in p$ and $t\in q$ such that $(st)[a\ssm u]\arrowp{\Pi}r$, and that $r\in p[a\ssm u]\app q[a\ssm u]$ when there exist $s\in p$ and $t\in q$ such that $s[a\ssm u]\,t[a\ssm u]\arrowp{\Pi}r$.
By Definition~\ref{defn.idiom}(\ref{item.idiom.app}) $(st)[a\ssm u]=s[a\ssm u]\,t[a\ssm u]$.

For the second part, $r\in (\nw a.p)\app(\nw a.q)$ when there exist $s$ and $t$ such that $\New{b}(b\ a)\act s\in p$ and $\New{b}(b\ a)\act t\in q$ and $r=st$.
It is a fact that the $\new$-quantifier distributes over conjunction \cite[Theorem~6.6]{gabbay:fountl} (provided $p$ and $q$ are small-supported; by Lemma~\ref{lemm.fresh.point.check}(\ref{point.finsupp}) they are) so also $\New{b}(b\ a)\act (st)\in p\app q$.
The result follows.
\end{proof}

\begin{corr}
\label{corr.lsm.app}
Suppose $p,q\in|\points_\Pi|$.
Then:
\begin{itemize*}
\item
$(p\tand q)[a\lsm u]=p[a\lsm u]\tand q[a\lsm u]$.
\item
$(p\app q)[a\lsm u]=p[a\lsm u]\app q[a\lsm u]$.
\end{itemize*}
\end{corr}
\begin{proof}
For the first part, by Definition~\ref{defn.some.ops} $p\tand q=p\cup q$.
We use Lemma~\ref{lemm.lsm.P}.

For the second part, we use Lemma~\ref{lemm.lam.point.alpha}(1) to assume without loss of generality that $a\#u$.
By Lemma~\ref{lemm.lsm.id} $(p\app q)[a\lsm u]=\nw a.((p\app q)[a\ssm u])$ and $p[a\lsm u]\app q[a\lsm u]=(\nw a.(p[a\ssm u]))\app(\nw a.(q[a\ssm u]))$.
We use Lemma~\ref{lemm.pq.ssm}.
\end{proof}

Proposition~\ref{prop.lsm.special.distrib} is the key technical result to proving \rulefont{\sigma\ppa} valid in Proposition~\ref{prop.points.coherent}:
\begin{prop}
\label{prop.lsm.special.distrib}
Suppose $u\in|\idiomprg|$ and $p\in|\points_\Pi|$.
Then
$$(b\uparrowp{\Pi}\ppa p)[a\lsm u]=b\uparrowp{\Pi}\ppa(p[a\lsm u]).
$$
\end{prop}
\begin{proof}
For the right-to-left inclusion $b\uparrowp{\Pi}\ppa(p[a\lsm u])\subseteq (b\uparrowp{\Pi}\ppa p)[a\lsm u]$ we reason as follows:
$$
\begin{array}{r@{\ }l@{\ }l}
b\uparrowp{\Pi}\ppa(p[a\lsm u])\subseteq& (b\uparrowp{\Pi}\ppa p)[a\lsm u]
\\
\liff&
p[a\lsm u]\subseteq (b\uparrowp{\Pi}\ppa p)[a\lsm u]\app (b\uparrowp{\Pi})
&\text{Lemma~\ref{lemm.ppa.app.lambda}}
\\
\liff&
p[a\lsm u]\subseteq (b\uparrowp{\Pi}\ppa p)[a\lsm u]\app (b\uparrowp{\Pi}[a\lsm u])
&\text{Lemma~\ref{lemm.lam.point.fresh}}
\\
\liff&
p[a\lsm u]\subseteq ((b\uparrowp{\Pi}\ppa p)\app b\uparrowp{\Pi})[a\lsm u]
&\text{Corollary~\ref{corr.lsm.app}}
\\
\Leftarrow &
p[a\lsm u]\subseteq p[a\lsm u]
&\text{Ls \ref{lemm.lsm.monotone} \& \ref{lemm.ppa.app.lambda}}
\end{array}
$$
For the left-to-right inclusion $(b\uparrowp{\Pi}\ppa p)[a\lsm u]\subseteq b\uparrowp{\Pi}\ppa(p[a\lsm u])$, using Lemma~\ref{lemm.lam.point.alpha}(1) to rename if necessary assume $a\#u$.
Using Theorem~\ref{thrm.no.increase.of.supp} also note that $a\#b\uparrowp{\Pi}\ppa(p[a\lsm u])$.
Thus by Lemma~\ref{lemm.lsm.ssm.iff} and part~3 of Lemma~\ref{lemm.ppa.app.lambda}
it suffices to prove that
$(b\uparrowp{\Pi}\ppa p)[a\ssm u]\subseteq b\uparrowp{\Pi}\ppa(p[a\ssm u])$.

Unpacking Definitions~\ref{defn.some.ops} (for $\ppa$) \ref{defn.suparrpi} (for $\uparrowp{\Pi}$) and~\ref{defn.p.ssm} (for $[a\ssm u]$ on points), this simplifies to showing that for all $s$,
$$
\Exists{s'}(s'b{\in} p\land s{=}s'[a\ssm u])
\limp
\Exists{t'}(t'{\in}p \land sb{=}t'[a\ssm u]) .
$$
So suppose $s'$ is such that $s'b{\in}p$ and $s{=}s'[a\ssm u]$.
Take $t'=s'b$.
So $t'[a\ssm u]=(s'b)[a\ssm u]=sb$ as required.
\end{proof}

%%%%%%%%%%%%%%%%%%%%%%%%%%%%%%%%%%%%%%%%%%%%%
\subsection{Operations on sets of points}

%%%%%%%%%%%%%%%%%%%%%%%%%%
\subsubsection{Basic definitions}

\begin{defn}
\label{defn.some.more.ops}
Suppose $p,q\in|\points_\Pi|$. 
Define the following operations:
\begin{frameqn}
\begin{array}{r@{\ }l}
\pp p=&\{q\in|\points_\Pi| \mid p\subseteq q\}
\\
p\bpp q=&\pp{(p\app q)}
\end{array}
\end{frameqn}
\end{defn}

\begin{rmrk}
$\pp p$ and $\bpp$ generate sets of points.
We use these to build a topology and a nominal spectral space with $\bpp$ over $\points_\Pi$ in Definition~\ref{defn.points.top}
and so to prove Theorem~\ref{thrm.Pi.completeness} (completeness).

$p\bpp q$ is a \emph{combination operator} in the sense of Definition~\ref{defn.bpp}, so that we get operations $\app$ and $\ppa$ on sets of points from Definition~\ref{defn.sub.sets.pp}.
Unpacking Definition~\ref{defn.some.ops} (for $p\app q$) and~\ref{defn.some.more.ops} (for $\pp{(p\app q)}$) it is not hard to see that $p\bpp q$ is the set of $\Pi$-points $r$ such that $\{st\mid s{\in}p, t{\in}q\}\subseteq r$.

Thus we have two operations named $\app$; one on points from Definition~\ref{defn.some.ops} and one on sets of points from Definition~\ref{defn.some.more.ops}.
Similarly, we have two operations named $\ppa$.
It will always be clear from context which is meant, and they are related by Proposition~\ref{prop.pp.p.commute}.
\end{rmrk}

\begin{rmrk}
It might be worth mentioning that $\pp p$ from Definition~\ref{defn.some.more.ops} is not a repeat of $\pp x$ from Definition~\ref{defn.pp}:
\begin{itemize*}
\item
In Definition~\ref{defn.pp} we assumed an underlying nominal distributive lattice with $\tall$.
\item
Definition~\ref{defn.some.more.ops} is constructed using $\points_\Pi$.
Now we consider Definition~\ref{defn.some.ops} and see that if we ignore $\app$ and $\ppa$ then $\points_\Pi$ is almost a nominal distributive lattice with $\tall$---but it lacks a disjunction (this is discussed in Remark~\ref{rmrk.odd.ops}).
We could reasonably call it a semilattice with $\tall$.
\end{itemize*}
So $\points_\Pi$ is a specific structure with its own properties.
\end{rmrk}

Lemma~\ref{lemm.pp.sub} is in the spirit of Lemma~\ref{lemm.completeness}:
\begin{lemm}
\label{lemm.pp.sub}
Suppose $p,q\in|\points_\Pi|$.
Then the following conditions are equivalent:
$$
p\in \pp q
\ \ \liff\ \
q\subseteq p
\ \ \liff\ \
\pp p\subseteq \pp q
$$
Furthermore, if $s\in|\idiom|$ then
$$
p\in\pp{s\uparrowp{\Pi}}\liff s\in p.
$$
\end{lemm}
(Recall that $s\uparrowp{\Pi}\in\points_\Pi$ by Definition~\ref{defn.suparrpi} and Lemma~\ref{lemm.suparrowp.point}.)
\begin{proof}
Unpacking Definition~\ref{defn.some.more.ops}, $p\in\pp q$ implies $q\subseteq p$.
It is a fact of sets that if $q\subseteq p$ then $p\subseteq p'$ implies $q\subseteq p'$, thus $\pp p\subseteq\pp q$.
Finally, if $\pp p\subseteq \pp q$ then since $p\in \pp p$, also $p\in \pp q$.

Now consider $p\in\pp{s\uparrowp{\Pi}}$.
By the previous paragraph this is if and only if $s\uparrowp{\Pi}\subseteq p$, and from condition~\ref{item.point.to} of Definition~\ref{defn.pi.point} this is if and only if $s\in p$.
\end{proof}

\begin{corr}
\label{corr.lam.pi.injective}
The assignment $p\longmapsto \pp p$ is injective from $|\points_\Pi|$ to $\nompow(\points_\Pi)$.
As a corollary, $\supp(\pp p)=\supp(p)$.
\end{corr}
\begin{proof}
The first part follows using Lemma~\ref{lemm.pp.sub}.
The corollary follows by part~\ref{item.conservation.of.support} of Theorem~\ref{thrm.no.increase.of.supp}.
\end{proof}

Corollary~\ref{corr.filter.point.property} should remind us of condition~\ref{sober.all} of Definition~\ref{defn.sober}.
We use it to get just that in Proposition~\ref{prop.points.sober}:
\begin{corr}
\label{corr.filter.point.property}
Suppose $p,q\in|\points_\Pi|$ and $\New{b}(b\ a)\act p\in\pp q$.
Then $p\in\pp{(\tall a.q)}$.
\end{corr}
\begin{proof}
Choose fresh $b$ (so $b\#p,q$), so that $(b\ a)\act p\in\pp q$.
We note two facts:
\begin{itemize*}
\item
By Proposition~\ref{prop.pi.supp} $a\#(b\ a)\act p$.
\item
By Lemma~\ref{lemm.tall.p.alpha} $a\#\tall a.q$, so that by Theorem~\ref{thrm.no.increase.of.supp} $a,b\#\pp{(\tall a.q)}$.
\end{itemize*}
Now we reason as follows:
$$
\begin{array}[b]{r@{\ }l@{\qquad}l}
(b\ a)\act p\in\pp q
\liff&
q\subseteq (b\ a)\act p
&\text{Lemma~\ref{lemm.pp.sub}}
\\
\liff&
\tall a.q\subseteq(b\ a)\act p
&\text{Lemma~\ref{lemm.tall.include.lam.point}(\ref{tall.include.fresh})}\ a\#(b\ a)\act p
\\
\liff&
(b\ a)\act p\in\pp{(\tall a.q)}
&\text{Lemma~\ref{lemm.pp.sub}}
\\
\liff&
p\in\pp{(\tall a.q)}
&\text{Corollary~\ref{corr.stuff}}\ a,b\#\pp{(\tall a.q)}
\end{array}
\qedhere$$
\end{proof}

The proof of Lemma~\ref{lemm.lam.pp.inj} is simple given what we have proved so far, but it is important; for instance it is the final step in the proof of Completeness in Theorem~\ref{thrm.Pi.completeness}.
\begin{lemm}
\label{lemm.lam.pp.inj}
The following conditions are equivalent:
$$
\pp{s\uparrowp{\Pi}}\subseteq\pp{t\uparrowp{\Pi}}
\ \ \liff\ \
s\uparrowp{\Pi}\in \pp t
\ \ \liff\ \
t\uparrowp{\Pi}\subseteq s\uparrowp{\Pi}
\ \ \liff\ \
t\in s\uparrowp{\Pi}
\ \ \liff\ \
s\arrowp{\Pi}t
\ \ \liff\ \
\Pi\cent s{\to}t
$$
\end{lemm}
\begin{proof}
We obtain
$
\pp{s\uparrowp{\Pi}}\subseteq\pp{t\uparrowp{\Pi}}
\liff
s\uparrowp{\Pi}\in \pp t
\liff
t\uparrowp{\Pi}\subseteq s\uparrowp{\Pi}
\liff
t\in s\uparrowp{\Pi}
$
from Lemma~\ref{lemm.pp.sub}.
Also, $t\in s\uparrowp{\Pi}\liff s\arrowp{\Pi}t$ is direct from Definition~\ref{defn.suparrpi}.
Finally, by Notation~\ref{nttn.Pi.cent} $s\arrowp{\Pi}t$ is synonymous with $\Pi\cent s{\to}t$.
\end{proof}

\begin{corr}
\label{corr.lam.pp.eq}
$\pp{s\uparrowp{\Pi}} = \pp{t\uparrowp{\Pi}}$ if and only if $s=_\Pi t$.
\end{corr}
\begin{proof}
From Lemma~\ref{lemm.lam.pp.inj}.
\end{proof}

%%%%%%%%%%%%%%%%%%%%%%%%%%
\subsubsection{Commutation properties}

Recall $q[a\lsm u]$ from Definition~\ref{defn.qasmu}, $p[u\ms a]$ from Definition~\ref{defn.p.action}, and (since by Corollary~\ref{corr.points.Pi.amgis} $\points_\Pi$ is an $\amgis$-algebra) $\pp p[a\sm u]$ from Definition~\ref{defn.sub.sets}.

Theorem~\ref{thrm.lam.pp.sm} is fairly easy to prove, but part~1 of it is key.
It relates the natural $\sigma$-action $\pp p[a\sm u]$ to the $\sigma$-action on points $p[a\lsm u]$ from Proposition~\ref{prop.points.Pi.sigma.algebra}.
Compare Theorem~\ref{thrm.lam.pp.sm} with Lemma~\ref{lemm.dup}:
\begin{thrm}
\label{thrm.lam.pp.sm}
Suppose $q\in|\points_\Pi|$ and $u\in|\idiomprg|$.
\begin{enumerate*}
\item
$\pp q[a\sm u]=\pp{(q[a\lsm u])}$.

As a corollary taking $p=s\uparrowp{\Pi}$,\ \ $\pp{s\uparrowp{\Pi}}[a\sm t]=\pp{s[a\ssm t]\uparrowp{\Pi}}$.
\item
$\pi\act(\pp q)=\pp{(\pi\act q)}$.
\end{enumerate*}
\end{thrm}
\begin{proof}
Consider some $p$; we wish to show that $p\in \pp q[a\sm u]\liff p\in \pp{(q[a\lsm u])}$.
By Lemmas~\ref{lemm.sigma.alpha} and~\ref{lemm.lam.point.alpha}(1) we may $\alpha$-rename $a$ in $\pp q[a\sm u]$ and $q[a\lsm u]$ to assume without loss of generality that $a\#u,p$.
We reason as follows; we use part~2 of Proposition~\ref{prop.amgis.iff} because by
Lemma~\ref{lemm.fresh.point.check}(\ref{point.finsupp}) $p$ and $q$ have small support:
$$
\begin{array}{r@{\ }l@{\qquad}l}
p\in \pp q[a\sm u]
\liff&
p[u\ms a]\in \pp q
&\text{Proposition~\ref{prop.amgis.iff}(2)}\ a\#u,p
\\
\liff&
q\subseteq p[u\ms a]
&\text{Definition~\ref{defn.some.more.ops}}
\\
\liff&
q[a\lsm u]\subseteq p
&\text{Proposition~\ref{prop.qasmu.iff}}\ a\#u,p
\\
\liff&
p\in\pp{(q[a\lsm u])}
&\text{Definition~\ref{defn.some.more.ops}}
\end{array}
$$
The corollary follows using Lemma~\ref{lemm.lsm.uparrow}.

The second part is proved by similar calculations, or directly from Theorem~\ref{thrm.equivar}.
\end{proof}

Corollary~\ref{corr.sm.lsm.cup} describes how the $\sigma$-action interacts with strictly supported sets union:
\begin{corr}
\label{corr.sm.lsm.cup}
Suppose $\mathcal P{\subseteq}|\points_\Pi|$ is strictly small-supported (in particular by Lemma~\ref{lemm.finite.strict} it suffices that $\mathcal P$ be finite).
Then
$$
(\bigcup_{p\in\mathcal P}\pp p)[a\sm u]=\bigcup_{p\in\mathcal P}\pp{(p[a\lsm u])}.
$$
\end{corr}
\begin{proof}
From Lemma~\ref{lemm.sub.bigcap}(\ref{sub.bigcup.sfs}) and Theorem~\ref{thrm.lam.pp.sm}.
\end{proof}

\begin{corr}
\label{corr.lsm.fresh}
\begin{enumerate*}
\item
$a\#\pp p$ then $(\pp p)[a\sm u]=\pp p$.
\item
If $b\#\pp p$ then $((b\ a)\act \pp p)[b\sm u]=\pp p[a\sm u]$.
\end{enumerate*}
\end{corr}
\begin{proof}
For the first part, suppose $a\#\pp p$.
By Corollary~\ref{corr.lam.pi.injective} $a\#p$.
We use Theorem~\ref{thrm.lam.pp.sm} and Lemma~\ref{lemm.lam.point.fresh}.

For the second part we reason similarly using Theorem~\ref{thrm.lam.pp.sm} and Lemma~\ref{lemm.lam.point.alpha}.
\end{proof}

\begin{prop}
\label{prop.pp.p.commute}
\begin{enumerate*}
\item
\label{item.pp.cap}
$\pp p\cap \pp q=\pp{(p\tand q)}$.
\item
\label{item.lam.pp.tall}
$\freshcap{a}\pp p=\pp{(\tall a.p)}$\ \ ($\tall a.p$ is from Definition~\ref{defn.some.ops}).
\item
\label{item.lam.pp.app}
$\pp p\app \pp q =\pp{(p\app q)}$.
\item
\label{item.lam.pp.ppa}
$\pp q\ppa \pp p =\pp{(q\ppa p)}$.
\end{enumerate*}
\end{prop}
\begin{proof}
We consider each case in turn; with what we have proved so far, the calculations are routine.
$u$ will range over elements of $|\idiom|$:
\begin{enumerate*}
\item
We reason as follows:
\begin{tab2a}
r\in \pp p\cap\pp q
\liff&
p\subseteq r\ \text{and}\ q\subseteq r
&\text{Definition~\ref{defn.some.more.ops}}
\\
\liff&
p\cup q\subseteq r
&\text{Fact}
\\
\liff&
p\in\pp{(p\tand q)}
&\text{Definition~\ref{defn.some.more.ops}}
\end{tab2a}
\item
We reason as follows:
\begin{tab2a}
\freshcap{a}\pp p =&
\bigcap_{u{\in}|\idiomprg|}\pp p[a\sm u]
&\text{Definition~\ref{defn.nu.U}}
\\
=&\bigcap_{u{\in}|\idiomprg|} \pp{(p[a\lsm u])}
&\text{Theorem~\ref{thrm.lam.pp.sm}}
\\
=&\pp{(\bigcup_{u{\in}|\idiomprg|} p[a\lsm u])}
&\text{Fact of Def~\ref{defn.some.more.ops}}
\\
=&\pp{(\tall a.p)}
&\text{Definition~\ref{defn.tall.p}}
\end{tab2a}
\item
We reason as follows:
\begin{tab2a}
r\in \pp p\app\pp q
\liff&
\Exists{p'{\in}\pp p, q'{\in}\pp q}r\in p'\bpp q'
&\text{Proposition~\ref{prop.amgis.iff.pp}}
\\
\liff&
\Exists{p',q'}p\subseteq p'\land q\subseteq q'\land r\in p'\bpp q'
&\text{Definition~\ref{defn.some.more.ops}}
\\
\liff&
\Exists{p',q'}p\subseteq p'\land q\subseteq q'\land p'\app q'\subseteq r
&\text{Definition~\ref{defn.some.more.ops}}
\\
\liff&
p\app q\subseteq r
&\text{Fact}
\\
\liff&
r\in\pp{(p\app q)}
&\text{Definition~\ref{defn.some.more.ops}}
\end{tab2a}
\item
We reason as follows:
\begin{tab2a}
r\in \pp q\ppa\pp p
\liff &
\Forall{q'{\in}\pp q}r\bpp q'\subseteq\pp p
&\text{Proposition~\ref{prop.amgis.iff.pp}}
\\
\liff &
\Forall{q'}q\subseteq q'\limp \pp{(r\app q')}\subseteq\pp p
&\text{Definition~\ref{defn.some.more.ops}}
\\
\liff &
\Forall{q'}q\subseteq q'\limp p\subseteq r\app q'
&\text{Definition~\ref{defn.some.more.ops}}
\\
\liff& p\subseteq r\app q
&\text{Fact}
\\
\liff& q\ppa p\subseteq r
&\text{Lemma~\ref{lemm.ppa.app.lambda}}
\\
\liff& r\in\pp{(q\ppa p)}
&\text{Definition~\ref{defn.some.more.ops}}
\qedhere\end{tab2a}
\end{enumerate*}
\end{proof}

Corollary~\ref{corr.filter.lam.pp.app.commute} resembles Lemma~\ref{lemm.filter.pp.app.commute} and is proved similarly:
\begin{corr}
\label{corr.filter.lam.pp.app.commute}
Suppose $p,q\in|\points_\Pi|$ and $u\in|\idiom|$.
Then $\app$ and $\ppa$ validate axioms \rulefont{\sigma\app} and \rulefont{\sigma\ppa} from Figure~\ref{fig.app.compatible}:
$$
\begin{array}{r@{\ }l}
(\pp p\app \pp q)[a\sm u] =& \pp p[a\sm u]\app \pp q[a\sm u]
\\
(\pp{b\uparrowp{\Pi}}\ppa \pp p)[a\sm u] =& \pp{b\uparrowp{\Pi}}\ppa \pp p[a\sm u]
\end{array}
$$
\end{corr}
\begin{proof}
We reason as follows:
\begin{tab2b}
(\pp p\app\pp q)[a\sm u]
=&\pp{(p\app q)}[a\sm u]
&\text{Part~\ref{item.lam.pp.app} of Proposition~\ref{prop.pp.p.commute}}
\\
=&\pp{((p\app q)[a\sm u])}
&\text{Part~1 of Theorem~\ref{thrm.lam.pp.sm}}
\\
=&\pp{(p[a\sm u]\app q[a\sm u])}
&\text{Corollary~\ref{corr.lsm.app}}
\\
=&\pp{(p[a\sm u])}\app \pp{(q[a\sm u])}
&\text{Part~\ref{item.lam.pp.app} of Proposition~\ref{prop.pp.p.commute}}
\\
=&\pp p[a\sm u]\app \pp q[a\sm u]
&\text{Part~1 of Theorem~\ref{thrm.lam.pp.sm}}
\\[1.5ex]
(\pp{b\uparrowp{\Pi}} \ppa\pp p)[a\sm u]
=&\pp{(b\uparrowp{\Pi}\ppa p)}[a\sm u]
&\text{Part~\ref{item.lam.pp.ppa} of Proposition~\ref{prop.pp.p.commute}}
\\
=&\pp{((b\uparrowp{\Pi}\ppa p)[a\sm u])}
&\text{Part~1 of Theorem~\ref{thrm.lam.pp.sm}}
\\
=&\pp{(b\uparrowp{\Pi}\ppa (p[a\sm u]))}
&\text{Proposition~\ref{prop.lsm.special.distrib}}
\\
=&\pp{b\uparrowp{\Pi}}\ppa \pp{(p[a\sm u])}
&\text{Part~\ref{item.lam.pp.ppa} of Proposition~\ref{prop.pp.p.commute}}
\\
=&\pp{b\uparrowp{\Pi}}\ppa \pp p[a\sm u]
&\text{Part~1 of Theorem~\ref{thrm.lam.pp.sm}}
\qedhere\end{tab2b}
\end{proof}

\begin{corr}
\label{corr.pp.commute.lam}
If $s',s'\in|\idiom|$ then $\pp{s'\uparrowp{\Pi}}\app \pp{s\uparrowp{\Pi}}=\pp{(s's)\uparrowp{\Pi}}$ and $\tlam a.(\pp{s\uparrowp{\Pi}})=\pp{(\lam{a}s)\uparrowp{\Pi}}$.
\end{corr}
\begin{proof}
We reason as follows:
\begin{tab2b}
\pp{s'\uparrowp{\Pi}}\app\pp{s\uparrowp{\Pi}}
=&\pp{(s'\uparrowp{\Pi}\app s\uparrowp{\Pi})}
&\text{Part~\ref{item.lam.pp.app} of Prop~\ref{prop.pp.p.commute}}
\\
=&\pp{(s's)\uparrowp{\Pi}}
&\text{Lemma~\ref{lemm.uparrow.app}}
\\[1.2ex]
\tlam a.(\pp{s\uparrowp{\Pi}})=&\tall a.(\pp{a\uparrowp{\Pi}}\ppa \pp{s\uparrowp{\Pi}})
&\text{Notation~\ref{nttn.lambda}}
\\
=&\tall a.\pp{(a\uparrowp{\Pi}\ppa s\uparrowp{\Pi})}
&\text{Part~\ref{item.lam.pp.ppa} of Prop~\ref{prop.pp.p.commute}}
\\
=&\pp{(\tall a.(a\uparrowp{\Pi}\ppa s\uparrowp{\Pi}))}
&\text{Part~\ref{item.lam.pp.tall} of Prop~\ref{prop.pp.p.commute}}
\\
=&\pp{(\lam{a}s)\uparrowp{\Pi}}
&\text{Proposition~\ref{prop.tall.lam.uparrow}}
\qedhere\end{tab2b}
\end{proof}

We conclude with a small calculation on $\ppa$:
\begin{lemm}
Suppose $p\in|\points_\Pi|$ and $s,t\in|\idiom|$.
Then the following are equivalent:
$$
p\in \pp{t\uparrowp{\Pi}}\ppa \pp{s\uparrowp{\Pi}}
\quad\Leftrightarrow\quad
s\in p\app t\uparrowp{\Pi}
\quad\Leftrightarrow\quad
\Exists{s'{\in}p}(s't\to_\Pi s).
$$
\end{lemm}
\begin{proof}
We reason as follows:
$$
\begin{array}[b]{r@{\ }l@{\quad}l}
p\in\pp{t\uparrowp{\Pi}}\ppa\pp{s\uparrowp{\Pi}}
\liff&
p\in\pp{(t\uparrowp{\Pi}\ppa s\uparrowp{\Pi})}
&\text{Part~\ref{item.lam.pp.ppa} of Proposition~\ref{prop.pp.p.commute}}
\\
\liff&
t\uparrowp{\Pi}\ppa s\uparrowp{\Pi}\subseteq p
&\text{Lemma~\ref{lemm.pp.sub}}
\\
\liff&
s\uparrowp{\Pi}\subseteq p\app t\uparrowp{\Pi}
&\text{Lemma~\ref{lemm.ppa.app.lambda}}
\\
\liff&
s\in p\app t\uparrowp{\Pi}
&\text{Part~\ref{item.point.to} of Definition~\ref{defn.pi.point}}
\\
\liff&
\Exists{s'{\in}p}(s't\to_\Pi s)
&\text{Unpacking Definition~\ref{defn.some.ops}}
\end{array}
\qedhere$$
\end{proof}

%%%%%%%%%%%%%%%%%%%%%%%%%%%%%%%%
\subsection{A topology}

\subsubsection{Giving $\points_\Pi$ a topology}

\begin{defn}
\label{defn.points.top}
Make $\points_\Pi$ into a nominal spectral space with $\bpp$ (Definition~\ref{defn.nom.top.spectral.bpp}):
\begin{enumerate*}
\item
The topology generated under small-supported unions by $\{\pp p\mid p\in|\points_\Pi|\}$.
\item
The combination operation $p\bpp q=\pp{(p\app q)}$ from Definition~\ref{defn.some.more.ops}.
\item
The pointwise actions from Definitions~\ref{defn.p.action} (for $\pi$ and $[u\ms a]$) and~\ref{defn.sub.sets.pp} (for $\app$ and $\ppa$).
\item\label{item.prg.points.Pi}
$\prg_{\points_\Pi} u=\pp{u\uparrowp{\Pi}}$ for $u\in|\idiomprg|$.
\end{enumerate*}
\end{defn}

\begin{rmrk}
\label{rmrk.lam.when.open}
So $U\in\otop{\points_\Pi}$ (meaning that $U$ is open) when there exists some small-supported $\mathcal P\subseteq|\points_\Pi|$ (Definition~\ref{defn.strictpow}) with $U=\bigcup\{\pp p\mid p\in\mathcal P\}$.
\end{rmrk}

\begin{prop}
\label{prop.points.sigma.bpp}
Definition~\ref{defn.points.top} does indeed determine a nominal $\sigma\bpp$-topological space in the sense of Definition~\ref{defn.nom.top.bpp}.
\end{prop}
\begin{proof}
We start with the conditions from Definition~\ref{defn.nom.top}.

By Corollary~\ref{corr.points.Pi.amgis} $(|\points_\Pi|,\act,\idiom,\amgis)$ is an $\amgis$-algebra.

Consider $U,V\in\otop{\points_\Pi}$.
As noted in Remark~\ref{rmrk.lam.when.open} $U=\bigcup\{\pp p\mid p\in\mathcal P\}$ and
$V=\bigcup\{\pp q\mid q\in\mathcal Q\}$
for some small-supported $\mathcal P,\mathcal Q\subseteq|\points_\Pi|$.
By Theorem~\ref{thrm.no.increase.of.supp} each $U$ is small-supported.
Furthermore:
\begin{enumerate*}
\item
\emph{$\varnothing$ and $|\points_\Pi|$ are open.}\quad
$\varnothing$ (the empty set of points) is open by construction of the topology, and $|\points_\Pi|=\bigcup\{\pp \varnothing\}$ is open---we noted in Remark~\ref{rmrk.varnothing.is.a.point} that $\varnothing$ the empty set of phrases is a point.
Then $\pp\varnothing$ is the set of all points and $\{\pp\varnothing\}$ is strictly small-supported (by $\varnothing$ the empty set of atoms).
\item
\emph{If $U$ and $V$ are open then so are $U\cap V$ and $U\cup V$.}\quad
It is a fact of sets that $U\cap V=\bigcup\{\pp p\cap\pp q\mid p\in\mathcal P,q\in\mathcal Q\}$.
We use part~\ref{item.pp.cap} of Proposition~\ref{prop.pp.p.commute} and Theorem~\ref{thrm.no.increase.of.supp}.

For $U\cup V$, we note that the union of two small-supported sets is small-supported.
\item
\emph{If $\mathcal U$ is a small-supported set of open sets then $\bigcup\mathcal U$ is open.}\quad
Using Theorem~\ref{thrm.no.increase.of.supp}.
\end{enumerate*}
Finally, we consider Definition~\ref{defn.nom.top.bpp} and note that $\bpp$ is a combination operator (that $\bpp$ is equivariant follows immediately from Theorem~\ref{thrm.equivar}).
\end{proof}

%%%%%%%%%%%%%%%%%%%%%%%%%%%%%%%%%
\subsubsection{The compact open sets of $\points_\Pi$}

\begin{lemm}
\label{lemm.p.in.U}
If $U\in\otop{\points_\Pi}$ then
$$
p\in U \quad\text{if and only if}\quad \pp p\subseteq U .
$$
\end{lemm}
\begin{proof}
Suppose $p\in U$.
By Definition~\ref{defn.points.top} $U=\bigcup\{\pp q\mid q\in \mathcal Q\}$ for some small-supported $\mathcal Q\subseteq|\points_\Pi|$, and there exists $q\in\mathcal Q$ with $p\in \pp q$.
By Lemma~\ref{lemm.pp.sub}, $\pp p\subseteq\pp q$.

The reverse implication is easy since $p\in\pp p$ by construction in Definition~\ref{defn.some.more.ops}.
\end{proof}

\begin{lemm}
\label{lemm.pp.p.compact}
If $p\in|\points_\Pi|$ then $\pp p$ is compact in $\points_\Pi$ with the topology from Definition~\ref{defn.points.top}.
In symbols: $\pp p\in\ctop{\points_\Pi}$.

As a corollary, if $s\in|\idiom|$ then $\pp{s\uparrowp{\Pi}}\in\ctop{\points_\Pi}$.
\end{lemm}
\begin{proof}
Suppose $\mathcal U$ covers $\pp p$.
Since $p\in\pp p$, also $p\in U$ for some $U\in\mathcal U$.
By Lemma~\ref{lemm.p.in.U} $\pp p\subseteq X$ and so $\pp p$ is covered by $\{X\}$.

The corollary follows from Lemma~\ref{lemm.suparrowp.point}.
\end{proof}

\begin{lemm}
\label{lemm.pp.U.compact}
If ${U\in\otop{\points_\Pi}}$ is compact then ${U=\bigcup\{\pp p\mid p\in\mathcal P\}}$ for some finite $\mathcal P\subseteq |\points_\Pi|$.
\end{lemm}
\begin{proof}
Suppose $U$ is compact.
By construction $U=\bigcup\mathcal P'$ for some small-supported (but not necessarily finite) $\mathcal P'\subseteq|\points_\Pi|$.
By compactness this has a finite subcover $\{\pp p_1,\dots,\pp p_n\}$.
We take $\mathcal P=\{p_1,\dots,p_n\}$.
\end{proof}

\begin{rmrk}
We have mentioned that the canonical model $\points_\Pi$ is not just a replay of $F(\ns D)$ from
the duality proof (Definition~\ref{defn.F}).
So compare
\begin{itemize*}
\item
Theorem~\ref{thrm.pp.x.clopen}, which identifies compacts in $F(\ns D)$ with sets of points of the form $\pp x$, with
\item
Lemma~\ref{lemm.pp.U.compact}, which identifies compacts in $\points_\Pi$ with \emph{finite unions} of sets of the form $\pp p$.
\end{itemize*}
This is related to issues discussed in Remarks~\ref{rmrk.odd.ops} and~\ref{rmrk.no.sigma.distrib.points}.
\end{rmrk}

\subsubsection{Interaction of $\app$ and $\ppa$ with $\cup$ and $\subseteq$}

Lemmas~\ref{lemm.some.monotone} and~\ref{lemm.bpp.is.app.for.points.Pi} are useful for Corollary~\ref{corr.bpp.ppa.points.Pi}:
\begin{lemm}
\label{lemm.some.monotone}
Suppose $p,q,q'\in|\points_\Pi|$ and $q\subseteq q'$.
Then
$p\app q\subseteq p\app q'$ and $p\bpp q'\subseteq p\bpp q$.
\end{lemm}
\begin{proof}
$p\app q$ is defined in Definition~\ref{defn.some.ops}, and the result follows direct from the definition.
By Definition~\ref{defn.some.more.ops} $p\bpp q'=\{r{\in}\points_\Pi\mid p\app q'\subseteq r\}$ and $p\bpp q=\{r{\in}\points_\Pi\mid p\app q\subseteq r\}$.
We use Lemma~\ref{lemm.pp.sub}.
\end{proof}

\begin{lemm}
\label{lemm.bpp.is.app.for.points.Pi}
If $X\in\otop{\points_\Pi}$ and $p,q\in|\points_\Pi|$
then the following conditions are equivalent:
$$
p\app q\in X
\quad\liff\quad
\pp{(p\app q)}\subseteq X
\quad\liff\quad
p\bpp q\subseteq X
\quad\liff\quad
p\bpp(\pp q)\subseteq X
$$
\end{lemm}
\begin{proof}
By Lemma~\ref{lemm.p.in.U} $p\app q\in X$ if and only if $\pp{(p\app q)}\subseteq X$.
By Definition~\ref{defn.some.more.ops} $\pp{(p\app q)}=p\bpp q$.

Suppose $p\bpp q\subseteq X$.
Consider any $q'\in\pp q$, meaning by Definition~\ref{defn.some.more.ops} that $q\subseteq q'$.
By Lemma~\ref{lemm.some.monotone} $p\app q'\subseteq p\bpp q$.
Since $q'$ was arbitrary, it follows from Definition~\ref{defn.YppaX} that $p\bpp(\pp q)\subseteq X$.

Conversely if $p\bpp(\pp q)\subseteq X$ then since (from Definition~\ref{defn.some.more.ops}) $q\in\pp q$, from Definition~\ref{defn.YppaX} $p\bpp q\subseteq X$.
\end{proof}

\begin{corr}
\label{corr.bpp.ppa.points.Pi}
If $X\in\otop{\points_\Pi}$ and $r,q\in|\points_\Pi|$
then $r\in (\pp q_1\cup\dots\cup\pp q_n)\ppa X$ if and only if $r\app q_i\in X$ for $1{\leq}i{\leq}n$.
\end{corr}
\begin{proof}
By Proposition~\ref{prop.amgis.iff.pp} $r\in (\pp q_1\cup\dots\cup\pp q_n)\ppa X$ if and only if
$r\bpp (\pp q_1\cup\dots\cup\pp q_n)\subseteq X$.
From Definition~\ref{defn.YppaX} this is if and only if $r\bpp\pp q_i\subseteq X$ for $1{\leq}i{\leq}n$.
We use Lemma~\ref{lemm.bpp.is.app.for.points.Pi}.
\end{proof}

\begin{corr}
\label{corr.ppq.ppa.finite.union}
$r\in \pp q\ppa (\pp p_1\cup\dots\cup\pp p_n)$ if and only if $r\in \pp q\ppa \pp p_i$ for some $1{\leq}i{\leq}n$.
\end{corr}
\begin{proof}
By Proposition~\ref{prop.amgis.iff.pp} $r\in \pp q\ppa (\pp p_1\cup\dots\cup\pp p_n)$ if and only if
$r\bpp q\subseteq\pp p_1\cup\dots\cup\pp p_n$.
By Lemma~\ref{lemm.bpp.is.app.for.points.Pi} this is if and only if
$r\app q\in \pp p_1\cup\dots\cup\pp p_n$.
It is a fact of sets that this is if and only if $r\app q\in \pp p_i$ for some $1{\leq}i{\leq}n$.
Using Lemma~\ref{lemm.bpp.is.app.for.points.Pi} and Proposition~\ref{prop.amgis.iff.pp} again, this is if and only if $r\in \pp q\ppa \pp p_i$ for some $1{\leq}i{\leq}n$, as required.
\end{proof}

\begin{corr}
\label{corr.unzip}
Suppose $\mathcal P,\mathcal Q\subseteq|\points_\Pi|$ are finite.
Then
$$
\bigcup\{\pp q\mid q\in\mathcal Q\}\ppa \bigcup\{\pp p\mid p\in\mathcal P\}=
\bigcap\{ \bigcup\{ \pp{(q\ppa p)}\mid p\in\mathcal Q\} \mid q\in\mathcal Q\}
.
$$
\end{corr}
\begin{proof}
We combine Corollaries~\ref{corr.bpp.ppa.points.Pi} and~\ref{corr.ppq.ppa.finite.union} with part~\ref{item.lam.pp.ppa} of Proposition~\ref{prop.pp.p.commute}.
\end{proof}

\subsubsection{Interaction of $\protect\freshcap{a}\text{}$ with unions}

Remarkably, $\freshcap{a}$ commutes with certain unions.
This is Lemma~\ref{lemm.tall.unions}.
The property is not valid in general in $\indiapp$; indeed this is not generally true in logic: $\forall x.(\phi\lor\psi)$ is not normally logically equivalent to $(\forall x.\phi)\lor(\forall x.\psi)$.
But, it holds in the canonical model $\points_\Pi$:
\begin{lemm}
\label{lemm.tall.unions}
Suppose $p_1,\dots,p_n{\in}\points_\Pi$.
Then
$$
\freshcap{a}\bigcup\nolimits_i(\pp p_i)=\bigcup\nolimits_i\freshcap{a}(\pp p_i).
$$
\end{lemm}
\begin{proof}
Suppose $q\in \freshcap{a}\bigcup_i\pp p_i$.
Using Lemma~\ref{lemm.all.alpha} rename to assume without loss of generality that $a\#q$.

By Definition~\ref{defn.nu.U} and Corollary~\ref{corr.sm.lsm.cup} $q\in \freshcap{a}\bigcup_i\pp p_i$ when $q\in \bigcup_i\pp{(p_i[a\lsm u])}$ for every $u\in|\idiomprg|$.
Choose fresh $b$ (so $b\#q,p_1,\dots,p_n$).
By facts of sets, there exists $1{\leq}i{\leq}n$ such that $q\in\pp{p_i[a\sm b]}$.
Using Lemma~\ref{lemm.p.sigma.pi} $q\in\pp{((b\ a)\act p_i)}$,
so by Theorem~\ref{thrm.equivar} $(b\ a)\act q\in \pp p_i$ and by Corollary~\ref{corr.filter.point.property} $q\in\pp{(\tall a.p_i)}$.
We use Proposition~\ref{prop.pp.p.commute}(\ref{item.lam.pp.tall}).

Conversely, it is easy to prove that $\freshcap{a}\pp p_i\subseteq \freshcap{a}\bigcup_i\pp p_i$, either using Lemma~\ref{lemm.tall.monotone} and Theorem~\ref{thrm.all.closed} since $\pp p_i\subseteq\bigcup_i\pp p_i$---or by an easy direct calculation from Definition~\ref{defn.nu.U} and Lemma~\ref{lemm.sub.bigcap}(\ref{sub.bigcap.monotone}).
\end{proof}

%[mjg ref to this remark from other paper]
\begin{rmrk}
\label{rmrk.any.two}
Lemma~\ref{lemm.tall.unions} depends on $q$ being small-supported (Lemma~\ref{lemm.fresh.point.check}(\ref{point.finsupp}))
and on the $p_i$ having a $\sigma$-action (Definition~\ref{defn.qasmu}).
This makes the canonical model $\points_\Pi$ powerfully well-behaved; neither property holds generally in $\amgis$-algebras.

One corollary of small support in particular is that the canonical model cannot support classical negation, for if we had classical negation then we could reason as follows (this is not intended to be fully formal, but it could be made so):
$$
p\in\freshcap{a}\pp q
\ \ \liff\ \
\New{b}(b\ a)\act p\in\pp q
\ \ \liff \ \
\New{b}(b\ a)\act p\not\in\points_\Pi{\setminus}\pp q
\ \ \liff\ \
p\in\freshcup{a}\pp q .
$$
In other words $\tall=\texi$.
Informally, this suggests we can have any two, but not all three, of the following qualities:
\begin{enumerate*}
\item
$\tall$ decomposed into $\new$, as in condition~\ref{filter.new} of Definition~\ref{defn.filter} or
condition~\ref{sober.all} of Definition~\ref{defn.sober} or Corollary~\ref{corr.filter.point.property}.
\item
Filters with small support.
\item
Classical negation (so that prime filters are ultrafilters).
\end{enumerate*}
In this paper we build models of the untyped $\lambda$-calculus in an ambient lattice which is a nominal distributive lattice with $\tall$.
It is natural to ask whether we could promote this lattice to be a Boolean Algebra---in other words, ``Can we add negation?''.
We think the answer is probably ``No, at least for the paper in its current form.'': negation would either cost us condition~\ref{filter.new} of Definition~\ref{defn.filter} (and all that depends on it), or it would cost us Lemma~\ref{lemm.tall.unions}, with undesirable consequences for the proofs to follow, starting with Lemma~\ref{lemm.points.ppa.some.more.compact}.
\end{rmrk}

%%%%%%%%%%%%%%%%%%%%%%%%%%%%%%%%%
\subsubsection{Proof that $\points_\Pi$ is coherent}

We saw in Proposition~\ref{prop.points.sigma.bpp} that $\points_\Pi$ is a nominal $\sigma\bpp$-topological space in the sense of Definition~\ref{defn.nom.top.bpp}.
We now show that it is coherent (Definitions~\ref{defn.coherent} and~\ref{defn.coherent.bpp}) and sober (Definition~\ref{defn.sober}).

\begin{lemm}
\label{lemm.points.ppa.some.compact}
Suppose $U$ and $V$ are compact in $\points_\Pi$ and suppose $u\in|\idiomprg|$.
Then:
\begin{itemize*}
\item
$|\points_\Pi|$ is compact.
\item
$U\cap V$ is compact.
\item
$U[a\sm u]$ is compact.
\end{itemize*}
\end{lemm}
\begin{proof}
By Lemma~\ref{lemm.pp.U.compact} we may assume
$U=\bigcup\{\pp p\mid p\in\mathcal P\}$ and $V=\bigcup\{\pp q\mid q\in\mathcal Q\}$ for finite $\mathcal P,\mathcal Q\subseteq|\points_\Pi|$.
We consider each part in turn:
\begin{itemize*}
\item
$|\points_\Pi|=\pp{\varnothing}$ (which is a point, as noted in Remark~\ref{rmrk.varnothing.is.a.point}).
We use Lemma~\ref{lemm.pp.p.compact}.
\item
Using part~\ref{item.pp.cap} of Proposition~\ref{prop.pp.p.commute}
$U\cap V=\bigcup\{\pp{(p\tand q)}\mid p\in\mathcal P,q\in\mathcal Q\}$.
By Lemma~\ref{lemm.pp.p.compact} each $\pp{(p\tand q)}$ is compact, and a finite union of compact sets is compact.
\item
By Lemma~\ref{lemm.sub.bigcap}(\ref{sub.bigcup.sfs}) and Theorem~\ref{thrm.lam.pp.sm} $U[a\sm u]=\bigcup\{\pp{(p[a\lsm u])}\mid p\in\mathcal P\}$.
By Lemma~\ref{lemm.pp.p.compact} each $\pp{(p[a\lsm u])}$ is compact, and a finite union of compact sets is compact.
\qedhere\end{itemize*}
\end{proof}

\begin{lemm}
\label{lemm.points.ppa.some.more.compact}
Suppose $U$ is compact in $\points_\Pi$.
Then continuing Lemma~\ref{lemm.points.ppa.some.compact}:
\begin{itemize*}
\item
$\freshcap{a}U$ is compact.
\end{itemize*}
\end{lemm}
\begin{proof}
By Lemma~\ref{lemm.pp.U.compact} we may assume
$U=\bigcup\{\pp p\mid p\in\mathcal P\}$ for finite $\mathcal P\subseteq|\points_\Pi|$.
By Lemma~\ref{lemm.tall.unions} and by part~\ref{item.lam.pp.tall} of Proposition~\ref{prop.pp.p.commute} $\freshcap{a}U=\bigcup\{\pp{(\tall a.p)}\mid p\in\mathcal P\}$.
By Lemma~\ref{lemm.pp.p.compact} each $\pp{(\tall a.p)}$ is compact, and a finite union of compact sets is compact.
\end{proof}

\begin{lemm}
\label{lemm.points.ppa.compact}
Suppose $U$ and $V$ are compact in $\points_\Pi$.
Then continuing Lemma~\ref{lemm.points.ppa.some.more.compact}:
\begin{itemize*}
\item
$U\app V$ is compact.
\item
$V\ppa U$ is compact.
\end{itemize*}
\end{lemm}
\begin{proof}
By Lemma~\ref{lemm.pp.U.compact} we may assume
$U=\bigcup\{\pp p\mid p\in\mathcal P\}$ and $V=\bigcup\{\pp q\mid q\in\mathcal Q\}$ for finite $\mathcal P,\mathcal Q\subseteq|\points_\Pi|$.
We consider each part in turn:
\begin{itemize*}
\item
By Lemma~\ref{lemm.bigcup.app} and part~\ref{item.lam.pp.app} of Proposition~\ref{prop.pp.p.commute}
$U\app V=\bigcup\{\pp{(p\app q)}\mid p\in\mathcal P,\ q\in\mathcal Q\}$.
By Lemma~\ref{lemm.pp.p.compact} each $\pp{(p\app q)}$ is compact, and a finite union of compact sets is compact.
\item
By Corollary~\ref{corr.unzip} $V\ppa U=\bigcap_{q\in\mathcal Q}\bigcup_{p\in\mathcal P}\pp{(q\ppa p)}$.
By Lemma~\ref{lemm.pp.p.compact} each $\pp{(q\ppa p)}$ is compact; by Lemma~\ref{lemm.points.ppa.some.compact} a finite intersection of compact sets is compact; and a finite unions of compact sets is compact.
\qedhere\end{itemize*}
\end{proof}

\begin{prop}
\label{prop.points.coherent}
$\points_\Pi$ is a coherent nominal $\sigma\bpp$-topological space.
\end{prop}
\begin{proof}
$\points_\Pi$ is a nominal $\sigma\bpp$-topological space by Proposition~\ref{prop.points.sigma.bpp}.
It remains to check the additional coherence conditions of Definitions~\ref{defn.coherent} and~\ref{defn.coherent.bpp}.

\emph{We check the conditions of Definition~\ref{defn.coherent}.}\quad
Suppose $U$ and $V$ are compact in $\points_\Pi$.
By Lemma~\ref{lemm.points.ppa.compact} $U[a\sm u]$, $|\points_\Pi|$, and $U\cap V$ are compact.
By Lemma~\ref{lemm.tall.unions} $\freshcap{a}U$ is compact.

By construction in Definition~\ref{defn.points.top} every open set is a small-supported union of sets of the form $\pp p$ for $p\in|\points_\Pi|$.
We note by Lemma~\ref{lemm.pp.p.compact} that $\pp p$ is compact, so every open set is also a small-supported union of compact sets.

\emph{We check the conditions of Definition~\ref{defn.coherent.bpp}.}\quad
Suppose $U$ and $V$ are compact.
By Lemma~\ref{lemm.points.ppa.compact} $U\app V$ and $V\ppa U$ are compact.

We now check conditions~\ref{item.nom.top.app.sub} and~\ref{item.nom.top.ppa.sub} of Definition~\ref{defn.coherent.bpp} (validity of \rulefont{\sigma\app} and \rulefont{\sigma\ppa}).
By Lemma~\ref{lemm.pp.U.compact} we may assume $U=\bigcup\{\pp p\mid p\in\mathcal P\}$ and $V=\bigcup\{\pp q\mid q\in\mathcal Q\}$ for finite $\mathcal P,\mathcal Q\subseteq|\points_\Pi|$.
We reason as follows:
$$
\begin{array}{@{\hspace{-1.5em}}r@{\ }l@{\quad}l}
(U\app V)[a\sm u]=&
\bigcup\{(\pp p\app \pp q)[a\sm u]\mid p{\in}\mathcal P,\ q{\in}\mathcal Q\}
&\text{Lemmas~\ref{lemm.bigcup.app} \& \ref{lemm.sub.bigcap}(\ref{sub.bigcup.sfs})}
\\
=&
\bigcup\{\pp p[a\sm u]\app \pp q[a\sm u]\mid p{\in}\mathcal P,\ q{\in}\mathcal Q\}
&\text{Corollary~\ref{corr.filter.lam.pp.app.commute}}
\\
=&
\bigcup\{\pp p[a\sm u]\mid p{\in}\mathcal P\}\app \{\pp q[a\sm u]\mid q{\in}\mathcal Q\}
&\text{Lemma~\ref{lemm.bigcup.app}}
\\
=&
U[a\sm u]\app V[a\sm u]
&\text{Lemmas~\ref{lemm.bigcup.app} \& \ref{lemm.sub.bigcap}(\ref{sub.bigcup.sfs})}
\\[1.5ex]
(\prg_{\points_\Pi} b\ppa U)[a\sm u]=&
\bigl(\bigcup_{p{\in}\mathcal P}\pp b\ppa \pp p\bigr)[a\sm u]
&\text{Def~\ref{defn.points.top} \& Cor~\ref{corr.unzip}}
\\
=&
\bigcup_{p{\in}\mathcal P}((\pp b\ppa\pp p)[a\sm u])
&\text{Lemma~\ref{lemm.sub.bigcap}(\ref{sub.bigcup.sfs})}
\\
=&
\bigcup_{p{\in}\mathcal P}\pp b\ppa \pp p[a\sm u]
&\text{Corollary~\ref{corr.filter.lam.pp.app.commute}}
\\
=&
\prg_{\points_\Pi} b\ppa (U[a\sm u])
&\text{Cor~\ref{corr.unzip} \& Lem~\ref{lemm.sub.bigcap}(\ref{sub.bigcup.sfs})}
\end{array}
$$
\end{proof}

%%%%%%%%%%%%%%%%%%%%%%%%%%%%%%%%%
\subsubsection{Proof that $\points_\Pi$ is sober, thus spectral}

Suppose $p\in|\points_\Pi|$.
Recall $\qq p$ from Definition~\ref{defn.qq}; given a point we form the prime filter of compact open sets containing it.
Lemma~\ref{lemm.points.qq.injective} is in the same spirit as Corollary~\ref{corr.lam.pi.injective}:
\begin{lemm}
\label{lemm.points.qq.injective}
The map $p\in|\points_\Pi|\longmapsto \qq p\in|FG(\points_\Pi)|$ is injective.
\end{lemm}
\begin{proof}
Suppose $\qq p=\qq q$ and suppose $s\in|\idiom|$.
We note the following:
$$
\begin{array}{r@{\ }l@{\qquad}l}
s\in p
\liff&
p\in\pp{s\uparrowp{\Pi}}
&\text{Lemma~\ref{lemm.pp.sub}}
\\
\liff&
\pp{s\uparrowp{\Pi}}\in \qq p
&\text{Definition~\ref{defn.qq}}
\end{array}
$$
Similarly $s\in q\liff \pp{s\uparrowp{\Pi}}\in\qq q$.
It follows that $s\in p\liff s\in q$, so that $p=q$.
\end{proof}

\begin{lemm}
\label{lemm.intersection.P}
Suppose $\mathcal P\subseteq|\points_\Pi|$.
Then $\bigcap\{\pp p\mid p\in\mathcal P\}=\pp{(\bigcup\mathcal P)}$.

(By Proposition~\ref{prop.points.bigcap} $\bigcup\mathcal P\in|\points_\Pi|$.)
\end{lemm}
\begin{proof}
By sets calculations using Lemma~\ref{lemm.pp.sub}.
\end{proof}

\begin{lemm}
\label{lemm.point.in.U}
Suppose $\mathcal U\subseteq\ctop{\points_\Pi}$ is a prime filter.
Then for every $U\in\mathcal U$ there exists $p\in|\points_\Pi|$ such that $\pp p\subseteq U$ and $\pp p\in\mathcal U$.
\end{lemm}
\begin{proof}
By Lemma~\ref{lemm.pp.U.compact} there exists a finite set of points $\mathcal P\subseteq|\points_\Pi|$ such that $U=\bigcup\{\pp p\mid p\in\mathcal P\}$ and by Lemma~\ref{lemm.pp.p.compact} $\pp p$ is compact for every $p\in\mathcal P$.
We assumed that $\mathcal U$ is a prime filter (Definition~\ref{defn.prime.filter}) and it follows that there exists some $p\in\mathcal P$ such that $\pp p\in\mathcal U$.
\end{proof}

\begin{corr}
\label{corr.bigcap.U}
Suppose $\mathcal U\subseteq\ctop{\points_\Pi}$ is a prime filter.
Then there exists $q\in|\points_\Pi|$ such that $\pp q=\bigcap\mathcal U$ and $\mathcal U=\qq q$.
\end{corr}
\begin{proof}
Write $\mathcal P=\{p\mid \pp p\in\mathcal U\}$ and $\pp{\mathcal P}=\{\pp p\mid p\in\mathcal P\}$.
From Lemma~\ref{lemm.point.in.U} $\bigcap\mathcal U=\bigcap\pp{\mathcal P}$.
By Lemma~\ref{lemm.intersection.P} $\bigcap\pp{\mathcal P}=\pp{(\bigcup\mathcal P)}$.
We take $q=\bigcup\mathcal P$.
\end{proof}

\begin{prop}
\label{prop.points.sober}
$\points_\Pi$ is sober (Definition~\ref{defn.sober}).
\end{prop}
\begin{proof}
By Lemma~\ref{lemm.coherent.iff} it suffices to check two conditions:
\begin{enumerate}
\item
The map $p\in|\points_\Pi|\longmapsto\qq p\in |FG(\points_\Pi)|$ is a bijection.
\item
If $p\in |\points_\Pi|$ and $U\in\ctop{\points_\Pi}$ then $\New{b}p\in (b\ a)\act U$ implies $p\in\freshcap{a}U$.
\end{enumerate}
Condition~1 is Lemma~\ref{lemm.points.qq.injective} and Corollary~\ref{corr.bigcap.U}.

For condition~2, by Lemma~\ref{lemm.pp.U.compact} $U=\bigcup\{\pp q\mid q\in\mathcal Q\}$ for some finite $\mathcal Q\subseteq |\points_\Pi|$.
Choose fresh $b$ (so $b\#U$ and $b\#p$\footnote{Recall from Lemma~\ref{lemm.fresh.point.check}(2) that $p\in|\points_\Pi|$ has small support here.} and $\Forall{q{\in}\mathcal Q}b\#q$).
It is a fact of sets that $p\in (b\ a)\act \pp q$ for one $q\in\mathcal Q$, so that by Theorem~\ref{thrm.new.equiv} also $\New{b}p\in(b\ a)\act\pp q$.
By Corollary~\ref{corr.filter.point.property} $p\in\pp{(\tall a.q)}\stackrel{P\ref{prop.pp.p.commute}(\ref{item.lam.pp.tall})}=\freshcap{a}\pp q$.
Now $\pp q\subseteq U$ and it follows that $p\in\freshcap{a}U$.
\end{proof}

\begin{thrm}
\label{thrm.points.spectral}
$\points_\Pi$ from Definition~\ref{defn.points.top} is indeed a nominal spectral space with $\bpp$.
\end{thrm}
\begin{proof}
By Proposition~\ref{prop.points.sigma.bpp} $\points_\Pi$ is a nominal $\sigma\bpp$-topological space.
By Proposition~\ref{prop.points.coherent} it is coherent and by Proposition~\ref{prop.points.sober} it is sober.
It remains to check that it is impredicative; that is, that $\prg_{\points_\Pi}=\pp{\text{-}}$ is a
morphism of $\sigma$-algebras (Definition~\ref{defn.morphism.sigma.alg}) from $\idiom$ to $\ctop{\points_\Pi}$.
This is just Theorem~\ref{thrm.lam.pp.sm}.
\end{proof}

%%%%%%%%%%%%%%%%%%%%%%%%%%%%%%%%%%%%%%%%%%
\subsection{Logical properties of the topology, and completeness}
\label{subsect.toplogical.completeness}

Recall that at the start of this section we fixed an idiom $\idiom$ (Definition~\ref{defn.idiom}) and a $\lambda$-reduction theory $\Pi$ over $\idiom$ (Definition~\ref{defn.lambda.reduction.theory}).

\begin{rmrk}
\label{rmrk.on.the.canonical.model}
By Proposition~\ref{prop.points.coherent} $\points_\Pi$ is coherent.
Thus from Definition~\ref{defn.G} and Theorem~\ref{thrm.T.to.G.obj} we have that $G(\points_\Pi)$ is an nominal distributive lattice with $\tall$; it consists of compact open sets in $\points_\Pi$ ordered by subset inclusion.
By Lemma~\ref{lemm.pp.U.compact}, each compact open set is a finite union of sets of the form $\pp p$ for $p\in|\points_\Pi|$.
\end{rmrk}

\begin{nttn}
\label{nttn.canonical.model}
We call $G(\points_\Pi)$ the \deffont{canonical model}.
\end{nttn}
We now set about proving Theorem~\ref{thrm.Pi.completeness}, which uses $G(\points_\Pi)$ to prove completeness---the converse direction to soundness from Theorem~\ref{thrm.lambda.soundness}.

Recall from Definition~\ref{defn.ddenot} the definition of $\pdenot{s}$ and recall from Definition~\ref{defn.lamtrm.pi} that $\lamtrm$ is the set of $\lambda$-terms.
Suppose $\idiom=\lamtrm$.
\begin{lemm}
\label{lemm.pdenot.pp}
$\pdenot{s}=\pp{s\uparrowp{\Pi}}$
and as a corollary $\pdenot{s}\in|G(\points_\Pi)|$.
\end{lemm}
\begin{proof}
By a routine induction on $\lambda$-terms:
\begin{itemize*}
\item
$\pdenot{a}
\stackrel{\text{Def~\ref{defn.ddenot}}}{=}
\prg_{G(\points_\Pi)} a
\stackrel{\text{Def~\ref{defn.points.top}}}{=}
\pp{a\uparrowp{\Pi}}$.
\item
$\pdenot{s's}
\stackrel{\text{Def~\ref{defn.ddenot}}}{=}
\pdenot{s'}\app\pdenot{s}
\stackrel{\text{ind hyp}}{=}
\pp{s'\uparrowp{\Pi}}\app \pp{s\uparrowp{\Pi}}
\stackrel{\text{Cor~\ref{corr.pp.commute.lam}}}{=}
\pp{(s's)\uparrowp{\Pi}}$.
\item
$\pdenot{\lam{a}s}
\stackrel{\text{Def~\ref{defn.ddenot}}}{=}
\tlam a.\pdenot{s}
\stackrel{\text{ind hyp}}{=}
\tlam a.(\pp{s\uparrowp{\Pi}})
\stackrel{\text{Cor~\ref{corr.pp.commute.lam}}}{=}
\pp{(\lam{a}s)\uparrowp{\Pi}}$.
\end{itemize*}
The corollary follows from Lemma~\ref{lemm.pp.p.compact}, since $|G(\points_\Pi)|$ is by definition the set of compact open sets of $\points_\Pi$.
\end{proof}

Recall from Notation~\ref{nttn.Pi.cent} and Definition~\ref{defn.ddenot} the notations $\Pi\cent s\to t$ and $G(\points_\Pi)\ment\Pi$.
\begin{prop}
\label{prop.points.ment}
$G(\points_\Pi)\ment\Pi$.
\end{prop}
\begin{proof}
Unpacking Definition~\ref{defn.ddenot}, we must show that $(s\to t)\in\Pi$ implies $\pdenot{s}\leq\pdenot{t}$, where $\leq$ means $\subseteq$.

So suppose $(s\to t)\in\Pi$.
By Notation~\ref{nttn.Pi.cent} this means precisely $s\arrowp{\Pi}t$ and it follows by Lemma~\ref{lemm.lam.pp.inj} that $\pp{s\uparrowp{\Pi}}\subseteq \pp{t\uparrowp{\Pi}}$, so by Lemma~\ref{lemm.pdenot.pp} $\pdenot{s}\subseteq\pdenot{t}$ as required.
\end{proof}

Recall $\mathcal T\cent s\to t$ from Notation~\ref{nttn.Pi.cent} and $\mathcal T\ment s\leq t$ from Definition~\ref{defn.ddenot}:
\begin{thrm}[Completeness]
\label{thrm.Pi.completeness}
Suppose $\mathcal T$ is a set of reduction axioms.
Then:
\begin{frameqn}
\mathcal T\cent s\to t 
\quad\text{if and only if}\quad
\mathcal T\ment s\leq t.
\end{frameqn}
\end{thrm}
\begin{proof}
The left-to-right implication is Theorem~\ref{thrm.lambda.soundness}.

Now suppose $\mathcal T\ment s\leq t$ and write $\Pi{=}\f{rew}(\mathcal T)$ (Definition~\ref{defn.rew.equ}).
By Proposition~\ref{prop.points.ment} $G(\points_\Pi)\ment\Pi$ so (since $\mathcal T{\subseteq}\Pi$) $\pdenot{s}\subseteq \pdenot{t}$.
By Lemma~\ref{lemm.pdenot.pp} $\pdenot{s}=\pp{s\uparrowp{\Pi}}$ and $\pdenot{t}=\pp{t\uparrowp{\Pi}}$.
By Lemma~\ref{lemm.lam.pp.inj} $\Pi\cent s{\to}t$, thus by Lemma~\ref{lemm.T.Pi} $\mathcal T\cent s{\to}t$ as required.
\end{proof}

\subsection{Interlude: an interesting disconnect}
\label{subsection.interlude.ii}

The duality theorem from Theorem~\ref{thrm.equivalence.pp} is more general than the completeness theorem needs it to be.
The completeness result of Theorem~\ref{thrm.Pi.completeness} is based on $\points_\Pi\in\inspectapp$ from Definition~\ref{defn.pi.point}.
Although $\points_\Pi$ is a spectral space (Theorem~\ref{thrm.points.spectral})
it has much more structure than that.
For instance:
\begin{itemize*}
\item
It is replete (Definition~\ref{defn.replete}).
\item
It has an existential quantifier, as we note later in Definition~\ref{defn.jj}.
\item
Perhaps most usefully, points in $\points_\Pi$ have small support (Lemma~\ref{lemm.fresh.point.check}(\ref{point.finsupp})).
As we noted in Remarks~\ref{rmrk.classes.of.amgis} and~\ref{rmrk.p.not.finite.support}, the full generality of our duality theorem encompasses spaces whose points do not necessarily have small support.
\end{itemize*}
Could we obtain a more specific duality result for structures that have more of the structure apparent in $\points_\Pi$?

We probably could.
However, we do not do it in this paper.
$\indiapp/\inspectapp$ have the minimal structure we need to interpret the $\lambda$-calculus and to carry out a filter-based duality proof.
The less structure we impose, the more general our duality result,\footnote{Broadly speaking, within a given class of structures, the less structure we assume the more challenging the duality result is to prove.  If duality theory were a competitive sport then it would be like golf: the lower your score the better your game.} and we have more representations.

But in the completeness result we are happy if the canonical model has more structure, since it suggests more programming and reasoning constructs; an existential quantifier, for example, suggests that the ambient meta-logic implicit in $\points_\Pi$ permits unconstrained search.
We do not care about any other structures because we have built one \emph{particular} concrete structure and having built it, we want to obtain as many bells and whistles from it for free as possible.\footnote{So canonical models are like tennis: a higher score is better.}

So on the one hand we have a world where less structure is good, because fewer assumptions means stronger theorems that are valid for a larger class of entities (provided we can still build the things we want to in those entities, i.e. interpret the $\lambda$-calculus, which we can), and on the other hand we have a world where more structure is good, because it gives us more tools to actually do things.

The apparent disconnect comes from a difference between two styles, each of which is optimised for its own purpose.

The general trend in this paper is a progression from the abstract and general, like $\inspectapp$, to the relatively more concrete and specific, like $\points_\Pi$.

%%%%%%%%%%%%%%%%%%%%%%%%%%%%%%%%%%%%%%
\section{Conclusions}

The semantics of this paper has the moderately unusual feature of being \emph{absolute}, meaning that variables are interpreted directly in the denotation and there is no (Tarski-style) valuation.

The reader may find this takes some getting used to, but it is actually simple and natural.

What corresponds to valuations is the $\sigma$-action, which allows us to take some $x$ and `evaluate' $a$ to $u$ in $x$ by forming $x[a\sm u]$.
This is an abstract nominal algebraic property of $x$; it is characterised by axioms (Figure~\ref{fig.nom.sigma}).
We do not necessarily have access to the internal structure of $x$.

However, we can certainly build concrete $\sigma$-algebras:
Two examples are $\lambda$-term syntax $\lamtrm$ from Definition~\ref{defn.lamtrm} and the canonical model $\points_\Pi$ from Subsection~\ref{subsect.toplogical.completeness}.
Another example is how we move from $\amgis$-algebra structure to $\sigma$-algebra structure (and back) using nominal powersets (Definition~\ref{defn.sub.sets}).
Nominal powersets have intersections, unions, and complements, and by combining all of these things we can interpret $\forall$ (Definition~\ref{defn.nu.U}).

In fact, it turns out that with a little more effort and just a bit more structure we can interpret application and $\lambda$ too.
This brings us on to another unusual feature of our topological semantics: it is \emph{purely} sets-based.

Algebraic (dually: topological) semantics for the $\lambda$-calculus exist, but
our semantics is this in a different and stronger sense than usual, because \emph{everything} is interpreted algebraically (dually: topologically), including variables, substitution, and $\lambda$-abstraction.

This paper gives a panoramic view of the interaction between nominal foundations and the $\lambda$-calculus.
This gives us something that shorter papers might not do so well: a feel for the overall point of view, and how the parts of the puzzle fit together.

%%%%%%%%%%%%%%%%%%%%%%%%
\subsection{Related work}

\subsubsection{Algebraic semantics}

Algebraic semantics for logics or calculi with binding include polyadic algebras \cite[Part~II]{halmos:algl}, cylindric algebras \cite{henkin:cyla}, and Lambda Abstraction Algebras \cite{salibra:algmlc}.
As far as we know, what is done in this paper has not been done in any of these (but see below).

We can suggest technical reasons for this.
Consider for instance the treatment of substitution in this paper.

For us substitution exists independently from $\beta$-reduction---this is the notion of $\sigma$-algebra from Subsection~\ref{subsect.sigma.amgis}.
This is important for our constructions to work.
For instance, Subsections~\ref{subsect.sigma.to.amgis} and~\ref{subsect.amgis.to.sigma} do not assume $\lambda$ and application, they only assume $\sigma$ and $\amgis$.
This is reflected in commutation results like Lemma~\ref{lemm.dup} and Theorem~\ref{thrm.lam.pp.sm}.

The commutations for $\lambda$ are later, and much harder: Lemmas~\ref{lemm.filter.pp.app.commute} and~\ref{lemm.filter.adjoint.compat} for $\indiapp$ and Corollary~\ref{corr.pp.commute.lam} for $\points_\Pi$.

It not obvious how substitution \emph{on its own} could be axiomatised without permutations and freshness side-conditions, i.e. without nominal algebra.
LAAs do not do this, neither do cylindric algebras.
Polyadic algebras \emph{assume} a monoid of substitutions.
This is tantalisingly close to finite permutations, but without their invertibility.

By enriching the foundation with names and binding, nominal techniques allow us to express new kinds of algebraic structures---such as substitution and its dual amgis (Figure~\ref{fig.nom.sigma}), and universal quantification (Figure~\ref{fig.genhnu}).
We exploited that this paper to break constructions up into more manageable parts:
\begin{itemize*}
\item
we split $\lambda$ into $\forall$ and $\ppa$ (Notation~\ref{nttn.lambda}),
\item
$\forall$ into $\new$ (see the discussion of condition~\ref{filter.new} of Definition~\ref{defn.filter} in Subsection~\ref{subsect.list.of.technical.features}),
\item
$\beta$-reduction into $\app$, $\ppa$, and $\sigma$ (Proposition~\ref{prop.lambda.beta.eta}), and then we split
\item
$\sigma$ (substitution) itself into permutation and freshness (Definitions~\ref{defn.sub.algebra} or~\ref{defn.sub.sets}).
\end{itemize*}
Consequences of this include that the theory of $\forall$ becomes partially independent of the theory of $\sigma$.\footnote{We see this in various places in this paper, and since this paper is large we point two of them out:

Condition~\ref{filter.new} of Definition~\ref{defn.filter} imposes a condition on filters relating $\tall$ to $\new$ and not to $\sigma$.
This is not the condition that one would expect from reading Definition~\ref{defn.nu.U}, but it is necessary, as discussed in Remark~\ref{rmrk.why.filter.new}.

Also, Figure~\ref{fig.nom.sigma} does not mention $\forall$ and Figure~\ref{fig.genhnu} (the algebraic presentation of $\forall$; the body of the paper uses nominal lattices instead, but this is equivalent) does not mention $\sigma$.
There is a connection between $\forall$ and $\sigma$ of course, which is expressed explicitly as a further axiom: see the notion of \emph{compatibility} in Definition~\ref{defn.fresh.continuous}, whose validity requires non-trivial work to check in Lemma~\ref{lemm.technical} and Proposition~\ref{prop.all.sub.commute}.}
So breaking constructions up into more manageable parts does not mean here that familiar proofs are carried out in smaller steps: it means that new kinds of proofs are made possible by assembling these parts into new kinds of structures with usefully different proof-dependencies to those we are used to seeing.

Representation theorems exist for cylindric algebras; for instance \cite{monk:reptca} gives a representation theorem for cylindric algebras, and \cite{pigozzi:lamaar} gives one for LAAs.
In both cases, an algebra is represented concretely as a set of valuations on the variables (on the indexes; the things that correspond to atoms in this paper).

This is typical.
Representation theorems for cylindric-algebra-style systems do all seem to use something corresponding to a set of valuations.
It works, but it is
a soundness and completeness proof with respect to Tarski-style semantics.
There is no duality result.

The only duality results we know of for logic were undertaken by Forssel \cite{forssell:firold} and by the first author \cite{gabbay:stodfo}.
See \cite{gabbay:stodfo} for a comparison of the two.

A \emph{Stone Representation} has been given for Lambda Abstraction Algebras.
This is a factorisation result in the style of the HSP theorem (also known as Birkhoff's theorem):
every LAA can be factored as a product of irreducible LAAs.
The factorisations are identified by \emph{central elements} of the algebra \cite{salibra:appual}.
The LAA is never represented as anything resembling a Stone space, and there is no duality.

In passing, we note that the HSP part of the LAA result also follows in this paper (for $\indiapp$) off-the-shelf, by the \emph{nominal} HSPA theorem \cite{gabbay:nomahs,gabbay:nomtnl}.
In other words, we get some of \cite{salibra:appual} for free, just by virtue of being nominal and using nominal algebra.
The HSPA factorisation is slightly better than the HSP factorisation (because it has an A in it: for atoms-abstraction).
Investigating any extra power this gives what that might mean for (nominal) LAAs, is an open problem.

We mention also \cite{kurz:unians}.
This is an attempt to encode what makes nominal techniques work using many-sorted universal algebra.
Equivariance, however, gets lost in the translation; a similar phenomenon was noted in \cite{gabbay:pnlthf} translating permissive-nominal logic (a first-order generalisation of nominal algebra) to higher-order logic.

\subsubsection{Absolute semantics}

Absolute semantics have appeared before.
Lambda-abstraction algebras (for the $\lambda$-calculus) and cylindric algebras and polyadic algebras (for first-order logic) are absolute.
Selinger made a case for using absolute semantics for the $\lambda$-calculus in \cite{selinger:lamca} (see Subsection~2.2); a line of thought echoed by the first author with Mulligan in \cite{gabbay:nomhss}.

Yet absolute semantics have not caught on.
We are inclined to believe that this is because the mathematical foundations to support it were not in place before, but they are now.
Now that we have nominal techniques we can make a lot of things work that would not work before.

Without nominal techniques things we use repeatedly in the current paper, like small support, freshness side-conditions, equivariance, the $\new$-quantifier, and even $\alpha$-equivalence, become challenging in various technical fiddly ways, and even the statements of  some properties become practically impossible to even write out.
For instance, how we might we render condition~\ref{filter.new} of Definition~\ref{defn.filter}, or condition~\ref{item.point.lam} of Definition~\ref{defn.pi.point}, or Definition~\ref{defn.qasmu}---to choose three out of many possible examples---without small support, freshness, and the $\new$-quantifier?
We would probably have to invent them first.

We also mention Kit Fine's \emph{arbitrary objects} \cite{fine:reaao} as an instance of a similar impulse towards absolute semantics, coming from philosophy.

There are precedents for this paper in the first author's work; indeed this paper is based on them.
The nominal semantics and duality results for first-order logic in \cite{gabbay:stodfo} and \cite{gabbay:semooc} are absolute, and are very much in the style and research programme of this paper.

Nominal algebra has helped us to reduce mathematical overhead and to simplify some technical manipulations that are otherwise all too easy to get bogged down in.
This is just what any good mathematical toolbox or foundation should do.

\subsubsection{$\eta$-expansion}
\label{subsect.extensional}

In Proposition~\ref{prop.lambda.beta.eta} we saw $\beta$-reduction and $\eta$-expansion appear spontaneously as corollaries of adjoint properties.
So our notion of $\lambda$-reduction theory is more general than an extensional $\lambda$-equality theory because reductions can go one way and not the other, but it is also more specific than just any set of reductions because it must contain $\eta$-expansion.

The reader used to seeing $\eta$ as a contraction rule in rewrite systems might be interested in a thread of publications by Barry Jay and Neil Ghani, which argues in favour of $\eta$-expansion from the point of view of rewriting, for better confluence and other properties which they list.
See \cite{ghani:virea} and \cite{ghani:etaedt}.

For us too, expansion rather than contraction seems to be the natural primitive.

We believe that we can remove $\eta$-expansion at some cost in complexity in the models.
We do this by considering \emph{two} application operations $\app'$ and $\app$; one intensional and one extensional; a detailed development is the topic of \cite{gabbay:simcmt}.

\subsubsection{Surjective pairing}
\label{subsect.surj.pairing}

To extend the $\lambda$-calculus with \deffont{surjective pairing} we add constants $\f{pair}$, $\f{proj}_1$, and $\f{proj}_2$ and equations $\f{proj}_1(\f{pair} s_1 s_2)=s_1$ and $\f{proj}_2(\f{pair} s_1 s_2)=s_2$ and \emph{surjectivity} $\f{pair}((\f{proj}_1 s)(\f{proj}_2 s))=s$ as axioms to the untyped $\lambda$-calculus.\footnote{\emph{Pairing} can be implemented definitionally as $\lambda a.\lambda b.\lambda f. f a b$ with first and second projections $\lambda a.a\,\f{true}$ and $\lambda a.a\,\f{false}$ where $\f{true}=\lambda a.\lambda b.a$ and $\f{false}=\lambda a.\lambda b.b$---but this is not \emph{surjective}; not every lambda term is a pair.
Surjective pairing is surjective by the surjectivity axiom above.
This is a proper extension of the $\lambda$-calculus \cite{barendregt:paiwcr}.}

Surjective pairing is of interest because the categorial semantics of the untyped $\lambda$-calculus naturally generates semantics for the $\lambda$-calculus with surjective pairing; nice discussions are in \cite[subsection~2.5]{scott:somacc} and (for a more detailed exposition) \cite[Part~I section~17]{lambek:inthoc}.

It seems straightforward to add $\f{pair}$, $\f{proj}_1$, and $\f{proj}_2$ with axioms as above to $\indiapp$ (Definitions~\ref{defn.D.impredicative} and~\ref{defn.FOLeq.pp}).
Determining what corresponding additional structure this would translate to in the dual category $\inspectapp$ (Definition~\ref{defn.nom.top.spectral.bpp}) is future work.

\subsubsection{Previous treatment of $\lambda$-calculus by the authors}

In \cite{gabbay:capasn,gabbay:oneaah} the first author and Mathijssen developed \emph{nominal algebra}
and axiomatised substitution and first-order logic, with completeness proofs.
Journal versions are \cite{gabbay:capasn-jv,gabbay:oneaah-jv}.

An axiomatisation, again with completeness proofs, for the $\lambda$-calculus followed in \cite{gabbay:lamcna,gabbay:nomalc}.

So the $\sigma$-axioms which appear in Figure~\ref{fig.nom.sigma} are taken from \cite{gabbay:capasn}, and the axioms for $\beta$ and $\eta$ are descended from \cite{gabbay:lamcna,gabbay:nomalc}.

In \cite{gabbay:stodfo} we applied duality theory to in nominal sets to the axiomatisation of \cite{gabbay:oneaah,gabbay:oneaah-jv}.
The main conceptual challenge (aside from the inherent difficulty of duality proofs) was to invent $\amgis$-algebras.\footnote{This took a couple of years: once the first author understood that for a duality result, a \emph{dual} to $\sigma$ was needed, the paper was easy to write.
At least, for a certain highly technical value of `easy'.}
The $\amgis$-axioms of Figure~\ref{fig.nom.sigma} are from \cite{gabbay:stodfo}.
We have taken this further in \cite{gabbay:semooc}.

This paper carries out a similar project to \cite{gabbay:stodfo}, but for the $\lambda$-calculus.
This has been a tougher target than first-order logic, which is unsurprising.
The main conceptual difficulty of this paper over the previous work is the treatment of application and $\lambda$ using adjoints and the logical quantifier.
The ideas for this are from \cite{gabbay:simcks} (see Figure~2, where $\ppa$ is written $\triangleright$).
The similarity with \cite{gabbay:simcks} is somewhat hidden just because it was written in a `modal logic' style.
That style has been replaced in this paper by the nominal foundations.

In summary, and at least in principle, this paper just combines \cite{gabbay:capasn}, \cite{gabbay:lamcna}, and \cite{gabbay:simcks} with \cite{gabbay:stodfo}.
(What could be simpler or more natural?)

\subsubsection{No conflict with topological incompleteness results}
\label{subsect.no.conflict}

The best-known models of the untyped $\lambda$-calculus are Scott's domain models and generalisations: graph semantics; filter semantics; stable semantics; strongly stable semantics; and so on.
An excellent discussion with references---an annotated bibliography and survey, in fact---appears in \cite{salibra:contlc} between Theorems~4.5 and~4.6.

These are all ordered structures, and this is key, since the idea is to reduce the function space using continuity conditions.

These semantics are all incomplete.
That is, domains-based denotational semantics proved the $\lambda$-calculus consistent, but results like \cite[Theorems~3.5 and~4.9]{salibra:topioi} proved that this is not the whole story: see also \cite{salibra:contlc}.\footnote{Page~2 of \cite{salibra:topioi} includes a brief but comprehensive history of such results.  The first incompleteness result was given in \cite{honsell:appttl} for the continuous semantics (Scott's construction).
This was followed by several generalisations.
Salibra's treatment has the benefit of covering a range of semantics in a uniform way.
}

The reader familiar with this literature and who has seen e.g. Theorem~3.5 of \cite{salibra:topioi} might be puzzled by Theorem~\ref{thrm.Pi.completeness}: the former states that no semantics in terms of partially-ordered models with a bottom element can be complete, whereas the latter claims to prove completeness for a semantics based on $\indiapp$, and an object of $\indiapp$ is a lattice and has a bottom element $\tbot$.

However, nothing insists that $\tbot$ should be a program.
That is, in the notation of Notation~\ref{nttn.impredicative.D}, it is perfectly possible that $\ns D\in\indiapp$ and $\tbot\not\in\prg\ns D$.

This illustrates that $\ns D$ is a logical structure---its dual is topological---and just a \emph{subset} $\prg\ns D$ is deemed to be `computational'.
The models of the $\lambda$-calculus live in $\prg\ns D\subseteq\ns D$, and $\prg\ns D$ need not be closed under meets or joins.

The formal sense in which this is intended is just that programs are the things that can be substituted for by the $\sigma$-action; so intuitively atoms in $\ns D\in\indiapp$ `range over' programs (more on this in Subsection~\ref{subsect.cardinality}).
In the light of this reading of the definitions, Definition~\ref{defn.replete} calls $\ns D$ \emph{replete} when its programs are Turing complete.

It remains to discover whether there exists a $\lambda$-equality theory such that if $\ns D\in\indiapp$ is a model of that theory then it can have no non-trivial order on its \emph{programs} (so if $x,y\in|\prg\ns D|$ then if $x\leq y$ then $y\leq x$).

\subsubsection{Game models}

A game semantics for the untyped $\lambda$-calculus is given in \cite{ker:inngmu-thesis,ker:inngmu} and is complete (unlike the models discussed in Subsection~\ref{subsect.no.conflict}), and \emph{also} abstract. %a term is identified with a strategy.

Let us briefly compare and contrast the following three complete semantics for the untyped $\lambda$-calculus: terms quotiented by equivalence, the games semantics, and the topological duality models of this paper.
Of these three, the last two are abstract, and informally speaking they are listed above in \emph{increasing} order of size---that is, a term model is a fairly small entity, a games semantics is only slightly larger, and the duality models may be very large indeed (see Subsection~\ref{subsect.cardinality}).
Also broadly speaking, if we wanted to compute on a $\lambda$-term then a term or games model would probably be the most efficient---if we wanted to investigate logical non-computational properties of models, then the models of this paper might be helpful.

Key technical moments in \cite{ker:inngmu-thesis} are Theorem~1.5.4 on page~13, the definition of an \emph{effectively almost-everywhere copycat} (EAC) strategy on page~57, and the discussion opening Subsection~5.2 on page~106.
These suggest that the game semantics works essentially by identifying the abstract properties that B\"ohm trees representing $\lambda$-terms possess such that these trees look game semantics. 

The precise relationship with the models in this paper is unclear, except perhaps to observe that the games model is `bottom-up', building semantics by abstracting from concrete tree structures, whereas our semantics is more `top-down' and algebraic, building semantics by using axioms and topologies to carve out well-behaved subspaces of huge powersets.
If there is any point of contact, the natural place to look for it would be in Section~\ref{sect.lambda.representation}, where we concretely build models $\points_\Pi$ out of syntax (see Subsection~\ref{subsect.fine.structure}).

\subsubsection{Sheaves}

We impose a topology on a set to reduce the size of the function-space by restricting to continuous functions.
Sheaves do much the same thing, but in more generality.

Nominal sets form a category which admits a sheaf presentation (a discussion specific to nominal techniques is in \cite{gabbay:fountl}).
Simplifying a little, this amounts to observing that equivariance (commuting with the permutation action) can be represented as a generalised `continuity' condition.
There is no need to stop there.
We could try to make `continuity' represent, for instance, \emph{compatibility} conditions such as \rulefont{\sigma\app} from Figure~\ref{fig.app.compatible}.

This is what is done by the \emph{Topological representation of the $\lambda$-calculus} considered in \cite{awodey:toprlc}.
Examining equation (15) of the paper we see that, essentially, an open set is a set of substitution instances of evaluations from variables to terms.
(The calculations are given only for the simply-typed $\lambda$-calculus.)
Continuity ensures that function application commutes with substitution, i.e. \rulefont{\sigma\app}.

The closest thing to the construction of \cite{awodey:toprlc} is the construction of $\points_\Pi$ in Definition~\ref{defn.pi.point}.

Both are representations of the (simply-typed) $\lambda$-calculus, and both are topological, but beyond that we see little resemblance between the two constructions.
Our consistency conditions are axiomatic, and we use the topology to do logic and so to break $\lambda$-down into $\forall$ and $\ppa$.
Substitution $\sigma$ is managed by axioms.

\subsubsection{Proof theory}

The design of the \emph{combination operator} $\bpp$ from Definitions~\ref{defn.bpp} and~\ref{defn.nom.top.bpp} goes back to the Kripke-style models of the untyped $\lambda$-calculus from \cite{gabbay:simcmt,mgabbay:simmtu}.
A discussion with specific references is in Remark~\ref{rmrk.where.from}.

These Kripke-style models were developed from a proof-theory~\cite{mgabbay:lambdacut} which views $\lambda$-reduction as a logical derivation rule, and gives it a cut-free sequent-style derivation system.
Thus this paper comes (in some sense) full circle when in Notation~\ref{nttn.lambda} we interpret $\lambda$ semantically in a way that explicitly contains a logical quantifier $\forall$.

The current paper, extensive as it is, is also embedded in and consistent with a broader research context.

\subsubsection{In what universe does this paper take place?}

The points built in Theorem~\ref{thrm.maxfilt.zorn} do not have small support, and in Definition~\ref{defn.bus.algebra} we assume a set with a permutation action but not necessarily a nominal set.
Thus, this paper does not take place entirely in the topos of nominal sets; we do whatever is convenient to get the results we need and do not commit to any specific logic when we get them, even though our main results can be stated entirely in the nominal sets universe.
In this we are being typical mathematicians, reasoning freely in English about informally but precisely specified mathematical objects.\footnote{Something similar happens in category theory when we talk about `the category of all sets'; what does that live in? This is usually left unspecified, which is generally fine, or at least, is generally not objected to.}

\subsubsection{Discussion of some decisions}
\label{subsubsect.design.decisions}

While formulating the mathematics in this paper we faced certain high-level design decisions.
For the reader's convenience we briefly mention some of these decisions along with our reasons for making them:
\begin{enumerate*}
\item
\emph{When building a nominal topological space it is natural that open sets (representing predicates) should be small-supported.  Should points also be small-supported?}

Not in general.
Intuitively this is because we want Zorn-style arguments building maximally consistent sets of predicates to be possible, such as appear in Theorem~\ref{thrm.maxfilt.zorn} in this paper and in Theorem~7.1.20 of \cite{gabbay:semooc}.
See Remark~\ref{rmrk.classes.of.amgis}.

However, the points of the canonical model $\points_\Pi$ (Subsection~\ref{subsect.toplogical.completeness}) \emph{are} small supported.
See the discussion in~\ref{rmrk.any.two}.
\item
\emph{When considering compactness, should covering sets be small-supported?  Or should they be strictly small-supported?}

The more convenient notion seems to be `small-supported'; see Definition~\ref{defn.n.closed}.
However, Proposition~\ref{prop.extra.filter} notes that in one important special case, the two are equivalent.
\item
\emph{What is the correct notion of `point' in the presence of a nominal universal quantifier?}

This paper contains two distinct notions of `point': one based on filters which is designed for the duality construction; the other based on sets of syntax closed under reduction which is designed for the canonical model.
Both require an non-evident condition to account for $\tall$.
See Definition~\ref{defn.filter} condition~\ref{filter.new}, and Definition~\ref{defn.pi.point} condition~\ref{item.point.lam}.
\item
\emph{In the presence of an amgis-action, what is an appropriate notion of freshness?}

The two notions of point give distinct answers to this question.
In Definition~\ref{defn.filter} the question is ill-formed because points are not necessarily small-supported.
In Definition~\ref{defn.pi.point} we assume small-supported, and the answer is expressed by Definition~\ref{defn.afreshp}, Definition~\ref{defn.pi.point} condition~\ref{item.point.lam}, and Proposition~\ref{prop.fresh.point}.
\end{enumerate*}

%%%%%%%%%%%%%%%%%%%%%%%%
\subsection{Future work}

\subsubsection{Fine structure of the canonical model}
\label{subsect.fine.structure}

In Subsection~\ref{subsection.interlude.ii} we noted that the canonical model $\points_\Pi$ has plenty of structure.
It remains to explore that structure: $\points_\Pi$ is a lattice and so contains a logic.
We know this has interesting structure, investigated from Subsection~\ref{subsect.Pi.points}.
That does not exhaust the possibilities: Appendix~\ref{subsect.existential} %and~\ref{subsect.existential.ii}
explores an existential quantifier on $\points_\Pi$, and this invites us to ask what the full logic of the canonical model $\points_\Pi$ is, and might we use it to investigate the $\lambda$-calculus.

Proposition~\ref{prop.points.Pi.sigma.algebra} notes that $\points_\Pi$ also has a $\sigma$-action, which we characterise in different ways in Subsection~\ref{subsect.two.chars.lsm}; the characterisation in Lemma~\ref{lemm.lsm.id} seems particularly appealing.
As we note in the body of the paper, there is probably a general theory here: a way of, given a $\sigma$-action on $\ns X$, building a $\sigma$-action on the nominal powerset of $\ns X$.
Such a theory was already undertaken in \cite{gabbay:stusun}, where constructions were applied to models of Fraenkel-Mostowski set theory; thus generating a huge class of huge $\sigma$-algebras, since there are many models of FM sets and many sets in each model.
The construction in Proposition~\ref{prop.points.Pi.sigma.algebra} suggests the possibility of a cleaner and/or alternative development of similar ideas.\footnote{\dots and this is exciting.  Most of this paper works by building various $\sigma$-algebras over relatively simple nominal algebraic structures like sets of points.  What more could be achieved if we gave ourselves an entire mathematical foundation structure to play with?}

$\points_\Pi$ is a remarkable mathematical object which may be somewhat overshadowed in this paper by the rest of the material.
The proofs in Section~\ref{sect.lambda.representation} have `coincidences' (two of several examples are the existence of the $\sigma$-action mentioned in the previous paragraph, and Lemma~\ref{lemm.tall.unions}), and some strong hints towards general theories (such as that noted in Subsection~\ref{subsect.additional.lemmas} of `nominal adjoints', that is; adjoint relations subject to freshness side-conditions).
These suggest richnesses of behaviour which may merit further examination.

\subsubsection{Weaken the axioms}

We have used nominal lattices to give semantics to the $\lambda$-calculus.
Part of our axiomatisation includes the conditions of Figures~\ref{fig.nom.sigma} and~\ref{fig.app.compatible}.

We can consider weakening these axioms, to obtain more general classes of structures.
For instance:
\begin{enumerate}
\item
\rulefont{\sigma\alpha} from Figure~\ref{fig.nom.sigma} expresses that $\sigma$ is a binder in the sense that the $a$ in $x[a\sm u]$ may be $\alpha$-converted.
If we remove this condition then we obtain a notion of `fusion' between names and values, but without then binding the name.
This was touched on in Subsection~\ref{subsect.further.remarks}; a concrete model is obtained by a simplified variant of Definition~\ref{defn.sub.sets} in which the $\new c$ quantification is removed (so that we lose Lemma~\ref{lemm.sigma.alpha}).
\item
The axiom \rulefont{\sigma\app} expresses that substitution commutes with application.

We discussed in detail in Remark~\ref{rmrk.explicit.substitutions} how this might be weakened in one direction to model an explicit substitution.

It might also be removed entirely, to get an environment which would allow functions to detect atoms (that is, variable symbols) in their arguments.

Meta-programming is a large field which has proven resistant so far even to precise categorisation.
Generalisations of $\indiapp$ without \rulefont{\sigma\app} might be one place to start looking for mathematical semantics.
\item
Further examples are easy to generate; it suffices to choose an axiom and modify or delete it.
For instance, relaxing \rulefont{\sigma\#} permits structures that are sensitive to associating a value to a fresh name (like a valuation context, a memory state, or indeed like the $\amgis$-algebras in this paper), and relaxing \rulefont{\sigma\sigma} permits structures that are sensitive to the order of association (like a process listening on a pair of ports), and so forth.
\end{enumerate}

\subsubsection{More structure in the existing axioms}

On a related note, in \cite{gabbay:simcks} we noted that $\lambda$ has a dual construction, $\tall a_1\dots a_n.(t\ppa s)$, of \emph{pattern-matching} (i.e. it applies to points in the `pattern' $t$ and outputs the same points in the pattern $s$).
Operational semantics taking pattern-matching as fundamental include for instance Jay's \emph{pattern calculus} \cite{jay:patc}, and there may be more to discover about what logics and calculi---other than the application to untyped $\lambda$-calculus contained in this paper---the structures in this paper can help develop.

The structures $\indiapp$ and $\points_\Pi$ are not just for the $\lambda$-calculus. 
Far from it: it is clear that they are rich and interesting environments, combining computational and logical structures.
We have used them as a bridge between lattice-theory and $\lambda$-calculus, but
we would go so far as to suggest that for some people this bridge might be just as interesting as the $\lambda$-calculus itself, and could be studied in its own right.

\subsubsection{Duality and cardinality}
\label{subsect.cardinality}

We discussed in Remarks~\ref{rmrk.size.issues} and~\ref{rmrk.size.set.up} that we insist that $\ns D$ be no larger than $\mathbb A$ (see Definitions~\ref{defn.D.impredicative} and~\ref{defn.impredicative.top}). 
This condition is needed just once in this paper, in the proof of Theorem~\ref{thrm.maxfilt.zorn}, which is key to the duality proof.

Traditionally nominal techniques take $\mathbb A$ to be countable, though the theory works perfectly well for larger sets of atoms, as we have seen in this paper.
Provided $\mathbb A$ is infinite, the proofs work, and this paper is parametric in the choice of size made in Definition~\ref{defn.atoms}.

An argument can be made that we intend $\lambda$-calculus models to represent computable objects so they should anyway be countable.
Then we can take $\f{size}(\mathbb A)$ in Definition~\ref{defn.atoms} to be $\aleph_0$ the first infinite cardinal (and the cardinality of $\mathbb N$), and we can take $\alpha$ in the proof of Theorem~\ref{thrm.maxfilt.zorn} to be $\omega_0$ the first infinite ordinal.

We also recall the L\"owenheim-Skolem theorems \cite[Corollaries~3.1.5 and~6.1.4]{hodges:modt} (specifically, that in suitable conditions a model exists if and only if a countable model does).
This suggests intuitively that countable models already have `all interesting structure'.

And indeed: this is reflected in the proofs.  Soundness (Theorem~\ref{thrm.lambda.soundness}) does not care if $\f{size}(\mathbb A)<\f{size}(\ns D)$, and completeness depends only on the syntactic idiom from Example~\ref{xmpl.canonical.idiom} (the canonical model is constructed from $\lambda$-syntax and is countable if its set of atoms is countable).

So for soundness and completeness we can take $\mathbb A$ to countable, even if $\ns D\in\india$ is uncountable. The proofs still work and even the text of the proofs remains unchanged: we just replace `small support' with `finite support' and $\f{size}(\mathbb A)$ with `countable'.

Only if we want a \emph{duality} result do we care that atoms be as large as the model as a whole.
We could simply insist that models be countable, but we prefer to allow a set of atoms of any size; it is the more general choice.

All this suggests a slogan: under duality, \emph{names are the dual to programs}.
This seems intuitively reasonable, and pursuing the maths here further is future work.

The L\"owenheim-Skolem theorems also suggest future work to see if we can transform a model constructed over a set of atoms with one cardinality, into one with more (or even fewer) atoms.
Methods for doing this in nominal algebra (see Appendix~\ref{subsect.alg}) already exist \cite{gabbay:finisn}.

\subsubsection{Size of permutations}
\label{subsect.perm.size}

We noted in Remark~\ref{rmrk.finite.perm} that we take permutations in Definition~\ref{defn.permutation} to be finite, not small.
This may seem like a mismatch between permutations (which are finite) and support (which is small, but may be infinite).
However, what governs the size of the permutations required is how powerful our \emph{binders} are, not how many \emph{free atoms} (i.e. how large a support) we need to allow.
$\lambda$-syntax is finite and only nests finitely many binders, so we only ever need to rename finitely many atoms at a time.
We allow potentially infinite support so we can use atoms to `name' elements, as discussed in Subsection~\ref{subsect.cardinality}.

It may be worth sparing a few words for the generalisation to infinite permutations.\footnote{Thanks to an anonymous referee for noting this subtle issue.}
Suppose the set of atoms $\atoms$ is uncountable, and suppose permutations and support are both taken to be countable.
Consider the set $\mathbb L$ of countable streams of distinct atoms; we can model this as $\mathbb N\hookrightarrow\mathbb A$ the set of injections from natural numbers to atoms.
Consider the relation ${\sim}\subseteq\mathbb L\times\mathbb L$ such that $l\sim l'$ when $\Exists{n{\in}\mathbb N}\Forall{n'{\geq}n}l(n'){=}l'(n')$; by a convenient abuse of notation write this as $\New{n}l(n){=}l'(n)$, and we call $l$ and $l'$ \deffont{asymptotically equal}.

Write $[l]_\sim\in\mathbb L/{\sim}$ for the equivalence class of $l$ under $\sim$.
So $l'\in[l]_\sim$ when $l'$ is asymptotically equal to $l$.
Note that $\mathbb L/{\sim}$ inherits the pointwise permutation action from $\mathbb L$, and $A\subseteq\atoms$ supports $[l]_\sim$ when (again, abusing notation) $\New{n}l(n)\in A$.

Following Definition~\ref{defn.supp} write $\supp([l]_\sim)$ for the intersection of all small $A$ supporting $[l]_\sim$ and write $a\#[l]_\sim$ for $a\not\in\supp([l]_\sim)$.
It is clear that $(a'\ a)\act[l]_\sim=[l]_\sim$ and $\supp([l]_\sim)=\varnothing$ and $a\#[l]_\sim$ for every pair of atoms $a$ and $a'$.

Fix two streams of distinct atoms $l,l'{\in}\mathbb L$.
Write $(l'\ l)$ for the (countable) permutation that swaps $l(n)$ with $l'(n)$ for every $n{\in}\mathbb N$.
It is not hard to check that $(l'\ l)\act l{=}l'$ and $(l'\ l)\act[l]_\sim{=}[l']_\sim\neq[l]_\sim$.
Note that $\nontriv((l'\ l))\cap\supp([l]_\sim)=\varnothing$.

Thus if $\pi$ is an infinite permutation, it is not in general the case that $\Forall{a{\in}\atoms}a\#x$ implies $\pi\act x=x$ \cite[Lemma~21]{gabbay:genmn}.
This implication only holds in general if $\pi$ is \emph{finite}.
This is investigated in \cite[\S6.2]{gabbay:genmn} where it is called \emph{fuzzy support}.
See also the discussion at the end of \cite[Subsection~2.4]{gabbay:nomtnl}.

Infinite binding, and more specifically $\mathbb L$, permutations like $(l'\ l)$, and infinite atoms-abstraction $[l]x$, are interesting: they arise from the study of name-generating processes running in finite but unbounded time; they arise independently in nominal semantics for nominal terms syntax \cite{gabbay:finisn,gabbay:metvil}; and they are interesting in their own right.
So in future work we might want to generalise parts of this paper to admit infinite permutations and infinite atoms-binding.

We believe that much, and perhaps all, of the underlying nominal machinery admits this generalisation.
Making this speculation into a theorem is, thankfully, not required for this paper.
If we choose to do so in future work then we would pay the following price in complexity and convenience: we either surrender the convenience of talking about a \emph{single} unique least supporting set $\supp(x)$; or we insist by definition that $\ns X$ is only a `nominal set' when its elements have a unique least small supporting set of atoms, which excludes examples like $\mathbb L/{\sim}$ but incurs an extra proof-obligation on the $\ns X$ we construct.

%\hyphenation{Mathe-ma-ti-sche}
%\providecommand{\bysame}{\leavevmode\hbox to3em{\hrulefill}\thinspace}
%\providecommand{\MR}{\relax\ifhmode\unskip\space\fi MR }
%% \MRhref is called by the amsart/book/proc definition of \MR.
%\providecommand{\MRhref}[2]{%
%  \href{http://www.ams.org/mathscinet-getitem?mr=#1}{#2}
%}
%\providecommand{\href}[2]{#2}

%\bibliographystyle{amsalpha}
%\bibliography{gabbaybib}

\renewcommand\href[2]{#2}
\hyphenation{Mathe-ma-ti-sche}
\providecommand{\bysame}{\leavevmode\hbox to3em{\hrulefill}\thinspace}
\providecommand{\MR}{\relax\ifhmode\unskip\space\fi MR }
\providecommand{\MRhref}[2]{%
  \href{http://www.ams.org/mathscinet-getitem?mr=#1}{#2}
}
\providecommand{\href}[2]{#2}

\appendix

\renewcommand*{\thesection}{\Alph{section}} %Fix unwanted "Appendix" added before Definition, Lemma and Theorem names. http://tex.stackexchange.com/questions/26333/elsarticle-appendix-and-a-table-of-contents

%%%%%%%%%%%%%%%%%%%%%%%%%%%%%%%%%%%%%%%%%%%%
\section{More on fresh-finite limits}

%%%%%%%%%%%%%%%%%%%%%%%%%%%%%%%%%%%%%%%%%%%%%%%%
\subsection{Nominal algebra axiomatisation of fresh-finite limits}
\label{subsect.alg}

A key definition in this paper has been Definition~\ref{defn.FOLeq}; this is `poset-flavoured', in the sense that $\tand$ and $\tall$ were characterised using fresh-finite limits (Definition~\ref{defn.fresh.finite.limit}).

It is interesting to ask whether fresh-finite limits can be rephrased in the syntax of nominal algebra, using equalities subject to freshness side-conditions.

This has implications, because if this can be done then $\india$ and $\indiapp$ are algebraic varieties, and satisfy the nominal HSP theorems from \cite{gabbay:nomahs,gabbay:nomtnl}.\footnote{The nominal HSPA theorem states that every nominal algebra model is a subobject of a homomorphic image of a cartesian product of atoms-abstractions of free algebras (atoms-abstraction for nominal algebras is defined in \cite{gabbay:nomahs}, as is `free algebra' and so on).  In spirit, this is like the factorisation of natural numbers into primes and similar factorisation theorems.  Such factorisation results are useful because they constrain the structure of models, and this is one of the applications of abstract algebraic techniques.

The result proved in \cite{gabbay:nomahs} considered an untyped syntax, but we expect it to generalise unproblematically to the typed case, if necessary.
}
This gives us off-the-shelf factorisation theorems similar to those considered in \cite{salibra:appual} (see especially Theorem~14), and in general, it is useful to know when a class of structures is an algebraic variety.

\begin{defn}
\label{defn.bounded.lattice}
A \deffont{bounded lattice} in nominal sets is a tuple
$$
\ns L=(|\ns L|,\act,\tand,\tor,\tbot,\ttop)
$$
where:
\begin{itemize*}
\item
$(|\ns L|,\act)$ is a nominal set which we may just write $\ns L$,
\item
$\tbot\in|\ns L|$ and $\ttop\in|\ns L|$ are equivariant \deffont{bottom} and \deffont{top} elements,
\item
$\tand,\tor:(\ns L\times\ns L)\equivarto\ns L$ are equivariant functions, such that $\tand$ and $\ttop$ form an idempotent monoid and $\tor$ and $\tbot$ form an idempotent monoid,
$$
\begin{array}{c@{\qquad}c@{\qquad}c@{\qquad}c}
(x\tand y)\tand z=x\tand(y\tand z)
&
x\tand y=y\tand x
&
x\tand x=x
&
x\tand \ttop=x
\\
(x\tor y)\tor z=x\tor(y\tor z)
&
x\tor y=y\tor x
&
x\tor x=x
&
x\tor \tbot=x
\end{array}
$$
\item
and $\tand$ and $\tor$ satisfy \emph{absorption}
$$
x\tand(x\tor y)=x
\qquad\qquad
x\tor(x\tand y)=x .
$$
\end{itemize*}
Here, $x,y,z$ range over elements of $|\ns L|$.
\end{defn}

A bounded lattice is a poset by taking $x\leq y$ to mean $x\tand y=x$ or $x\tor y=y$ (the two conditions are provably equivalent).
Definition~\ref{defn.bounded.lattice} is the usual definition of a bounded lattice, but over a nominal set; but we have not done anything with it yet.

Definition~\ref{defn.hnu.L} exploits the nominal set structure to algebraise the universal quantifier:
\begin{defn}
\label{defn.hnu.L}
Suppose $\ns L$ is a nominal poset.

A \deffont{(nominal) universal quantifier} $\tall$ on $\ns L$ is an equivariant map
$\tall: (\mathbb A\times\ns L)\equivarto\ns L$
satisfying the equalities \rulefont{\tall\alpha} to \rulefont{\tall{\leq}} in Figure~\ref{fig.genhnu}.
\end{defn}

\begin{figure}[t]
$$
\begin{array}{l@{\qquad}l@{\quad}r@{\ }l}
\rulefont{\tall\alpha}
&
b\#x\limp& \tall b.(b\ a)\act x=&x
\\
\rulefont{\tall{\tand}}
&
& \tall a.(x\tand y)=&(\tall a.x)\tand(\tall a.y)
\\
\rulefont{\tall{\tor}}
&
a\#y\limp& \tall a.(x\tor y)=&(\tall a.x)\tor y
\\
\rulefont{\tall{\leq}}
&
&\tall a.x\leq& x
\end{array}
$$
\caption{Nominal algebra axioms for $\tall$}
\label{fig.genhnu}
\end{figure}

\begin{lemm}
\label{lemm.a.fresh.tall.a}
$a\#\tall a.x$.
\end{lemm}
\begin{proof}
Choose fresh $b$ (so $b\#x$).
By Proposition~\ref{prop.pi.supp} $a\#(b\ a)\act x$.
It follows by Theorem~\ref{thrm.no.increase.of.supp} that $a\#\tall b.(b\ a)\act x$.
By \rulefont{\tall\alpha} $\tall b.(b\ a)\act x=\tall a.x$.
\end{proof}

\begin{prop}
\label{prop.bounded.simple}
Suppose $\ns L$ is a bounded lattice with $\tall$ and $x\in|\ns L|$.
Then $a\#x$ implies $\tall a.x=x$
\end{prop}
\begin{proof}
We reason as follows:
$$
\tall a.x=
\tall a.(x\tor x)
\stackrel{\rulefont{\tall{\tor}}}{=}
(\tall a.x)\tor x
$$
So $x\leq\tall a.x$.
Furthermore by \rulefont{\tall{\leq}} $\tall a.x\leq x$, and we are done.
\end{proof}

\begin{corr}
\label{corr.a.fresh.bounded}
Suppose $\ns L$ is a bounded lattice and suppose $\ns L$ has a nominal universal quantifier $\tall$.
Then $\tall a.x$ is the $a\#$limit for $\{x\}$.
\end{corr}
\begin{proof}
By \rulefont{\tall{\leq}} $\tall a.x\leq x$ and by Lemma~\ref{lemm.a.fresh.tall.a} $a\#\tall a.x$.

Suppose $z\in|\ns L|$ and $a\#z$ and $z\leq x$.
It follows using \rulefont{\tall{\tand}} that $\tall a.z\leq \tall a.x$ and by Proposition~\ref{prop.bounded.simple} $\tall a.z=z$.
\end{proof}

\begin{prop}
The notion of a nominal distributive lattice with $\tall$ from Definition~\ref{defn.FOLeq} is characterised in nominal algebra as:
\begin{quote}
a bounded lattice (Definition~\ref{defn.bounded.lattice}) with a nominal universal quantifier (Definition~\ref{defn.hnu.L}) that is distributive (Definition~\ref{defn.distrib}) and has a compatible $\sigma$-action (Definition~\ref{defn.fresh.continuous}).
\end{quote}
\end{prop}
\begin{proof}
By routine calculations.
The interesting part is to check that the characterisation of $\tall a$ from Definition~\ref{defn.fresh.finite.limit} as a fresh-finite limit coincides with the algebraic characterisation of $\tall a$ in Figure~\ref{fig.genhnu}.
The meat of this is handled by Proposition~\ref{prop.bounded.simple}.
\end{proof}

%%%%%%%%%%%%%%%%%%%%%%%%%%%%%%%%%%%%%%%
\subsection{More on fresh-finite limits}
\label{subsect.ffl}

The quantifier $\freshwedge{a}x$ from Definition~\ref{defn.fresh.finite.limit} is greatest in the set $\{x'\mid x'{\leq} x\land a\#x'\}$, but we have seen other ways to characterise quantification too: Definition~\ref{defn.nu.U} and Proposition~\ref{prop.these.are.equivalent} characterise the quantifier in $\sigma$-powersets (of filters); Propositions~\ref{prop.char.freshwedge} and~\ref{prop.char.freshwedge.names} do something similar in the abstract; Lemma~\ref{lemm.tall.include.lam.point} does it again, for points.

A higher-level view is possible of some general principles behind these results.
Propositions~\ref{prop.char.freshwedge} and~\ref{prop.char.freshwedge.names} are part of this higher-level picture, and we can also observe:
\begin{itemize}
\item
Definition~\ref{defn.freshwedges} and Proposition~\ref{prop.freshwedge.strict} characterise $\freshwedge{a}x$ as a limit of a strictly small-supported set ($\{x'\mid x'\leq x\land a\#x'\}$ is small-supported but not \emph{strictly} small-supported).
\item
Definition~\ref{defn.fresh.orbit} and Proposition~\ref{prop.freshwedgeo} characterise $\freshwedge{a}x$ as a limit of a permutation orbit and using the $\new$-quantifier.
This ties $\freshwedge{a}x$ to Proposition~\ref{prop.these.are.equivalent} and Remark~\ref{rmrk.prop.ii.remarkable} and to the opening discussion of Subsection~\ref{subsect.filters}.
\end{itemize}
See also \cite{gabbay:stusun,gabbay:frenrs} and computational studies such as \cite{klin:townc}, with their emphasis on studying nominal sets in terms of their permutation orbits.

\begin{defn}
\label{defn.freshwedges}
Suppose $\mathcal L$ is a nominal poset and $x\in|\mathcal L|$.
Consider the set
$$
\begin{array}{r@{\ }l}
B=&\{x'{\in}|\mathcal L| \mid \supp(x'){\subseteq}S\ \wedge\ x'{\leq} x\}  .
\end{array}
$$
Then write $\freshwedges{S}x$ for the $\leq$-greatest element of $B$, if this exists.
We call this the \deffont{$S$-strict limit} of $x$.
\end{defn}

\begin{rmrk}
So:
\begin{itemize*}
\item
$\freshwedge{a}x$ is the greatest $x'$ beneath $x$ such that $a\not\in\supp(x')$.
\item
$\freshwedges{\supp(x){\setminus}\{a\}}x$ is greatest $x'$ beneath $x$ such that $\supp(x'){\subseteq}\supp(x){\setminus}\{a\}$.
\end{itemize*}
It is not \emph{a priori} evident that these two notions must coincide.
However, they often do, as we will now show.
\end{rmrk}

\begin{prop}
\label{prop.freshwedge.strict}
If $\freshwedge{a}x$ exists then so does $\freshwedges{\supp(x){\setminus}\{a\}}x$ and they are equal.
\end{prop}
\begin{proof}
Suppose $\freshwedge{a}x$ exists.
By construction $\supp(\freshwedge{a}x){\subseteq}\supp(x){\setminus}\{a\}$ and $\freshwedge{a}x\leq x$.
Therefore $\freshwedge{a}x\in B$ (notation from Definition~\ref{defn.freshwedges}).
Also by construction $x'\leq\freshwedge{a}x$ for every $x'\in B$, since if $\supp(x')\subseteq\supp(x){\setminus}\{a\}$ then certainly $a\#x'$.
It follows that $\freshwedge{a}x$ is greatest in $B$.
\end{proof}

\begin{rmrk}
It is easy to prove that if $x\leq y$ then $\freshwedges{S}x\leq \freshwedges{S}y$ (so $\freshwedges{S}$ is monotone or \emph{functorial}).
Another plausible definition for $\freshwedge{a}x$ is to set
\begin{equation}
\label{eq.suggested}
\freshwedge{a} x
\ =\ \text{the $\leq$-greatest element of}\
\{x' \mid \supp(x')\subseteq\supp(x){\setminus}\{a\} \land x'{\leq} x\}.
\end{equation}
However, monotonicity is then not so obvious, though in view of Proposition~\ref{prop.freshwedge.strict} monotonicity does hold of \eqref{eq.suggested}, for the cases we care about.
\end{rmrk}

Recall from Notation~\ref{nttn.fix} the definition of $\fix$.
\begin{defn}
\label{defn.fresh.orbit}Following \cite{gabbay:stusun,gabbay:frenrs} define $x\ii{a}$ by
$$
\begin{array}{r@{\ }l}
x\ii{a}=&\{\pi\act x \mid \pi\in\fix(\supp(x){\setminus}\{a\})\}  .
\\
=&\{x\}\cup\{(b\ a)\act x\mid b\#x\}
\end{array}
$$
Write $\bigwedge x\ii{a}$ for the $\leq$-greatest lower bound of $x\ii{a}$, if this exists.
\end{defn}

\begin{rmrk}
So we can rewrite this as follows:
\begin{itemize*}
\item
$\freshwedge{a}x$ is the greatest $x'$ beneath $x$ such that $a\not\in\supp(x')$.
\item
Definition~\ref{defn.fresh.orbit} specifies the greatest $x'$ beneath $x$ and beneath $(b\ a)\act x$ for every $b\#x$.
\end{itemize*}
\end{rmrk}

\begin{lemm}
\label{lemm.fresh.ii}
$a\#x\ii{a}$.
\end{lemm}
\begin{proof}
By the pointwise action $\pi'\act x\ii{a}=\{(\pi'\circ\pi)\act x\mid \pi\in\fix(\supp(x){\setminus}\{a\})\}$.
We choose a fresh $b$ (so $b\#x$, that is, $b{\not\in}\supp(x)$) and use Corollary~\ref{corr.stuff}(\ref{stuff.freshness.criterion}) and routine calculations.
\end{proof}

\begin{prop}
\label{prop.freshwedgeo}
Suppose $a{\in}\mathbb A$ and $x{\in}|\mathcal L|$.
Then:
\begin{enumerate}
\item
If $\freshwedge{a}x$ exists then so does $\bigwedge x\ii{a}$, and they are equal.
\item
Suppose $\mathcal L$ has a monotone $\sigma$-action (Definition~\ref{defn.fresh.continuous}).
Then if $\bigwedge x\ii{a}$ exists, then so does $\freshwedge{a}x$, and they are equal.
\end{enumerate}
\end{prop}
\begin{proof}
\begin{enumerate}
\item
By Definition~\ref{defn.nom.poset} $\freshwedge{a}x\leq x$ and $a\#\freshwedge{a}x$.
It follows by equivariance of $\leq$ and Corollary~\ref{corr.stuff} that $\freshwedge{a}x\leq \pi\act x$ for every $\pi\in\fix(\supp(x){\setminus}\{a\})$.
Therefore $\freshwedge{a}x$ is a lower bound for $x\ii{a}$.

Now consider $z$ some other lower bound for $x\ii{a}$, so that $\Forall{\pi{\in}\fix(\supp(x){\setminus}\{a\})}z\leq\pi\act x$.
Choose fresh $b$ (so $b\#x,z$); by Theorem~\ref{thrm.equivar} $\freshwedge{b}(b\ a)\act x$ exists (and by Lemma~\ref{lemm.freshwedge.alpha} it is equal to $\freshwedge{a}x$).
It follows by equivariance of $\leq$ and Corollary~\ref{corr.stuff} (since $a,b\#z$) that $z\leq (b\ a)\act x$ so $z\leq\freshwedge{b}(b\ a)\act x\stackrel{\text{L\ref{lemm.freshwedge.alpha}}}=\freshwedge{a}x$.
\item
By Lemma~\ref{lemm.fresh.ii} and Theorem~\ref{thrm.no.increase.of.supp} $a\#\bigwedge x\ii{a}$ so $\bigwedge x\ii{a}$ is an $a\#$lower bound for $x$.

Consider any other $z$ such that $z\leq x$ and $a\#z$.
By \rulefont{\sigma\#} $z[a\sm n]=z$ for every $n{\in}\mathbb A$.
By monotonicity $z=z[a\sm n]\leq x[a\sm n]$ for every $n{\in}\mathbb A$.
It follows using Lemma~\ref{lemm.sub.alpha} that $z\leq \bigwedge x\ii{a}$.
\qedhere\end{enumerate}
\end{proof}

\begin{rmrk}
Nominal posets offer a richness of what we can collectively term `fresh limits':
\begin{itemize*}
\item
Fresh-finite limits from Definition~\ref{defn.fresh.finite.limit}.
\item
\emph{Strict} fresh-finite limits from Definition~\ref{defn.freshwedges}.
\item
Limits of permutation orbits.
\item
If a $\sigma$-algebra structure is present, then we have even more: limits of all substitution instances (Proposition~\ref{prop.char.freshwedge}; see also Definition~\ref{defn.nu.U}) and substitution for all names or all fresh names (see Propositions~\ref{prop.char.freshwedge.names} and~\ref{prop.these.are.equivalent}).
\end{itemize*}
For a nominal poset with a monotone $\sigma$-action, these are all equal where they exist.
\end{rmrk}

We return briefly to Proposition~\ref{prop.freshwedge.strict}, which stated that in a nominal poset $\mathcal L$ if $\freshwedge{a}x$ exists then so does $\freshwedges{\supp(x){\setminus}\{a\}}x$ and they are equal.
What can we say about the other way around?
When does the existence of $\freshwedges{S}x$ imply the existence of $\freshwedge{a}x$?

We can note Proposition~\ref{prop.support.interpolation} and Remark~\ref{rmrk.not}:
\begin{lemm}
\label{lemm.supp.freshwedges}
$\supp(\freshwedges{S}x)\subseteq S$.
\end{lemm}
\begin{proof}
By Lemma~\ref{lemm.strict.support} $\supp(B)\subseteq S$ ($B$ is from Definition~\ref{defn.freshwedges}).
We use Theorem~\ref{thrm.no.increase.of.supp}.
\end{proof}

\begin{defn}
Say $\mathcal L$ has \deffont{support interpolation} when for every $x$ and $y$,\ if $x\leq y$ then there exists a $z$ such that
$$
x\leq z
\quad
\text{and}
\quad
z\leq y
\qquad\text{and}\qquad
\supp(z)\subseteq\supp(x){\cap}\supp(y).
$$
(\emph{Interpolation} is used here by analogy with the concept in logic \cite{gabbay:intdmi}; here we interpolate on the set of names.)
\end{defn}

\begin{prop}
\label{prop.support.interpolation}
Suppose $\mathcal L$ has support interpolation and suppose $\freshwedges{S}x$ exists for all $x$ and $S{\subseteq}\mathbb A$.

Then $\freshwedge{a}x$ exists for all $x$ and $a$ and is equal to $\freshwedges{\supp(x){\setminus}\{a\}}x$.
\end{prop}
\begin{proof}
Suppose all strict fresh-finite limits $\freshwedges{S}x$ exist.
Consider some $a$ and $x$.
\begin{itemize*}
\item
\emph{We show that if $x'\leq x$ and $a\#x'$ then $x'\leq\freshwedges{\supp(x){\setminus}\{a\}}x$.}

Suppose $x'\leq x$ and $a\#x'$.
By support interpolation there exists a $z$ with $x'\leq z\leq x$ and $\supp(z)\subseteq\supp(x)\cap\supp(x')$.
By assumption $z\leq\freshwedges{\supp(x){\setminus}\{a\}}x$ and so
$
x'\leq \freshwedges{\supp(x){\setminus}\{a\}}x.
$
\item
\emph{We show that $\freshwedges{\supp(x){\setminus}\{a\}}x$ is least with this property.}

By Lemma~\ref{lemm.supp.freshwedges} $\supp(\freshwedges{\supp(x){\setminus}\{a\}}x)\subseteq\supp(x){\setminus}\{a\}$.
Thus in particular $a\#\freshwedges{\supp(x){\setminus}\{a\}}x$.
Also by construction $\freshwedges{\supp(x){\setminus}\{a\}}x\leq x$.

Thus for any other greatest element $y$ in $\{x'\mid x'{\leq}x\land a\#x'\}$,\ $\freshwedges{\supp(x){\setminus}\{a\}}x\leq y$.
\end{itemize*}
Thus, $\freshwedge{a}x$ exists and is equal to $\freshwedges{\supp(x){\setminus}\{a\}}x$.
\end{proof}

\begin{rmrk}
\label{rmrk.not}
In the absence of support interpolation Proposition~\ref{prop.support.interpolation} may fail.
For instance, consider a nominal poset consisting of singleton atoms $\{a\}$ and unordered pairs of atoms $\{a,b\}$ and an element $\ast$, such that:
\begin{itemize*}
\item
$\{a,b\}\leq \{a\}$ for all (distinct) $a$ and $b$.
\item
$\ast\leq \{a\}$ for all $a$.
\end{itemize*}
We assume the natural pointwise permutation actions and $\pi\act\ast=\ast$, so that $\supp(\{a,b\})=\{a,b\}$ and $\supp(\{a\})=\{a\}$ and $\supp(\ast)=\varnothing$.

Then $\ast=\freshwedges{\varnothing}\{a\}$ but $\freshwedge{a}\{a\}$ does not exist since $\ast\not\leq\{a,b\}$.
\end{rmrk}

%%%%%%%%%%%%%%%%%%%%%%%%%%%%%%%%%%%%%%
\section{Additional properties of the canonical model}

We noted in Subsection~\ref{subsect.commutes} that $\points_\Pi$ has structure above and beyond being in $\indiapp$.
We conclude with some further reflection on this.

These properties were not needed for our main results, but they seem striking enough to merit a note in an Appendix.

\subsection{$\tlam$ commutes with unions}

We consider $\tlam$ from Notation~\ref{nttn.lambda} in $\points_\Pi$ and prove Corollary~\ref{corr.unexpected}: $\tlam$ commutes with sets union (cf. a similar property for $\tall$ proved in Lemma~\ref{lemm.tall.unions}).

This property is not valid in general in $\indiapp$, because $\ppa$ and $\tall$ do not commute with $\tor$ in general,\footnote{The closest we get to this in the general case is \rulefont{\ppa{\tor}} from Figure~\ref{fig.app.compatible} and \rulefont{\tall\tor} from Figure~\ref{fig.genhnu}} but it does hold in the canonical model $\points_\Pi$.

\begin{lemm}
\label{lemm.bpp.pp.points}
$q\bpp \pp p = \pp{(q\app p)}$.
\end{lemm}
\begin{proof}
We reason as follows:
\begin{tab2b}
q\bpp \pp p =&
\bigcup\{q\bpp p'\mid p\subseteq p'\}
&\text{Defs~\ref{defn.some.more.ops} \&~\ref{defn.YppaX}}
\\
=&
\bigcup\{\pp{(q\app p')}\mid p\subseteq p'\}
&
\text{Definition~\ref{defn.some.more.ops}}
\\
=&
\pp{(q\app p)}
&
\text{Fact of Def~\ref{defn.some.more.ops}}
\qedhere\end{tab2b}
\end{proof}

\begin{lemm}
\label{lemm.ppa.unions}
$\pp p\ppa (\bigcup_i\pp p_i) = \bigcup_i(\pp p\ppa\pp p_i)$.
\end{lemm}
\begin{proof}
We reason as follows:
\begin{tab2b}
q\in \pp p\ppa(\bigcup_i\pp p_i)
\liff&
q\bpp \pp p\subseteq \bigcup_i\pp p_i
&\text{Proposition~\ref{prop.amgis.iff.pp}}
\\
\liff&
\pp{(q\app p)}\subseteq\bigcup_i\pp p_i
&\text{Lemma~\ref{lemm.bpp.pp.points}}
\\
\liff&
\Forall{r}(q\app p{\subseteq} r\limp\Exists{i}p_i{\subseteq} r)
&
\text{Definition~\ref{defn.some.more.ops}}
\\
\liff&
\Exists{i}p_i\subseteq q\app p
&\text{Suffices to take }r=q\app p
\\
\liff&
\Exists{i}\pp{(q\app p)}\subseteq \pp p_i
&\text{Definition~\ref{defn.some.more.ops}, fact}
\\
\liff&
\Exists{i}q\bpp p\subseteq \pp p_i
&\text{Definition~\ref{defn.some.more.ops}}
\\
\liff&
\Exists{i}q\in \pp p\ppa \pp p_i
&\text{Proposition~\ref{prop.amgis.iff.pp}}
\\
\liff&
p\in \bigcup_i(\pp p\ppa \pp p_i)
&\text{Fact}
\qedhere\end{tab2b}
\end{proof}

\begin{corr}
\label{corr.unexpected}
$\tlam a.\bigcup_i \pp p_i=\bigcup_i\tlam a.\pp p_i$.
\end{corr}
\begin{proof}
We recall $\tlam$ from Notation~\ref{nttn.lambda} and using Lemmas~\ref{lemm.ppa.unions} and~\ref{lemm.tall.unions} deduce that $\tlam a.\bigcup_i \pp p_i=\bigcup_i\tlam a.\pp p_i$.
\end{proof}

%%%%%%%%%%%%%%%%%%%%%%%%%%%%%%%%%%%%
\subsection{An existential quantifier}
\label{subsect.existential}

$\points_\Pi$ from Definition~\ref{defn.pi.point} has a universal quantifier, defined in Definition~\ref{defn.some.ops}.
Proposition~\ref{prop.points.bigcap} hints that an existential quantifier might exist.\footnote{Why?  Because the proof of Proposition~\ref{prop.points.bigcap} would be simple and direct if we assume that $\mathcal P$ is strictly small-supported.  The fact that it works for \emph{all} small-supported $\mathcal P$ suggests there might be some way of removing atoms from the support of $p\in\mathcal P$ while also making $p$ smaller as a set.  This is exactly what an existential quantifier on points would do (recall that everything is inverted/dual in $\points_\Pi$).}
And indeed, it is not hard to construct as the natural dual %We note that $\texi a.p$ is dual
to $\tall a.p$ from Definition~\ref{defn.some.ops}:
\begin{defn}
\label{defn.jj}
Suppose $p\in|\points_\Pi|$ and $A\subseteq\mathbb A$.
Define $\texi a.p$ by:
$$
\begin{array}{r@{\ }l}
\texi a.p=&\bigcap_{u{\in}|\idiomprg|} p[a\lsm u]
\end{array}
$$
\end{defn}

\begin{lemm}
\label{lemm.jj.point}
If $p\in|\points_\Pi|$ then $\texi a.p\in|\points_\Pi|$.
\end{lemm}
\begin{proof}
From Proposition~\ref{prop.points.bigcap}.
\end{proof}

A dual version of Lemma~\ref{lemm.tall.include.lam.point} exists for $\texi$ instead of $\tall$, and is easy to prove.
We omit details.

The universal and existential quantifiers $\tall$ and $\texi$ interact with $\app$ similarly to how we saw $\tall$ interact with $\tor$ in \rulefont{\tall\tor} of Figure~\ref{fig.genhnu} or \rulefont{distrib\tall} of Definition~\ref{defn.distrib}; this is Lemmas~\ref{lemm.tall.app.lam} and~\ref{lemm.jj.app.lam}.

We need Lemma~\ref{lemm.app.continuous.lam} as a simple technical lemma:
\begin{lemm}
\label{lemm.app.continuous.lam}
If $p\subseteq p'$ then $p\app q\subseteq p'\app q$.
\end{lemm}
\begin{proof}
Routine from Definition~\ref{defn.some.ops}.
\end{proof}

\begin{lemm}
\label{lemm.tall.app.lam}
If $a\#q$ then $\tall a.(p\app q)\subseteq (\tall a.p)\app q$.
\end{lemm}
\begin{proof}
By Theorem~\ref{thrm.no.increase.of.supp} and Lemma~\ref{lemm.tall.p.alpha}, $a\#(\tall a.p)\app q$.
Thus by part~\ref{tall.include.fresh} of Lemma~\ref{lemm.tall.include.lam.point}, $\tall a.(p\app q)\subseteq (\tall a.p)\app q$ if and only if
$p\app q \subseteq (\tall a.p)\app q$.
By part~\ref{tall.include.1} of Lemma~\ref{lemm.tall.include.lam.point} and Lemma~\ref{lemm.app.continuous.lam}, $p\app q \subseteq (\tall a.p)\app q$ is true.
\end{proof}

\begin{lemm}
\label{lemm.jj.app.lam}
If $a\#q$ then $(\texi a.p)\app q\subseteq \texi a.(p\app q)$.
\end{lemm}
\begin{proof}
Suppose $v\in (\texi a.p)\app q$.
So there exist $s'$ and $t$ such that $\Forall{u}s'[a\ssm u]\in p\land t\in q\land s't\arrowp{\Pi}v$.
In particular, $s't\arrowp{\Pi}v$.
It follows by condition~\ref{item.somerel.ssm} of Definition~\ref{defn.lambda.reduction.theory} that $s'[a\ssm u](t[a\ssm u])\arrowp{\Pi}v[a\ssm u]$ for every $u$.
But $a\#q$ so by Proposition~\ref{prop.fresh.point} $a\fresh q$ so $t[a\ssm u]\in q$ for every $u$.
Thus, $v\in \texi a.(p\app q)$.
\end{proof}

\end{document}